\documentclass[english,11pt]{article}
\usepackage[utf8]{inputenc}
\usepackage[T1]{fontenc}
\usepackage[margin=1in]{geometry}
\usepackage{tikz}
\usepackage{pgfplots}

\usepackage{amsmath}
\usepackage{amsfonts}
\usepackage{amssymb}
\usepackage{amsthm}
\usepackage{comment}
\usepackage{algorithm}
\usepackage{algorithmicx}
\usepackage{algpseudocode}
\usepackage{enumerate}
\usepackage{mathtools}
\usepackage{mathrsfs}
\usepackage{csquotes}
\usepackage{apptools}
\usepackage{graphicx}
\usepackage{verbatim}
\usepackage{float}
\usepackage{chngcntr}
\usepackage{xr}
\usepackage{bm}
\usepackage{todonotes}
\usepackage[parfill]{parskip}
\usepackage[shortlabels]{enumitem}
\usepackage{tikz}
\DeclareUnicodeCharacter{1EF3}{\'y}
\usepackage[numbers]{natbib}
\usepackage{thm-restate}

\usepackage[hypertexnames=false]{hyperref}
\usepackage[capitalize, nameinlink]{cleveref}
\usepackage{blkarray}

\usepackage{arydshln}
\makeatletter
\renewcommand*\env@matrix[1][*\c@MaxMatrixCols c]{%
	\hskip -\arraycolsep
	\let\@ifnextchar\new@ifnextchar
	\array{#1}}
\makeatother

\numberwithin{equation}{section}

\crefname{prob}{Problem}{Problems}
\crefname{thm}{Theorem}{Theorems}

\crefname{eq}{Equation}{Equations}
\crefname{ineq}{Inequality}{Inequalities}
\creflabelformat{ineq}{#2{\upshape(#1)}#3} 

\newcommand{\declarecolor}[2]{\definecolor{#1}{RGB}{#2}\expandafter\newcommand\csname #1\endcsname[1]{\textcolor{#1}{##1}}}
\declarecolor{White}{255, 255, 255}
\declarecolor{Black}{0, 0, 0}
\declarecolor{LightGray}{216, 216, 216}
\declarecolor{Gray}{127, 127, 127}
\declarecolor{Orange}{237, 125, 49}
\declarecolor{LightOrange}{251,229, 214}
\declarecolor{Yellow}{255, 192, 0}
\declarecolor{LightYellow}{255, 242, 200}
\declarecolor{Blue}{91, 155, 213}
\declarecolor{LightBlue}{222, 235, 247}
\declarecolor{Green}{112, 173, 71}
\declarecolor{LightGreen}{226, 240, 217}
\declarecolor{Navy}{68, 114, 196}
\declarecolor{LightNavy}{218, 227, 243}

\hypersetup{
	colorlinks=true,
	pdfpagemode=UseNone,
	citecolor=Green,
	linkcolor=Navy,
	urlcolor=Black
}

\algnewcommand\algorithmicforeach{\textbf{for each}}
\algdef{S}[FOR]{ForEach}[1]{\algorithmicforeach\ #1\ \algorithmicdo}

\def\R{{\mathbb{R}}}

\def\N{{\mathbb{N}}}

\theoremstyle{plain}
\newtheorem{theorem}{Theorem}[section]

\newtheorem{corollary}[theorem]{Corollary}
\newtheorem{lemma}[theorem]{Lemma}

\newtheorem{invariant}[theorem]{Invariant}
\theoremstyle{remark}
\newtheorem{rem}[theorem]{Remark}

\theoremstyle{definition}
\newtheorem{definition}[theorem]{Definition}
\newtheorem{example}[theorem]{Example}

\AtAppendix{\counterwithin{theorem}{section}}

\usepackage{babel}

\counterwithin{figure}{section}

\DeclareMathOperator{\new}{new}

\DeclareMathOperator{\nnz}{nnz} % number of nonzero entries
\DeclareMathOperator*{\argmin}{arg\,min}

\DeclareMathOperator{\poly}{poly} % polynomial in [argument]
\DeclareMathOperator{\polylog}{polylog} % poly-logarithmic in [argument]
\DeclareMathOperator{\layer}{\mathsf{Layer}}

\newcommand\norm[1]{\left\lVert#1\right\rVert}

\newcommand\mc[1]{\mathcal{#1}}

\newcommand{\T}{\top}
\newcommand{\eps}{\varepsilon}
\newcommand{\grad}{\nabla}

\DeclareMathOperator{\tw}{tw}

\newcommand{\wt}{\widetilde}

\global\long\def\defeq{\stackrel{\mathrm{{\scriptscriptstyle def}}}{=}}

\def\=#1{\bm{#1}}
\def\+#1{\mathcal{#1}}
\def\-#1{\mathbb{#1}}

\newcommand{\block}[1]{[#1]}
\global\long\def\leq{\leqslant}%
\global\long\def\hes#1#2{\nabla^{2}#1(#2)}%
 
\global\long\def\wt#1{\widetilde{#1}}%
\global\long\def\new#1{#1^{\mathrm{new}}}%
\global\long\def\norm#1{\|#1\|}%
\global\long\def\d{\mathrm{d}}%
\global\long\def\R{\mathbb{R}}%
\global\long\def\Rn{\mathbb{R}^{n}}%
\global\long\def\defeq{\overset{\mathrm{def}}{=}}%

\global\long\def\oxs{(\overline{x},\overline{s})}%
 
\global\long\def\ox{\overline{x}}%
\global\long\def\os{\overline{s}}%
\global\long\def\ot{\overline{t}}%
\global\long\def\nnz{\operatorname{nnz}}%
\global\long\def\dom{\operatorname{dom}}%
\global\long\def\tw{\mathrm{tw}}%

\global\long\def\eps{\varepsilon}%
\global\long\def\poly{\operatorname{poly}}%
\global\long\def\tw{\operatorname{tw}}%
\global\long\def\dom{\operatorname{dom}}%

\global\long\def\ts{\mathcal{S}}%
\global\long\def\tt{\mathcal{T}}%
\global\long\def\depth{\mathsf{depth}}%
\global\long\def\lca{\mathsf{LCA}}%
\global\long\def\low{\mathsf{low}}%
\global\long\def\block_low{\overline{\mathsf{low}}}%

\graphicspath{./}
 
\begin{document}
	
\title{A Nearly-Linear Time Algorithm for Linear Programs with Small Treewidth: A Multiscale Representation of Robust Central Path}
\author{
	Sally Dong\\
	University of Washington
	\and
	Yin Tat Lee\\
	University of Washington \\ \& Microsoft Research
	\and
	Guanghao Ye\\ 
	University of Washington
}
\maketitle
\begin{abstract}
Arising from structural graph theory, treewidth has become a focus of study in fixed-parameter tractable algorithms in various communities including combinatorics, integer-linear programming, and numerical analysis. 
Many NP-hard problems are known to be solvable in  $\widetilde{O}(n \cdot 2^{O(\mathrm{tw})})$ time, where $\mathrm{tw}$ is the treewidth of the input graph. Analogously, many problems in P should be solvable in $\widetilde{O}(n \cdot \mathrm{tw}^{O(1)})$ time; however, due to the lack of appropriate tools, only a few such results are currently known. \cite{fomin2018fully} conjectured this to hold for as broadly as all linear programs; in our paper, we show this is true:

   Given a linear program of the form $\min_{Ax=b,\ell \leq x\leq u} c^{\top} x$, and a width-$\tau$ tree decomposition of a graph $G_A$ related to $A$, we show how to solve it in time
   $$\widetilde{O}(n \cdot \tau^2 \log (1/\varepsilon)),$$
where $n$ is the number of variables and $\varepsilon$ is the relative accuracy. 
Combined with recent techniques in vertex-capacitated flow \cite{bernstein2021deterministic}, this leads to an algorithm with $\widetilde{O}(n^{1+o(1)} \cdot \mathrm{tw}^2 \log (1/\varepsilon))$ runtime.
Besides being the first of its kind, our algorithm has runtime nearly matching the fastest runtime for solving the sub-problem $Ax=b$, under the assumption that no fast matrix multiplication is used.

We obtain these results by combining recent techniques in interior-point methods (IPMs), sketching, and a novel representation of the solution under a multiscale basis similar to the wavelet basis. 
%This representation further yields the first IPM with $o(\rank(A))$ time per iteration when the treewidth is small. 

% In our framework, treewidth is used as a sparsity measure connected to Cholesky factorizations; we further design central path maintenance techniques to leverage the associated sparse structures. 
%	Our result extends to a class of convex programs that captures significant theoretical and real-world applications, such as signal and image processing.
\end{abstract}

\newpage
\tableofcontents
\newpage

\section{Introduction}

Linear programming is one of the most fundamental problems in computer
science and optimization. General techniques for solving linear programs,
such as simplex methods~\cite{dantzig1951maximization}, ellipsoid
methods~\cite{khachiyan1980polynomial} and interior point methods~\cite{karmarkar1984new},
have been developed and continuously refined since the 1940s, and
have later been found to be useful in a wide range of problems spanning
optimization, combinatorics, and machine learning.

For an arbitrary linear program $\min_{Ax=b,\ell \leq x\leq u}c^{\top}x$ with
$n$ variables and $d$ constraints, the current fastest algorithms
take either $\widetilde{O}(n^{2.373}\log(1/\eps))$
time~\cite{DBLP:conf/stoc/CohenLS19, jiang2020faster}
or $\widetilde{O}((\sqrt{nd \cdot \nnz(A)}+d^{2.5})\log(1/\eps))$ time~\cite{van2020solving, brand2020maxflow},
where $\eps$ is the accuracy parameter\footnote{The current fastest exact algorithms for linear program take either
$2^{O(\sqrt{d\log\frac{n-d}{d}})}$ time~\cite{hansen2015improved}, or the runtime depends on the magnitude of entries of $A$.}. 
When $A$ is a dense matrix, these runtimes are close to optimal,
as they nearly match the runtime $\widetilde{O}((\text{nnz}(A)+d^{\omega})\log(1/\eps))$
to solve the subproblem $Ax=b$, where $\omega \approx 2.373$ is the matrix multiplication exponent. When $A$ is sparse, as is the 
case in many problems arising from both theory and applications, we ask 
if much faster runtimes are possible. 

When $n$ and $d$ are the same order, this problem is highly non-trivial,
even for linear systems. It is only recently known how to solve a
sparse linear system in slightly faster than $d^{\omega}$ time~\cite{peng2020solving},
and sub-quadratic time is insurmountable under the current techniques.
It turns out in practice however, sparse linear systems often have
low \emph{treewidth}, a condition much stronger than mere sparsity;
for example, many of the linear programs
in the Netlib repository have sublinear treewidth
(\cref{fig:practice}). For low treewidth linear systems, a small
polynomial dependence on treewidth still implies a much faster than
quadratic runtime, hence making them a particularly suitable target
of study.

Beyond the practical consideration, whether there is a $\widetilde{O}(n \cdot \tw^{O(1)})$ LP algorithm is important in parameterized complexity. 
Most algorithms designed for low treewidth graphs rely on dynamic programming,
which naturally give algorithms with runtime exponential in treewidth even for problems in P, 
such as reachability and shortest paths~\cite{akiba2012shortest,chatterjee2013faster,chaudhuri2000shortest,planken2012computing}.
There are only a few problems in P that we know how to solve in $\widetilde{O}(n \cdot \tw^{O(1)})$ time \cite{fomin2018fully}. 
We refer to \cref{sec:exp-treewidth-algs,sec:poly-treewidth-algs} for a discussion of these problems.

Recently,~\cite{fomin2018fully} posed exactly this question\footnote{We add the $\log(1/\eps)$ term to their original conjecture. Without
this term, the conjecture will imply the existence of strongly polynomial
time algorithms for linear programs, one of Smale's 18 unsolved problems
in mathematics.}:
\begin{center}
Can linear programs be solved in $\widetilde{O}(n\cdot\tw^{O(1)}\log(1/\eps))$
time?
\par\end{center}

We answer the question affirmatively in this paper:
\begin{theorem}\label{thm:main}
Given a linear program $\min_{Ax=b,\ell \leq x\leq u}c^{\top}x$, where
$A\in\R^{d\times n}$ is a full rank matrix with $d\leq n$, define the \emph{dual graph} $G_{A}$\footnote{There are different ways of associating a graph with the matrix $A$ (see~\cite{jansen2015structural,fomin2018fully}). We adopt the one used in the ILP community~\cite{jansen2015structural, eisenbrand2019algorithmic}. 
%The resulting treewidth is in fact only off by a factor of 2 from the one used in the conjecture posed by~\cite{fomin2018fully}. 
We choose this definition so that when applied to linear programming formulations of flow problems, in which the constraint matrix $A$ is the incidence
matrix of the input graph $G$, we have $G_{A} = G$, and hence the treewidth of the LP is most meaningfully related to the flow problem.}
to be the graph with vertex
set $\{1,\dots,d\}$, such that $ij\in E(G_{A})$ if there is a column
$r$ such that $A_{i,r}\neq0$ and $A_{j,r}\neq0$.
Suppose that
\begin{itemize}
\item a tree decomposition of $G_{A}$ with width $\tau$ is given, and 
%TODO: Haotian wants a definition of "width" here.
\item $r$ is the inner radius of the polytope, namely, there is $x$ such that $Ax=b$ and $\ell + r \leq x \leq u - r$.
\end{itemize}
Let $L = \|c\|_2$ and $R = \|u-\ell\|_2$. For any $0<\eps\leq1/2$, we can find $x$ such that $Ax=b$ and $\ell \leq x\leq u$ such
that
\[
c^{\top}x\leq\min_{Ax=b,\, \ell \leq x \leq u}c^{\top}x+\eps\cdot L R
\]
in expected time
\[
\widetilde{O}(n\cdot \tau^{2}\log(R/(r \eps))).
\]
\end{theorem}

To keep this paper simple, we refrain from using fast matrix multiplication. Under this restriction, we note that our runtime is tight, 
since it nearly matches the fastest runtime for solving the subproblem $Ax=b$ (\cref{cor:chol-time}).

Our algorithm involves a pre-processing component: We need to find
some suitable reordering of the rows of $A$, known as an \emph{elimination
order}, so that matrices in later computations will have certain desired
sparsity patterns. In practice, there are various efficient algorithms for finding a good reordering, 
such as minimum degree orderings~\cite{george1989evolution,amestoy1996approximate,fahrbach2018graph} 
and nested dissection algorithms~\cite{george1973nested,lipton1979generalized,karypis1998fast}.
In theory, there are also different ways to compute the reordering. In the previous version of this paper, we applied techniques in \cite{arora2016combinatorial} and \cite{brandt2017approximating} to give two reordering algorithms with suboptimal bounds. They are removed to shorten the paper. After our paper, \cite{bernstein2021deterministic} gave an almost-linear time algorithm to compute a tree decomposition that is a $\polylog n$ factor from optimal.
%$\widetilde{O}(\tw(G_A))$-sized balanced vertex separator. This implies the following:
%When a tree decomposition is given, it is relatively straightforward; otherwise, we draw on two existing results on vertex separators to directly compute the ordering: One is \cite{brandt2017approximating} which computes a crude $\widetilde{O}(\tw(G_A)^{2})$-sized separator in $\widetilde{O}(n\cdot\tw(G_A)^{4})$ time. Another one is \cite{arora2016combinatorial} which reduces the computation of small vertex separators to the maximum flow problem. Applying these two results to \cref{thm:main} implies the following:
This implies the following:
\begin{theorem}\label{thm:main2}
Applying the algorithms in \cite{bernstein2021deterministic}, the runtime in \cref{thm:main} becomes
\[
\widetilde{O}\left( \left( (n \cdot \tw)^{1+o(1)} + n \cdot \tw^{2}\right)  \log(1/\eps) \right) = 
\widetilde{O}\left( \left(n^{1+o(1)}  \cdot \tw^{2}\right)  \log(1/\eps) \right).
\]
%or
%\[
%\widetilde{O}(n\cdot\tw(G_{A})^{2}\log(1/\eps) + \mathcal{T}_{\mathrm{maxflow}}),
%\]
%where $\mathcal{T}_{\mathrm{maxflow}}$ is the time it takes to compute a weighted
%maximum flow on an auxiliary directed graph with $O(|V(G_{A})|)$ vertices,
%$O(|E(G_{A})|)$ edges, edge capacities bounded by $O(|V(G_{A})|)$,
%and treewidth $O(\tw(G_{A})).$
\end{theorem}

%Using $\mathcal{T}_{\text{maxflow}}=\widetilde{O}(m + n^{1.5})$~\cite{brand2020maxflow}, our algorithm takes $\widetilde{O}(n\cdot\tw(G_A)^{2}\log(1/\eps)  + n^{1.5})$ when $\tw(G_{A})$ is large. This is faster than the current fastest algorithm for linear programs for the case $\tw(G_{A})\leq n^{0.68}$ with $d = \Theta(n)$. This case is common among Netlib instances (\cref{fig:practice}).

Detailed discussions can be found in literature (e.g.~\cite{renegar1988polynomial}
and \cite[Sections E, F]{lee2013path}) on converting an approximation solution
to an exact solution. To summarize, for integral $A,b,c$, it suffices
to pick $\eps=2^{-O(L)}$ to get an exact solution, where $L=\log(1+d_{\max}+\|b\|_{2}+\|c\|_{2})$
is the bit complexity and $d_{\max}$ is the largest absolute value
of the determinant of a square sub-matrix of $A$. For many combinatorial
problems, $L=O(\log(n+\|b\|_{2}+\|c\|_{2}))$.

%Finally, we note that the notion of treewidth for a linear program is robust: Introducing a new constraint increases treewidth by at most one, and introducing a new variable increases treewidth by a factor of at most $O(\log d)$ after reformulation. This property is particularly useful for problems with a standard combinatorial structure plus a small number of additional restrictions. For example, finding a max-flow with limited flow on a subgraph would simply require an additional constraint to the standard max-flow formulation, and the runtimes would not be severely affected in our framework.

\subsection{Convex Generalization}

\cref{thm:main} and \cref{thm:main2}  generalize to a class of convex optimization problem
as follows:
%\begin{theorem}
%Consider a convex program of the form
%\begin{equation}
%\min_{Ax=b,\,x_{i}\in K_{i}\text{ for }i\in[m]}c^{\top}x\label{eq:prob-1}
%\end{equation}
%where $A$ is a full rank $d\times n$ constraint matrix and $K_{i}\in\R^{n_{i}}$
%are compact convex sets.
%We identify the columns of $A$ in blocks, such that the block $i$
%contains the $n_{i}$ columns corresponding to $x_{i}$. The generalized
%dual graph$G_{A}$ has $d$ vertices $\{1,\dots,d\}$, and $ij\in E(G_{A})$
%if there exist a block $r$ such that $A$ has a non-zero entry in
%block $r$ in row $i$, as well as row $j$. We define $\tw(A)$ to
%be the treewidth of $G_{A}$. If $n_{i}=O(1)$ for each $i\in[m]$,
%then we can optimize (\ref{eq:prob-1}) to $\eps$ accuracy as described
%in Theorem 1 in time
%\end{theorem}
% Preview source code from paragraph 0 to 6
\begin{theorem}\label{thm:conv_main}
Given a convex program 
\begin{align}\label{eq:conv_form}
\min_{Ax=b,x_{i}\in K_{i}\text{ for }i\in[m]}c^{\top}x
\end{align}
where $A\in\R^{d\times n}$ is a full rank matrix with $d\leq n$
and $K_{i}\subset\R^{n_{i}}$ are convex sets, with $\sum_{i=1}^m n_i = n$. We identify the columns
of $A$ in blocks, such that block $i$ contains the $n_{i}$
columns corresponding to $x_{i}$. We define the \emph{generalized dual graph} $G_{A}$ to be the graph with vertices set $\{1,\cdots d\}$,
such that $ij\in E(G_{A})$ if there is a block $r$ such that
$A_{i,r}\neq \boldsymbol{0}$ and $A_{j,r}\neq\boldsymbol{0}$. We define the product convex set $K = \Pi_{i=1}^m K_i$.
Suppose that 
\begin{itemize}
\item we are given a tree decomposition of $G_{A}$ with width $\tau$,
\item $R$ is the diameter of the set $K$,
\item There exists $z$ such that $Az=b$ and $B(z,r) \subset K$,
\item $n_{i}=O(1)$ for all $i\in[m]$,
\item we are given initial points $x_{i} \in \mathbb{R}^{n_i}$ such that  $B(x_{i},r)\subset K_{i}$ for each $i$,
\item we can check if $y\in K_{i}$ in $O(1)$ time for all $i \in [m]$.
\end{itemize}
Then, for any $0<\eps\leq 1/2$, we can find $x\in K$ with $Ax=b$
such that
\[
c^{\top}x\leq\min_{Ax=b,x \in K}c^{\top}x+\eps\cdot\|c\|_{2}\cdot R
\]
in expected time 
\[
\widetilde{O}(n\cdot\tau^{2}\log(R/r) \log(R/(r\eps))).
\]
\end{theorem}

\begin{rem}
The proofs for the convex program and the linear program are almost identical.
Any operation involving the entry $A[i,j]$ in the linear program setting 
is generalized to operations on the $1\times n_{j}$ submatrix of $A$ from row $i$ and block
$j$. Since each block has size $O(1)$, the overall runtime relating to all matrix
operations is maintained. We analyze our interior point method directly
using this generalized formulation in this paper; the linear programming
formulation follows as a special case.
\end{rem}

%This result is a simple consequence of grouping the columns of $A$
%into blocks -- any operations pertaining to the entry $A[i,j]$ in
%the linear programming case is generalized to operations pertaining
%to the $1\times n_{j}$ submatrix of $A$ from row $i$ and block
%$j$, where we treat each entry in the submatrix individually. Since
%each block has size $O(1),$the overall runtime relating to all matrix
%operations is maintained. We analyze our interior point method directly
%using this generalized formulation in Section ; the linear programming
%formulation follows as a special case.

This natural convex generalization in fact captures a large number of
problem formulations. We illustrate with one example from signal processing, the fused lasso model for denoising~\cite{tibshirani2005sparsity}: Given a 1-D input signal $u_{1},u_{2},\cdots,u_{n}$,
find an output $x$ that minimizes the potential 
\[
V(x) = \sum_{i=1}^{n}(x_{i}-u_{i})^{2}+\lambda \sum_{i=1}^{n-1}|x_{i+1}-x_{i}|,
\]
where the first term restricts the output signal to be close to the
input, and the second term controls the amount of irregularity, and
$\lambda$ is the regularization parameter. To relate it back to our
problem \cref{eq:conv_form}, we consider a generalized formulation:
Given a family of convex functions $\phi_{1},\dots,\phi_{N}$ of $x=(x_{1},\dots,x_{n})$,
where for each $i,$ the function $\phi_{i}(x)=\phi_{i}(x_{S_{i}})$
only depends on the variables $\{x_{j}:j\in S_{i}\}$ for some subset
$S_{i}\in[n]$, we want to solve the problem
\begin{align}\label{eq:conv_form2}
\min_{x\in\mathbb{R}^{n}}\sum_{i=1}^{N}\phi_{i}(x_{S_{i}}).
\end{align}
By creating extra variables $y_{i,j}$ for all $i\in[n]$ and $j\in S_{i}$,
we can write the problem as $\min\sum_{i}t_{i}$, subjected to $y_{i,j}=x_{j}$
and $t_{i}\geq\phi_{i}(y_{i,j})$ for all $i$ and all $j\in S_{i}$.
The inequality constraints is equivalent to requiring that $(t,y)$
lie in the convex set $\{(t,y):t_{i}\geq\phi_{i}(y_{i,j})\}$.  This
is exactly in the form of \cref{eq:conv_form}. The dual graph $G_{A}$
of this problem is closely related to the intersection graph $G_{\mathcal{I}}$
of the set family $\{S_{i}\}_{i\in N}$: Specifically, each set of constraints $y_{i,j}=x_{j}$ corresponds to 
$|S_{i}|$ many vertices in $G_{A}$, and contracting
each such set into one vertex produces $G_{\mathcal{I}}$.
Hence, we have that the treewidth $\tw(G_A)$ of this convex program
is at most the treewidth of $G_{\mathcal{I}}$. For the denoising
problem above, the intersection graph is in fact close to a path and
has constant treewidth. Therefore, our result shows that this problem
can be solved in nearly-linear time, without relying on the specific
formula or structure. 

\subsection{Difficulties}

In this section, we discuss a few alternate approaches to our problem
and why they likely prove unfruitful. We will illustrate using problems of the form \cref{eq:conv_form2}
 when it is more straightforward.

\subsubsection{Dynamic Programming} \label{sec:DP}

Dynamic programming is a natural first approach,
as has been applied to other low treewidth problems. To explain the
difficulty of achieving fully-polynomial-time fixed-parameter tractability
in the optimization setting, we consider the following simplified
problem: Given a graph $G=(V,E)$ with a convex function $f_{e}:\R^{2}\rightarrow\R$
for every edge $e\in E$, consider the objective function on $\R^{V}$
defined by
\begin{align} \label{eq:f_graph}
f_{G}(x)=\sum_{ij\in E}f_{ij}(x_{i},x_{j}).
\end{align}
To divide the problem into smaller one, we consider any small
balanced vertex separator $S\subset V$; namely $V$ is partition
into three sets $S$, $L$ and $R$ such that there are no edges between
$L$ and $R$. We can write the objective function $f(x)$ by
\[
f_{G}(x)=f_{L}(x)+f_{R}(x)+f_{G-E(L)-E(R)}(x),
\]
where $f_{T}(x)=\sum_{ij\in E(T)}f_{ij}(x_{i},x_{j})$ and $E(T)$ is the set of edges with at least one end point in $T$. To minimize
$f_{G}$, it suffices to fix $x_{S}$ and recursively minimizing $x$
on $L$ and $R$, and minimize over all fixed $x_{S}$. Namely,
\[
\min_{x}f_{G}(x)=\min_{x_{S}}f_{G-E(L)-E(R)}(x_S)+ \widetilde{f_{L}}(x_{S}) + \widetilde{f_{R}}(x_{S}).
\]
where $\widetilde{f_{L}}(x_{S})=\min_{x_{L}}f(x_{S},x_{L})$ and $\widetilde{f_{R}}(x_{S})=\min_{x_{R}}f(x_{S},x_{R})$. Here, we crucially use the fact that $f_{G-E(L)-E(R)}(x)$ depends only on the variables in $S$, but not $L$ and $R$; the term $f_{L}$ depends only on the variables in $L$ and $S$, but not $R$; similarly for $f_{R}$. In general, if $f$ is convex, then both $\widetilde{f_{L}}$ and $\widetilde{f_{R}}$ are convex functions on $\R^{S}$. Hence,
the formula shows that we can solve the optimization problem by first
constructing the reduced problem on $G[L]$ and $G[R]$, then solve a size $|S|$
optimization problem.

If the $f_{ij}$'s are all quadratic functions, then both $\widetilde{f_{L}}$
and $\widetilde{f_{R}}$ are quadratic functions, and it turns out
they can be stored as matrices known as Schur complements. Hence,
we can solve the problem with the approach described above; in fact,
algebraic manipulation gives the sparse Cholesky factorization algorithm
with runtime $\widetilde{O}(n\cdot\tau^{2})$.

However, for general convex function $f_{G}$, it is not known how
to store the functions $\widetilde{f_{L}}$ and $\widetilde{f_{R}}$
efficiently, and this will likely require runtime exponential in treewidth.
%TODO: Haotian asks for more explaination on why takes exponential time.
At a high level, the reason is that before we solve the outer problem
$f_{G}$, we do not know at which fixed $x_{S}$ we should recurse
on for $\widetilde{f_{L}}$ and $\widetilde{f_{R}}$. It is known that
without adaptivity, exponentially many oracle calls are needed to
minimize a general convex function~\cite{nemirovski1994parallel, balkanski2018parallelization, bubeck2019complexity}.
This suggests we should compute $\widetilde{f_{L}}$ and $\widetilde{f_{R}}$
recursively for each different $x_{S}$. However, it is likely that
we need to access at least two different points $x_{S}$, and this
already leads to runtime recursion $T(n)\geq4T(n/2)+O(1)$ which is
at least $n^{2}$. Therefore, dynamic programming appears to be inefficient
for general convex optimization.

\subsubsection{Scanning Through Variables}

When the underlying structure of the variable dependencies is simple
enough, a simple scan through the variables may suffice for the problem
at hand; for example,~\cite{durfee2019efficient} successfully applies
this approach for function-fitting problems on a path. To illustrate,
consider a problem of the form
\[
\min_{x}F(x)\defeq f_{1}(x_{1},x_{2},\cdots,x_{k})+f_{2}(x_{2},\cdots,x_{k+1})+f_{3}(x_{3},\cdots,x_{k+2})+\cdots.
\]
Suppose $x^{*}$ is the unique minimizer of the function and 
$x_{1}^{*},x_{2}^{*},\cdots,x_{k-1}^{*}$ are given.
By looking at the gradient of the function above at the first coordinate,
we know that
\[
\frac{\partial}{\partial x_{1}}F(x)=\frac{\partial}{\partial x_{1}}f_{1}(x_{1}^{*},x_{2}^{*},\cdots,x_{k}^{*})=0.
\]
Since $x_{1}^{*},x_{2}^{*},\cdots,x_{k-1}^{*}$ is given, this is
a one variable non-linear equation on $x_{k}^{*}$ and it has a unique
solution under mild assumptions. Solving these equations, we obtain
$x_{k}^{*}$. Now, looking at $\frac{\partial}{\partial x_{2}}F(x)$,
we have that
\[
\frac{\partial}{\partial x_{2}}F(x)=\frac{\partial}{\partial x_{2}}f_{1}(x_{1}^{*},x_{2}^{*},\cdots,x_{k}^{*})+\frac{\partial}{\partial x_{2}}f_{2}(x_{2}^{*},x_{3}^{*},\cdots,x_{k+1}^{*})=0.
\]
Since we already know $x_{1}^{*},\cdots,x_{k}^{*}$, this is again
a one variable non-linear equation. Therefore, we can solve this problem
one variable at a time. 

% I still confused by the pargraph in general. 
% So, I only cut out the sentence I understand.
%

This approach can be modelled by an underlying graph structure in the following sense: 
Each variable $x_i$ is represented by vertex $i$ of the graph, and $i\sim j$ if there is some a term $f_k$ dependent on both $i$ and $j$. We say a vertex $i$ is solved if we know $x_i^*$. In the example above, the graph is a thick path, and if the first $k-1$ vertices are solved at the beginning, then we can follow the path to solve the remaining vertices one by one.
\begin{comment}
Suppose each $f_{i}$ depends on variables $\{x_{j}:j\in S_{i}\}$
for some index set $S_{i}\subseteq[n]$. We say a variable
$x_{j}$ \emph{is solved} if we know $x_{j}^{*}$. 
If there exists some $j$ such that all except one variable in $S_{i}$ with
$j\in S_{i}$ is solved, then we can solve for the missing
variable using $\frac{\partial}{\partial x_{j}}F(x)$. In the example
above, after knowing $x_{1}^{*},\dots,x_{k-1}^{*}$, the path structure
of $G$ allows us to iterate through $j$ from 1 to $n$ and move
along the path to solve the remaining variables one by one.
\end{comment}

Unfortunately, this type of scan-based algorithm cannot be generalized.
Consider a convex function of the form \cref{eq:f_graph}
where the graph $G$ is a complete binary tree with $n$ leaves. Let $i$ be a vertex such that the subtree rooted at $i$ is of height two containing four leaves. Observe that we cannot solve for the children of $i$ by case analysis, if both $i$ and the leaves are unsolved. Since there are $n/4$ many subtree of height two in $G$, at least $n/4$ many variables must be known at the beginning, before we can follow the graph structure to solve for the remaining variables. As such, this approach does not produce any meaningful simplification.

\subsubsection{Tightening the Iterations Bounds for Interior Point Methods}\label{sec:IPM_lower}

Another natural approach for attacking the conjecture is to prove
that existing polynomial time methods for linear program run faster
automatically for graphs with low treewidth. Currently, there are
two family of polynomial time algorithms -- the ellipsoid method
(more generally cutting plane methods) and interior point methods.
%For the ellipsoid method, it is known that $n^{2}$ iterations is
% needed regardless of the problem. 
For cutting plane methods, $n$
iterations are needed in general, since the method only obtains one
hyperplane per iteration, and we need $n$ hyperplane simply to represent
the solution even for the case $\tw(A)=O(1)$. In general, these hyperplanes are represented by dense vectors and will probably take $n^2$ time in total.

For interior point methods, the iteration bound is less clear since
there is no information obstruction. In general, it is known that
$O(\sqrt{n}\log(1/\eps))$ iterations are needed to solve a linear
program, and each iteration involves solving a linear system. For
the case $d=\Theta(n)$ in particular, this bound has not been improved
since the '80s. In fact, it has been shown that the standard interior
point method used in practice indeed takes $\Omega(\sqrt{n}\log(1/\eps))$
iterations in the worst case~\cite{mut2014tight, allamigeon2018log}, and some of these
constructions have $O(1)$ treewidth. Even for concrete problems such
as maximum flow, difficult instances for iterative methods often have
treewidth $O(1)$~\cite{kelner2014almost}. These lower bounds suggest that obtaining an optimization
method with $\widetilde{O}(\tw^{O(1)}(A))$ iterations requires a substantively new algorithm.

\subsubsection{Faster Iterations via Inverse Maintenance}

Dual to the previous approach is the idea of speeding up each iteration
of interior point methods. Each iteration of these methods require some computation or maintenance
involving a term $(AH^{-1}A^{\top})^{-1}$; previous work on linear
programming focused on inverse maintenance techniques to accomplish
this either explicitly or implicitly. In~\cite{DBLP:conf/stoc/CohenLS19, van2020solving, jiang2020faster},
the inverse is explicitly maintained and this takes at least $d^{2}$
time in total. ~\cite{brand2020bipartite, brand2020maxflow} focus on IPM for the
bipartite matching problem and the maximum flow problem, where a sparsified Laplacian system $AH^{-1}A^{\top}x=b$
is solved directly in each iteration and hence the whole algorithm
takes at least $d$ per step and $d^{1.5}$ time in total, where $d$ is the number of vertices.
It seems that either approach cannot lead to nearly linear time (when
$n=\Theta(d)$). 

In our setting, one natural approach is to maintain the Cholesky factorization
$LL^{\top}=AH^{-1}A^{\top}$. This can be done in nearly-linear time
in total, by combining ideas from numerical methods~\cite{Davis2006book}
and previous algorithms mentioned above. Unfortunately, in general,
almost any sparse update in $H$ leads to $\Omega(d)$ changes in
$L^{-1}$. Hence, it seems difficult to get a runtime faster than $d^{1.5}$
by just combining inverse maintenance with current knowledge of sparse
Cholesky factorization.

\subsection{Related Works}

\subsubsection{Algorithms With Runtime at Least Exponential to Treewidth}\label{sec:exp-treewidth-algs}

The notion of treewidth is closely tied to vertex separators; specifically,
low treewidth graphs have small vertex separators, and this structure
is amenable to a dynamic-programming approach for various
problems. A number of NP-hard problems such as $\textsc{Independent Set}$, $\textsc{Hamiltonian Circuit}$,
$\textsc{Steiner Tree}$, and $\textsc{Travelling Salesman}$ can be
solved with runtimes that depend only linearly on the problem size
and exponentially on treewidth~\cite{bodlaender1994tourist} as the
result of dynamic programming. They are extensively studied as part
of the class of fixed-parameter tractible problems. 
In general, dynamic programming style approaches based on the tree decomposition
unfortunately almost always lead to an exponential dependence on treewidth,
even for polynomial-time solvable problems. 

We point to one particular recent result here, which is a $2^{O(k^{2})}\cdot n$
time algorithm to find $k$ disjoint paths given $k$ vertex pairs
on a planar graph by~\cite{lokshtanov2020exponential}; it appears
to be one of the first algorithms to exploit treewidth in a completely
different way from dynamic programming. 

\subsubsection{Algorithms with Runtime Polynomial to Treewidth}\label{sec:poly-treewidth-algs}

When the problem is linear algebraic, such as solving linear systems and computing rank or determinant,
the dynamic programming approaches often leads to runtime polynomial in treewidth.

For linear systems $Ax=b$,
George first developed the method of nested dissection in~\cite{george1973nested},
which leveraged the underlying graph structure of $A$ for the case
where it is a grid. This was generalized by the seminal work of Lipton,
Rose and Tarjan in~\cite{lipton1979generalized}, to solving systems
where $A$ is any symmetric positive-definite matrix whose underlying
graph has good balanced vertex separators. This was further extended
by~\cite{alon2013matrix}, to apply to non-singular matrices over any field. 
The Cholesky factorization of $A$ is a key part of all aforementioned
results; it has a long line of study in numerical analysis~\cite{Davis2006book},
and is used as the default sparse linear system solver in various languages such as Julia, $\textsc{Matlab}$ and Python.
Our algorithm heavily relies on the machineries developed in this line of work.
% and we will explain this in details in \cref{sec:tw_tex} and \cref{sec:chol_tex}.

Recently, \cite{fomin2018fully} shows several problems can be reduced to matrix factorizations efficiently, including computing determinant, computing rank, and finding 
maximum matching, and this leads to $O(\tau^{O(1)}\cdot n)$ time algorithms where $\tau$ is the width of the given tree decomposition of the graph.
The only non-linear algebraic $O(\tau^{O(1)}\cdot n)$ time problem
we are aware of is {\sc Unweighted Maximum Vertex-Flow}~\cite{fomin2018fully},
which makes use of the crucial fact that the vertex separator size
is directedly connected to the flow size to achieve a $\widetilde{O}(\tau^{2}\cdot n)$
runtime.

%When a tree decomposition of width $\tau$ is given for
%the underlying graph of an $n\times n$ matrix $A$,~\cite{fomin2018fully}
%shows how to compute the determinant and rank of $A$, as well as
%solve a linear system $Ax=b$ in $O(\tau^{3}\cdot n)$ time. Finding
%the cardinality of a maximum matching and constructing it can be reduced
%to linear algebraic problems, and they also give $\widetilde{O}(\tau^{3}\cdot n)$
%and $\widetilde{O}(\tau^{4}\cdot n)$ time algorithms respectively,
%where $\tau$ is the width of the tree decomposition of the graph.
%The only other fully-polynomial time fixed parameter tractible problem
%we are aware of is $\textsc{Unweighted Maximum Vertex-Flow}$~\cite{fomin2018fully},
%which makes use of the crucial fact that the vertex separator size
%is directedly connected to the flow size to achieve a $\widetilde{O}(\tau^{2}\cdot n)$
%runtime, when given a tree decomposition of width $\tau$. 

When we are not restricted to nearly linear-time algorithms, 
~\cite{kyng2018incomplete} combines nested
dissection with support theory to solve the class of linear systems
where $A$ can be viewed as a higher dimensional graph Laplacian.
For semidefinite
programming,~\cite{zhang2018sparse} shows that interior point methods
can solve certain classes of sparse semidefinite programs in $O(\tau^{6.5}n^{1.5}\log(1/\eps))$
time, where $\tau$ is a sparsity parameter for SDPs analogous to
treewidth for LPs. Both algorithms require solving super-logarithm many linear systems.

As far as we know, there is no previous work on linear programming in direct relation to treewidth. 

\subsubsection{Related Works in Optimization}

A long line of work in the integer-linear programming community studies solving ILPs with respect to fixed treedepth, a parameter related but more restrictive than treewidth; indeed, ILPs can be weakly NP-hard even on instances with treewidth at most two. 
For an ILP with treedepth denoted $\mathrm{td}(A)$, 
\cite{eisenbrand2019algorithmic} gives a weakly polynomial ILPs algorithm running in time $O(g(\min\{\mathrm{td}(A),\mathrm{td}(A^{\top})\})\cdot\text{poly}(n))$,
where $g$ is at least some doubly-exponential function. 
This is followed-up
by~\cite{cslovjecsek2020blockstructured}, giving a strongly polynomial
algorithm running in $2^{O(\mathrm{td}\cdot2^{\mathrm{td}})}\Delta^{O(2^{\mathrm{td}})}n^{1+o(1)}$time,
where $\Delta$ is an upper-bound on the absolute value of an entry of $A$. 
~\cite{eisenbrand2019algorithmic} also discusses how an algorithm for ILP may be used to solve LP,~\cite{brand2019parameterized} builds on this to show an algorithm solving mixed integer-linear programs in time $f(a, \mathrm{td}(A)) \poly(n)$, where $a$ is the largest coefficient of the constraint matrix. 

The optimization work in this paper is mainly inspired by techniques
for general interior point methods, where the first proof of a polynomial
time algorithm was due to Karmarkar~\cite{karmarkar1984new}. After
multiple running time improvements~\cite{karmarkar1984new, renegar1988polynomial, vaidya1989new, nesterov1991acceleration, DBLP:journals/corr/abs-1910-08033, DBLP:conf/stoc/CohenLS19, DBLP:conf/colt/LeeSZ19, brand2020bipartite, van2020solving},
the current fastest IPMs are the results of~\cite{brand2020maxflow}
and~\cite{jiang2020fasterSDP}. We build on this recent line of work,
where ideas from interior point methods, Johnson-Lindenstrauss sketching, and linear algebraic data structures are
combined. For our dynamic data structure, we inspired by ideas similar to wavelets
commonly found in signal processing~\cite{rioul1991wavelets}, where
we maintain IPM information across iterations at different scales,
and process updates in every level of resolution.

\section{Overview of Our Approach}

In this section, we provide a high-level explanation of the overall approach and the techniques used. 
We discuss the more general convex formulation given in~\cref{thm:conv_main}, but for simplicity, \emph{we assume each $n_i = 1$ and $m = n$} in the statement of the theorem; this allows us to directly refer to coordinates of all relevant matrices and vectors, rather than blocks. We revert back to blocks for the detailed proofs in later sections.

%TODO: put this back in again
%For any $n=\sum_{i=1}^{^{m}}n_{i}$-dimensional vector, we view it as a concatenation of $m$ vectors, where the $i$-th one is of length $n_{i}$ and referred to as \emph{block} $i$. 
%Since each block is already assumed to be constant size, we further assume they are size one for all relevant matrix operations for simplicity.

Our algorithm is based on interior point methods~\cite{nesterov1994interior}.
These methods solve the convex program by alternating between taking
a gradient step, and projecting back to the constraint set $Ax=b$
under a suitable norm. The movement of $x$ follows some path $x(t)$
inside the interior of the domain $K$, with $t$ decreasing by
a $1-\Theta(1/\sqrt{n})$ factor every iteration, starting at some point $x(1)\in K$
and ending at the solution $x(0)$ we want to find. We use the common
central path defined by
\begin{equation}
x(t)=\arg\min_{Ax=b}c^{\top}x+t\sum_{i=1}^{m}\phi_{i}(x_{i})\label{eq:ipm_formula-1}
\end{equation}
where $\phi_{i}$ is a self-concordant barrier function 
(\cref{defn:self-concordant}) on
$K_{i}$ that blows up on $\partial K_{i}$, namely, $\phi_{i}(x_{i})\rightarrow+\infty$
as $x_{i}\rightarrow\partial K_{i}$. 
%For the convex setting, we assume each $n_{i}\in O(1)$, and $\phi_{i}$ is $O(1)$-self-concordant. 
%Since $\phi_{i}$ blows up on $\partial K_{i}$,
%if the interior of the domain is non-empty, then $x(t)$ lies in the
%interior for all $t>0$. Also, by the definition of $x(t)$, we have
%that $x(0)$ is a minimizer of the problem \ref{eq:problem}. 
We simultaneously require the dual central path $s(t)$ \cref{eq:KKT}, where $s$ is maintained similarly to $x$. 
%The weights $w\in\R_{>0}^{m}$
%are fixed throughout the algorithm and defined according to the predetermined cost of updating block $i$. %for the sake of simplicity, we will assume $w=\mathbf{1}_{m}$. %TODO: maybe this is not necessary for overview

The main difficulty is in following the path $x(t)$ efficiently.
At timestep $t$ of the central path, the current point $x$ is updated
by $x\leftarrow x+\delta_{x}$, where
\begin{equation}\label{eq:ipm-idea}
\delta_{x}=\left( H_x^{-1}-H_x^{-1}A^{\top}(AH_x^{-1}A^{\top})^{-1}AH_x^{-1}\right) \delta_{\mu}(x,s,t)
\end{equation}
for some non-negative diagonal matrix $H_x$ dependent on $x$, 
and vector $\delta_{\mu}$ dependent on $(x,s,t)$.

Our work therefore focuses on how to quickly and approximately maintain~\cref{eq:ipm-idea}, and the accumulation of $\delta_x$ over the entire central path for the final solution $x$ (and analogously for the dual solution $\delta_s$ and $s$). In \cref{sec:overview_RCPM}, we follow existing results and approximate $AH_{x}^{-1}A^{\top}$ by $AH_{\ox}^{-1}A^{\top}$ where $\ox$ is an approximation of $x$. This ensures the change in $AH_{\ox}^{-1}A^{\top}$ is low-rank in each iteration, which allows us to update $(AH_{\ox}^{-1}A^{\top})^{-1}$ implicitly and efficiently using existing results in Cholesky decomposition, outlined in  \cref{sec:overview_chol}. Unfortunately, the change of $\delta_{x}$ is dense even under a sparse change of $\ox$. In \cref{sec:overview_representation}, we propose a novel representation of $\delta_{x}$, where only $\widetilde{O}(n \tau \log(1/\eps))$ ``coefficients'' are changed during the central path. This allows us to maintain the solution $x$ implicitly during the whole algorithm using only $\widetilde{O}(n \tau^2 \log(1/\eps))$ time. Finally, to maintain $AH_{\ox}^{-1}A^{\top}$ close to $AH_{x}^{-1}A^{\top}$, we show how to detect large coordinate changes in this new representation in \cref{sec:overview_ds}.

The main contributions of this paper is the novel representation of the central path and the data structure to maintain and detect changes under this representation. We believe that this representation will be of independent interest beyond convex programs with low treewidth.

\subsection{Robust Central Path Method}\label{sec:overview_RCPM}
Although each entry of $H_x$ and $\delta_{\mu}$ is updated at every step due to the dense update of $x$,
a \emph{robust central path} circumvents the need to recompute them completely in every iteration, and thus lowers the cost of each step.
This idea has been used since the first interior point method~\cite{karmarkar1984new}, and has led to significant recent progress in convex optimization~\cite{DBLP:conf/stoc/CohenLS19, van2020deterministic, van2020solving, brand2020bipartite, jiang2020faster, jiang2020improved, brand2020maxflow}.

%We update $(H_x)_{i,i}$ only when the coordinate $x_{i}$ has changed a lot, and ensure that $H$ admits only $\widetilde{O}(n)$ coordinate changes throughout the whole algorithm. 
%We approximate $\delta_{\mu}(t)$ by a product of polynomials of $t$, whose coefficients admit sparse changes between steps.
In \cref{sec:Robust-IPM}, we give our robust central path algorithm (\cref{alg:IPM_framework}), which is a slight variant of the one presented in~\cite{DBLP:conf/colt/LeeSZ19}. The changes are needed to support some extra approximation required by our new representation. \cref{thm:IPM_framework} shows that to solve problem \cref{eq:conv_form}, it suffices to implement $\wt O(\sqrt{n} \log(1/\eps))$ approximate steps
\begin{align}
x & \leftarrow x + (H_{\ox}^{-1}-H_{\ox}^{-1}A^{\top}(AH_{\ox}^{-1}A^{\top})^{-1}AH_{\ox}^{-1})\delta_{\mu}(\ox,\os,\ot) \label{eq:central-path-update}\\
s & \leftarrow s + tA^{\top}(AH_{\ox}^{-1}A^{\top})^{-1}AH_{\ox}^{-1}\delta_{\mu}(\ox,\os,\ot) \nonumber
\end{align}
where $\ox, \os$ are vectors close to $x,s$, and $\ot$ is a scalar close to $t$.

We only need to output $(x,s)$ at the end, and do not need their exact values during the algorithm. Instead, it suffices to detect which coordinate has changed too much and update the approximation $(\ox,\os)$ accordingly. %TODO: section below is not needed..? update the ref
For interior point methods, if updated lazily, there are only a nearly linear number of coordinate changes to $\ox$ and $\os$ during the whole algorithm:
$$\sum_{\text{IPM step } k} \| \ox^{(k+1)} - \ox^{(k)} \|_0 + \sum_{\text{IPM step } k} \| \os^{(k+1)} - \os^{(k)} \|_0 = \wt O(n \log(1/\eps)).$$
Since $\ox, \os$ are $n$-dimensional vectors, every coordinate is updated only roughly $\log(1/\eps)$ times on average, and hence it allows for very efficient updates of the approximate steps. In particular, we have the following:
\begin{equation}
\text{Throughout the algorithm, there are only }\widetilde{O}(n\log(1/\eps))\text{ coordinate updates to }H_{\ox}.\label{eq:overview_H_change}
\end{equation}

\subsection{Cholesky Decomposition}\label{sec:overview_chol}

In recent IPM works, each iteration involves either
computing or maintaining $(AH_{\ox}^{-1}A^{\top})^{-1}$
of the update given in~\cref{eq:central-path-update}. 
However, this is too expensive for our setting, even for the case of constant treewidth.
The change of the inverse between consecutive steps is usually a dense matrix (possibly small rank) which takes at least $\Omega(d)$ space to represent.
In our algorithm, we instead maintain the sparse Cholesky decomposition.

$AH_{\ox}^{-1}A^{\top}$ is a positive-definite matrix, and therefore
admits a unique \emph{Cholesky decomposition} \\
$AH_{\ox}^{-1}A^{\top}=LL^{\top}$,
where $L$ is a lower-triangular matrix with positive diagonal entries.
The diagonal matrix $H_{\ox}$ changes throughout the algorithm, but crucially,
this only changes the entries of $L$, not its non-zero pattern.
In~\cref{sec:tw_tex}, we discuss how to compute a permutation of the rows
of $A$ (and correspondingly entries of \textbf{$b$}), and an associated
elimination tree $\mathcal{T}$ of $A$, which reflects the non-zero pattern of
$L$. Suppose the rows of $A$ has been reordered, and then $AH_{\ox}^{-1}A^{\top}$is
factored into $LL^{\top}$. Let $\tau$ be the height of the elimination
tree $\mathcal{T}.$ The following properties hold (\cref{thm:main-tw-Cholesky}):
\begin{itemize}
\item $\mathcal{T}$ is a tree on $d$ vertices $\{1,\dots,d\}$, with vertex
$i$ representing row/column $i$ of $L$.
\item The columns of $A,\,L,$ and $L^{-1}$ are all $\tau$-sparse.
\item The non-zero entries of $L^{-1}e_{i}$ and $Le_{i}$ are respectively
subsets of the path from vertex $i$ to the root of $\mathcal{T}$.
Furthermore, they can be computed in $\tau^{O(1)}$ time.
\item For a single coordinate change in $H_{\ox}$, it takes $\tau^{O(1)}$ time
to update $L$ exactly.
\end{itemize}
Now, we can rewrite $(AH_{\ox}^{-1}A^{\top})^{-1}$ as $L^{-\top}L^{-1}$, and take advantage of the sparsity of $L$ via $\mathcal{T}$ in the algorithm. In particular, by \cref{eq:overview_H_change}, we have the following:
\begin{equation}
\text{Throughout the algorithm, there are only }\widetilde{O}(n \tau^{O(1)}\log(1/\eps))\text{ coordinate updates to }L.\label{eq:overview_L_change}
\end{equation}

\subsection{Multiscale Representation of the Central Path}\label{sec:overview_representation}
To implement the central path steps, we want all variables to change in a sparse way, 
so we can update quickly between iterations. In particular, we want to represent $x$ (similarly $s$) implicitly by
$$ x = x_0 +  B h$$
for some vectors $x_0, h$ and some basis matrix $B$. 
%We also want to make sure all variables $x_0, h, B$ changes sparsely every iteration.
%change in h is not sparse^

When $H_{\ox}$ and $\delta_{\mu}$ admit only sparse changes between steps,
the first term ($H_{\ox}^{-1} \delta_{\mu}$) of~\cref{eq:central-path-update} is easy to compute explicitly, 
which we do and maintain as part of $x_0$.
Part of the second term given by $h\defeq L^{-1}AH_{\ox}^{-1/2}\delta_{\mu}$ is similarly easy to maintain,
due to the fact that each column of $L^{-1}$ and $A$ has sparsity
$\tau$ and can be obtained in $\tau^{O(1)}$ time. However, computing
and maintaining $H_{\ox}^{-1} A^{\top}L^{-\top}h$ explicitly is still costly. The first key observation of this paper is that the representation 
$$x = x_0 +  H_{\ox}^{-1} A^{\top}L^{-\top}h$$ 
has the following properties:

\begin{enumerate}
\item For any $i$, we can compute $x_i$ in $\tau^{O(1)}$ time.

Note that $x_i = (x_0)_i + h^\top L^{-1} A H_{\ox}^{-1} e_i$. Since we know each column of $A$ is $\tau$-sparse, we can compute $A H_{\ox}^{-1} e_i$ in $O(\tau)$ time and it is $O(\tau)$ sparse. Hence, $L^{-1} A H_{\ox}^{-1} e_i$ is just a mixture of $O(\tau)$ many columns of $L^{-1}$ and since each column of $L^{-1}$ is $O(\tau)$ sparse, we can compute it in $\tau^{O(1)}$ time. This gives a $\tau^{O(1)}$ time algorithm to compute $x_i$.

\item After a sparse update to $L$ and $H_{\ox}$, we can maintain the representation in $\tau^{O(1)}$ time.

More precisely, given $x = x_0 +  H_{\ox}^{-1} A^{\top}L^{-\top}h$, $\new{L} = L + \Delta L$, $\new{H_{\ox}} = H_{\ox} + \Delta H_{\ox}$, then we can find $\new{x_0}$ and $\new{h}$ in $\tau^{O(1)}$ time such that $x = \new{x_0} +  (\new{H_{\ox}})^{-1} A^{\top}(\new{L})^{-\top}\new{h}$.

For the change of $\new{H_{\ox}}$, we can simply set $\new{x_0} = x_0 + (H_{\ox}^{-1}- (\new{H_{\ox}})^{-1}) A^{\top}(\new{L})^{-\top}h$. Since $(H_{\ox}^{-1}- (\new{H_{\ox}})^{-1})$ is sparse, we can compute the term $(H_{\ox}^{-1}- (\new{H_{\ox}})^{-1}) A^{\top}(\new{L})^{-\top}h$ by the approach from Property 1 (computing the formula from left to right).

For the change of $\new{L}$, we simply need to find $\new{h}$ such that $(\new{L})^{-\top}\new{h} = L^{-\top}h$. Rearranging, we have $\new{h} = (\new{L})^{\top} L^{-\top}h = h + (\Delta L)^{\top} L^{-\top}h$. Again, since $(\Delta L)^{\top}$ is sparse, we can compute it from left to right.
\end{enumerate}

From now on, we call $h$ the \emph{multiscale coefficients}. Since there are only $\wt O(n \tau^{O(1)} \log(1/\eps))$ coordinates in $H_{\ox}$ and $L$ (\cref{eq:overview_H_change}, \cref{eq:overview_L_change}), Property 2 shows that we can maintain this representation in $\wt O(n \tau^{O(1)} \log(1/\eps))$ time. Furthermore, we have:
\begin{equation}
\text{Throughout, there are only }\widetilde{O}(n \tau^{O(1)}\log(1/\eps))\text{ coordinates updates to the multiscale coefficients.}\label{eq:overview_h_change}
\end{equation}
Finally, Property 1 shows that these multiscale coefficients is as good as explicit representation since we can read any entry in $\tau^{O(1)}$ time. Suppose we know which coordinates of $x$ deviated from $\ox$ significantly, then we can simply use Property 1 to update $\ox$. 

Combining this with heavy-hitter ideas, we can easily get an algorithm of time $\widetilde{O}(n^{1.25} \tau^{O(1)}\log(1/\eps))$ (See~\cite{ye2020fast} for an earlier draft version of this paper).

\subsection{Data Structures for Maintaining Multiscale Representation}\label{sec:overview_ds}

A key component of our algorithm revolves around finding which coordinates of $x$ deviate significantly from $\ox$. Specifically, we want to find large coordinate in $H_{\ox}^{1/2}(x - \ox)$, where the term $H_{\ox}^{1/2}$ is to measure the deviation in a correct norm required by the interior point method.

Similar to the discussion above, we can maintain $H_{\ox}^{1/2}(x - \ox)$ implicitly as $x_0 + \mathcal{W}^\top h$ %todo: maybe change this x_0 notation to match up with the algo
for some sparsely changing vectors $x_0$, where $\mathcal{W} \defeq L^{-1} A H_{\ox}^{-1/2}$ and $h \defeq L^{-1}AH_{\ox}^{-1/2}\delta_\mu$. Here, we focus on discussing the change of the term $\mathcal{W}^\top h$; 
analogous ideas are used for $x_0$.

First, observe that we cannot maintain $\mathcal{W}^\top h = H_{\ox}^{-1/2} A^{\top}L^{-\top}h$, as
the rows of $A$ and $L^{-1}$ may be dense. However,
it is relatively easy to maintain $v^{\top}\mathcal{W}^\top h$ for
any vector $v$, since $v^{\top}\mathcal{W}^\top h = h^\top \mathcal{W} v = h^{\top}L^{-1}AH_{\ox}^{-1/2}v$,
and we can exploit the column sparsity of $A$ and $L^{-1}$.
If we use a Johnson-Lindenstrauss sketching
matrix $\Phi$ in place of $v$ and maintain $\Phi \mathcal{W}^\top h$, then this allows us to quickly estimate $\|\mathcal{W}^\top h\|_{2}^{2}$.

We construct a data structure called the \emph{sampling tree} $\mathcal{S}$ (\cref{def:sampling-tree}), 
based on the elimination tree $\mathcal{T}$, to
store a family of sketches of the form $\Phi \mathcal{W}^\top h$. 
%The choice of sketches is motivated by standard techniques in heavy-hitters.
Specifically, $\mathcal{S}$ is a constant-degree tree with leaves given by the set $[n]$,
where leaf $i$ corresponds to $(\mathcal{W}^\top h)_i$.
For any node $v \in V(\mathcal{S})$, let $\chi(v) \subseteq [n]$ denote the set of all leaves in the subtree rooted at $v$, and let $\Phi_{\chi(v)}$
denote the JL sketching matrix restricted to the indices given by $\chi(v)$.
Then at node $v$, we maintain $\Phi_{\chi(v)} \mathcal{W}^\top h$. 
By JL properties, we can estimate $\|(\mathcal{W}^\top h)|_{\chi(v)}\|_{2}^{2}$ at each node $v$; 
in other words, we have the approximate $\ell_2$-norm of various subvectors of $\mathcal{W}^\top h$ of different lengths. 
Using this information, we can apply the standard sampling technique of walking down $\mathcal{S}$ from the root to a leaf:
\begin{equation}\label{eq:overview-sample-tree}
\text{We can sample for a coordinate $i$ proportional to $(\mathcal{W}^\top h)_i^{2}$ in $O(\mathrm{height}(\mathcal{S})) \leq \widetilde{O}(\tau)$ \emph{steps}.}
\end{equation}
A large coordinate $(\mathcal{W}^\top h)_i$ means $x_i$ and $\ox_i$ differ significantly. Then we compute $x_i$ exactly  and update $\ox_i \leftarrow x_i$.

$\ox$ and $\delta_\mu$ are updated every iteration, 
hence, we must maintain the latest $\mathcal{W}$ and $h$ to support sampling using $\mathcal{S}$.
As there are $\wt O(n\tau)$ nodes in $\mathcal{S}$,
we do not have enough time to update $\Phi_{\chi(v)}\mathcal{W}^\top h$ at every node $v$ every iteration.
However, observe that we only need to know the latest value of $\norm{(\mathcal{W}^\top h)|_{\chi(v)}}_{2}^{2}$ during the sampling procedure, 
and as $\mathcal{S}$ is a constant-degree tree, sampling once only visits $O(\mathrm{height}(\mathcal{S}))$ nodes in $\mathcal{S}$. So we may rely on a form of lazy maintenance. Here, we focus on discussing a coordinate change in $\ox$; analogous ideas are used for changes in $\delta_\mu$. 

For a single coordinate change in $\ox_i$, we need to update $H_{\ox}^{-1/2}$ and $L$, but crucially the update only affects $\Phi_{\chi(v)} \mathcal{W}^\top h$ at select nodes of $\mathcal{S}$.
Specifically, $H_{\ox}^{-1/2}$ changes by a single entry, which causes the value of $\Phi_{\chi(v)} \mathcal{W}^\top h$ to change only if $i \in \chi(v)$. Hence, for each entry update of $H_{\ox}^{-1/2}$, we only need to update a path in $\mathcal{S}$.
On the other hand, a change in $\ox_i$ causes $O(\tau)$ columns of $L$ to update. 
Each column of $L$ has a corresponding node $u \in \mathcal{S}$, such that for any $v \in \mathcal{S}$, the value $\Phi_{\chi(v)}\mathcal{W}^\top h$ maintained at $v$ changes only if $u$ is an ancestor or descendant of $v$.
Hence, for each column update of $L$, we split the effect into two:
\begin{enumerate}
	\item ``upwards'' effect: The updates to ancestors of $u$. Since $u$ has at most $\mathrm{height}(\mathcal{S})$ many ancestors, we have sufficient time to update these sketches immediately. 
	\item ``downwards'' effect: The updates to descendants of $u$. We cannot afford to update the whole subtree rooted at $u$; hence, we delay the update. 
	A node can have $\mathrm{height}(\mathcal{S})$-many delayed updates, with one per ancestor. 
	Then, we can perform all the delayed updates in $\tau^{O(1)}$ time when it is accessed during sampling.
\end{enumerate}
Combined with~\cref{eq:overview-sample-tree}, we have:
\begin{equation}
	\text{$\mathcal{S}$ can be maintained to support sampling a coordinate $i$ in $\tau^{O(1)}$ amortized time}.
\end{equation}

%\subsection{Improved Data Structures}\label{sec:overview_improved_ds}
To further lower the cost for the case that $\tau$ is large, we present a more involved construction of a sampling tree using heavy-light decompositions with height $O(\log n)$ (\cref{subsec:balanced-sampling-tree}).
% Add some sentence for ideas for getting the tight t depdence explain n t^2 here
% > path overlap a lot
% > heavy light decompositions 
% Put the formal statement here as restatable?
% Maybe not restable but an informal version?

\subsection{Proofs of Main Theorems}
%TODO: clarify on eps1 and eps in IPM
We now link the various pieces of this paper together to prove~\cref{thm:main,thm:main2,thm:conv_main}.

All three settings require a preprocessing step to find a suitable reordering
of the constraints $Ax=b$. Constructing the graph $G_A$ from the non-zero pattern of $AH^{-1}_{\ox}A^\top$ takes $O(n\tau^2)$ time. Then by~\cref{thm:main-tw-Cholesky}, we can find a reordering 
of the rows of $A$ and a binary elimination tree $\mathcal{T}$ for the corresponding
Cholesky decomposition: when a width-$\tau$ tree decomposition of $G_A$ is given as in~\cref{thm:main}, this takes $\widetilde{O}(n \tau)$ time and produces an elimination tree of height $\widetilde{O}(\tau)$. Otherwise, we use \cite{bernstein2021deterministic} to obtain a tree of height $\widetilde{O}(\tw(G_A))$ which takes  $\widetilde{O}(n \cdot \tw(G_A))$ time.

We can reduce the linear program of~\cref{thm:main} to a convex program of the form~\cref{eq:problem}, before invoking~\cref{thm:IPM_framework} for the interior point method.
Specifically, for the LP given in~\cref{thm:main}, each convex set $K_{i}$ is the interval $[u_{i},l_{i}]$ with $n_i = 1$; we have that $\phi_{i}(x_{i})=-\log(u_{i}-x_{i})-\log(x_{i}-l_{i})$
is a 1-self-concordant barrier function for $K_{i}$ minimized by $x_i = (l_i+u_i)/2$; without loss of generality, we may set $w=\mathbf{1}_{m}$ and have $\kappa=n$.

For~\cref{thm:conv_main}, we can  invoke~\cref{thm:IPM_framework} directly. When the barrier functions are not given, we use the universal barrier $\phi_i$ with self-concordance $n_i$ for each $i$ (\cref{subsec:IPM-universal-barrier}); since $n_i = O(1)$, we can find the minimizer $x_i$ of $\phi_i$ as a preprocessing step in $O(1)$ time. As in the LP case, we set $w = \mathbf{1}_m$ and have $\kappa = \sum_{i=1}^m w_i n_i = n$.

\cref{thm:IPM_framework} shows that the robust interior point method given as~\cref{alg:IPM_framework} produces the approximate solution as required, and terminates within $O(\sqrt{\kappa}\log(\kappa/\eps \cdot R/r))=O(\sqrt{n}\log(R/(\eps r)))$ steps. 
%To begin,~\cref{alg:IPM_framework} makes a small reduction in order to find an initial point $(x,s)$ for the central path; it is clear this change (\cref{sec:initial-point-reduction}) can be performed in linear time, and does not affect the order of $d,m,n$ or the elimination tree.  In the parameters of~\cref{alg:IPM_framework}, the central path itself begins at timestep $t_{\max} = 1$ and terminates at some $t_{\min} \geq \eps^2/(360 n^4 \kappa) = O(\eps^2/n^5)$ (\cref{line:IPM-iterations}).

The data structure \textsc{CentralPathMaintenance} is used to perform one step of the central path exactly as we need.
The cost of a step is analyzed in~\cref{thm:central-path-algo}. 
Let $\tau$ denote the height of the elimination tree $\mathcal{T}$ computed during preprocessing, and let $N=O(\sqrt{n} \log(n/\eps \cdot R/r))$ denote the number of central path steps.
To begin, we initialize the data structure via $\textsc{Initialize}$
in time $O(n\tau^{2}\log^4(N))$. At timestep $t$,~\cref{alg:IPM_framework} needs to find $\ox,\os, \overline{t}$ and compute updates to $x,s,t$, which is all accomplished by invoking $\textsc{MultiplyAndMove}(t)$. 
As $\textsc{MultiplyAndMove}$ is called $N$ times over the entire algorithm, 
the total runtime is 
$O({Nn^{1/2}}+n{\log(t_\textnormal{max}/t_\textnormal{min})})\cdot \tau^2\poly\log(N)=\widetilde{O}(n\tau^{2}\log(1/\eps))$.
At the very end, $\textsc{Output}$ outputs the result $(x,s)$
exactly in time $O(n\tau^{2})$. 

Finally, for the setting of~\cref{thm:conv_main}, since we use universal barrier functions $\phi_i$ for $i \in [m]$, computing $\grad \phi_i$ and $\grad^2 \phi_i$ as part of the IPM take $O(\log (Rn/r))$ time by~\cref{rem:IPM-universal-barrier}. Hence, we incur an additional $\log(R/r)$ factor in the overall runtime.
\begin{comment}
Lastly, we note one additional condition of our central path data
structure: Suppose $t_{0}$ denote the first timestep of the interior
point algorithm. $\textsc{MultiplyAndMove}$ can support all queries
$t$ satisfying $t\in(t_{0}/2,t_{0})$. As we decrease $t$ by a constant
factor each step, when $t$ becomes smaller than $t_{0}/2$, we must
restart the data structure with an updated value of $t_{0}$; in other
words, we output the exact $(x(t),s(t))$ at the current timestep
$t$, and reinitialize the data structure with initial point $(x(t),s(t))$,
and $t_{0}=t$. As there are $\widetilde{O}(\sqrt{n}\log(1/\eps))$
steps of the central path, we restart the data structure $O(\log n)$
times, a factor which we hide in the overall runtime. 
\end{comment}

\subsection{Wavelet Interpretation}
\begin{figure}[t]
\centering
\includegraphics[width=0.9\textwidth]{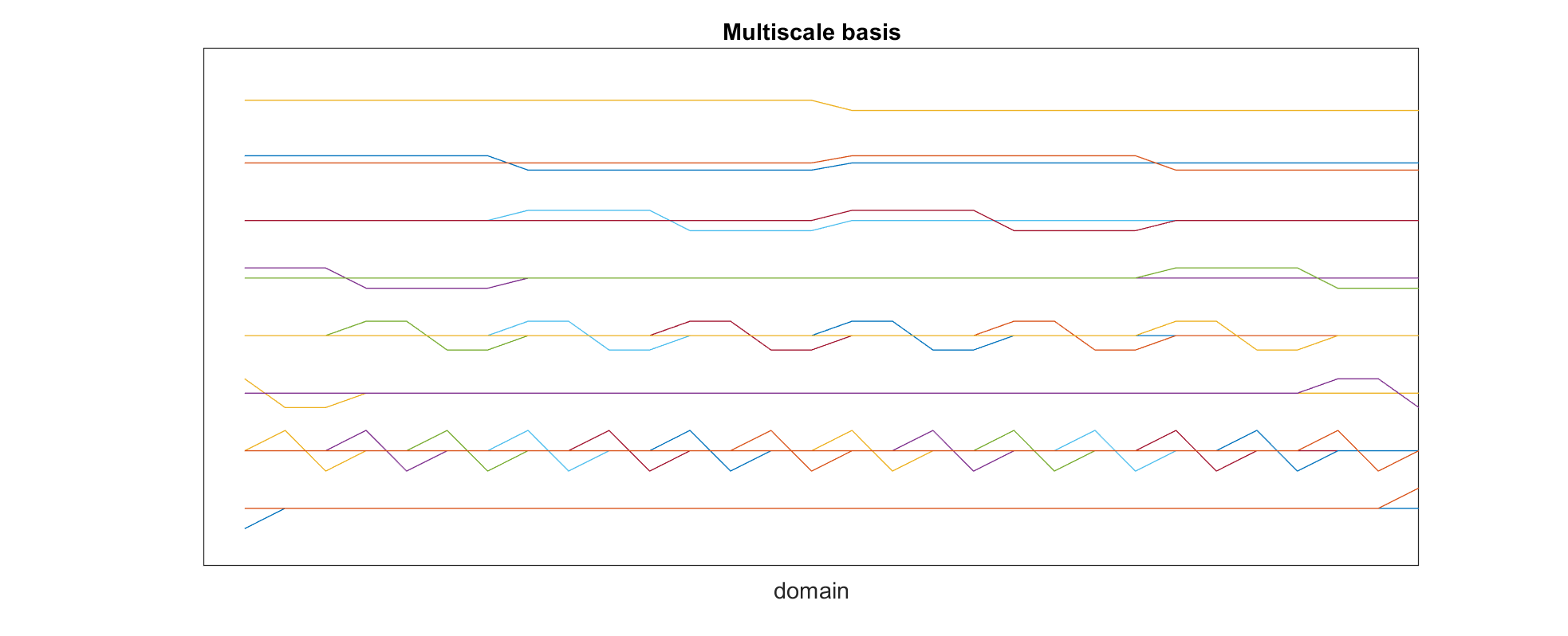}
\caption{The multiscale basis $\{ A^{\T} L^{-1} e_i \}_i$ where $A$ is the incidence matrix of a path. We group the basis by size and shift the basis according to its size for clarity.\label{fig:wavelet}} %not sure what the waves actually mean? :P
\end{figure}
Now, we explain the geometric meaning of this multiscale representation and its connection to wavelets. The rest of this subsection can be safely skipped as this view is not used in any proof. 

In wavelet theory, a complex signal is
represented as a linear combination of shifted and scaled versions
of a simple signal. In our context, we are representing the pending
change $\delta_{x}$ by various linear combinations of vectors $H_{\ox}^{-1} A^{\top}L^{-\top} e_i$. In the case that
$A$ is the incidence matrix of a path, the elimination tree is simply the complete binary tree, that when flattened in a breadth-first fashion, returns the original path. Here the vertices at different levels of the elimination tree exactly correspond to dyadic intervals of different lengths, while vertices at the same level correspond to dyadic intervals of the same length but with a ``time'' shift. When $H = I$, the vector $H^{-1} A^{\top}L^{-\top} e_i$ in fact looks quite similar to the Haar basis (See \cref{fig:wavelet}).

More precisely, we define the wavelet transform $\mathcal{W} \defeq L^{-1} A H^{-1/2}$, where $\mathcal{W}$ maps the signal in the original space to the coefficient space, with the basis elements corresponding to vertices of the elimination tree. Here we list only some similarities between this and the standard wavelet transform: %TODO: confirm the interpretation?

\begin{enumerate}
\item Applying the wavelet and inverse wavelet transform recovers the signal:
\[
\mathcal{W}^{\top}\mathcal{W}h=h\text{ for any }h\in\text{Range}(\mathcal{W}^{\top}).
\]
\item For each point in the original space, it is only covered by a few basis elements with different scales:
\[
\mathcal{W}e_{i}\text{ lies on }O(\tau)\text{ paths on the elimination tree}.
\]
\item For each basis element, it covers the original space with different scales:
\[
\text{The support of }\mathcal{W}^{\top}e_{i}\text{ is roughly a subtree}.
\]
\item There is a fast wavelet transform:
\[
\text{We can apply \ensuremath{\mathcal{W}\text{ and }}\ensuremath{\mathcal{W^{\top}}\text{to any vector in }}}n\tau^{O(1)}\text{ time}.
\]
\end{enumerate}
The key difference is that our wavelet basis does not represent any signal, but only the signal in the range of $H^{-1/2} A^\top$.

\begin{figure}[t]
\centering
\includegraphics[width=\textwidth]{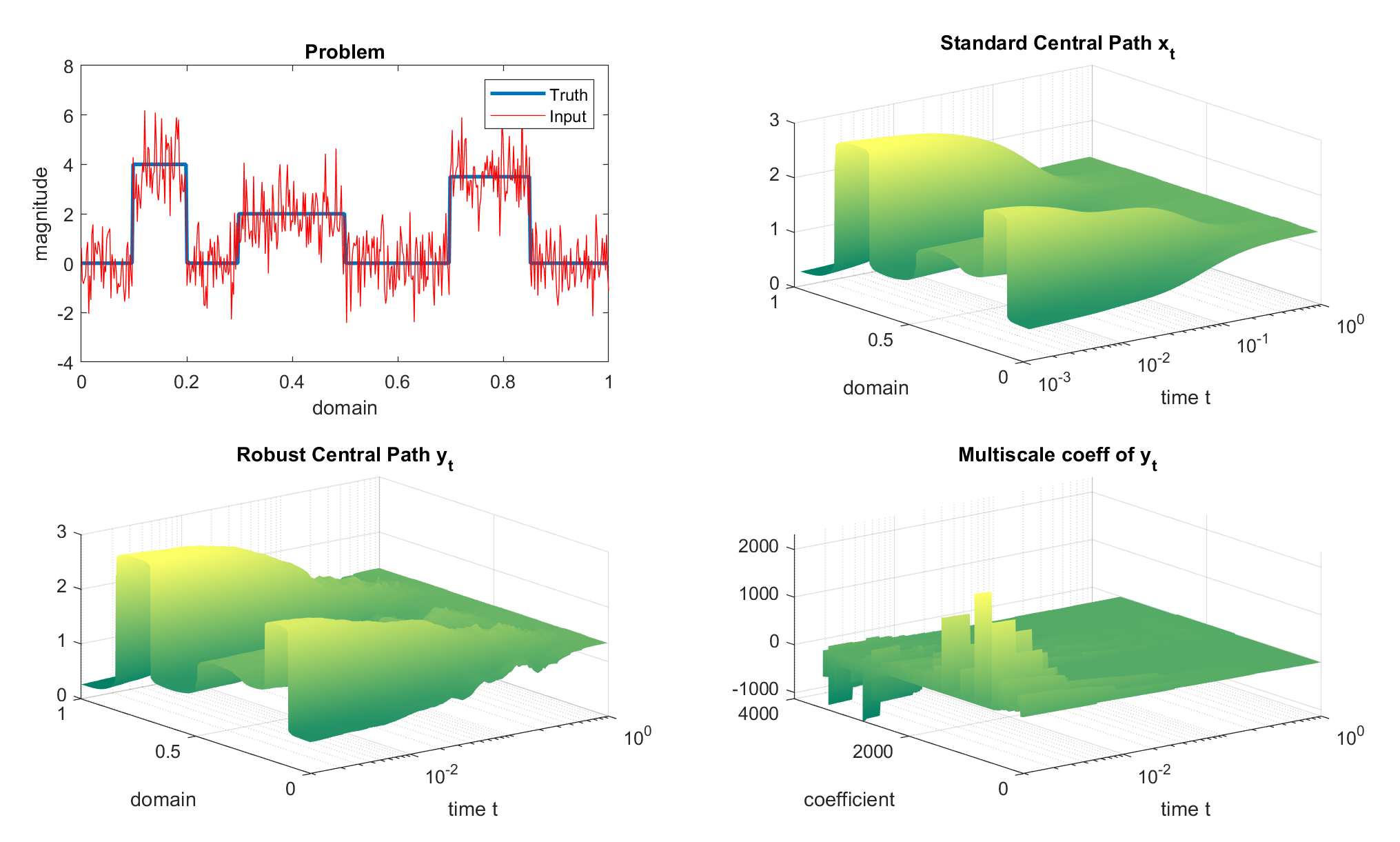}
\caption{Consider the fused lasso problem, with the upper left figure showing the input and true signals. The upper right figure shows the standard central path for this problem. The lower left figure shows the robust central path implicitly maintained in our algorithm. The lower right figure shows the multiscale coefficients we explicitly maintain in the algorithm. \label{fig:IPM}}
\end{figure}

To illustrate the multiscale coefficient, we consider the fused lasso problem\footnote{We pick this problem because it is easy to represent the whole central path as a surface plot.} in ~\cref{fig:IPM}. The central path is a 1-D signal that smoothly moves from a constant at $t=1$ to a recovered signal at $t = 10^{-3}$. Following this smooth transition is expensive, hence we consider the robust central path, which is noisier but converges to the same recovered signal. The noise comes from the approximation of $x$ by $\ox$\footnote{We emphasize $x_t$ is the robust central path, not $\ox_t$. We only use $\ox$ to approximate the linear systems. We cannot use $\ox$ as the solution because it does not satisfy the condition $Ax=b$.}. However, maintaining this robust central path is still quite expensive, because all coordinates change at every step in the robust central path. The crux of this paper is representing the robust central path by the multiscale basis $\{ A^{\T} L^{-1} e_i \}_i$, under which the coefficient changes sparsely.

Finally, we note that choosing this wavelet basis is quite natural from the view of physical science. Consider applying the interior point method for the maximum flow problem on a path: In this case, the linear system $AH_x^{-1}A^{\top}$ is simply a weighted Laplacian on a path, and the central path is simply the solution of some partial differential equations. Numerical differential equations in general face the same computation issues as us, that is, to represent the solution in a sparse way. It has been known since the '90s that both the Laplacian (more generally, elliptical differential equation), its inverse, and the solution can be represented sparsely using the wavelet basis such as~\cite{beylkin1991fast, dahmen1997multiscale}. The idea of using wavelets to approximate the solution has been applied to many partial differential equations~\cite{schaeffer2013sparse}. Arguably, this paper shows that the idea also applies to the ``partial differential equation'' defined by a central path, where the self-concordance theory ensures everything is well-behave enough for this to happen.

\section{Preliminaries} \label{sec:prelim} % TODO: maybe change the name to Notations?

In this section, we introduce the notations we used throughout the paper.

We say a symmetric matrix $A\in \R^{n\times n}$ is positive semidefinite (PSD) if $x^\top Ax\geq 0$ for all $x\in \R^n$ and positive definite (PD) if $x^\top Ax > 0$ for all $x\in \R^n$.
For symmetric matices $A,B\in \R^{n\times n}$, we use $A\succeq B$ to indicate that $A-B$ is a PSD matrix. We define operators $\preceq,\succ,\prec$ analogously.  

For a vector $v\in \R^{n}$, we use $\|v\|_2$ to denote its euclidean norm. 
We use $\|v\|_0$ to denote the number of non-zero entries in $v$.
For a PSD matrix $A \in \R^{n\times n}$, we let $\|v\|_A = \sqrt{v^\top A v}$.

We use $e_i$ to denote the standard unit vector. 
We use $\boldsymbol{0}_n,\boldsymbol{1}_n$ to denote all-zero and all-one vectors in $\R^n$. 
We define $\boldsymbol{0}_{m\times n}$ and $\boldsymbol{1}_{m\times n}$ analogously. 
We write $I_n\in \R^n$ to denote the identity matrix.
When dimensions are clear in the context, we drop the subscripts.

We use upper case letters to denote matrices, and lower cases for vectors and scalars. We use $A\cdot B$ to denote the matrix-matrix multiplication and $A\cdot x$ to denote the matrix-vector multiplication for readability.
When readability is not an issue, the operator $\cdot$ is omitted. To distinguish from the vector dot product, we always use $x^\top y$.

For any matrix $A\in \R^{m \times n}$, we use $A_S$ to denote the matrix restricted to the column (block) indices given by the set $S$. We say a block diagonal matrix $A \in \oplus_{i=1}^m \R^{n_i \times n_i}$ if $A$ can be written as
\begin{align*}
A = \begin{bmatrix}
A_1 & & & \\
& A_2 & & \\
& & \ddots &  \\
& & & A_m
\end{bmatrix}
\end{align*}
where $A_1 \in \R^{n_1 \times n_1}$, $A_2 \in \R^{n_2 \times n_2}$, $\dots$, and $A_m \in \R^{n_m \times n_m}$.

For any matrix $M$, we use $M_i$ to denote the $i$-th column or column block, and $\mathcal{M}_i$ to denote the non-zero pattern of the $i$-th column or column block, i.e. it is a set of row indices. For example, $j \in \mathcal{A}_i$ if row $j$ of $A$ is non-zero in a column in block $i$. We use $M^i$ to denote the $i$-th row of $M$ and $\mathcal{M}^i$ to denote the non-zero pattern of the $i$-th row. 

We use $\widetilde O(\cdot)$ to hide $\log^{O(1)}(n)$ and $(\log \log(1/\eps))^{O(1)}$ factors. We similarly define $\widetilde \Omega$ and $\widetilde \Theta$. 
For any positive integer $n$, we let $[n]$ denote the set $\{1,2,\ldots, n\}$.  
We use $\sinh x$ to denote $\frac{ e^x - e^{-x} }{2}$ and $\cosh x$ to denote $\frac{ e^x + e^{-x} }{2}$.

For a tree $\mathcal{T}=(V,E)$, we write $v\in \mathcal{T}$ or $v\in V(\mathcal{T})$ interchangeably to denote $v\in V$. 
For a rooted tree $\mathcal{T}$, we say a set $S$ \emph{lies on a path of $\mathcal{T}$} if there is a path $\mathcal{P}$ from the root of $\mathcal{T}$ to some node in $\mc{T}$, and $S\subseteq \mathcal{P}$.  
We use $\mathcal{P}(i)$ to denote the set of vertices on the path from vertex $i$ to the root in $\mathcal{T}$, and
$\mathcal{D}(i)$ to denote the set of vertices in the subtree rooted at $i$.

We use the convention that a tree consisting of a single vertex has height 1.

In our pseudocode, we use $\textsf{font}$ to denote data structure objects,  $\textsc{Font}$ to denote functions and object types, and regular math font to denote other variables stored in a data structure. 
Throughout our algorithms, we assume there is a basic object type {\sc List} which gives us random access to all its elements.
We write {\sc DataStructureA} extends {\sc DataStructureB} in the object-oriented programming sense: that is, {\sc DataStructureA} contains all the variables and functions from {\sc DataStructureB}, accessible either directly by name when there is no naming conflict, or with the keyword $\textsf{super}$.

%Diag
%Induced Norm
%PSD notation
\section{Elimination Tree}\label{sec:tw_tex}

Any positive-definite matrix $M$ admits a unique \emph{Cholesky factorization} $M = LL^\T$, where $L$ is a lower-triangular matrix with real and positive diagonal entries. 
In this section, we review some existing techniques~\cite{bodlaender1995approximating, Davis2006book} for computing a permutation of the linear constraints $Ax=b$, for $A \in \mathbb{R}^{d \times n}$. 
Our goal is to ensure that after permuting the rows of $A$, the Cholesky factorization $LL^\top = AH_{\ox}^{-1}A^\top$ will have certain desired sparsity patterns, 
which is then reflected in an associated \emph{elimination tree}. 

%The canonical Cholesky factorization algorithm can be viewed as a modified symmetric Gaussian elimination of $M$. We hope to make use of this factorization of $M := AH^{-1}A^\T$; however, this takes cubic time in the case that $L$ is dense. 

%A well-known technique is to permute the rows and columns of $M$, so that the Cholesky factorization $PMP^\T = L L^\T$ has some desired sparsity property that will help speed up computations.

Let the rows of $A$ be labelled $1, 2, \dots, d$. Recall we are given block-diagonal structure $n = \sum_{i=1}^m n_i$ for $A$ and $H_{\ox}$. We identify $A$ in \emph{column blocks}, 
with $A_i$ denoting the $n_{i}$ columns in block $i$. 
We simply use $H$ in the remainder of this section, as we only require its non-zero pattern which is independent of $\ox$; $H$ is an $n \times n$ block-diagonal positive-definite matrix, and without loss of generality, we may assume all entries in each block of $H$ are non-zero. In this case, observe that the $n_i$ columns in block $i$ of $AH^{-1/2}$ all have the same non-zero pattern, which we denote by $\mathcal{A}_i \subseteq [d]$. We use the convention that a tree on one vertex has height 1.

%\subsection{Cholesky Factorization}

%Before stating the main result of this section, we introduce the necessary definitions. 
%The application of these concepts in numerical analysis is well-studied.

The main results of this section is as follows. We give two cases for the runtime, corresponding to~\cref{thm:main} with a given tree decomposition, and ~\cref{thm:main2} without the decomposition. 
\begin{theorem}\label{thm:main-tw-Cholesky}
	Let $A$ be a $d \times n$ matrix with block structure $n = \sum_{i=1}^m n_i$, and suppose we are given the generalized dual graph $G_A$.
	We can compute a permutation $P$ of the rows of $AH^{-1/2}$ (equivalently, an ordering $\pi : [d] \mapsto [d]$), and a tree $\mathcal{T}$ on $d$ vertices, so that in the Cholesky factorization $PAH^{-1}A^\T P^\T = LL^\T$, 
	\begin{itemize}
		\item each vertex of $\mathcal{T}$ corresponds to a row/column of the Cholesky factor $L$, and
		\item the non-zero entries of $Le_i, L^{-1}e_i$ are respectively subsets of the path from vertex $i$ to the root in $\mathcal{T}$.
	\end{itemize} 
	The second property implies the column sparsity of $L$ and $L^{-1}$ are bounded by $\mathrm{height}(\mathcal{T})$. The following runtimes and associated tree height are possible:
	\begin{enumerate}
	\item $\widetilde{O}(n \cdot \tau)$ if a tree decomposition of the dual graph $G_A$ of width $\tau$ is given.  $\mathrm{height}(\mathcal{T}) = O(\tau \log{n})$.
	\item $\widetilde{O}((n \cdot tw(G_A))^{1+o(1)})$ without a given tree decomposition. $\mathrm{height}(\mathcal{T}) = O(tw(G_A) \polylog{n})$, where $tw(G_A)$ is the treewidth of $G_A$\footnote{Here, we defined the treewidth of a directed graph by simply ignoring the directions of the edges. This definition is compatible with first writing the directed max-flow as an LP, and then taking the treewidth of the dual graph of the constraint matrix.}.
	%\item  $\widetilde{O}(\mathcal{T}_{\mathrm{maxflow}})$ without a given tree decomposition, where $\mathcal{T}_{\mathrm{maxflow}}$ is the time required to compute the max flow on an auxilliary directed graph with $O(d)$ vertices and $\widetilde{O}(tw(G_A))$ treewidth
	%$\mathrm{height}(\mathcal{T}) = \widetilde{O}(tw(G_A))$.
	\end{enumerate}
\end{theorem}

Proving \cref{thm:main-tw-Cholesky} requires a number of concepts that may be unfamiliar to the reader. We begin by presenting them and their basic properties in the subsections below.

\subsection{Dual Graph and Treewidth}

We begin with the necessary definitions for completion.

\begin{definition}
	Recall the $\emph{generalized dual graph}$ of the matrix $A \in \mathbb{R}^{d \times n}$ with block structure $n = \sum_{i=1}^m n_i$ is the graph $G_A = (V,E)$ with $V = \{1,\dots, d\}$, and $ij \in E$ if and only if $A_{i,r} \neq \mathbf{0}$ and $A_{j,r} \neq \mathbf{0}$ for some $r$, where we use $A_{i,r}$ to mean the submatrix of $A$ in row $i$ and column block $r$. 
	
	Equivalently, $G_A$ is the dual graph of $AH^{-1/2}$ by the definition in~\cref{thm:main}. In particular, the non-zero pattern of $(AH^{-1/2})(AH^{-1/2})^\T$ is precisely the adjacency matrix of $G_A$.
\end{definition}

\begin{definition}
	A \emph{tree-decomposition} of a graph $G$ is a pair $(X, T)$, where $T$ is a tree, and $X : V(T) \mapsto 2^{V(G)}$ is a family of subsets of $V(G)$ called \emph{bags}  labelling the vertices of $T$, such that
	\begin{enumerate}
		\item $\bigcup_{t \in V(T)} X(t) = V(G)$,
		\item for each $v \in V(G)$, the nodes $t \in V(T)$ with $v \in X(t)$ induces a connected subgraph of $T$, and
		\item for each $e = uv \in V(G)$, there is a node $t \in V(T)$ such that $u,v \in X(t)$. 
	\end{enumerate}
	
	The \emph{width} of a tree-decomposition $(X, T)$ is $\max \{ |X(t)|-1 \; : \; t \in T\}$. The \emph{treewidth} of $G$ is the minimum width over all tree-decompositions of $G$. Intuitively, the treewidth of a graph captures how close the graph is to being a tree.
\end{definition}

The following structural results about treewidth are elementary.
\begin{lemma}\label{lem:tw-edge-bound}
	If $G$ is a graph on $d$ vertices and $tw(G) = \tau$, then $|E(G)| \leq d\tau$. \qed
\end{lemma}

\begin{lemma}
	If $G'$ is a subgraph of $G$, then $tw(G') \leq tw(G)$.	\qed
\end{lemma}

\begin{lemma}
	For the complete graph on $k$ vertices, $tw(K_k) = k-1$. \qed
\end{lemma}
There are some basic relations between the sparsity of a matrix $A$ and the treewidth of its dual graph:
\begin{lemma}
	Any block of $A$ with sparsity $\tau$ induces a clique of size $\tau$ in $G_A$. It follows that $\max\{|\mathcal{A}_i| : A_i \text{ a column block of } A\} \leq tw(A) + 1$. 
	\qed
\end{lemma}

%treewidth, that generalizes the property of having small separators. Intuitively, a graph has small treewidth if it can berecursively decomposed intosmall subgraphs that have small overlap, or even more intuitively, if the graph resembles a ‘fat tree’. Many problems that are NP-hard for general graphs can be solved in polynomial time for graphs withsmall treewidth. 

Treewidth is a natural structural parameter of a graph, with close connections to graph algorithms of a recursive nature. At a high level, it is generalized by the notion of well-separable graphs. We are particularly interested in its connection to vertex separators.

\subsection{Balanced Vertex Separator}

\begin{definition}
Let $G = (V,E)$ be a graph. For any $W \subseteq V$ and $1/2 \leq \alpha < 1$, an \emph{$\alpha$-vertex separator of $W$} is a set $S \subseteq V$ of vertices such that every connected component of the graph $G[V - S]$ contains \emph{at most} $\alpha \cdot |W|$ vertices of $W$. In the particular case when $W = V$, we call the separator an \emph{$\alpha$-vertex separator of $G$}. The \emph{separator number} of $G$ is the maximum over all subsets $W$ of $V$ of the size of the smallest $1/2$-vertex separator of $W$ in $G$. 

We sometimes denote an $\alpha$-vertex separator $S$ by $(G_1, S, G_2)$, where $V(G_1) \cup S \cup V(G_2) = V(G)$, and $G_1$ and $G_2$ are disconnected in $G \setminus S$.
\end{definition}

Similar to treewidth, separator numbers are monotone.
\begin{lemma}
	Let $G'$ be a subgraph of $G$. For any constant $1/2 \leq \alpha < 1$, the size of the smallest $\alpha$-vertex separator of $G'$ is at most that of $G$. \qed
\end{lemma}

The following theorem relates the treewidth of a graph and the separator number. 
\begin{theorem}[\cite{bodlaender1995approximating}, Lemma 6]
	If $G$ is a graph with treewidth $\tau$, then there exists a 1/2-balanced separator of $G$ of size at most $\tau+1$. \qed
\end{theorem}

Now we return to ideas for computing the permutation and elimination tree.

\subsection{Elimination Tree}
    Let $G = (V,E)$ be the generalized dual graph of $A$, that is, its adjacency matrix is given by the non-zero pattern of $AH^{-1}A^\top$. 
    Let $\pi :  V \mapsto [d]$ be an ordering of the vertices of $G$, which we will call an \emph{elimination order}. 
    We say a vertex $v \in V$ is eliminated at step $\pi(v)$. 
    The \emph{filled graph of $G$ corresponding to $\pi$}, denoted by $G_\pi^+$, is constructed as follows:
    
    \begin{algorithm}[H]\caption{Construct $G_\pi^+$} \label{alg:make-filled-graph}
    	\begin{algorithmic}
    		\State $G_\pi^+ \leftarrow (V, E)$
    		\For{$i$ from 1 to $n$}
    			\ForEach{$v \in V$ such that $\pi(v) > i$}
    				\If{$\exists$ a path $P$ from $\pi^{-1}(i)$ to $v$ in $G$, and all $u \in P - v$ satisfies $\pi(u) \leq i$}
    					\State add an edge between $\pi^{-1}(i)$ and $v$ in $G_\pi^+$
    				\EndIf
    			\EndFor
    		\EndFor
    		\State \Return $G_\pi^+$
    	\end{algorithmic}
    \end{algorithm}

	This construction of $G_\pi^+$ is also known as the elimination game on $G$, which intuitively models the canonical Cholesky factorization algorithm on $PAH^{-1}A^\T P^\T = LL^\T$, where $P$ is the permutation matrix for $\pi$: 
	Indeed, eliminating the vertex $\pi^{-1}(i)$ at the $i$-th iteration of the elimination game can be viewed as moving the $\pi^{-1}(i)$-th row of $A$ to the $i$-th row in the factorization algorithm, and adding the edge between $\pi^{-1}(i)$ and $v$ for the specified vertices $v \in V$ indicates that the $vi$-th entry of $L$ is non-zero in the factorization algorithm. 	
	It turns out the adjacency matrix of the filled graph $G_\pi^+$ precisely gives the nonzero structure of the triangular factor $L$. Hence, our goal is to choose $\pi$ to decrease the number of edges in $G_\pi^+$.

	Formally, $u,v \in V(G_\pi^+)$ are adjacent if and only if there is a path $P$ from $u$ to $v$ in $G$, such that all interior vertices $w$ on $P$ satisfies $\pi(w) < \min \{\pi(u), \pi(v)\}$.
	
	\begin{definition}[Elimination Tree]
	The \emph{elimination tree corresponding to $\pi$} is the tree $\mathcal{T}$ defined by the following parent-children relation: 
	For a vertex $v \in V$, its parent is $\argmin \{\pi(w) : w \in N_{G^+_{\pi}}(v), \; \pi(w) > \pi(v)\}$; 
	in words, it is the vertex $w$ that is eliminated earliest after $v$, that is reachable from $v$ in $G$ using a path whose interior vertices are all eliminated before $v$.
	Different elimination orders give rise to different elimination trees. The height of the shortest elimination tree over all choices of $\pi$ is the \emph{minimum etree height}.
	\end{definition}
	
	When the rows of $A$ are reordered according to $\pi$, the elimination tree reflects the non-zero pattern in the Cholesky factor.

	\begin{lemma}[\cite{Schreiber1982}]\label{lem: nnz path}
		Let $L$ be the Cholesky factor for the matrix $AH^{-1}A^\T$. Let $L_j$ denote the $j$-th column block. The non-zero pattern of $L_j$ is a subset of the vertices on the path from $j$ to the root in the elimination tree corresponding to the identity permutation.
	\end{lemma}
	
	\begin{example} 
		 The figure below shows the relationship between a matrix, its Cholesky factor, and the corresponding elimination tree. On the left is a $10 \times 10$ matrix $AA^\T$, with rows labelled $\{1, \dots, 10\}$. In the middle is the Cholesky factor $L$ of $AA^\T$. On the right is the elimination tree, where node $i$ in the tree corresponds to row $i$ of the matrices $AA^\T$ and $L$.
		 \begin{figure}
		 	\centering 
		 	\includegraphics[scale=0.9]{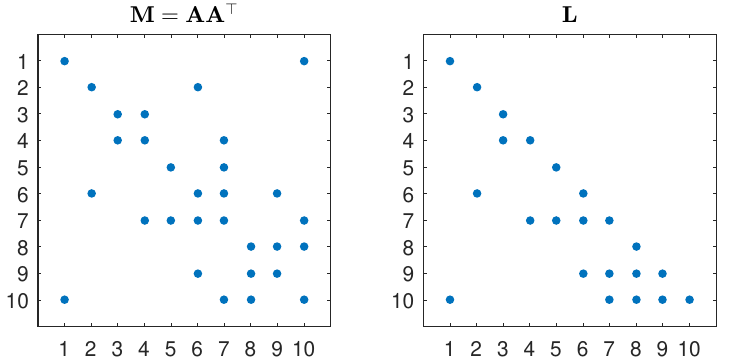}
		 	\includegraphics[scale=1]{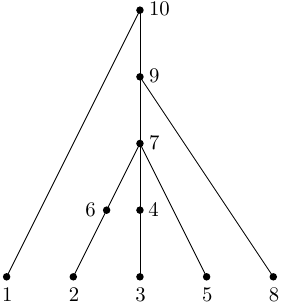}
		 	\caption{Each blue dot represents a non-zero entry in the matrix.} 
		 \end{figure}
	\end{example}

	\begin{lemma}\label{lem:clique-in-elim-tree}
		If $uw$ is an edge in $G$, then in any elimination tree $\mathcal{T}$ of $G$, there is an ancestor-descendant relationship between $u$ and $w$. It follows that if $K$ is a clique in $G$, then in any elimination tree $\mathcal{T}$ of $G$, the vertices of $K$ all lie on the same path from some leaf of $\mathcal{T}$ to the root.
		\qed
	\end{lemma}

The various parameters presented above are related by the following result:
\begin{theorem}[\cite{bodlaender1995approximating}, Theorem 12]
	Every graph $G$ on $n$ vertices satisfies
	\[
	\text{separator number}-1  \leq \text{treewidth} \leq \text{min elimination tree height} \leq \text{separator number} \cdot \log n .
	\]
\end{theorem}

This structural theorem indicates that we can construct an elimination tree $\mathcal{T}$ of $G_A$ and bound its height as a function of $tw(A)$. Specifically, we use the standard technique of recursively computing vertex separators, and using them to generate an ordering $\pi$ of the vertices of $V(G_A)$. 
Rather than constructing the elimination tree according to the definition however, we construct a slightly taller bounded-degree tree, and show it still reflects the sparsity conditions of the Cholesky factor.

In~\cref{alg:make-elim-tree}, we use a list notation $(v_1, \dots, v_k)$ to denote the ordering $\pi$ of a set of vertices $\{v_1, \dots, v_k\}$ with $\pi(v_i) = i$. We use $+$ to denote the concatenation of two lists.
\begin{algorithm}[H]
	\caption{Constructing an Elimination Order and Tree} 
	\label{alg:make-elim-tree}
	\begin{algorithmic}[1]
		\Procedure {makeElimOrderAndTree}{$G$}
		\If {$|V(G)|\leq f(\tau)$}
			\State let $\pi$ be an arbitrary ordering of $V(G)$
			\State construct a path on $V(G)$ according to $\pi$, let $u$ be the last vertex of the ordering/path
			\State \Return $(\pi, u)$
		\EndIf
		\State $(G_1, S, G_2) \leftarrow \textsc{approxBalancedSeparator}(G)$
		\State $(\pi_1,v_1) \leftarrow \textsc{makeElimOrderAndTree}(G_1)$ 
		\State $(\pi_2,v_2) \leftarrow \textsc{makeElimOrderAndTree}(G_2)$
		\State $\pi \leftarrow $ arbitrary ordering of $S$
		\State construct a path on $S$ according to $\pi$, let $u$ be the first vertex of the ordering/path and $v$ the last
		\State set $u$ as the parent of $v_1$ and $v_2$ 
		\State \Return $(\pi_1 + \pi_2 + \pi, v)$
		\EndProcedure
		
		\begin{comment}
		\Procedure{\textsc{makeElimTree}}{$G, \pi$}
			\State \textsc{//} We use \textsc{union-find} with an additional vertex tag on sets indicating the root of the subtree, as follows:
			\State \textsc{//} \textsc{MakeSet}$(v, x)$ makes the singleton set with element $v$ and tag $x$
			\State \textsc{//} \textsc{Union}$(u,v,x)$ takes the union of the sets containing $u$ and $v$, and tags the resulting set with $x$
			\State \textsc{//} \textsc{Find}$(v)$ returns $(v_s,x)$ where $v_s$ is the representative of the set containing $v$ with tag $x$
			\For {$i = 1$ to $n$}
				\State $v \leftarrow \pi^{-1}(i)$
				\State \textsc{MakeSet}$(v, v)$
				\For {$w \in N_G(v)$}
					\State $(w_s,x) \leftarrow \textsc{Find}(w)$
					\State $(v_s, y) \leftarrow \textsc{Find}(v)$
					\If {\textsc{$w_s$ != $v_s$} and \textsc{$w_s$ != null}}
						\State Set $x$'s parent to be $v$ in $T$
						\State \textsc{Union}$(w_s, v_s, v)$
					\EndIf
				\EndFor
			\EndFor
		\EndProcedure
		\end{comment}
	\end{algorithmic}
\end{algorithm}

\begin{theorem}
	Let $G = (V,E)$ be a graph on $n$ vertices with treewidth $\tau$, and let $f(\tau)$ be some function of $\tau$. Suppose \emph{$\textsc{approxBalancedSeparator}$} is an algorithm that, given a graph $H$ on $k$ vertices, computes $(H_1, S, H_2)$ where
	\begin{enumerate}
		\item $H_1, H_2 \subseteq H$ are subgraphs of $H$, and $V(H_1) \cup S \cup V(H_2) = V(H)$, %TODO: Is this reductant?
		\item $S$ is an $\alpha$-vertex separator of $H$ for some universal constant $1/2 \leq \alpha < 1$, and $|S| \leq f(\tau)$,
		\item the algorithm runs in time $T_{sep}(k)$.
	\end{enumerate}	
	Then~\cref{alg:make-elim-tree} constructs an elimination order and a binary tree $\mathcal{T}$ for $G$ of height at most $O(f(\tau) \cdot \log n)$, in time $\widetilde{O}(T_{sep}(n))$.
\end{theorem}

\begin{proof}

We have the following straightforward analysis of $\textsc{makeElimOrderAndTree}$: Let $T(k)$ denote the runtime on a graph with $k$ vertices. Then
\[\begin{cases}
	T(k) = O(1) \qquad \qquad &k \leq f(\tau) \\
	T(k) \leq  T(\alpha' k) + T((1-\alpha')k) + T_{sep}(k) \qquad &k > f(\tau)
\end{cases}\]
Solving the recurrence, we have $T(n) = \widetilde{O}(T_{sep}(n))$.

In total, $\textsc{makeElimOrderAndTree}$ recurses to a depth of $O(\log n)$, and at each recursive iteration, the contribution to the elimination tree height is the size of the separator $|S|\leq f(\tau)$ computed in the iteration.
\end{proof}

\begin{comment}
After computing the elimination order, it is a fairly straightforward algorithms exercise to construct the elimination tree quickly; we include the pseudocode for completeness. The main idea here is that at iteration $i$, each set represents a connected component $K$ in $G[\{ v \;:\; \pi(v) < i\}]$, and the tag on a set represents the unique element in the component without a parent in $\mathcal{T}$ yet; in other words, it is the root of the subtree of $\mathcal{T}$ containing exactly all vertices in $K$. The full explanation is omitted. $\textsc{makeElimTree}$ has run-time
\[
	\sum_{v \in V(G)} \sum_{w \in N(v)} O(\alpha(n)) \leq 2|E(G)| O(\log n) \leq \widetilde{O}(n\tau),
\]
where $\alpha(n)$ is the inverse-Ackermann function.
\end{comment}

%\subsection{Balanced Separator via Tree Decomposition}

A standard implementation of \textsc{approxBalancedSeparator} given a tree decomposition of $G$ is as follows:
\begin{theorem}
	Let $(X,T)$ be a width-$\tau$ tree decomposition of a graph $G$ on $n$ vertices. Then in $O(n \tau)$ time, we can find a $2/3$-vertex separator $(G_1, S, G_2)$ of $G$, and tree decompositions $(X_1, T_1)$ of $G_1$ and $(X_2, T_2)$ of $G_2$ each of width at most $\tau$.
\end{theorem}
\begin{proof}
	We assume $T$ has $O(n)$ nodes to start (a transformation can be made in $O(\tau \cdot |V(T)|)$ time in the recursive iterations, see e.g.~\cite{fomin2018fully} Definition 2.4). 
	By scanning through the bags of $T$ in $O(n \tau)$ time, we can find a node $t \in T$ such that $T \setminus t$ is two disjoint subtrees $T_1, T_2$, with $|\bigcup_{s \in T_1} X(s) \setminus X(t)| \leq 2/3 n$, and similarly for $T_2$. Then $X(t) \subseteq V(G)$ is a $2/3$-vertex separator of $G$.
	By removing the vertices $X(t)$ from all the bags in $T_1$ and $T_2$, we get the tree decompositions of $G_1$ and $G_2$ respectively, both of width at most $\tau$.
\end{proof}

When the tree decomposition is not given, we can use \cite{bernstein2021deterministic} to find it approximately.
\begin{theorem}[\cite{bernstein2021deterministic}]
Given a graph with $m$ edges and $n$ vertices, we can compute a width-$O(tw(G) \log^{3}{n})$ tree decomposition in $O(m^{1+o(1)} \polylog{n})$ time.
\end{theorem}

\begin{proof}[Proof of~\cref{thm:main-tw-Cholesky}]
	
	It remains to show that the tree $\mathcal{T}$ returned by~\cref{alg:make-elim-tree} satisfies the sparsity properties specified in~\cref{thm:main-tw-Cholesky}.
	
	Let $\mathcal{T}'$ be the true elimination tree corresponding to the elimination order $\pi$ computed by~\cref{alg:make-elim-tree}. 
	Note that in a recursive iteration $\textsc{makeElimOrderAndTree}(H)$, the subroutine $\textsc{approxBalancedSeparator}(H)$ will return $(H_1, S, H_2)$, such that \emph{in the original graph $G$}, vertices in $H_1$ are only connected to vertices in $H_2$ via a path containing vertices in $S$. 
	Hence, vertices in $H_1$ have no ancestors in $H_2$ in $\mathcal{T}'$, and vice versa. Any path in $\mathcal{T'}$ from a vertex $i \in H_1$ to the root goes through some higher-ordered vertices in $H_1$ followed by a subset of the vertices $S$; this is contained in the path in $\mathcal{T}$ from $i$ to the root, which includes all higher-ordered vertices in $H_1$ and all of $S$.
	By~\cref{lem: nnz path}, 
	after the permuting according to $\pi$, for each $j \in [d]$, the non-zero pattern of $L_j$ is a subset of the path from $j$ to the root in $\mathcal{T}'$; it is therefore also true in $\mathcal{T}$.
	
	Plugging in $f(\tau) = \tau$ when a width-$\tau$ decomposition is given, and $f(\tau) = O(tw(G) \log^{3}{n})$ when not, and using the monotonicity property of treewidth and separator size, we get the conclusions of~\cref{thm:main-tw-Cholesky} immediately.
\end{proof}

\section{Sparsity Patterns and Maintaining the Cholesky Factorization}\label{sec:chol_tex}

In this section, we discuss the sparsity properties of all the matrices we work with for the central path, and the required runtime for their computations and maintenance. All of these properties are known (see textbooks~\cite{george1994computer,Davis2006book} for more complete introductions). We include some algorithms and proofs to familiarize readers for the techniques we will use. As in the previous sections, we have the constraint matrix $A \in \mathbb{R}^{d \times n}$ whose rows are permuted according to \cref{thm:main-tw-Cholesky}. Let $L$ be the Cholesky factor of $AH^{-1}A^\T$, and let $\mathcal{T} =(\{1,\dots, d\}, E)$ be the elimination tree for $L$ of height $\tau$. Note we use the convention that a tree consisting of a single vertex has height 1.

Recall $\mathcal{P}(i)$ denote the set of vertices on the path from vertex $i$ to the root in $\mathcal{T}$, and
$\mathcal{D}(i)$ denotes the set of vertices in the subtree rooted at $i$ (including $i$) in $\mathcal{T}$.
For any matrix $M$, we use $M_i$ to denote the $i$-th column or block, and $\mathcal{M}_i$ to denote the non-zero pattern of the $i$-th column or block (i.e. it is a set of row indices). For example, $j \in \mathcal{A}_i$ if row $j$ of $A$ is non-zero in a column in block $i$. We use $M^i$ to denote the $i$-th row of $M$ and $\mathcal{M}^i$ to denote the non-zero pattern of the $i$-th row. 

We begin with basic properties of $A$ and $L$:

\begin{lemma}\label{lem:A-sparsity-pattern}
	If $tw(A) = \tau$, then $nnz(A_i) \leq \tau$ for all $i \in [n]$. In particular, $\mathcal{A}_i$ is a subset of some path from a leaf to the root of $\mathcal{T}$.
\end{lemma}
\begin{proof}
	By construction, $\mathcal{A}_i$ form a clique in the dual graph $G_A$, and $tw(A)$ is lower-bounded by the size of the largest clique in $G_A$. By construction of the elimination tree, any clique in $G_A$ must lie on one path from a leaf to the root of $\mathcal{T}$.
\end{proof}

\begin{lemma}[{\cite[proposition 5]{Schreiber1982}}]\label{lem:L-column-sparsity-pattern}%[\cite{Schreiber1982}]
	$\mathcal{L}_i \subseteq \mathcal{P}(i)$ for each $i$. In particular, the height of the elimination tree satisfies $\tau \geq \max\{|\mathcal{L}_i| : i \in [d]\}$. \qed
\end{lemma}

As a corollary, this relation between the non-zero pattern of the columns of $L$ and $\mathcal{T}$ further allow us to characterize the non-zero pattern of the rows of $L$:

\begin{lemma}\label{lem:L-row}
	$\mathcal{L}^i \subseteq \mathcal{D}(i)$. \qed
\end{lemma}

\subsection{Solving Triangular Systems}

Now, we discuss the cost of solving triangular systems, in connection to the elimination tree $\mathcal{T}$.

	\begin{algorithm}[H]\caption{Solving $Lx = v$}\label{alg:L_solve}
		\begin{algorithmic}[1]
			\State $x \leftarrow \mathbf{0}_d$
			\For {increasing $j$ with $v_j \neq 0$}
			\State $x_j \leftarrow v_j/L_{jj}$
			\State $v \leftarrow v - x_j L_j$
			\EndFor
			\State \Return $x$
		\end{algorithmic}
	\end{algorithm}	
	
\begin{lemma} \label{lem:Linv-sparsity-time}
Let $x=L^{-1}v$, and let $\mathcal{S}$ be the non-zero pattern of $v$. Then, the nonzero pattern of $x$ is a subset of $\bigcup_{i\in \mathcal{S}}\mathcal{P}(i)$. Furthermore, we can solve for $L^{-1}v$ in $O(\|L^{-1}v\|_0 \cdot \tau)$ time.  In particular, if the non-zero pattern of $v$ is a subset of some path $\mathcal{P}$ from a leaf to the root in $\mathcal{T}$, then the non-zero pattern of $L^{-1}v$ is also a subset of $\mathcal{P}$, and we can solve for $L^{-1}v$ in $O(\tau^2)$ time. 
\end{lemma}
\begin{proof}
We prove the sparsity pattern by inspecting~\cref{alg:L_solve}. Note that the $x_i\neq 0$ if $v_i \neq 0$ or ($x_j \neq 0$ and $L_{ij} \neq 0$). \cref{lem:L-column-sparsity-pattern} shows that $L_{ij} \neq 0$ implies $i$ is an ancestor of $j$. Hence, the non-zeros in $x$ can only propagate to its ancestors from the non-zeros of $v$. As a result, the non-zero pattern of $x$ is a subset of $\bigcup_{i\in \mathcal{S}}\mathcal{P}(i)$.

For the runtime, we note that $L_j$ has $\tau$ non-zero entries and hence each step takes $O(\tau)$ time. Since the number of steps is exactly $\|L^{-1}v\|_0$, we have the runtime $O(\|L^{-1}v\|_0 \cdot \tau)$.
\end{proof}

\begin{lemma}\label{lem:L-inv-transpose-coordinate}
	For any $v$, we can solve for $(L^{-\T}v)_i$ in time $O(\tau^2)$.
\end{lemma}
\begin{proof}
	Note that $(L^{-\T}v)_i = e_i^\T L^{-\T} v$. By \cref{lem:Linv-sparsity-time}, computing $e_i^\T L^{-\T}$ takes $O(\tau^2)$ time and the resulting vector has $\tau$ sparsity, so the subsequent multiplication with $v$ also takes $\tau = O(\tau^2)$ time.
\end{proof}

\begin{lemma}\label{lem:Linv-path}
	Let $S \subseteq [d]$ be a subset of the vertices on some path $\mathcal{P}$ from a leaf to the root in $\mathcal{T}$. Then for any $y$, we can compute the subvector $(L^{-\top}y)|_S = y^\top L^{-1}|_S$ in $O(\tau^2)$ time, where $L^{-1}|_S$ denotes $L^{-1}$ restricted to the columns given by $S$.
\end{lemma}
\begin{proof} Let $S' = V(\mathcal{P})$, so we have $S \subseteq S'$. \cref{lem:Linv-sparsity-time} shows that $L^{-1}e_i$ is supported on $S'$ for any $i \in S'$. It follows that for any $i \in S$, we have $e_i^{\T} L^{-\T} y = y^{\T} L^{-1} e_i = y|_{S'}^{\T} (L^{-1} e_i)|_{S'}$. Hence, $(L^{-\T}y)|_{S}$ only depends on the entries of $y$ on $S'$. 

This allows us to write $(L^{-\T}y)|_S = (L^{-\T})|_{S'\times S'} y|_{S'} = (L_{S'\times S'})^{-\T} y|_{S'}$. Finally, we note that $L_{S'\times S'}^{\T}$ is a $(\tau+1) \times (\tau+1)$ upper triangular matrix and hence we can solve it in $O(\tau^2)$ time.

%	We can compute $L^{-1}|_S$ in $O(\tau^2)$ time. In Algorithm~\cref{alg:L_solve}, observe that entries of $L^{-1}_j$ is modified in two ways, relying only on entries from $L|_S$: 
%	\begin{enumerate}
%		\item in iteration $j$, where we update $L^{-1}_{k,j}$ if and only if $L_{k,j} \neq 0$,
%		\item in iteration $i > j$, where we update $L^{-1}_{k,j}$ if and only if $i,k \in \mathcal{P}(j)$ and $k$ is an ancestor of $i$.
%	\end{enumerate}
%	Hence, each column $L^{-1}_j$ is computed with $O(\tau)$ operations. Since $|S| \leq \tau$, computing $L^{-1}|S$ requires $O(\tau^2)$ time in total.
%	
%	Since $S$ lie on the path $\mathcal{P}$, we know that $\mathcal{L}^{-1}_i \subseteq \mathcal{P}$ for each $i \in S$. Hence, $y^\top L^{-1}|_S$ is simply a linear combination of at most $\tau$ non-zero rows of $L^{-1}|_S$, and can be computed in $O(\tau^2)$ time. 
\end{proof}

\subsection{Computing and Updating the Cholesky Factorization}

Next, we study the cost of computing and updating the Cholesky factorization. The crux for efficient implementation of sparse Cholesky factorization is that both the matrix $M$ and its Cholesky decomposition $M = L L^\top$ are sparse, and hence the operations involving $0$ can be skipped. There are many different algorithms for this; the following is one of them.
	\begin{algorithm}[H]\caption{Cholesky factorization of a matrix $M$}\label{alg:compute-chol}
		\begin{algorithmic}[1]
			\For {$j = 1$ to $d$}
			\State $L_{j,j} \leftarrow \sqrt{M_{j,j} - \sum_{k=1}^{j-1} L_{j,k}^2}$
			\For {$i = j+1$ to $d$}
			\State $L_{i,j} \leftarrow \frac{1}{L_{j,j}} \left(M_{i,j} - \sum_{k=1}^{j-1} L_{i,k} L_{j,k} \right)$
			\EndFor
			\EndFor
			\State \Return $L$
		\end{algorithmic}
	\end{algorithm}

By analyzing the number of non-zeros operations of the above algorithm (or other Cholesky factorization algorithms), one can show the following:
\begin{lemma}[{\cite[Theorem 2.2.2]{george1994computer}}]\label{lem:chol-time}
	For a positive definite matrix $M$, we can compute its Cholesky factorization $M = LL^\top$ in time 
	\[
	\Theta(\sum_{j=1}^d |\mathcal{L}_j|^2),
	\] 
	where $|\mathcal{L}_j|$ denotes the number of nonzero entries in the $j$-th column of $L$. \qed
\end{lemma}

\begin{corollary}\label{cor:chol-time}
	The Cholesky factorization $AH^{-1}_{\ox}A^\top = LL^\top$ can be computed in $O(n\tau^2)$ time.
\end{corollary}

\begin{proof}
	By the definition of tree height, for any vertex $i$ in the tree, the length of the path from $i$ to root is less than $\tau$. Then \cref{lem: nnz path} implies $|\mathcal{L}_j| \leq \tau$ for all $j$. Hence, it takes $O(n\tau^2)$ time to compute $AH^{-1}_{\ox}A^\top$ explicitly. Then, we can apply \cref{lem:chol-time} to compute Cholesky factorization and it takes time $O(d\tau^2) = O(n \tau^2)$.
\end{proof}

%why this is needed??
\begin{comment}
\begin{theorem}\label{thm: chol-update-time}
	Given a positive definite matrix $\=M \in \R^{n\times n}$, its elimination tree $\mathcal{T}$ of height $t$, and the corresponding Cholesky factorization $\=M = \=L\=L^\top$, for any $w \in \mathbb{R}^n$, we can compute the Cholesky factorization of $\=M+ww^\top$ in $\+O(t^2)$ time.
\end{theorem}
\begin{proof}
	By Lemma 4.2 of \cite{Davis2003} and \cref{lem: nnz path}.
\end{proof}

After an update $\=M + ww^\top$, the elimination tree height doesn't change
\end{comment}

The following two lemmas involve rank-1 updates of the Cholesky factorization, one regarding the sparsity pattern and one the update time. We state a simplified version of ~\cite{Davis2003}, which makes a further sparsity assumption. We include the proof of first lemma for intuition.
% I change the statement so that it can be also used for the block case
\begin{lemma}[{\cite[Section 5]{Davis2003}}]\label{lem:cholesky-update-sparsity-pattern}
Given a positive definite matrix $M \in \R^{d\times d}$, its elimination tree $\mathcal{T}$ of height $\tau$, and the corresponding Cholesky factorization $M = L L^\top$. Let $(L+\Delta L)(L+\Delta L)^\top$ be the new Cholesky factorization of $M + w w^\top$. Suppose that the sparsity pattern of $M$ and $M + w w^\top$ are same. If we let $\mathcal{S}$ be the index set of columns of $L$ that are updated, i.e. $\mathcal{S} =\{j \in [d] \mid \Delta L_j \neq \mathbf{0}\}$, then $\mathcal{S}$ is a subset of some path from $k$ to the root in $\mathcal{T}$ where $k$ is the first non-zero index in $w$. Consequently, the row and column sparsity of $\Delta L$ are bounded by $\tau$, and $\nnz(\Delta L) = O(\tau^2)$. 

The same holds for $M - w w^\top$ as long as $M - w w^\top$ is positive definite.
\end{lemma}
\begin{proof}

	Since $w w^\top$ is a clique in the graph associated with non-zeros of $M$, it is contained in a path from $k$ to the root in $\mathcal{T}$ where $k$ is the first non-zeros in $w$. Let $\mathcal{I}$ be the set of indices in this path. It follows that $M$ is only changed in the $\mathcal{I} \times \mathcal{I}$ block. Now, we run \cref{alg:compute-chol} twice, once on $M$ and once on $M + ww^\top$ and prove that the difference in $L$ is in the $\mathcal{I}\times \mathcal{I}$ block. The formulas in \cref{alg:compute-chol} show that the changes to $M$ and $L$ are propagated in the $L$ in the next step in the following ways:
\begin{itemize}
\item Updating the entry $M_{ij}$ causes $L_{ij}$ to update.\\
This case is good because we know $i,j\in \mathcal{I}$.
\item Updating the entry $L_{jk}$ causes $L_{jj}$ to update.\\
By induction, in the last step $L_{jk}$ is updated implies $j,k\in \mathcal{I}$.
Hence, $(j,j)\in \mathcal{I}\times \mathcal{I}$ (only entries in the $\mathcal{I}\times \mathcal{I}$ submatrix of $L$ are updated).
\item Updating the entry $L_{ik}$ and $L_{jk}\neq0$ causes $L_{ij}$
to update.\\
By induction, we know $i,k\in \mathcal{I}$. Since $L_{jk}\neq0$, \cref{lem:L-column-sparsity-pattern}
shows $j$ is on the path of $k$ to the root. Since $k\in \mathcal{I}$, we
have $j\in \mathcal{I}$. Hence, $(i,j)\in \mathcal{I}\times \mathcal{I}$ again.
\item Updating the entry $L_{jk}$ and $L_{ik}\neq0$ causes $L_{ij}$
to update.\\
Same argument as above.
\end{itemize}
In all the cases, the change in $L$ is restricted to the $\mathcal{I} \times \mathcal{I}$ submatrix.
	% The following proof is wrong.
%	It follows that $w w^\top$ is non-zero in a $\tau \times \tau$ submatrix, with the rows and columns corresponds to the non-zeros of $w$. The formulas above show that the changes to $M$ are propagated in $L$ in two ways: updating the entry $M_{j,j}$ causes the value of each non-zero entry of $L_j$ to update; and updating the entry $M_{i,j}$ causes $L_{i,j}$ and subsequently $L_{i,i}$ to update. Since the change in $M$ is restricted to a $\tau \times \tau$ square submatrix, with indices lying on a path in $\mathcal{T}$, we conclude that the change in $L$ is restricted to the same $\tau \times \tau$ submatrix.
\end{proof}

At a high level, since we know $L$ is changed in a $\tau \times \tau$ sized block, we only need to update the factorization on that block. Similarly to matrix inverse, there are simple algorithms for rank-1 update for factorization in time linear to the square of the dimension.

\begin{lemma}[{\cite[Section 5]{Davis2003}}]\label{lem:cholesky-update-time}
Given a positive definite matrix $M \in \R^{d\times d}$, its elimination tree $\mathcal{T}$ of height $\tau$, and the corresponding Cholesky factorization $M = L L^\top$. Let $(L+\Delta L)(L+\Delta L)^\top$ be the new Cholesky factorization of $M + w w^\top$. Suppose that the sparsity pattern of $M$ and $M + v v^\top$ are same. Then, we can compute $\Delta L$ in $O(\tau^2)$ time.

The same holds for $M - w w^\top$ as long as $M - w w^\top$ is positive definite.
\qed
\end{lemma}
% No proof is needed, this is exactly what they show

%\begin{lemma}\label{thm:chol-update-time}
%	Let $H$ be an $n \times n$ diagonal matrix with non-negative entries. Suppose $AHA^\T$ has Cholesky factorization $LL^\T$, and the corresponding elimination tree $\mathcal{T}$ has height $\tau$. Then for any $\lambda$ with $\lambda > -H_{ii}$, we can compute the Cholesky factorization and elimination tree of $A(H+\lambda e_i e_i^\T)A^\T$ in $O(\tau^2)$ time. Furthermore, the resulting elimination tree does not change.
%\end{lemma}
%\begin{proof}
%	First, note that since $AHA^\T$ and $A(H + \lambda e_i e_i^\T)A^\T$ have the same non-zero pattern, their Cholesky factors will also have the same non-zero pattern, and their corresponding elimination trees will be identical. Only the values of some entries in the Cholesky factors will differ. By Lemma 4.2 of~\cite{Davis2003}, the time to compute this is $O(\tau^2)$.
%\end{proof}
% \input{data-structure.tex}

\section{Robust Central Path Maintenance} \label{sec:data_structure}

In this section, we present a data structure $\textsc{CentralPathMaintenance}$
to efficiently perform the robust central path step needed in \cref{alg:IPM_framework}. 
Specifically, we will prove the following theorem.

\begin{restatable}[Robust Central Path Step]{theorem}{centralPathAlg}
	\label{thm:central-path-algo} 
	Suppose~\cref{alg:IPM_framework} is run on the convex program~\cref{eq:problem}.
	Given the constraint matrix $A \in\R^{d\times n}$ with block-diagonal structure $n=\sum_{i=1}^{m}n_{i}$, its binary elimination tree $\mathcal{T}$ of height $\tau$, 
	and parameters $\lambda, \overline{\eps}, \eps_t, \alpha, w = \mathbf{1}_m$ as defined in~\cref{alg:IPM_framework},
	the randomized data structure $\textsc{CentralPathMaintenance}$
	(\cref{alg:CentralPath,alg:CentralPath-1}) 
	implicitly maintains the central path primal-dual solution pair $(x,s)$ (\cref{alg:IPM_framework} \cref{line:move_xs})
	and explicitly maintains its approximation $\oxs$ (\cref{alg:IPM_framework} \cref{line:maintain_xs}) 
	using the following functions:
	\begin{itemize}
	\item $\textsc{Initialize}(x,s,t_{0},k)$: Initializes
	the data structure with initial primal-dual solution pair $(x,s)$, initial central path timestep $t_0$, and a runtime tuning parameter $k$ in $O(n\tau^{2}\log^4(n))$ time.
	\item $\textsc{MultiplyAndMove}(t)$: It implicitly maintains 
	\begin{equation}\label{eq:robust-central-path-update}
	\begin{aligned}
	x & \leftarrow x+H_{\ox}^{-1/2}(I-P_{\ox})H_{\ox}^{-1/2}\delta_{\mu}(\overline{x},\overline{s},\overline{t})
	\\
	s & \leftarrow s+tH_{\ox}^{1/2}P_{\ox}H_{\ox}^{-1/2}\delta_{\mu}(\overline{x},\overline{s},\overline{t})
	\end{aligned}
	\end{equation}
	where 
	$H_{\ox} \defeq \nabla^{2}\phi(\overline{x})$,
	$P_{\ox} \defeq H_{\ox}^{-1/2}A^{\top}(AH_{\ox}^{-1}A^{\top})^{-1}AH_{\ox}^{-1/2}$,
	and $\overline{t}$ is some earlier timestep satisfying $|t-\overline{t}|\leq \eps_t \cdot \overline{t}$.
	
	It also explicitly maintains $(\ox,\os)$ such that $\|\overline{x}_{i}-x_{i}\|_{\overline{x}_{i}}\leq\overline\eps$
	and $\|\os_{i}-s_{i}\|_{\ox_{i}}^{*}\leq t \overline\eps w_{i}$ for all $i\in[m]$
	with probability at least 0.9. 
	% otherwise we will have loglog(1/eps) 
	
	Assuming the function is called at most $N$ times and $t$ is monotonically decreasing from $t_{\max}$ to $t_{\min}$, the total running time is
	\[
		O\left(\left({Nn^{1/2}}+n{\log(t_\textnormal{max}/t_\textnormal{min})}\right)\tau^2\poly\log(N)\right).
	\]
	\item $\textsc{Output}$: It computes $(x,s)$ exactly and outputs them in $O(n\tau^{2})$ time.
	\end{itemize}
\end{restatable}

% This implies the run-time is independent of number of iteration in IPM. 
\begin{rem}
	The $N$ dependence in the runtime is a result of parameter tuning. 
	If the IPM takes more than $\wt O(\sqrt{n}\log(1/\eps))$ steps, 
	the data structure can still run in $\wt O(n \tau^2 \log(1/\eps))$ by choosing a larger value for the parameter $k$ in \textsc{Initialize}.
\end{rem}

\subsection{Multiscale Representation of the Central Path Dynamic}

In any call to \textsc{MultiplyAndMove}, we want to update the central path primal-dual solution pair according to~\cref{eq:robust-central-path-update}, as well as the approximation pair. Here, we introduce the multiscale representation used in these computations:

\begin{definition}[Multiscale Basis]
	At any step of the robust central path with approximate primal-dual solution pair $(\ox,\os)$, we define 
	\[
	\mathcal{W}\defeq L_{\ox}^{-1}AH_{\ox}^{-1/2}
	\]
	where $H_{\ox}=\nabla^{2}\phi(\overline{x})$ and 
	$L_{\ox}$ is the lower Cholesky factor of $AH_{\ox}^{-1}A^{\top}$.
\end{definition}
Intuitively, the basis element are rows of $\mathcal{W}$, which are represented by vertices in the elimination tree $\mathcal{T}$. 
Note that our data structure never computes or stores  $\mathcal{W}$ explicitly, as it is a costly operation.

\begin{definition}[Multiscale Coefficients]
	At any step of the robust central path with approximate primal-dual solution pair $(\ox,\os)$, we define 
	\[
		h\defeq L_{\ox}^{-1}AH_{\ox}^{-1} \delta_\mu(\ox,\os,\overline{t})
	\]
	where $H_{\ox}=\nabla^{2}\phi(\overline{x})$, and 
	$L_{\ox}$ is the lower Cholesky factor of $AH_{\ox}^{-1}A^{\top}$.
\end{definition}

Now, we can rewrite the central path update from~\cref{eq:robust-central-path-update} using the multiscale representation:
\begin{equation}\label{eq:central-path-multiscale-rep}
\begin{aligned}
x & \leftarrow x+H_{\ox}^{-1} \delta_\mu(\overline{x},\overline{s},\overline{t}) - H_{\ox}^{-1/2} \mathcal{W}^\top h \\
s & \leftarrow s+t H_{\ox}^{1/2} \mathcal{W}^\top h.
\end{aligned}
\end{equation}

\subsection{Implicit Representation of $(x,s)$}
For the first part of proof of \cref{thm:central-path-algo},
we demonstrate how to obtain an implicit representation of the robust central path pair $(x,s)$, using the explicitly maintained approximation pair $(\ox, \os)$. 
Rather than directly working with the expression in~\cref{eq:central-path-multiscale-rep},
we rewrite $(x,s)$ in terms of variables that admit sparse changes between consecutive steps in the central path, in order to more efficiently maintain them.

%In fact, when we fixed $t$, the central path can be represented by a slowly changing vector in $\ell_{0}$ norm,
%plus a slowly changing vector in $\ell_{0}$ norm under $\mathcal{W}$-basis. 
\begin{theorem} 
\label{thm:W-representation}
Given constraint matrix $A$ and its binary elimination tree $\mathcal{T}$ with height $\tau$, the data structure $\textsc{MultiscaleRepresentation}$
(\cref{alg:W-Representation,alg:W-Representation-1})
implicitly maintains the primal-dual pair $(x,s)$ as defined by~\cref{eq:central-path-multiscale-rep}, computable via the expression
\begin{equation}\label{eq:implicit-representation}
\begin{aligned}
{x} & ={\widehat{x}}+H_{\ox}^{-1/2}\beta_{x}{c_x} - H_{\ox}^{-1/2}\mathcal{W}^{\top}(\beta_{x}{h}+{\eps_{x}}) \\
{s} & ={\widehat{s}}+H_{\ox}^{1/2}\mathcal{W}^{\top}(\beta_{s}{h}+{\eps_{s}}),
\end{aligned}
\end{equation}
by maintaining the variables $\widehat{x}, \beta_x, c_x, \eps_x, \widehat{s}, \beta_s, \eps_s, h, H_{\ox}$ and $L_{\ox}$. Note that the variables $\eps_{x}$ and $\eps_{s}$ here denote the accumulated error of $\beta_x h$ and $\beta_s h$; they are not necessarily small. 

The data structure supports the following functions:
\begin{enumerate}
\item $\textsc{Initialize}(x,s,\ox,\os,\overline{t})$
: Initializes the data structure in $O(n\tau^{2})$ time, 
with initial value of the primal-dual pair $(x,s)$, its initial approximation $(\ox,\os)$, and initial approximate timestep $\overline{t}$.
\item $\textsc{Move}()$: Moves $(x,s)$ according to~\cref{eq:central-path-multiscale-rep} in $O(1)$ time by updating its implicit representation.
\item $\textsc{Update}(\new{\ox},\new{\os})$: Updates the approximation pair $\oxs$ to $(\new{\ox},\new{\os})$.

Let $\mathcal{S}=\{i\in[m]\mid\new{\ox_{i}}\neq\ox_{i}\text{ or }\new{\os_{i}}\neq\os_{i}\}$. 
Then each call to \textsc{Update} takes $O(|\mathcal{S}| \cdot \tau^{2})$ time, and each variable
in~\cref{eq:implicit-representation} except $\mathcal{W}$ changes in $O(|\mathcal{S}|\cdot \tau)$ many entries.
\end{enumerate}
\end{theorem}

\begin{algorithm}
	\caption{Multiscale Representation Data Structure - Initialize and Move
	\label{alg:W-Representation}}
	
	\begin{algorithmic}[1]
	
	\State \textbf{datastructure }$\textsc{MultiscaleRepresentation}$
	
	\State \textbf{private : member}
	
	\State \hspace{4mm} Constraint matrix $A\in\R^{d\times n}$, elimination tree $\mathcal{T}$ \Comment{Fixed global constants}
	
	\State \hspace{4mm} $\ox,\os\in\R^{n}$ \Comment{Approximate
	primal dual pair of $(x,s)$}
	
	\State \hspace{4mm} $H_{\ox} \in\oplus_{i\in[m]}\R^{n_{i}\times n_{i}}$
	\Comment{Hessian matrix $H_{\ox}=\nabla^{2}\phi(\overline{x})$}
	
	\State \hspace{4mm} $L_{\ox}\in\R^{d\times d}$ \Comment{Lower Cholesky
	factor of $AH_{\ox}A^{\top}$}
	
	\State \hspace{4mm}
	$\widehat{x},\widehat{s}, c_{x}\in\R^{n},\; \eps_{x},\eps_{s}, h \in\R^{d},\; \beta_{x},\beta_{s}\in\R$
	\Comment{Implicit representation of $(x,s)$ as in \cref{eq:implicit-representation}}
	
	\State \hspace{4mm} $\overline{\alpha}\in\R,\overline{\delta_{\mu}}\in\R^{n}$
	\Comment{Implicit representation of $\delta_{\mu}$ as in~\cref{inv:representation-invariant}}
	
	\State \hspace{4mm} 
	$\overline{t} \in \R_{+}$ 
	\Comment{Central path timestep parameter}
	
	\State \textbf{end members}
	
	\Procedure{$\textsc{Initialize}$}{$x\in\R^{n},s\in\R^{n},\ox\in\R^{n},\os\in\R^{n},\overline{t}\in\R_{+}$}
	
	\State $\ox\leftarrow\ox,\os\leftarrow\os, \overline{t} \leftarrow \overline{t}$
	
	\State $\widehat{x}\leftarrow x,\widehat{s}\leftarrow s$
	
	\State $\eps_{x}\leftarrow\boldsymbol{0},\eps_{s}\leftarrow\boldsymbol{0}$
	
	\State $\beta_{x}\leftarrow0,\beta_{s}\leftarrow0$
	
	\State $H_{\ox}\leftarrow\nabla^{2}\phi(\ox)$
	
	\State Find lower Cholesky factor $L_{\ox}$ where $L_{\ox}L_{\ox}^{\top}=AH_{\ox}^{-1}A^{\top}$
	using $\mathcal{T}$ \Comment{By~\cref{cor:chol-time}}
	
	\State $\textsc{\textsc{Initialize}}h(\ox,\os,H_{\ox},L_{\ox})$
	
	\EndProcedure

	\Procedure{$\textsc{Initialize}h$}{$\ox,\os,H_{\ox},L_{\ox}$}\Comment{\cref{lem:W-rep-initialize-h}}
	
	\For{$i\in[m]$}
	
	\State $(\overline{\delta_{\mu}})_{i}\leftarrow-\frac{\alpha \sinh(\frac{\lambda}{w_{i}}\gamma_{i}(\overline{x},\overline{s},\overline{t}))}{\gamma_{i}(\overline{x},\overline{s},\overline{t})}\cdot\mu_{i}(\ox,\os,\overline{t})$
	\State $\overline{\alpha}\leftarrow\overline{\alpha}+\alpha^{2}\cdot w_{i}^{-1}\cosh^{2}(\frac{\lambda}{w_{i}}\gamma_{i}(\ox,\os,\overline{t}))$
	\Comment{$\lambda,w,\gamma,\mu$ as defined in~\cref{alg:IPM_framework}}
	\EndFor
	
	\State $c_{x}\leftarrow H_{\ox}^{-1/2}\overline{\delta_{\mu}}$
	
	\State $h\leftarrow L_{\ox}^{-1}AH_{\ox}^{-1}\overline{\delta_{\mu}}$
	
	\EndProcedure
	
	\Procedure{$\textsc{Move}$}{}
	
	\State $\beta_{x}\leftarrow\beta_{x}+(\overline{\alpha})^{-1/2}$ \label{line:move_beta_x}
	
	\State $\beta_{s}\leftarrow\beta_{s}+\overline{t}\cdot(\overline{\alpha})^{-1/2}$ \label{line:move_beta_s}
	
	\EndProcedure
	
	\end{algorithmic}
	\end{algorithm}
	
	\begin{algorithm}
	\caption{Multiscale Representation Data Structure - Update \label{alg:W-Representation-1}}
	
	\begin{algorithmic}[1]
	
	\State \textbf{datastructure }$\textsc{MultiscaleRepresentation}$
	
	\Procedure{$\textsc{Update}$}{$\new\ox$,$\new\os$} \Comment{\cref{lem:W-update}}
	
	\State $\new H\leftarrow\nabla^{2}\phi(\new\ox)$
	
	\State Find lower Cholesky factor $\new L$ where $\new L(\new L)^{\top}=A\new HA^{\top}$
	\Comment{\cref{lem:cholesky-update-time}}
	
	\State $\textsc{Update}h(\new\ox,\new\os,\new H, \new L)$

	\State $\textsc{Update}\mathcal{W}(\new L,\new H)$
	
	\State $\ox\leftarrow\new\ox,\; \os\leftarrow\new\os$
	
	\State $H_{\ox}\leftarrow\new H,\; L_{\ox}\leftarrow\new L$
	
	\EndProcedure
	
	\Procedure{$\textsc{Update}h$}{$\new\ox,\new\os,\new H, \new L$} \Comment{\cref{lem:W-rep-update-h}}
	
	\State $S\leftarrow\{i\in[m]\mid\new{\ox_{i}}\neq\ox_{i}\text{ or }\new{\os_{i}}\neq\os_{i}\}$
	
	\State $\new{\overline{\alpha}}\leftarrow\overline{\alpha},\; \new{\overline{\delta_{\mu}}}\leftarrow\overline{\delta_{\mu}}$
	
	\For{$i\in S$}
	
		\State \Comment{$\lambda,w,\gamma,\mu$ as defined in~\cref{alg:IPM_framework}}
		\State $\new{\overline{\alpha}}\leftarrow\new{\overline{\alpha}}-\alpha^{2}\cdot w_{i}^{-1}\cosh^{2}(\frac{\lambda}{w_{i}}\gamma_{i}(\ox,\os,\overline{t}))+\alpha^{2}\cdot w_{i}^{-1}\cosh^{2}(\frac{\lambda}{w_{i}}\gamma_{i}(\new\os,\new\ox,\overline{t}))$
	
		\State $(\new{\overline{\delta_{\mu}}})_{i}\leftarrow-\frac{\alpha\sinh(\frac{\lambda}{w_{i}}\gamma_{i}(\new{\overline{x}},\new{\overline{s}},\overline{t}))}{\gamma_{i}(\new{\overline{x}},\new{\overline{s}},\overline{t})}\cdot\mu_{i}(\new{\overline{x}},\new{\overline{s}},\overline{t})$
	\EndFor
	
	\State $\new{c_{x}}\leftarrow(\new H)^{-1/2} \new{\overline{\delta_{\mu}}}$
	
	\State $\new h\leftarrow (\new L)^{-1}A(\new H)^{-1}\new{\overline{\delta_{\mu}}}$
	
	\State $\new{\eps_{x}}\leftarrow\eps_{x}+\beta_{x}(\new h-h)$ \label{line:multiscale-update-19}
	
	\State $\new{\eps_{s}}\leftarrow\eps_{s}+\beta_{s}(\new h-h)$
	
	\State $\new{\widehat{x}}\leftarrow\widehat{x}+\beta_{x}(H_{\ox}^{-1/2}c_{x}-(\new H)^{-1/2}\new{c_{x}})-(H_{\ox}^{-1/2}-(\new H)^{-1/2})\mathcal{W}^{\top}(\beta_{x}h+\eps_{x})$
	
	\State $\new{\widehat{s}}\leftarrow\widehat{s}+(H_{\ox}^{1/2}-(\new H)^{1/2})\mathcal{W}^{\top}(\beta_{s}h+\eps_{s})$
	\label{line:updateh-new}
	
	\State $c_{x}\leftarrow\new{c_{x}},h\leftarrow\new h$
	
	\State $\eps_{x}\leftarrow\new{\eps_{x}},\eps_{s}\leftarrow\new{\eps_{s}}$
	\State $\widehat{x}\leftarrow\new{\widehat{x}},\widehat{s}\leftarrow\new{\widehat{s}}$
	
	\EndProcedure
	
	\Procedure{$\textsc{Update}\mathcal{W}$}{$\new L,\new H$} \Comment{\cref{lem:W-rep-updateW}}
	
	\State $\new{\widehat{x}}\leftarrow\widehat{x}-(\new H)^{-1/2}((\new H)^{-1/2}-H_{\ox}^{-1/2})A^{\top}L^{-\top}(\beta_{x}h+\eps_{x})$
	
	\State $\new{\widehat{s}}\leftarrow\widehat{s}-(\new H)^{1/2}((\new H)^{-1/2}-H_{\ox}^{-1/2})A^{\top}L^{-\top}(\beta_{s}h+\eps_{s})$
	
	\State $\new \eps_{x}\leftarrow \eps_{x}+ (\new L-L)^{\top}L^{-\top}(\beta_{x} h+\eps_{x})$
	\State $\new \eps_{s}\leftarrow \eps_{s}+ (\new L-L)^{\top}L^{-\top}(\beta_{s} h+\eps_{s})$
	
	\State $\widehat{x}\leftarrow\new{\widehat{x}}, \widehat{s}\leftarrow\new{\widehat{s}}$
	
	\State $\eps_{x} \leftarrow \new{\eps_{x}},\eps_{s} \leftarrow \new{\eps_{s}}$
	\EndProcedure
	
	\end{algorithmic}
	\end{algorithm}
	
\subsubsection*{Proof of~\cref{thm:W-representation}}
We prove the correctness and running time for each operation
of $\textsc{MultiscaleRepresentation}$, and that they respect~\cref{inv:representation-invariant}. The correctness of the implicit representation in~\cref{eq:implicit-representation} then follows immediately. 

\begin{invariant}
\label{inv:representation-invariant}
After the data structure $\textsc{MultiscaleRepresentation}$ is initialized, the correct central path pair $(x,s)$ is always implicitly maintained and can be computed according to~\cref{eq:implicit-representation}.
Moreover, the following additional invariants are maintained:
\begin{align*}
	\overline{\alpha} &= \sum_{j=1}^m w_j^{-1}\cosh^2(\frac{\lambda}{w_j}\gamma_i(\ox,\os,\overline{t})) \tag{i}\\
	\overline{\delta_\mu} &= \overline{\alpha}^{1/2}\cdot\delta_\mu(\ox,\os,\overline{t})\tag{ii}\\
	c_x &= H_{\ox}^{-1/2} \overline{\delta_\mu} \tag{iii}\\
	h &= L_{\ox}^{-1}AH_{\ox}^{-1}\overline{\delta_{\mu}}. \tag{iv} 
\end{align*}
\end{invariant}

\begin{lemma}[$\textsc{Initialize}$]
\label{lem:W-initialize}The data structure $\textsc{MultiscaleRepresentation}$
takes $O(n\tau^2)$ to initialize. Moreover, \cref{inv:representation-invariant} is satisfied after initialization.
\end{lemma}

\begin{proof}
\textbf{Proof of Correctness:} 
We initialize $\widehat{x}$ to $x$ and $\widehat{s}$
to $s$ and all other terms from~\cref{eq:implicit-representation} to zero. Hence, this is the correct representation. Next, we call the helper function \textsc{Initialize}$h$, and the remainder of~\cref{inv:representation-invariant} is guaranteed by~\cref{lem:W-rep-initialize-h}.

\textbf{Proof of Runtime: }Since $n_{i}=O(1)$ for all $i\in[m]$,
we can compute $\nabla^{2}\phi(\ox)$ in $O(n)$ time. By \cref{cor:chol-time}, we can find the lower Cholesky factor in $O(n\tau^{2})$
time. 
By \cref{lem:W-rep-initialize-h}, $\textsc{Initialize}h$ takes $O(n\tau^2)$ time. 
Hence, the initialization takes $O(n\tau^{2})$ time.
\end{proof}

\begin{lemma}[$\textsc{Initialize}h(\ox,\os,H_{\ox},L_{\ox})$]\label{lem:W-rep-initialize-h}
	Given approximate central path pair $(\ox,\os)$, the Hessian matrix $H_{\ox}=\nabla^2 \phi(\ox)$, 
	and lower Cholesky factor $L_{\ox}$, the data structure takes $O(n\tau^2)$ time to perform $\textsc{Initialize}h$. 
	Moreover, (i)-(iv) of \cref{inv:representation-invariant} are satisfied after initialization.
\end{lemma}
\begin{proof}
	\textbf{Proof of Correctness:} The invariants directly follow from the definition.
	
	\textbf{Proof of Runtime:} Since $n_i=O(1)$ for all $i\in [m]$, each iteration of the for-loop takes $O(1)$ time. 
	Then, it takes $O(n)$ time to compute $\overline{\alpha}$ and $\overline{\delta_\mu}$. Since $H_{\ox}$ is a block-diagonal matrix, we can compute
	$c_x=H_{\ox}^{-1/2}\overline{\delta_\mu}$ in $O(n)$ time. Finally, $L_{\ox}^{-1}AH_{\ox}^{-1}\overline{\delta_\mu}$ can be computed in time $O(n\tau^2)$ by 
	\cref{lem:A-sparsity-pattern,lem:Linv-sparsity-time}.
\end{proof}

\begin{lemma}[$\textsc{Move}$] 
\label{lem:W-move}Under \cref{inv:representation-invariant},
the data structure $\textsc{MultiscaleRepresentation}$
takes $O(1)$ time to move the current central path pair $(x,s)$ by one step according to~\cref{eq:central-path-multiscale-rep}.
\begin{comment}
to
\begin{align*}
x\leftarrow x+H_{\ox}^{-1/2}(I-P_{\ox})H_{\ox}^{-1/2}\delta_\mu(\ox,\os,\overline{t}), \\
\quad s\leftarrow s+tH_{\ox}^{1/2}P_{\ox}H_{\ox}^{-1/2}\delta_\mu(\ox,\os,\overline{t}).
\end{align*}
\end{comment}
Moreover,~\cref{inv:representation-invariant}
is preserved afterwards.
\end{lemma}

\begin{proof}
\textbf{Proof of Correctness: }Let $\new x,\new s$ be the updated values of $x,s$ after $\textsc{Move}$ is performed. We check that the implicit representation from~\cref{eq:implicit-representation} is indeed the correct expression for $\new x$ by comparing it to $x$:
\begin{align*}
\new x-x & =H_{\ox}^{-1/2}\overline{\alpha}^{-1/2}c_x -H_{\ox}^{-1/2}\mathcal{W}^{\top}(\overline{\alpha}^{-1/2}h)\\
 & =H_{\ox}^{-1/2}\overline{\alpha}^{-1/2}H_{\ox}^{-1/2}\overline{\delta_\mu}-H^{-1/2}\mathcal{W}^{\top}(\overline{\alpha}^{-1/2}L_{\ox}^{-1}AH_{\ox}^{-1}\overline{\delta_\mu})\\
 & =H_{\ox}^{-1}\delta_\mu(\ox,\os,\overline{t})-H_{\ox}^{-1/2}\mathcal{W}^{\top}L_{\ox}^{-1}AH^{-1}\delta_\mu(\ox,\os,\overline{t})\\
 &=H_{\ox}^{-1}\delta_\mu(\ox,\os,\overline{t})-H_{\ox}^{-1/2}\mathcal{W}^{\top}h,
\end{align*}
where the first step follows by \cref{line:move_beta_x}, the second by (3) and (4) of \cref{inv:representation-invariant}, the third by (2) of \cref{inv:representation-invariant}, 
and the fourth step follows by the definition of $h$. This difference is exactly as given in~\cref{eq:central-path-multiscale-rep}.

Similarly, for $\new s$, we have 
\begin{align*}
\new s-s & =H_{\ox}^{1/2}\mathcal{W}^{\top}(\overline{t}\cdot \beta_s)\\
& =\overline{t} H_{\ox}^{1/2}\mathcal{W}^\top (\overline{\alpha}^{-1/2}L_{\ox}AH_{\ox}^{-1}\overline{\delta_\mu})\\
& =\overline{t}H_{\ox}^{1/2}\mathcal{W}^\top L_{\ox}AH_{\ox}^{-1}\delta_\mu(\ox,\os,\overline{t})\\
&=\overline{t}H_{\ox}^{1/2}\mathcal{W}^\top h,
\end{align*}
exactly as given in~\cref{eq:central-path-multiscale-rep}. The first step follows from \cref{line:move_beta_s}, the second
and third steps from (4) and (2) of \cref{inv:representation-invariant}, and the fourth step from the definition of $h$.

\textbf{Proof of Runtime: } The operation only uses addition and taking square roots of real numbers. 
\end{proof}

\begin{lemma}[$\textsc{Update}(\new{\ox},\new{\os})$]
 \label{lem:W-update}Under \cref{inv:representation-invariant},
the data structure $\textsc{MultiscaleRepresentation}$
takes $O(|S|\cdot\tau^{2})$ time to move the approximation
pair $\oxs$ to $(\new{\ox},\new{\os})$, where $S=\{i\in[m]\mid\new{\ox_{i}}\neq\ox_{i}\text{ or }\new{\os_{i}}\neq\os_{i}\}$.
%It also updates the variables used in the implicit representation of $(x,s)$ such that
~\cref{inv:representation-invariant}
is preserved at the end of the function call.
 
Moreover, the total number of coordinate changes in the variables involved in the implicit representation is bounded by $O(|S|\cdot \tau)$.
\end{lemma}

\begin{proof}
We can update $\ox,\os$ trivially. Immediately afterwards, we must update $H_{\ox}$, $L$ and $h$ in the data structure, so they correspond correctly to $\new \ox$. As a result of these updates, \cref{eq:implicit-representation} will no longer hold, so we must then adjust the other variables $\widehat{x}, \widehat{s}, c_x, \eps_x, \eps_s, \beta_x, \beta_s$ used in the implicit representation, in order to restore the invariant. 
To simplify the presentation, we accomplish this via two helper functions, \textsc{Update}$h$ and \textsc{Update}$\mathcal{W}$.

By combining \cref{lem:W-rep-update-h,lem:W-rep-updateW}, we show that the implicit representation expression holds after all variables are updated. Furthermore, they show the required bound on the total number of coordinate changes in all the implicit representation variables.

For the runtime, we can compute $\new L$ in  $O(\tau^{2}\cdot(\|\new\ox-\ox\|_{0}+\|\new\os-\os\|_{0}))$ time by \cref{lem:cholesky-update-time}.
Furthermore, $\textsc{Update}h$ takes $O(|S| \cdot \tau^2)$ time.
For $\textsc{Update}\mathcal{W}$, we can split the update of $L$ into $|S|$ many rank-$1$ updates by updating $A((\new H_{i})^{-1}-H_{i}^{-1})A^{\top}$ in time
$O(\tau^{2})$ for each $i \in S$, where $H_i = \nabla^2 \phi_i(x_i)$. 
By \cref{lem:cholesky-update-sparsity-pattern}, the non-zero columns of $\Delta L\defeq \new L - L$ lie on a path of $\mathcal{T}$. 
Then, each call of $\textsc{Update}\mathcal{W}$ takes $O(\tau^{2})$ time by \cref{lem:W-rep-updateW}. 
Hence, the total runtime is $O(|S|\cdot\tau^{2})$.
\end{proof}

\begin{lemma}[$\textsc{Update}h(\new{\ox},\new{\os},\new{H},\new{L})$]\label{lem:W-rep-update-h}
	Under \cref{inv:representation-invariant}, given the new approximation pair $\new{\ox},\new{\os}$, $\new H=\nabla^{2}\phi(\new{\ox})$, and the lower Cholesky factor $\new L$ of $A(\new H)^{-1/2}A^{\top}$, the function \textsc{Update}$h$ 
	updates the implicit representation such that (i)-(iv) of~\cref{inv:representation-invariant} are preserved,
	and at the end of the function call, the central path pair are given by 
	\begin{align*}
		x & ={\widehat{x}}+(\new{H})^{-1/2}\beta_x c_{x}-(\new{H})^{-1/2}\mathcal{W}^{\top}(\beta_x h+\eps_{x}),\\
		s & =\widehat{s}+(\new{H})^{1/2}\mathcal{W}^\top(\beta_s h+ \eps_{s}).
	\end{align*} Moreover, it takes $O(|S|\cdot \tau^2)$ time to perform $\textsc{Update}h$ where $S=\{i\in[m]\mid \new{\ox}_i \neq \ox_i \text{ or }\new{\os}_i\neq \os_i\}$,
	and all the variables in \cref{eq:implicit-representation} change in at most $O(|S|\cdot \tau)$ many entries.
\end{lemma}

\begin{proof}
	\textbf{Proof of Correctness: }First, we check (i)-(iv) of  \cref{inv:representation-invariant}:
	For (i) and (ii), note that the values of $\mu_i$ and $\gamma_i$ only depend on $\ox_i$, $\os_i$ and $\overline{t}$,
	so it suffices to update only the entries of $\overline{\alpha}$ and $\overline{\delta_\mu}$ with indices in $S$. For (iii) and (iv), they are trivially satisfied by definition.

	After~\cref{line:updateh-new}, we have computed new versions of the variables $\widehat{x}, \widehat{s}, c_x, h,\eps_x,\eps_s$. For the implicit representation of $x$, they satisfy: 
	\begin{align*}
		&\quad \new{\widehat{x}}+(\new{H})^{-1/2}\beta_x\new{c_{x}}-(\new{H})^{-1/2}\mathcal{W}^{\top}(\beta_x\new{h}+\new{\eps_{x}}) \\
		&=\widehat{x}+\beta_{x}(H_{\ox}^{-1/2}c_{x}-(\new H)^{-1/2}\new{c_{x}})-(H_{\ox}^{-1/2}-(\new H)^{-1/2})\mathcal{W}^{\top}(\beta_{x}h+\eps_{x}) \\
		&\quad+(\new{H})^{-1/2}\beta_x\new{c_{x}}-(\new{H})^{-1/2}\mathcal{W}^{\top}(\beta_x\new{h}+\new{\eps_{x}}) \\ 
		&= \widehat{x} + \beta_x H_{\ox}^{-1/2}c_x -H_{\ox}^{-1/2}\mathcal{W}^\top(\beta_x h+\eps_x) \\ 
		&= x,
	\end{align*} 
	where the first step follows by the definition of $\new{\widehat{x}}$ and the second step follows by $\beta_x \new{h}+\new{\eps_x} = \beta_x h+\eps_x$ from~\cref{line:multiscale-update-19}. 
	The proof of implicit representation of $s$ is identical; we omit it here. 
	
	The remainder of the function updates the variables to their new versions, giving the desired conclusion of the lemma.
	 
	\textbf{Proof of Runtime: } Since $n_i=O(1)$ for all $i\in [m]$, it takes $O(|S|)$ time to compute $\new{\overline{\alpha}}$ and $\new{\overline{\delta_\mu}}$.
	Since $H$ is a block-diagonal matrix, it takes $O(|S|)$ time to compute $\new{c_x}$. 
	Similarly, it takes $O(|S|)$ time to compute $(\new{H})^{-1}\new{\overline{\delta_\mu}}$.
	We use $\hat{\mu}$ to denote $({H})^{-1}{\overline{\delta_\mu}}$ and $\new{\hat{\mu}}$ to denote  $(\new{H})^{-1}\new{\overline{\delta_\mu}}$.
	Then, we can compute $\new{h}$ by computing $\new{h}-h$. We note that 
	\begin{align*}
		\new{h} - h &= (\new L)^{-1}A(\new H)^{-1}\new{\hat{\mu}} -  L^{-1}AH^{-1}{\hat{\mu}} \\
		&=\underbrace{(\new L)^{-1}A(\new H)^{-1}\new{\hat{\mu}} -  L^{-1}A(\new H)^{-1}\new{\hat{\mu}}}_{\delta_h^{(1)}} + \underbrace{L^{-1}A(\new H)^{-1}\new{\hat{\mu}} -  L^{-1}AH^{-1}{\hat{\mu}}}_{\delta_h^{(2)}}.
	\end{align*}
	We note that $\delta_h^{(2)}$ can be computed in  $O(|S|\cdot \tau^2)$ time using \cref{lem:A-sparsity-pattern,lem:Linv-sparsity-time}. By definition of $\new L$, we can compute $\delta_h^{(1)}$ in $O(|S|\cdot \tau^2)$ time. Let $S_{L} \defeq \{i\in[d]\mid\new L_{i}\neq L_{i}\}.$ By \cref{lem:cholesky-update-sparsity-pattern}, $S_{L}$ can be covered by $|S|$ many paths on the elimination tree.  	This also shows $\new{h}-h$ has $O(|S|\cdot \tau)$ many non-zero entries. Hence, we can compute $\new{\eps_x}$ and $\new{\eps_s}$ in $O(|S|\cdot \tau)$ time.
	Finally, since $\nnz(H-\new{H})=O(|S|)$, we can compute $(H^{-1/2}-(\new{H})^{-1/2})\mathcal{W}^\top$ and $(H^{1/2}-(\new{H})^{1/2})\mathcal{W}^\top$ 
	in $O(|S|\cdot \tau^2)$ time by \cref{lem:A-sparsity-pattern,lem:Linv-sparsity-time}. 

	\textbf{Number of Coordinate Changes:} Recall that $\new{h}-h$ has $O(|S|\cdot \tau)$ many non-zero entries, so the number of coordinate changes in $\eps_x,\eps_s$ 
	is also bounded by $O(|S|\cdot \tau)$. 
	The number of coordinate changes in $\widehat{x}$ and $\widehat{s}$ is bounded by $O(|S|\cdot \tau)$ 
	since $(H^{-1/2}-(\new{H})^{-1/2})\mathcal{W}^\top$ and $(H^{1/2}-(\new{H})^{1/2})\mathcal{W}^\top$ both have $O(|S|\cdot \tau)$ many non-zero entries by 
	\cref{lem:A-sparsity-pattern,lem:Linv-sparsity-time}, and $\|\new{c_x}-c_x\|_0=O(|S|)$.
\end{proof}

\begin{lemma}[$\textsc{Update}\mathcal{W}(\new L,\new H)$]
\label{lem:W-rep-updateW} 
Given $\new H=\nabla^{2}\phi(\new{\ox})$,
the lower Cholesky factor $\new L$ of $A(\new H)^{-1/2}A^{\top}$, 
and current implicit representation of $(x,s)$ given by 
\begin{align*}
x & =\widehat{x}+(\new{H})^{-1/2}\beta_x c_{x}-(\new{H})^{-1/2}\mathcal{W}^{\top}(\beta_x h+\eps_{x})\\
s & =\widehat{s}+(\new{H})^{1/2}\mathcal{W}^{\top}(\beta_s  h+\eps_{s})
\end{align*}
$\textsc{Update}\mathcal{W}$ takes $O((|\mathcal{S}|+|\mathcal{S}_{L}|)\cdot\tau^{2})$ time to update the variables maintained by $\textsc{MultiscaleRepresentation}$, such that at the end of the function call, the central path pair is given by
\begin{align*}
x & =\new{\widehat{x}}+(\new{H})^{-1/2}\beta_x c_{x}-(\new{H})^{-1/2}(\new{\mathcal{W}})^{\top}(\beta_x h+\new{\eps_{x}})\\
s & =\new{\widehat{s}}+(\new{H})^{1/2}(\new{\mathcal{W}})^{\top}(\beta_s h +\new{\eps_{s}}),
\end{align*}
where $\new{\mathcal{W}}\defeq(\new L)^{-1}A(\new{H})^{-1/2}$,
$\mathcal{S} \defeq\{i\in[m]\mid\new{\ox_{i}}\neq\ox_{i}\text{ or }\new{\os_{i}}\neq\os_{i}\}$,
and $\mathcal{S}_{L} \defeq \{i\in[d]\mid\new L_{i}\neq L_{i}\}.$ 

Moreover, when viewed as a set of vertices in $\mathcal{T}$, if $\mathcal{S}_{L}$ can be covered 
by $O(|\mathcal{S}|)$ many paths in $\mathcal{T}$, 
then the running time of $\textsc{Update}\mathcal{W}$ is $O(|\mathcal{S}|\cdot\tau^{2})$ and the number of coordinate changes in $\widehat{x},\widehat{s},\eps_x$ and $\eps_s$ is bounded by $O(|\mathcal{S}|\cdot \tau)$.
\end{lemma}

\begin{proof}
\textbf{Proof of Correctness:}

First, we examine the reason behind the definition of $\new {\eps}_{x}$: We want to find $\new{\eps}_{x}$ such that
\[
(\new L)^{-\top}(\beta_x h+ \new \eps_{x}) =L^{-\top}(\beta_x h+\eps_{x}).
\]
Rearrange, we get
\begin{align*}
\new \eps_{x} &={(\new{L})}^{\top}L^{-\top}(\beta_x h+\eps_{x})-\beta_x h\\
& =(\new{L}-L+L)^{\top}L^{-\top}(\beta_x h+\eps_{x})- \beta_x h\\
& =(\new{L}-L)^{\top}L^{-\top}(\beta_x h+\eps_{x}) + L^\top L^{-\top}(\beta_x h + \eps_x) - \beta_x h\\
& =\eps_x + (\new{L}-L)^{\top}L^{-\top}(\beta_x h+ \eps_x).
\end{align*}

Now, we check the implicit representation of $x$. At the end of the function, we have
\begin{align*}
	&\quad \new{\widehat{x}}+(\new{H})^{-1/2}\beta_x c_{x}-(\new{H})^{-1/2}(\new{\mathcal{W}})^{\top}(\beta_x h+\new{\eps_{x}}) \\ 
	&=\new{\widehat{x}}+(\new{H})^{-1/2}\beta_x c_{x}-(\new{H})^{-1/2} (\new{H})^{-1/2}A^\top (\new{L})^{-\top} (\beta_x h+\new{\eps_x}) \\
	&=\new{\widehat x}+(\new{H})^{-1/2}\beta_x c_{x}-(\new{H})^{-1}A^\top L^{-\top}(\beta_x h + \eps_x) \\ 
	&=x,
\end{align*}
where the first step follows by definition of $\new{\mathcal{W}}$, the second step follows by the property of $\new{\eps}_x$ above, 
and the last step follows by definition of $\new{\widehat{x}}$.

The proofs for $\new \eps_s$ and $s$ are identical; we omit them here.

\textbf{Proof of Runtime: }
Note that $\Delta {\eps_x} \defeq \new \eps_x - \eps_x =((\beta_x h + \eps_x)^\top L^{-1}(\new{L}-L))^\top$.
Then, we can compute
$L^{-1}(\new L-L)$ in $O(|\mathcal{S}_{L}|\cdot\tau^{2})$ time by \cref{lem:L-column-sparsity-pattern,lem:Linv-sparsity-time}, and therefore compute $\Delta {\eps_x}$ in $O(|\mathcal{S}_{L}|\cdot\tau^{2})$ time.
By \cref{lem:A-sparsity-pattern,lem:Linv-sparsity-time}, we can compute $((\new{H})^{-1/2}-H^{-1/2})A^{\top}L^{-\top}$
in $O(|\mathcal{S}|\cdot\tau^{2})$ time, and the result has sparsity $O(|\mathcal{S}|\cdot \tau)$. Thus, we can compute $\new{\widehat{x}}$
and $\new{\widehat{s}}$ in $O(|\mathcal{S}|\cdot\tau^{2})$ time.
In total, the function runs in $O((|\mathcal{S}|+|\mathcal{S}_{L}|)\cdot\tau^{2})$ time.

%When $S_{L}$ lies on a path of $\mathcal{T}$, 
For each path $\mathcal{P}$,
we can directly compute $(L^{-\top}(\beta_x h +\eps_{x}))|_{\mathcal{P}}$ in time $O(\tau^{2})$ by \cref{lem:Linv-path}.
Then, it takes $O(\tau^{2})$ to compute $\Delta \eps_x$ for each path.
Hence, the update time in this case is bounded by $O(|\mathcal{S}|\cdot\tau^{2})$.

\textbf{Number of Coordinate Changes:}
By \cref{lem:A-sparsity-pattern,lem:Linv-sparsity-time}, $((\new{H})^{-1/2}-H^{-1/2})A^{\top}L^{-\top}$
has sparsity $O(|\mathcal{S}|\cdot \tau)$. For each path $\mathcal{P}$, the solution of $L^{-1}(\new{L}-L)$ is a 
$\tau\times\tau$ submatrix by \cref{lem:L-column-sparsity-pattern,lem:Linv-sparsity-time},
leading to $\new \eps_x - \eps_x$ and $\new {\eps}_s - \eps_s$ having sparsity $O(|\mathcal{S}|\cdot \tau)$. 
\end{proof}

\subsection{Approximating A Sequence of Vectors}
The central path maintenance involves a number of dynamic vectors, e.g. $\widehat{x},\widehat{s},c_x$ from \cref{eq:implicit-representation}. These can essentially be viewed as online sequences of vectors, where the sequence length is the number of central path steps. 
To work with these vector variables efficiently over the central path steps, we maintain their $\ell_\infty$-approximations. 

In this section, we introduce the techniques for obtaining $\ell_{\infty}$-approximations of an online sequence of vectors using a \emph{sampling tree} data structure, crucially avoiding reading the input vectors in full at all times to lower the runtime.
The underlying idea is standard in sampling, heavy-hitters, and sketching, see e.g.\ \cite{DBLP:journals/jal/CormodeM05}. We explain how it is used in the context of central path maintenance in subsequent sections.

\begin{definition}
\label{def:sampling-tree}A \emph{sampling tree} $(\mathcal{S},\chi)$ of $\mathbb{R}^n$ 
consists of a constant degree rooted tree $\mathcal{S}=(V,E)$
and a labelling of the vertices $\chi:V \to2^{[n]}$, such that:
\begin{itemize}
\item $\chi(\textrm{root})=[n]$,
\item If $v$ is a leaf node of $\mathcal{S}$, then $|\chi(v)|=1$,
\item For any node $v$ of $\mathcal{S}$, the set $\{\chi(c)\mid\text{$c$ is a child of $v$}\}$
forms a partition of $\chi(v)$.
\end{itemize}
\end{definition}

\begin{theorem}
	\label{thm:approximate-ell-infty}
	Given a sampling tree $(\mathcal{S},\chi)$ with height $\eta$, 
	some $0<\eps_{\mathrm{apx}}, \delta_{\mathrm{apx}}<1$,
	length of input sequence $k$,
	a fixed but unknown JL-matrix $\Phi\in\R^{r\times n}$ where $r=\Omega(\eta^2\log(nk/\delta_{\mathrm{apx}}))$,
	and upper bound $\zeta>0$ such that the sequence $\{y^{(\ell)}\}_{\ell=1}^k$ satisfies $\|y^{(\ell)}-y^{(\ell-1)}\|_{2}\leq \zeta$ for all $\ell \in [k]$,
	the data structure $\textsc{\ensuremath{\ell_{\infty}}-Approximates}$ (\cref{alg:approximate-ell_infty,alg:approximate-ell_infty-1}) supports $k$ calls to \textsc{Query}, such that:
	
	\item In the $\ell$-th call to \textsc{Query}, 
	the data structure can indirectly access $\{y^{(i)}\}_{i=1}^\ell$ using the list of oracles  $\{\mathcal{O}[y^{(i)}]\}_{i=1}^{\ell}$ as follows:
	\begin{itemize}
		\item[] $\mathcal{O}[y^{(i)}].\textsc{TypeI}(v)$: access to the vector $\Phi_{\chi(v)}y^{(i)}$ for node $v \in \mathcal{S}$,
		\item[] $\mathcal{O}[y^{(i)}].\textsc{TypeII}(j)$: access to entry $y^{(i)}_{j}$ for $j \in [n]$,
	\end{itemize}
	and returns $z^{(\ell)}$ such that  $\|z^{(\ell)}-y^{(\ell)}\|_{\infty}\leq\eps_{\mathrm{apx}}$ with probability at least $1 - \delta_{\mathrm{apx}}/k$.

    Over the entire input sequence, the data structure makes $O(\eta \cdot \zeta^2k^{2}/\eps_{\mathrm{apx}}^{2}\cdot\poly\log({nk\zeta}/({\eps_{\mathrm{apx}}\cdot \delta_{\mathrm{apx}}})))$
	type-I oracle calls and $O(\zeta^2k^{2}/\eps_{\mathrm{apx}}^{2}\cdot\poly\log({nk\zeta}/({\eps_{\mathrm{apx}}\cdot \delta_{\mathrm{apx}}})))$
	type-II oracle calls, 
	with $O(\eta\cdot r \cdot \zeta^2 k^{2}/\eps_{\mathrm{apx}}^{2}\cdot\poly\log({nk\zeta}/({\eps_{\mathrm{apx}}\cdot \delta_{\mathrm{apx}}})))$
	additional computation time.
	It maintains $\{z^{(\ell)}\}_{\ell=1}^{k}$ such that
	$\|z^{(\ell)}-y^{(\ell)}\|_{\infty}\leq\eps_{\mathrm{apx}}$ for all $\ell \in [k]$ with success probability at least $1-\delta_{\mathrm{apx}}$.
	\end{theorem}
	
	\begin{algorithm}
	\caption{$\ell_{\infty}$ Maintenance Data Structure -- Initialize and Query \label{alg:approximate-ell_infty}}
	
	\begin{algorithmic}[1]
	
	\State \textbf{data structure }$\textsc{\ensuremath{\ell_{\infty}}-Approximates}$
		
	\State \textbf{private : members}
	
	\State \hspace{4mm} Sampling tree $(\mathcal{S},\chi)$ 
	\Comment{Fixed global constant}
	
	\State \hspace{4mm} $\eps_{\mathrm{apx}},\delta_{\mathrm{apx}}\in(0,1)$ 
	\Comment{error parameter and failure probability parameter}
	
	\State \hspace{4mm} $\ell, k \in \N$ \Comment{counter and total length of sequence}

	\State \hspace{4mm} $\zeta \in \R_+$ \Comment{upper bound of $\|y^{(\ell+1)} - y^{(\ell)}\|_2$}
	
	\State \hspace{4mm} \textsf{list} $\{\mc{O}\{y^{(i)}\}\}_{i=0}^k$ \Comment{sequence of oracles of input vectors $y^{(i)}$}
	
	\State \hspace{4mm} \textsf{list} $\{z^{(i)}\}_{i=0}^k$ \Comment{sequence of approximations}

	\State \textbf{end members}
	
	\Procedure{$\textsc{Initialize}$}{$\mathcal{S}, \chi, \eps_{\mathrm{apx}}\in(0,1),\delta_{\mathrm{apx}}\in(0,1), \zeta \in \R_+, k\in \mathbb{N}$}
	
	\State $\ell\leftarrow0, k \leftarrow k$
	
	\State $\eps_{\mathrm{apx}}\leftarrow\eps_{\mathrm{apx}},\; \delta_{\mathrm{apx}}\leftarrow \delta_{\mathrm{apx}},\;
	\zeta \leftarrow \zeta$
	
	\State $y^{(0)} \leftarrow \boldsymbol{0}, z^{(0)}\leftarrow\boldsymbol{0}$
	
	\EndProcedure
	
	\Procedure{$\textsc{Query}$}{$\mc{O}[\new y]$}
	
	\State $\ell\leftarrow\ell+1$
	
	\State $\mc{O}[y^{(\ell)}] \leftarrow \mc{O}[\new y]$ \Comment{store new oracle to list}
	
	\State $z^{(\ell)}\leftarrow z^{(\ell-1)}$ \Comment{first set the approximation to be the same as the previous}
	
	\State $\mathcal{I}\leftarrow\emptyset$ 
	\Comment{set of indices $i$ where we may need to update $z^{(\ell)}_i$}
	
	\For{$j=0,1,\ldots,\lfloor\log\ell\rfloor$}
		\State $\mathcal{I}_{j}\leftarrow\emptyset$ 
		\If{$\ell\equiv0\mod\;2^{j}$} \label{line:start:num-of-oracle-queries}
			\For{$O(4^{j}\zeta^2/\eps_{\mathrm{apx}}^{2}\cdot\log^{3}k \cdot \log(nk\zeta/({\eps_{\mathrm{apx}}  \delta_{\mathrm{apx}}})))$ many
	times}  
				\State $\mathcal{I}_j\leftarrow\mathcal{I}_j\cup\{\textsc{Sample}(\ell-2^{j}+1,\ell)\}$
				\label{line:oracle-query}
	
			\EndFor
	
		\EndIf\label{line:end:num-of-oracle-queries2}
		\State $\mathcal{I} \leftarrow \mathcal{I} \cup \mathcal{I}_j$
	
	\EndFor
	
	\ForAll{$i \in \mathcal{I}$}
		\State $z_{i}\leftarrow \mc{O}[y^{(\ell)}].\textsc{TypeII}(i)$
		\label{line:linf_approx_28}
		\If{$|z_{i}-(z^{(\ell)})_{i}|>\eps_{\mathrm{apx}}$} \Comment{Set $z_{i}^{(\ell)}\leftarrow y_{i}^{(\ell)}$
	when error is larger than $\eps_{\mathrm{apx}}$}
	
			\State $(z^{(\ell)})_{i}\leftarrow z_{i}$ 
			\label{line:linf_approx_30}
		\EndIf
	
	\EndFor
	
	\State \Return $z^{(\ell)}$ 	
	
	\EndProcedure
	
	\end{algorithmic}
	\end{algorithm}

	\begin{algorithm}
		\caption{$\ell_{\infty}$ Maintenance Data Structure -- Sample and Estimate \label{alg:approximate-ell_infty-1}}
		
		\begin{algorithmic}[1]
		
		\Procedure{$\textsc{Sample}$}{$a\in[k],b\in[k]$} \label{line:ell_infty-sample}
		
		\Repeat \label{line:repeat-start}
		
			\State $v\leftarrow\mathrm{root}(\mathcal{S}),p\leftarrow 1$
		
			\While{$v$ is not a leaf node} \label{line:descent-start}
		
				% \State $s\leftarrow \sum_{\text{child $c$ of $v$}}\textsc{Estimate}(a,b,c)$
			
				% \State Sample $q$ uniformly at random from interval $[0,s]$
		
				\State Sample child $v'$ of $v$ with probability $\displaystyle p_{v'} \defeq \frac{\textsc{Estimate}(a,b,v')}
					{\sum_{u \text{ a child of } v} \textsc{Estimate}(a,b,u)}$
				\label{line:linf_approx_sample}
				\State $\displaystyle p \leftarrow p \cdot p_{v'}$ 
	
				\State $v\leftarrow v'$
		
			\EndWhile \label{line:descent-end}
			\State \Comment{Since $v$ is a leaf node, $\chi(v)$ consists of a single index from $[n]$}
			\State \Comment{For notational purposes, suppose $\chi(v) = \{i\}$}
			\State $y^{(a)}_i\leftarrow \mc{O}[y^{(a)}].\textsc{TypeII}(i)$ \label{line:linf_approx_44}
			\State $y^{(b)}_i\leftarrow \mc{O}[y^{(b)}].\textsc{TypeII}(i)$ \label{line:linf_approx_45}
		
			\State
			with probability $(y^{(a)}_i-y^{(b)}_i)^{2}/(10\cdot p\cdot \textsc{Estimate}(a,b,\operatorname{root}(\mathcal{S})))$, \Return $i$ 
		
		\Until{\textbf{false}} \label{line:repeat-end}

		\EndProcedure
		
		\Procedure{$\textsc{Estimate}$}{$a\in[k],b\in [k],v\in \mathcal{S}$}
		% \State $v_{a}\leftarrow\mathcal{O}[a].\textsc{TypeI}(v),v_{b}\leftarrow\mathcal{O}[b].\textsc{TypeI}(v)$
		
		\State \Return $\|\mc{O}\{y^{(a)}\}.\textsc{TypeI}(v)-\mc{O}\{y^{(b)}\}.\textsc{TypeI}(v)\|_{2}^{2}$
		\EndProcedure
		\end{algorithmic}
		\end{algorithm}

	For any vector $y$ in the sequence, we show that $\Phi_{\chi(v)}y$ allows us to obtain a $(1\pm \frac{1}{\eta})$-approximation of $\|y|_{\chi(v)}\|_2^2$.
	With such estimations, we can sample a coordinate of $y$ with probability proportional to $y_i^2$ using $O(\eta)$ many oracle calls, using 
	a random descent on the sampling tree, where we choose each child with probability proportional to their estimation. 
	This further enables us to obtain a $(1\pm \eps)$-approximation of $y$ in the $\ell_\infty$-norm using
	$O({\|y\|_2^2}/{\eps^2}\log(\|y\|_2/\eps))$ type-II oracle calls by the coupon collection problem.
	By linearity of $\Phi$, this then allows us to approximate $y^{(b)}-y^{(a)}$. 

	Instead of directly estimating $y^{(\ell)}$ for each $\ell$, we obtain a $\ell_\infty$-approximation for all $y^{(b)}-y^{(a)}$, where $[a,b] \defeq \{a,a+1,\dots, b-1,b\}$ is in the set of \emph{dyadic intervals} of $[k]$.
	\begin{definition}
		Let $k$ be a positive integer. The set of \emph{dyadic intervals} of $[k]$ is
		\[
			\{[i\cdot 2^j+1,(i+1)\cdot 2^j] \mid i,j\in\mathbb{N};(i+1)\cdot 2^j\leq k\}.
		\]
	\end{definition}
	The following lemma tells us why dyadic intervals help to keep the error sub-linear to the size of the intervals.
	\begin{lemma}[folklore]\label{lem:dyadic-interval-cover}
		Any interval $[a,b]$ in $[k]$ can be partitioned into at most $2\log k$ dyadic intervals.
	\end{lemma}
	Hence, it suffices to obtain a $(1\pm\frac{\eps}{2\log k})$ approximation for every dyadic interval.

	Before we prove \cref{{thm:approximate-ell-infty}}, we show that the function $\textsc{Sample}(a,b)$ in the data structure 
	indeed samples a coordinate $i$ of $(y^{(b)}-y^{(a)})$ with probability proportional to $(y^{(b)}-y^{(a)})_{i}^{2}$.
	
	\begin{lemma}[$\textsc{Sample}$]\label{lem:ell-infty-sample}
	Under the same setting as \cref{thm:approximate-ell-infty}, 
	conditioned on the function $\textsc{Estimate}(a,b,v)$ always returns $\|\Phi_{\chi(v)}(y^{(a)}-y^{(b)})\|_2^2 = (1\pm \frac{1}{2\eta})\|(y^{(a)}-y^{(b)})|_{\chi(v)}\|_2^2$,
	the function $\textsc{Sample}(a,b)$ on \cref{line:ell_infty-sample} samples a coordinate
	$i$ proportional to $(y^{(b)}-y^{(a)})_{i}^{2}$ with
	$O(\eta \cdot r)$ expected running
	time, 
	and makes $O(\eta)$ many type-I oracle calls and $O(1)$ many type-II oracle calls in expectation.
	\end{lemma}
	
	\begin{proof}
	\textbf{Proof of Correctness:}
	Let $\delta_{y}$ denote $y^{(b)}-y^{(a)}$ in this proof.
	
	%First, observe that the helper function \textsc{Estimate}$(a,b,v)$ returns $\norm{\Phi_{\chi(v)}\delta_y}_2^2$.
	
	Let $v_{1},v_{2},\ldots,v_{m}\in V$
	be the sequence of nodes visited in the while-loop on~\crefrange{line:descent-start}{line:descent-end}, where $v_{1}$ is the root node of $\ts$
	and $v_{m}$ is the leaf node with $\chi(v_{m})=\{i\}$, for some $i \in [n]$ and $m \leq \eta$. 
	Observe that at end of the while-loop, $p$ is exactly the probability that $v_m$ is sampled. Hence, the algorithm outputs $v_m$ with probability
	\[
		\Pr[i\text{ is outputed}] = \Pr[v_m \text{ is sampled}] \cdot \Pr[i\text{ is returned}\mid v_m] = p\cdot \frac{(y^{(a)}_i-y^{(b)}_i)^2}{10\cdot p \|\Phi\delta_y\|_2^2}=\frac{\delta_{y,i}^2}{10 \cdot \|\Phi\delta_y\|_2^2}.
	\]
	This shows  $\textsc{Sample}$ outputs a coordinate $i$ with probability
	proportional to $\delta_{y,i}^{2}$.
	
	\textbf{Proof of Runtime}: 
	Consider one iteration of the repeat-loop  (\crefrange{line:repeat-start}{line:repeat-end}), 
	where the inner while-loop visits the node sequence $v_1,\dots, v_m$. 
	The probability to choose child $v_{j+1}$ from node $v_{j}$ is \\ $\|\Phi_{\chi(v_{j+1})}\delta_{y}\|_{2}^{2}/\sum_{\text{$u$ a child of $v_j$}}\|\Phi_{\chi(u)}\delta_{y}\|_{2}^{2}.$ By our condition on the function $\textsc{Estimate}(a,b,v)$, we have $\sum_\text{$c$ a child of $v_j$} \|\Phi_{\chi(u)}\delta_y\|_2^2 = (1\pm \frac{1}{\eta}) \|\Phi_{\chi(v_j)}\delta_{y}\|_2^2$. Hence,
	\[
		\Pr[v_{j+1}\mid v_j] \leq \left(1+\frac{1}{\eta}\right)\frac{\|\Phi_{\chi(v_{j+1})}\delta_y\|_2^2}{\|\Phi_{\chi(v_{j})}\delta_y\|_2^2}.
	\] 
	
	Taking the telescoping product, we have 
	\[
		p\leq (1+\frac{1}{\eta})^\eta \frac{\|\Phi_{\chi(v_{2})}\delta_{y}\|_{2}^{2}}{\|\Phi_{\chi(v_{1})}\delta_{y}\|_{2}^{2}}\times\frac{\|\Phi_{\chi(v_{3})}\delta_{y}\|_{2}^{2}}{\|\Phi_{\chi(v_{2})}\delta_{y}\|_{2}^{2}}\times\ldots\times\frac{\|\Phi_{\chi(v_{m})}\delta_{y}\|_{2}^{2}}{\|\Phi_{\chi(v_{m-1})}\delta_{y}\|_{2}^{2}} \leq 3\cdot \frac{\|\Phi_{\{i\}}\delta_{y}\|_{2}^{2}}{\|\Phi\delta_{y}\|_{2}^{2}}.
	\]
	This iteration of the repeat-loop returns with probability 
	\[
		\frac{\delta_{y,i}^2}{10p\|\Phi\delta_{y}\|_{2}^{2}} \geq \frac{\delta_{y,i}^2}{10\|\Phi\delta_{y}\|_{2}^{2}} \cdot \frac{\|\Phi\delta_{y}\|_{2}^{2}}{3\|\Phi_{\{i\}}\delta_{y}\|_2^2} \geq \frac{1}{50},
	\]
	where we used $\|\Phi_{\{i\}}\delta_{y}\|_2^2 \leq (1+\frac{1}{2\eta})\delta_{y,i}^2$.  Hence, the expected number of iteration is $O(1)$.
	
	For each iteration of the repeat-loop, the inner while-loop on~\crefrange{line:descent-start}{line:descent-end} traverses a path from the root to a leaf node in $\mathcal{S}$, so we can bound the number of iterations by $\eta$. At each node $v_j$ along the path with $c$ children, sampling on~\cref{line:linf_approx_sample}  requires $O(r)$ time and $2c$ many type-I queries. At the end of the descent to a leaf, we make 2 type-II queries on~\cref{line:linf_approx_44,line:linf_approx_45} and then 2 type-I queries. 
	Hence, each call of \textsc{Sample} takes $O(\eta \cdot r)$ time, $O(\eta)$ type-I queries, and $O(1)$ type-II queries in expectation.
	\end{proof}
	
	\begin{proof}[Proof of \cref{thm:approximate-ell-infty}]
	\textbf{Proof of Runtime:} For a fix $j \in [ \lfloor \log k\rfloor]$, the data structure makes
	$O(4^{j}\zeta^2\cdot\log^{3}k\cdot\log(nk\zeta/(\eps_{\mathrm{apx}}\cdot \delta_{\mathrm{apx}}))/\eps_{\mathrm{apx}}^{2})$ many calls to $\textsc{Sample}$, for each $2^{j}$ calls to $\textsc{Query}$.
	Since there are $k$ calls to \textsc{Query} in total, the total number of \textsc{Sample} call is 
	\[
	\sum_{j=1}^{\lfloor\log k\rfloor}\frac{k}{2^{j}}\cdot O\left(\frac{4^{j}\zeta^2}{\eps_{\mathrm{apx}}^2}\cdot\log^{3}k \cdot \log(\frac{nk\zeta}{\eps_{\mathrm{apx}}\cdot \delta_{\mathrm{apx}}})\right)
	=O\left(\frac{k^2\zeta^2\cdot \log^{3}k \cdot \log(nk\zeta/(\eps_{\mathrm{apx}}\delta_{\mathrm{apx}}))}{\eps_{\mathrm{apx}}^2}\right).
	\]
	Combining the bound above and \cref{lem:ell-infty-sample}, the total runtime is $O(\frac{\eta r k^{2}\zeta^2}{\eps_{\mathrm{apx}}^{2}}\poly\log({nk\zeta}/({\eps_{\mathrm{apx}}\cdot \delta_{\mathrm{apx}}})))$ in expectation, 
	with $O(k^{2}\zeta^2/\eps_{\mathrm{apx}}^{2}\cdot \poly\log({nk\zeta}/({\eps_{\mathrm{apx}}\cdot \delta_{\mathrm{apx}}}))))$ type-II queries and $O(\eta k^{2}\zeta^2/\eps_{\mathrm{apx}}^{2}\cdot \poly\log({nk\zeta}/({\eps_{\mathrm{apx}}\cdot \delta_{\mathrm{apx}}})))$ type-I queries in expectation.

	\textbf{Proof of Correctness}: 
	By the coupon collection problem, for any vector $v$, 
	we can find all coordinates $i$ of $v$ such that $|v_{i}|\geq\frac{1}{a}\|v\|_{2}$ with high probability, 
	by sampling $O(a^{2}\log a)$ many coordinates, each time sampling coordinate $i$ with probability proportional to $v_i^2$. 
	By our choice of $r=\Omega(\eta^2\log(nk/\delta_{\mathrm{apx}}))$ and union bound over all type-I queries, 
	we have the function $\textsc{Estimate}(a,b,v)$ always return $\|\Phi_{\chi(v)}(y^{(a)}-y^{(b)})\|_2^2 = (1\pm \frac{1}{2\eta})\|(y^{(a)}-y^{(b)})|_{\chi(v)}\|_2^2$
	with probability at least $1-\frac{\delta_{\mathrm{apx}}}{2}$.
	
	Fix $j \in [\lfloor{\log \ell}\rfloor]$. We sample $O(4^{j}\zeta^2 \cdot\log^{3}k \cdot \log(nk\zeta/(\eps_{\mathrm{apx}}\delta_{\mathrm{apx}}))/\eps_{\mathrm{apx}}^{2})$
	many coordinates for $y^{(\ell)}-y^{(\ell-2^{j})}$ (\crefrange{line:start:num-of-oracle-queries}{line:end:num-of-oracle-queries2}). 
	By the triangle inequality, and the property that consecutive $y^{(\ell)}$'s change slowly, 
	we have $\|y^{(\ell)}-y^{(\ell-2^{j})}\|_{2}\leq2^{j}\cdot\zeta$. 
	Hence, $\mathcal{I}_{j}$ contains all coordinates $i$ such that $|(y^{(\ell)}-y^{(\ell-2^{j})})_{i}|>O(\eps_{\mathrm{apx}}/\log k)$
	with success probability $1-\delta_{\mathrm{apx}}/\poly(nk)$. 
	
	Taking the union bound over all $j$, we have that $\mathcal{I} = \bigcup \mathcal{I}_j$ contains all coordinates $i$ such that $|(y^{(\ell)}-y^{(\ell-2^{j}+1)})_{i}|>O(\eps_{\mathrm{apx}}/\log k)$ for all $\ell\in[k]$ and $j\in [ \lfloor \log k\rfloor]$ with probability at least $1-\frac{\delta_{\mathrm{apx}}}{2}$. 
	
	Consider the data structure at the end of iteration $\ell$. Fix a coordinate $i$, and let $\ell_{i}$ be the iteration when the value of $z_{i}$ was set by \cref{line:linf_approx_28,line:linf_approx_30}. In other words, $z_{i}^{(\ell)} = y_{i}^{(\ell_{i})}$. Let $\delta_{y,i}^{(t)} = (y^{(t)} - y^{(t-1)})_i$ for any $t \in [k]$. Then,
	\[
	|y^{(\ell)}_{i}-z^{(\ell)}_{i}|=|y_{i}^{(\ell)}-y_{i}^{(\ell_{i})}|
	=\Bigg\vert\sum_{t=\ell_{i}+1}^{\ell}\delta_{y,i}^{(t)}\Bigg\vert=\Bigg\vert\sum_{s=1}^{2\log k}\sum_{t=a_{s}2^{j_{s}}+1}^{(a_{s}+1)2^{j_{s}}}\delta_{y,i}^{(t)}\Bigg\vert\leq2\log k\cdot O(\eps_{\mathrm{apx}}/\log k)\leq\eps_{\mathrm{apx}},
	\] %TODO: this dyadic interval notation is confusing. maybe just directly say sum over at most 2logk dyadic intervals
	where the third step follows by \cref{lem:dyadic-interval-cover},
	and the fourth step follows by the fact that $\ell_{i}$ is the
	last time $z_{i}$ was updated, so for any $[a\cdot2^{j}+1,(a+1)2^{j}]\subset[\ell_{i},\ell]$,
	we have $|(y^{((a+1)2^{j}+1)}-y^{(a\cdot2^{j})})_{i}|<O(\eps_{\mathrm{apx}}/\log k)$.

	Taking the union bound again, we have the data structure succeeds with probability at least $1-\delta_{\mathrm{apx}}$.%TODO: not sure if I follow this last step?

	\end{proof}

\subsection{Sketching A Sequence of Vectors}

In this section, we show how to construct an oracle used in \cref{thm:approximate-ell-infty}
that supports type-I queries for a sequence of vectors.
\begin{theorem}
\label{thm:vector-sketch}
Given a sampling tree $(\mathcal{S},\chi)$ of $\R^{n}$ with height $\eta$, and 
a JL sketching matrix $\Phi\in\R^{r \times n}$, 
the data structure $\textsc{VectorSketch}$ maintains $y_v \defeq \Phi_{\chi(v)} h$ for all nodes $v$ in the sampling tree through the following functions:
\begin{itemize}
\item $\textsc{Initialize}(\mathcal{S},\chi,\Phi\in\R^{r\times n},h)$:
Initializes the data structure in time $O(n\cdot\eta\cdot r)$, so that node $v \in \mc{S}$ maintains $y_v$.
\item $\textsc{Update}(\new h\in\R^{n})$: Maintains the data structure for $h \leftarrow \new h$ in $O(\eta\cdot r \cdot\|\new h-h\|_{0})$ time.
\item $\textsc{Query}(v\in V(\mathcal{S}))$: Outputs $\Phi_{\chi(v)}h$
in $O(r)$ time.
\end{itemize}
\end{theorem}

\begin{algorithm}
\caption{Vector Sketching Data Structure \label{alg:vector-sketch}}

\begin{algorithmic}[1]

\State \textbf{datastructure }$\textsc{VectorSketch}$

\State \textbf{private : members}

\State \hspace{4mm} $\Phi\in\R^{r\times n}$ \Comment{JL
matrix}

\State \hspace{4mm} Sampling tree $(\mathcal{S},\chi)$

%\State \hspace{4mm} $\ell \in \N$ \Comment{iteration counter}

\State \hspace{4mm} $h \in \R^n$ \Comment{latest input vector}

%\State \hspace{4mm} \textsf{list}  $\{y^{(i)}\}_{i=0}^k$ 
%\Comment{array storing input sequence of vectors}

\State \hspace{4mm} \textsc{List} $\{y_v\}_{v \in V(\mathcal{S})}$ 
\Comment{sketches indexed by nodes of $\mathcal{S}$, where $y_v \defeq \Phi_{\chi(v)}h$}

\State \textbf{end members}

\Procedure{$\textsc{Initialize}$}{$\mathcal{S},\chi,\Phi\in\R^{r\times n}, h \in\R^{n}$}

\State $(\mathcal{S},\chi)\leftarrow(\mathcal{S},\chi)$

\State $\Phi\leftarrow\Phi$

\State $h \leftarrow h$

\ForAll{$v\in\mathcal{S}$}

	\State $y_v \leftarrow \Phi_{\chi(v)}h$

\EndFor

\EndProcedure

\Procedure{$\textsc{Update}$}{$\new h$}

%\State $\ell\leftarrow\ell+1$
\ForAll{$i$ such that $\new h_{i}\neq h_{i}$}

	\State Find leaf node $v$ of $\mc{S}$ such that $\chi(v)=\{i\}$

	\ForAll{node $u \in \mathcal{P}^{\mathcal{S}}(v)$} 
	\Comment{where $\mathcal{P}(v)$ is the path from $v$ to the root in $\mathcal{S}$}

		\State $y_v\leftarrow y_v-\Phi_{\{i\}}h+\Phi_{\{i\}}\new h$

	\EndFor

\EndFor
\State $h \leftarrow \new h$

\EndProcedure

\Procedure{$\textsc{Query}$}{$v\in \mathcal{S}$}

	\State \Return $y_v$

\EndProcedure

\end{algorithmic}
\end{algorithm}

\begin{proof}
\textbf{Correctness:} 
In \textsc{Initialize}, we calculate $y_v$ directly for all $v$. In \textsc{Update} when $h$ is updated to $\new h$, 
note that $y_v$ maintained at a node $v$ needs to be updated if and only if $(\new h-h)_{\chi(v)}\neq\boldsymbol{0}$.
Hence, for each index $i$ with $(\new h - h)_{i}\neq0$, we need to update at all nodes $v$ with $i\in\chi(v)$, which is exactly the path from the leaf node $u$ with $\chi(u)=\{i\}$ to the root of $\mathcal{S}$. Moreover, the update due to coordinate $i$ is precisely $\Phi_{\{i\}}(\new h - h)$.

\textbf{Runtime:} For \textsc{Initialize}, let $\layer^\mathcal{S}(i)\defeq\{v\in V(\mathcal{S})\mid\depth(v)=i\}$
be the set of nodes in the $i$-th layer of the sampling tree. 
By property (3) of the sampling tree in \cref{def:sampling-tree}, $\chi(v)\cap\chi(u)=\emptyset$ for any $u,v\in\layer^\mathcal{S}(i)$. 
Hence, we can compute $\Phi_{\chi(v)}h$ for all $v\in\layer^\mathcal{S}(i)$ in time $O(n\cdot r)$. 
Since $\mathcal{S}$ has height $\eta$, initialization takes $O(n\cdot \eta \cdot r)$ time.

For \textsc{Update}, observe that the outer for-loop runs $\|\new h - h\|_{0}$ times. 
The inner for-loop iterates at most $\eta$ times, as it  traverses up a path from a leaf node to the root in $\mathcal{S}$. 
For each node on the path, we need to compute $z_v-\Phi_{\{i\}}h+\Phi_{\{i\}}\new h$
which takes $r$ time. 
Thus, we can bound the total update time
by $O(\eta \cdot r\cdot\|\new h - h\|_{0})$. 

To find the leaf node $v$ such that $\chi(v)=\{i\}$,
we note the function $\chi$ is fixed, so it can be pre-process during initialization in $O(n)$ time.

The query time follows by the fact that $y_v$ is an $r$-dimensional vector.
\end{proof}

\subsection{Sketching the Multiscale Representation via Simple Sampling Tree} \label{subsec:simple-sampling-tree}

The previous section shows how to construct an oracle used in \cref{thm:approximate-ell-infty} that supports type-I queries for a sequence of slowing changing vectors. However, not all vector variables in our main central path maintenance data structure change slowly across consecutive central path steps. 
In particular, we also want to maintain the sketches of matrix-vector products involving $\mathcal{W}^\top$, such as $\mathcal{W}^{\top}h, \mathcal{W}^{\top}\eps_x$ and $\mathcal{W}^{\top}\eps_s$ from \cref{eq:implicit-representation}.

Consider maintaining $\ell_\infty$-approximations of the sequence of $\mathcal{W}^\top h$: Using $\textsc{VectorSketch}$ presented in \cref{thm:vector-sketch} directly yields a data structure whose update time at iteration $\ell+1$ is a function of $\norm{(\mathcal{W}^\top h)^{(\ell+1)} - (\mathcal{W}^\top h)^{(\ell)}}_0$. 
Recall that $\mathcal{W}$ and $h$ change between central path steps as a function of changes in $\ox$; unfortunately, even if $\ox$ only changes in a single coordinate, $\mathcal{W}^{\top}h$ can change densely. 
Hence, we would like to design a modified data structure whose update time is a function of $\norm{\ox^{(\ell+1)} - \ox^{(\ell)}}_0$ and $\norm{h^{(\ell+1)}-h^{(\ell)}}_0$ instead.
%if $x_{i}$ changes and the non-zero pattern of $A_{i}$ is close to the root of elimination tree. %TODO: this is not very clear
In this section and the next, we present sketching data structures that serve as the oracle needed in \cref{thm:approximate-ell-infty} for type-I queries, specifically for the case when the online sequence of vectors is of the form $\{(\mathcal{W}^\top h)^{(\ell)}\}_{\ell=0}^k$, for dynamic $\mathcal{W}$ and $h$.
%TODO: would like to not use h for this entire section...

To this end, we have to crucially utilize the structure of the  lower Cholesky factor $L$ and the elimination tree $\mathcal{T}$.  
In this section, we present a simple construction of a sampling
tree which preserves the structural property of the elimination tree $\mathcal{T}$,
and an intuitive implementation of the sketching maintenance data structure with $\wt O(\poly(\tau))$ amortized runtime per update.
In \cref{subsec:balanced-sampling-tree}, we show a more involved data structure to lower the runtime, using many of the same ideas.

This section mainly serves to illustrate our approach to maintaining the sketches, so \emph{we assume each $n_i=1$ and $m=n$} in the block structure of $A$ for simplicity of presentation; the assumption is removed in \cref{subsec:balanced-sampling-tree}.

\begin{theorem}
Given the constraint matrix $A$, its binary elimination tree $\mathcal{T}$ with height $\tau$, a JL matrix $\Phi\in\R^{r \times n}$,
and a sampling tree $(\mathcal{S},\chi)$ with height $\eta \leq O(\tau+\log(n))$ constructed as in \cref{subsec:simple-sampling-tree-construction},
the data structure $\textsc{SimpleSketch}$ (\cref{alg:simple-W-sketch,alg:simple-W-sketch-1}) 
maintains the sketch $\Phi_{\chi(v)}\mathcal{W}^{\top}h = \Phi_{\chi(v)}H_{\ox}^{-1/2}A^\top L_{\ox}^{-\top}$ at every node $v \in \mc{S}$ through following operations:
\begin{itemize}
\item $\textsc{Initialize}(S,\chi,\Phi,\ox,h)$: Initializes the data structure in $O(n\cdot \tau\cdot \eta\cdot r)$ time, so that node $v \in \mathcal{S}$ maintains the sketch $\Phi_{\chi(v)} \mathcal{W}^\top h = \Phi_{\chi(v)}H_{\ox}^{-1/2}A^\top L_{\ox}^{-\top}$.
\item $\textsc{Update}(\new \ox, \new h)$: Updates all sketches in $\mc{S}$ to reflect $\mathcal{W}$ updating to $\new {\mathcal{W}}$ and $h$ to $\new h$, where $\new{\mathcal{W}}$ is given implicitly by $\new \ox$.
This function runs in $O(\|\new \ox-\ox\|_{0}\cdot\tau^{2}\cdot \eta \cdot r) + O(\|\new h-h\|_{0}\cdot\tau \cdot r)$ time.
\item $\textsc{Query}(v)$: Outputs $\Phi_{\chi(v)}\mathcal{W}^{\top}h$ in $O(\tau^{2}\cdot r)$ time.
\end{itemize}
\end{theorem}

In \cref{subsec:simple-sampling-tree-construction}, we give
the construction of the sampling tree. In \cref{subsec:simple-ds},
we give the analysis of each individual function.

\begin{table}[H]
	\begin{centering}
		\begin{tabular}{|c|c|}
			\hline 
			\textbf{Symbol} & \textbf{Definition} \tabularnewline
			\hline 
			$\text{\ensuremath{\mathcal{T}}}$ & elimination tree with vertex set $\{1,\dots,d\}$ \tabularnewline
			\hline 
			$(\mathcal{S},\chi)$ & sampling tree. By convention, we call vertices of $\mc{S}$ \emph{nodes} \tabularnewline
			\hline 
			$\mathcal{D}(v)$ & set of nodes in the subtree rooted at $v$ (inclusively)\tabularnewline
			\hline 
			$\mathcal{P}(v)$ & set of nodes on the path from $v$ to the root (inclusively)\tabularnewline
			\hline 
			$\depth(v)$ & depth of node $v$ in tree ($\depth(\textrm{root})=1)$\tabularnewline
			\hline 
			$\low(a)$ & the lowest node in tree in the nonzero pattern of a vector $a$\tabularnewline
			\hline 
			$\block_low(A)$ & the lowest node in tree in the nonzero column pattern of $A$\tabularnewline
			\hline 
			$\lca(u,v)$ & the lowest common ancestor of $u$ and $v$ in tree \tabularnewline
			\hline 
			$A_{S}$ & matrix $A$ restricted to coordinates/blocks indexed by nodes in set $S$\tabularnewline
			\hline 
			$f^{\mathcal{S}}(v)$ or $f^{\mathcal{T}}(v)$ & function $f(v)$ for the specified tree $\mathcal{S}$ or $\mathcal{T}$ \tabularnewline
			\hline 
		\end{tabular}
		\par\end{centering}
	\caption{Notations in this section}
\end{table}

\subsubsection{Simple Sampling Tree Construction} \label{subsec:simple-sampling-tree-construction}

We begin with the construction of a \emph{simple sampling tree} $(\ts,\chi)$ which has $n$ leaf nodes, based on the elimination tree ${\mathcal{T}}$. 
Recall $\mathcal{T}$ has $d$ vertices given by the set $\{1,2,\ldots,d\}$, where vertex $i$ correspond to row $i$ of $A$.

First, we define the function
$\mathsf{low}^{{\mathcal{T}}}(a):\R^{d}\to[d]$ by 
\[
\mathsf{low}^{{\mathcal{T}}}(a)=\arg\max_{i\in\{i\mid a_{i}\neq0\}}\mathsf{depth}^{{\mathcal{T}}}(i),
\]
%TODO: actually, this is just: max {i st. a_i \neq 0}
which gives the node $i$ at the lowest level in $\mathcal{T}$ such that $a_i \neq 0$. 
Note that $\mathsf{low}^{{\mathcal{T}}}(A_{j})$ is well-defined, since the non-zero pattern of $A_{j}$ is a subset of a path in ${\mathcal{T}}$ by \cref{lem:A-sparsity-pattern}. 

For each vertex $i \in \mathcal{T}$,
we construct a balanced binary tree on all new nodes, rooted at node $i'$ and with leaf nodes given by the set $F_i = \{c_j \mid \mathsf{low}^{{\mathcal{T}}}(A_{j})=i\}$. Observe if $i$ is a leaf node of $\mathcal{T}$, then $F_i$ is non-empty.
We construct a new tree $\ts$ by beginning with $\ts \leftarrow \mathcal{T}$, and then for each $i \in \mathcal{T}$, attaching the new subtree rooted at $i'$ under $i$ in $\ts$. After this, the set of leaf nodes in $\mathcal{S}$ is given by $\bigcup_{i \in \mathcal{T}} F_i = \{c_j \mid j \in [n]\}$, with every leaf corresponding to a distinct column of $A$.

For each $v \in \ts$, we recursively define its labelling $\chi(v)$ by:
\[\chi(v) =
\begin{cases}
	i &\text{if $v = c_i$ is a leaf node} \\
	\bigcup_{u \text{ a child of $v$ in $\mathcal{S}$}} \chi(u) \qquad  &\text{else}
\end{cases}
\]
In particular, for $i \in \mc{S} \cap \mc{T} = \{1,\dots,d\}$, $\chi(i)$ satisfies:
\begin{equation} \label{eq:sampling-tree-chi}
\chi(i)=\{j\mid\mathsf{low}^{{\mathcal{T}}}(A_{j})=i\}\cup\bigcup_{\text{$j$ a child of \ensuremath{i} in \ensuremath{\mathcal{T}}}}\chi(j).
\end{equation}

Let this newly constructed $(\mathcal{S},\chi)$ be the \emph{simple sampling tree}. Since $\mathcal{T}$ is a binary tree, $\mathcal{S}$ is a degree-3 tree.
An example is shown in \cref{fig:simple-sampling-tree-example}.

\begin{definition} [$\mathcal{T}(v)$]
	Recall we have $V(\mc{T}) = \{1,\dots,d\} \subset V(\mc{S})$. For a node $v \in \mc{S}$, we define its \emph{$\mc{T}$-ancestor} $\mc{T}(v)$ to be the lowest ancestor of $v$ in $\mc{S}$ that is also in $\mc{T}$. In particular, if $v \in \{1,\dots, d\}$, then $\mathcal{T}(v) = v$. 
	
	For example, in \cref{fig:simple-sampling-tree-example}, $\mc{T}(5) = 5$ and $\mc{T}(c_3) = 4$, where $c_3$ is the bottom left node in $\mc{S}$.
\end{definition}

\begin{comment}
\[
{{\mathcal{T}}(v)=\begin{cases}
i & \text{if exists \ensuremath{i} such that \ensuremath{\ts(i)=v}}\\
{{\mathcal{T}}(\text{parent node of $v$ in $\ts$})} & \text{otherwise},
\end{cases}}
\] 
\end{comment}

%already in table
%We use $\mathcal{P}^{\ts}(v)$ (resp. $\text{\ensuremath{\mathcal{P}}}^{\mathcal{T}}(v)$) to denote the set of vertices on the path from node $v$ to the root on the tree $\mathcal{S}$ (resp. tree $\mathcal{T}$).

\begin{theorem}\label{thm:simple-sampling-tree-construct-time}
	Given an elimination tree $\mathcal{T}$ with height $\tau$, the simple sampling tree with height $\tau + O(\log n) = O(\tau)$ can be constructed in $O(n\tau+n\log n)$ time. 
\end{theorem}
\begin{proof}
	Since the newly added balanced binary tree under each $v\in \mathcal{T}$ has height at most $O(\log n)$, 
	the height of the sampling tree is bounded by $\tau+O(\log n)=O(\tau)$.
	
	For each column $A_i$, we can find $\low(A_i)$ in time $O(\tau)$ since $\nnz(A_j)=O(\tau)$ by \cref{lem:A-sparsity-pattern}. 
	Hence, we can find $F_i$ for every $i \in \mathcal{T}$ in $O(n\tau)$ time in total. 
	Constructing the balanced binary tree rooted at every $i'$ takes at most $O(n\log n)$ total time. Finding the sets $\chi(v)$ at every $v \in \ts$ takes $O(n \log n)$ total time.
	Hence, we can construct $(\mathcal{S},\chi)$ in $O(n\tau + n\log n)$ time.
\end{proof}

\begin{figure}
\begin{center}
	\includegraphics[scale=0.35]{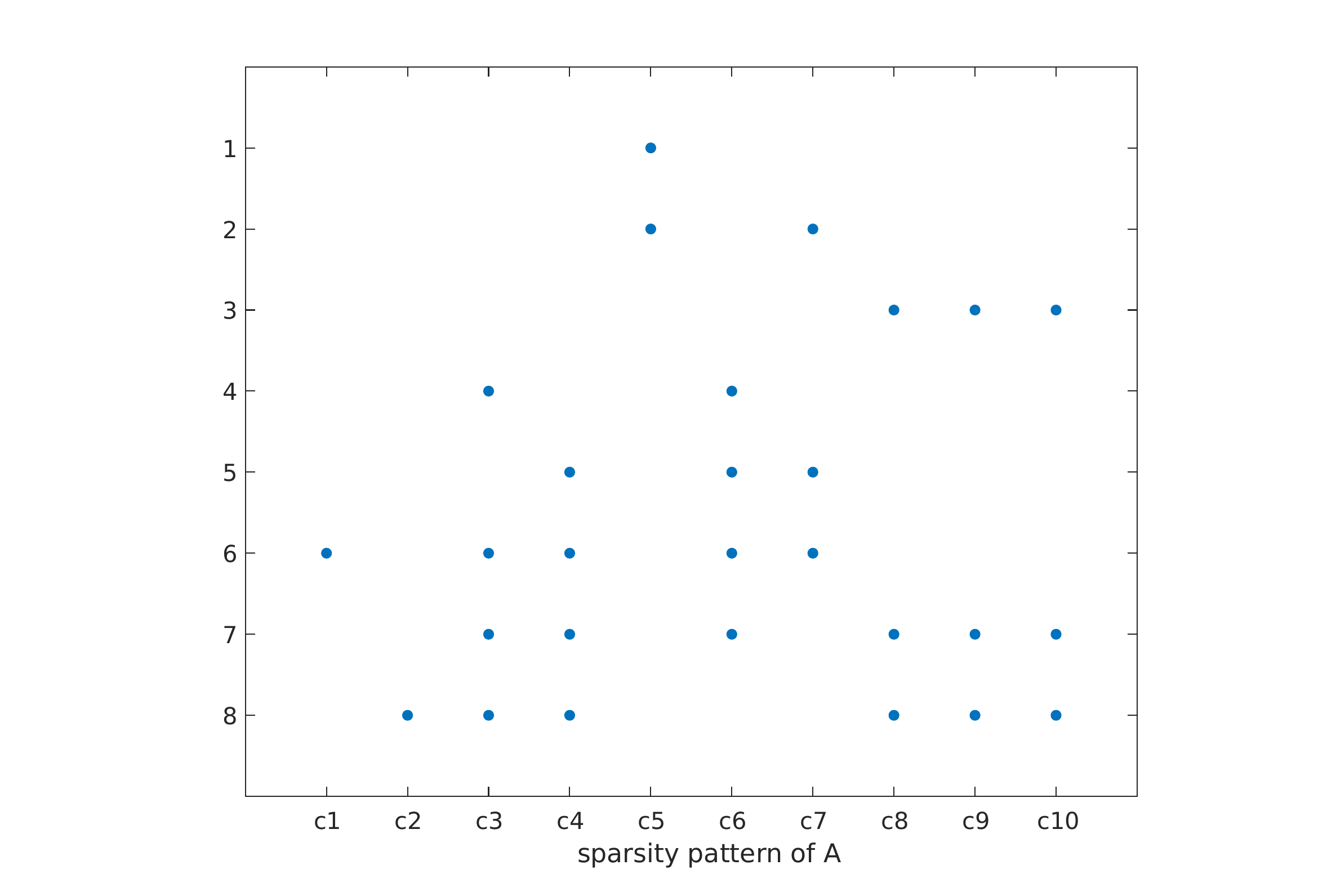}
	\includegraphics[width=16cm]{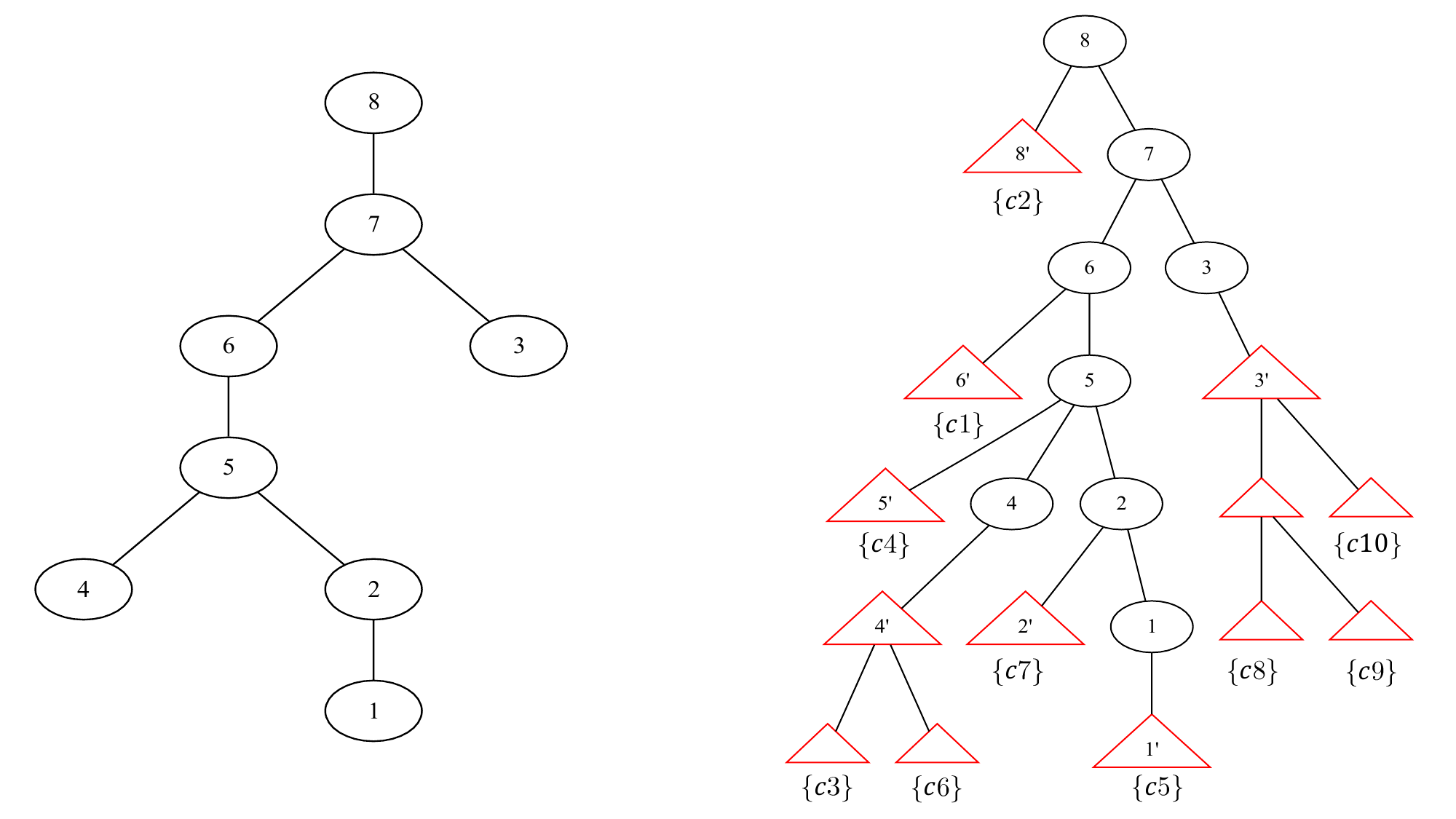}
\end{center}
\caption{Example of simple sampling tree: The tree on the left is the elimination tree $\mathcal{T}$ for constraint matrix $A$, whose sparsity pattern is shown on the right ($n_i=1$ for all $i$). 
The tree on the right is the simple sampling tree, where
red triangles are newly added nodes, and the bracket under each leaf node denotes the column that the node will maintain in the data structure.}\label{fig:simple-sampling-tree-example}
\end{figure}

\subsubsection{Data Structure for Sketching}\label{subsec:simple-ds}

Now, we discuss how to maintain the sketches $\Phi_{\chi(v)} \mathcal{W}^\top h$ at every node $v \in \mathcal{S}$.
Recall the non-zero pattern of the Cholesky factor $L$ is reflected in the elimination tree. Specifically, the non-zero pattern of column $L_i$ is a subset of the path from $i$ to the root of $\mc{T}$. 
Since we have constructed $\mc S$ to preserve the ancestor-descendant relationships from $\mc T$, 
we will be able to update the sketches in $\mc S$ in a more clever way.

To better utilizing the structural relationship between the lower Cholesky factor and the sampling tree, for any $v \in \mc{S}$, we rewrite the sketches $\Phi_{\chi(v)}\mathcal{W}^{\top}h = \Phi_{\chi(v)}H^{-1/2}A^{\top}L^{-\top}h$ using the following notation:
	
\begin{definition}[{$J_v, Z_v^*, y_v$}]
	For each $v \in \mc{S}$, let
	\[
	J_v \defeq \Phi_{\chi(v)} H^{-1/2}A^{\top}, \quad
	Z^*_v \defeq J_v \cdot L^{-\top}, \quad
	y_v \defeq Z^*_v \cdot h = \Phi_{\chi(v)} \mc{W}^{\top} h.
	\]
	At every node $v$, we will maintain $J_v$, and some variant of $Z_v^*$ and $y_v$ discussed later.
\end{definition}

Let us first examine the sparsity pattern of $J_v$ and $Z^*_v$:
\begin{lemma}[{Sparsity pattern of $J_v$}] \label{lem:J-sparsity}
	Let $v\in\mathcal{S}$, and suppose $J_v$ satisfies (i) of \cref{inv:W-simple-invariant}. 
	Let $S$ be the non-zero column pattern of $J_v$, i.e.~$S=\{j\in[d]\mid (J_v)_{j}\neq\boldsymbol{0}\}$.
	If $v \in \mc{S} \setminus \mc{T}$, then $S \subseteq \mc{P}^{\mc{T}}(\mc{T}(v))$. On the other hand, if $v \in \mc{T}$, then $S\subseteq\mathcal{D}^{\mathcal{T}}(\mathcal{T}(v))\cup\mathcal{P}^{\mathcal{T}}(\mathcal{T}(v))$.
\end{lemma}

\begin{proof}
	First, note that for any $i \in [n]$, the non-zero column pattern of $\Phi_{\{i\}}H^{-1/2}A^{\top}$ is the non-zero pattern of column $A_i$. More generally, the non-zero column pattern of $J_v = \Phi_{\chi(v)}H^{-1/2}A^\top$ is given by the union of the non-zero pattern of columns $A_j$ such that $j \in \chi(v)$ for any $v \in \mc{S}$.
	
	In the case that $v \in \mc{S} \setminus \mc{T}$, let $i = \mc{T}(v)$ be the $\mc{T}$-ancestor of $v$. By construction of $(\mc{S},\chi)$, we have $\chi(v) \subseteq \{j \in [n] \mid \mathsf{low}^{\mathcal{T}}(A_{j}) = i\}$. Hence, the sparsity pattern of any column $A_j$ with $j \in \chi(v)$ is a subset of $\mc{P}^{\mc{T}}(i)$ by \cref{lem:A-sparsity-pattern}. 
	Since $v$ is a descendant of $i$ in $\mc{S}$, we have $\chi(v) \subset \chi(i)$ by property of $\chi$. Therefore, $S \subseteq \mc{P}^{\mc{T}}(i)$, as required.
	
	In the case that $v =i \in \mc{T}$, observe that as a consequence of \cref{eq:sampling-tree-chi}, we have 
	\[
	\chi(i) = \{j \in [n] \mid \mathsf{low}^{\mathcal{T}}(A_{j})= k, \text{ where } k \in \mc{D}^{\mc{T}}(i)\}.
	\]
	By \cref{lem:A-sparsity-pattern}, for any $j \in \chi(i)$, the sparsity pattern of $A_j$ is a subset of a path in $\mc{T}$ containing $i$, that is, it is contained in $\mc{D}^{\mc{T}}(i) \cup \mc{P}^{\mc{T}}(i)$.
\end{proof}

\begin{lemma}[{Sparsity pattern of $Z^*_v$}]
	\label{lem:Z-sparsity}
	Let $v \in \mathcal{T}$, and let $S$ be the non-zero pattern of the columns of $Z_v$, i.e.~$S=\{i\in[d]\mid (Z_v)_{i}\neq\boldsymbol{0}\}$.
	Then, $S\subseteq\mathcal{D}^{\mathcal{T}}(v)\cup\mathcal{P}^{\mathcal{T}}(v)$.
\end{lemma}

\begin{proof}
	This directly follows by \cref{lem:J-sparsity,lem:Linv-sparsity-time}.
\end{proof}

%Recall that there is no bound on the row sparsity of $A$, 
%so if there is a dense row $A_{i}$ such that $\chi(\mathcal{S}(i))=\Omega(n)$, 
%then we have to update almost all the sketches in $\mathcal{S}$. %TODO: This is not clear, maybe just take out

At this point, we have all the tools to answer queries for the sketch at a node $v \in \mathcal{S} \setminus \mathcal{T}$: Given $J_v$ at $v \in \mathcal{S} \setminus \mc{T}$, the non-zero columns of $J_v$ is a subset of $\mathcal{P}^{\mathcal{T}}(\mathcal{T}(v))$ by \cref{lem:J-sparsity}. Hence, we can compute the sketch $y^*_v=J_v \cdot L^{-\top}h$ in $\poly(\tau)$ time during a query. The sketch at a node $v \in\mathcal{T}$ needs to be maintained more carefully.

Updates to $\mc{W}$ via $\ox$ causes a corresponding update to the Cholesky factor $L$.
We will show later that if column $j$ of $L$ changes, then the sketches that change are at nodes of $\mc{S}$ in the subtree rooted at $j$, and the path from $j$ to the root;
we delay the updates of $L$ at nodes on the path $\mathcal{P}^{\mathcal{T}}(j)$. 

\begin{definition}[{$L[t], t_v, Z_v$}]
	Let $\{L[t]\}_{t \geq 0}$ be a list of Cholesky factors computed at different iterations during the maintenance, 
	such that $L[\ell]$ is the Cholesky factor computed at iteration $\ell$, for an internal iteration counter $\ell \geq 0$ that advances whenever $L$ is updated.

	At every node $v \in \mc{S} \cap \mc{T}$, we maintain a time stamp $t_v \geq 0$.
	Furthermore, we maintain a modified $Z_v$ that depends on $L$ from an earlier iteration given by $t_v$, that is,
	\[
		Z_v = J_v \cdot L[t_v]^{-\top}.
	\]
\end{definition}
\begin{rem}
	For any $v \in \mc{T}$, note that $Z_v$ and $Z_v^*$ have the same non-zero pattern, as the non-zero pattern of $L$ is constant throughout the algorithm.
\end{rem}

Similarly, updates to $h$ may cause sketches at many nodes of $\mc{S}$ to change. Again, we implement lazy updating for the part of the sketch $y_v = Z_v^* \cdot h$ involving $h$. 
%: If the updated coordinate of $h$ corresponds to the root of $\mathcal{T}$, then we have to update all the sketches. 
\begin{definition}[{$y_v^{\triangledown}$}]
	For $v \in \mc{T}$, let
	\[
	y^{\triangledown}_v\defeq Z_v\cdot(I-I_{\mathcal{P}^{\tt}(v)})h.
	\]
	At every node $v \in \mc{S} \cap \mc{T}$, we maintain $y_v^{\triangledown}$. 
	When $Z_v = Z_v^*$, observe that we can write $y_v$ as
	\begin{equation}\label{eq:y_v}
		y_v = Z^{*}_v\cdot I_{\mathcal{P}^{{\tt }}(v)}h + y_v^{\triangledown}.
	\end{equation}
\end{definition}
It follows that to obtain the latest sketch $y_v$ at node $v$, we can update $Z_v \leftarrow Z_v^*$, update $y_v^{\triangledown}$ accordingly, and then compute $y_v$ by \cref{eq:y_v}, noting that the first term can be computed in $\poly(\tau)$ time given $Z_v^*$ and $h$.

Now, we list the invariants our data structure maintains during
the algorithm.

\begin{invariant}\label{inv:W-simple-invariant} The variables maintained in the data structure \textsc{SimpleSketch}, as given in \cref{alg:simple-W-sketch}, always preserve the following invariant before and after each function call:
		\begin{alignat*}{2}
		J_v & =\Phi_{\chi(v)}H^{-1/2}A^{\top} & v\in\mathcal{S} \tag{i}\\
		Z_v & =J_v \cdot L[t_v]^{-\top} & v\in\mathcal{T} \tag{ii} \\
		\boldsymbol{0} & =(L[\ell]-L[t_v])_{\mathcal{D}^{\mathcal{T}}(v)\setminus\{v\}} \qquad & v\in\mathcal{T} \tag{iii}\\
		y^{\triangledown}_v & = Z_v \cdot (I - I_{\mathcal{P}^{\tt}(v)}) h & v\in\mathcal{T} \tag{iv}
		\end{alignat*}
		where $H =\nabla^{2}\phi(\ox)$, 
		$L$ is the lower Cholesky factor such that $LL^{\top}=AH^{-1/2}A^{\top}$, and $t_v$ is the time stamp of $v$.
\end{invariant}

Finally, suppose the current iteration counter is $\ell$. the following lemma tells us how to compute the latest value $Z^{*}_v =J_v \cdot L[\ell]^{-\top}$, using $Z_v = J_v \cdot L[t_v]^{-\top}$ maintained by the data structure:
\begin{lemma}
\label{lem:Z*-formula}Suppose \cref{inv:W-simple-invariant}
is satisfied for node $v\in\mathcal{T}$, then 
\[
Z^{*}_v=Z_v-\left(L[\ell]^{-1}(L[\ell]-L[t_v])_{\mathcal{P}^{\mathcal{T}}(v)}\cdot Z_v^{\top}\right)^{\top}.
\]
\end{lemma}

\begin{proof}
Let us denote $\Delta L = L[\ell]-L[t_v]$. Then $(Z^*)^\top=(L[\ell]+\Delta L)^{-1} J_v^\top$, 
and we want to find $\Delta Z$ such that 
\begin{align*}
	Z_v^\top + (\Delta Z)^\top = (Z^*)^\top &= (L[t_v]+\Delta L)^{-1} J_v^\top.\\
\intertext{We have}
 	(L[t_v]+\Delta L)(Z_v^\top + (\Delta Z)^\top) &= J_v^\top = L[\ell] Z_v^\top \\
(\Delta Z)^{\top} &= -(L[t_v]+\Delta L)^{-1} (\Delta L) Z_v^\top \\
Z^{*}_v - Z_v = \Delta Z &= -\left(L[\ell]^{-1}(L[\ell]-L[t_v])Z_v^{\top}\right)^{\top}.
\end{align*}
We split $\Delta L$ into three parts:
\[
\Delta L=(I_{\mathcal{P}^{\mathcal{T}}(v)}+
I_{\mathcal{D}^{\mathcal{T}}(v)\setminus\{v\}}+
I_{\mathcal{T}\setminus(\mathcal{D}^{\mathcal{T}}(v)\cup\mathcal{P}^{\mathcal{T}}(v))})\Delta L.
\]
By \cref{lem:Z-sparsity}, 
the non-zero columns of $Z_v=J_v \cdot L[t_v]^{-\top}$ is a subset of $\mathcal{D}^{\mathcal{T}}(v)\cup\mathcal{P}^{\mathcal{T}}(v)$.
Hence,  $I_{\mathcal{T}\setminus(\mathcal{D}^{\mathcal{T}}(v)\cup\mathcal{P}^{\mathcal{T}}(v))}\cdot Z_v^{\top}=\boldsymbol{0}$.
By (iii) of \cref{inv:W-simple-invariant}, $(L[\ell]-L[t_v])\cdot I_{\mathcal{D}^{\mathcal{T}}(v)\setminus\{v\}}=\boldsymbol{0},$ implying that
\[
L[\ell]^{-1}(L[\ell]-L[t_v])Z_v^{\top}=L[\ell]^{-1}(L[\ell]-L[t_v])_{\mathcal{P}^{\mathcal{T}}(v)}\cdot Z_v^{\top}.
\]
\end{proof}

\begin{algorithm}
\caption{Simple Multiscale Representation Sketching Data Structure --
Initialize and Query\label{alg:simple-W-sketch}}

\algnewcommand{\LeftComment}[1]{\State \(\triangleright\) #1} 

\begin{algorithmic}[1]

\State \textbf{datastructure }$\textsc{SimpleSketch}$

\State \textbf{private : members}

\State \hspace{4mm} $\Phi\in\R^{r\times n}$ \Comment{JL
	matrix}

\State \hspace{4mm} sampling tree $(\mathcal{S},\chi)$
\Comment{constructed according to~\cref{subsec:simple-sampling-tree-construction}}

\State \hspace{4mm} elimination tree $\mathcal{T}$

\State \hspace{4mm} $\ell \in\mathbb{N}$ \Comment{iteration counter}

\State \hspace{4mm} $h \in\R^{d}$

\State \hspace{4mm} $\ox \in\R^{n}$ 
\Comment{$\mathcal{W}$ given implicitly by $\ox$}

\State \hspace{4mm} $H \in\R^{n \times n}$\Comment{Hessian  $H=\nabla^{2}\phi(\ox)$}

\State \hspace{4mm} \textsc{List} $\{L[t] \in \R^{d\times d}\}_{t \geq 0}$ 
\Comment{sequence of Cholesky factors at various iterations $t$}
\State \hspace{4mm} \textsc{List} $\{J_v \in \R^{r\times d}\}_{v \in \mc{T}}$
\Comment{$J_v =\Phi_{\chi(v)}H^{-1/2}A^{\top}$}

\State \hspace{4mm} \textsc{List} $\{Z_v \in\R^{r\times d}\}_{v\in\mc{S}}$
\Comment{$Z_v=J_v\cdot L[t_v]^{-\top}$ } 

\State \hspace{4mm} \textsc{List} $\{y^{\triangledown}_v \in\R^{r}\}_{v \in \mc{T}}$
\Comment{$y^{\triangledown}_v=Z_v\cdot(I-I_{\mathcal{P}^{{\tt}}(v)})h$}

\State \hspace{4mm} \textsc{list} $\{t_v \in\mathbb{N}\}_{v \in\mc{T}}$ 
\Comment{$t_v$ is the time of the last update at a node $v$}

\State \textbf{end members}

\Procedure{$\textsc{Initialize}$}{$\ts, \chi, \Phi\in\R^{r\times n},\ox\in\R^{n},h\in\R^{d}$}
\Comment{\cref{lem:simple-initialize}}

\State $(\ts, \chi) \leftarrow (\ts, \chi)$

\State $\Phi\leftarrow\Phi$

\State $\ell\leftarrow0,h \leftarrow h$

\State Compute $H\leftarrow\nabla^{2}\phi(\ox)$

\State Find the lower Cholesky factor $L[\ell]$ of $AH^{-1}A^{\top}$

\ForAll{$v\in\mathcal{S}$}

	\State $J_v\leftarrow\Phi_{\chi(v)}H^{-1/2}A^{\top}$
	\Comment{compute $J_v$ for all $v \in \mc{S}$}

\EndFor

\ForAll{$v \in \mathcal{T}$}
\Comment{these nodes store additional partial computations}
	\State $Z_v\leftarrow J_v\cdot L[\ell]^{-\top}$

	\State $y^{\triangledown}_v\leftarrow Z_v\cdot(I-I_{\mathcal{P}^{{\tt }}(v)})h$

	\State $t_v \leftarrow\ell$
	\Comment{record that $Z_v$ and $y_v^{\triangledown}$ were last updated at time $\ell$}

\EndFor
\EndProcedure

\Procedure{$\textsc{Query}$}{$v\in\mathcal{S}$} \Comment{\cref{lem:simple-query}}

	\If{$v\in\mathcal{S}\setminus\mathcal{T}$}
	\Comment{directly compute and return the value of sketch}
		\State \Return $J_v \cdot L[\ell]^{-\top}h$
	\EndIf

	\LeftComment{for $v \in \mc{T}$, we make use of existing partial computations}
	\State 
	$\Delta L \leftarrow(L[\ell]-L[t_v])_{\mathcal{P}^{\mathcal{T}}(v)}$

	\State $Z_v\leftarrow Z_v-(L[\ell]^{-1} \cdot \Delta L \cdot Z_v^{\top})^{\top}$
	\Comment{Update $Z_v$ to correspond to $L[\ell]$, that is, $Z_v = Z^*_v$}
	\State $y^{\triangledown}_v \leftarrow Z_v \cdot (I - I_{\mathcal{P}^{\tt}(v)}) \cdot h$
	
	\LeftComment{$Z_v$ and $y^{\triangledown}_v$ now correspond to the latest $L[\ell]$, so we update the time stamp of $v$}
	\State $t_v\leftarrow\ell$

	\State \Return $Z_v \cdot I_{\mathcal{P}^{\mathcal{T}}(v)}\cdot h + y^{\triangledown}_v$

\EndProcedure

\end{algorithmic}
\end{algorithm}

\begin{algorithm}
	\caption{Simple Multiscale Representation Sketching Data Structure --
		Updates\label{alg:simple-W-sketch-1}}
	
	\begin{algorithmic}[1]
		
		\State \textbf{datastructure }$\textsc{SimpleSketch}$
		
		\State
		
		\Procedure{$\textsc{Update}$}
		{$\new{\ox}\in\R^{n}$, $\new h \in \R^n$}
		
		\Comment{\cref{lem:simple-updateW}}
		
		\For{$i\in[n]$ where $\new \ox_{i}\neq \ox_{i}$}
		
		\State $\textsc{UpdateCoordinate}(\new{\ox}, i)$
		\Comment{break up the update into single-coordinate updates}
		\EndFor
		
		\ForAll{$\new{h_{i}}\neq h_{i}$}
		
		\ForAll{$v\in\mathcal{P}^{\mathcal{\mathcal{T}}}(i)$}
		
		\State $y^{\triangledown}_v\leftarrow y^{\triangledown}_v + Z_v \cdot I_{\{i\}}\cdot(\new h-h)$
		
		\EndFor
		
		\EndFor

		\State $h\leftarrow\new h$
		
		\EndProcedure
		
		\Procedure{$\textsc{UpdateCoordinate}$}{$\new{\ox} \in \R^n, i \in [n]$}\Comment{\cref{lem:update-coordinate}}
		
		\State $\ox_{i}\leftarrow\new{\ox_{i}}$
		
		\State $\new H=\nabla^{2}\phi(\ox)$
		
		\State $\ell\leftarrow\ell+1$ 
		\Comment{increment iteration before computing a new Cholesky factor}
		\State Find lower Cholesky factor $L[\ell]$ of $A(\new H)^{-1}A^{\top}$
		
		\State $\textsc{Update}L(\mathcal{P}^{\mathcal{T}}(\low^{\mathcal{T}}(A_{i})))$
		
		\State $\textsc{Update}H(\new H_i, i)$
		
		\EndProcedure
		
		\Procedure{$\textsc{Update}L$}{$S \subseteq \mc{T}$}  
		\Comment{$S$ is a path in $\mc{T}$, \cref{lem:update-L}}
		
		\ForAll{$v\in S$}
		\State \Comment{We update $Z_v$ to $Z_v^*$ in two steps: first from $L[t_v]$ to $L[\ell-1]$, then from $L[\ell-1]$ to $L[\ell]$}
		\State $Z_v \leftarrow Z_v-\left(L[\ell-1]^{-1} \cdot (L[\ell-1]-L[t_v])_{\mathcal{P}^{\mathcal{T}}(v)}\cdot Z_v^{\top}\right)^{\top}$
		
		\State $Z_v\leftarrow Z_v-(L[\ell]^{-1}\cdot (L[\ell]-L[\ell-1])\cdot Z_v^{\top})^{\top}$
		
		\State $y^{\triangledown}_v \leftarrow Z_v \cdot(I-I_{\mathcal{P}^{\mathcal{T}}(v)})\cdot h$
		
		\State $t_v \leftarrow \ell$
		
		\EndFor
		
		\EndProcedure
		
		\Procedure{$\textsc{Update}H$}{$\new H$}
		\Comment{\cref{lem:update-H}}
		
		\State $\Delta H=\new H - H$
		
		\ForAll{$i\in[n]$ such that $(\Delta H)_{i}\neq\boldsymbol{0}$}
		
		\State Find $v$ such that $\chi(v)=\{i\}$
		
		\ForAll{$u\in\mathcal{P}^{\mathcal{S}}(v)$}
		
		\State $J_v\leftarrow\Phi_{\chi(v)}(H +\Delta H\cdot I_{\{i\}})^{-1/2}A^{\top}$
		
		\If{$u\in\mathcal{T}$}
		
		\State $Z_v\leftarrow J_v\cdot L[t_v]^{-\top}$
		
		\State $y^{\triangledown}_v\leftarrow Z_v\cdot(I-I_{\mathcal{P}^{\mathcal{T}}(v)})\cdot h$
		
		\EndIf
		
		\EndFor
		
		\EndFor
		\State $H \leftarrow \new H$
		\EndProcedure
		
	\end{algorithmic}
\end{algorithm}

Now we are ready to prove the correctness and runtime of each function in the data structure. The correctness of the overall maintenance data structure then follows immediately from the invariants.

\begin{lemma}[$\textsc{Initialize}$]
\label{lem:simple-initialize}
Given initial $\ox$ and $h$, the
JL matrix $\Phi\in\R^{r\times n}$, and the elimination
tree $\mathcal{T}$ with height $\tau$, the data structure \textsc{SimpleSketch} initializes the sketches in the sampling tree in $O(n\cdot \tau\cdot \eta\cdot r)$ time. 
Moreover, the internal state of the
data structure satisfies \cref{inv:W-simple-invariant}
after initialization.
\end{lemma}

\begin{proof}
The correctness directly follows by the setup of~\cref{inv:W-simple-invariant}.

\textbf{Runtime:} 
%We can construct the sampling tree in time $O(n\eta)$ by \cref{thm:simple-sampling-tree-construct-time}.
By \cref{cor:chol-time}, we can
find $L[\ell]$ in time $O(n\tau^{2})$. For any non-leaf node $v\in\mathcal{S}$,
we note that $J_v=\sum_{\text{$u$ a child of $v$}}J_u$.
For a leaf node $v\in\mathcal{S},$ we have $\chi(v)=\{i\}$ for some $i \in [n]$, so we can compute $J_v$ in time $O(\tau \cdot r)$ by \cref{lem:A-sparsity-pattern}. Then, we
can compute all $J_v$ for non-leaf nodes by summing the values of its children, iteratively up the tree.
Since the height of tree $\mathcal{S}$ is $\eta$,
we can compute $J_v$ for all $v\in\mathcal S$ in time $O(|V(\mathcal{S})|\cdot\tau\cdot \eta \cdot r)$. 

%Recall that for each $v\in\mathcal{T}$, we have $\chi(v)=\{i\mid\mathsf{low}^{{\mathcal{T}}}(A_{i})=v\}\cup\bigcup_{\text{child \ensuremath{u} of \ensuremath{v} in \ensuremath{\mathcal{T}}}}\chi(u)$.
For $v \in \mc{T}$, by \cref{eq:sampling-tree-chi}, we have
\begin{align*}
Z_v &=  \left(\Phi_{\{i\mid\mathsf{low}^{{\mathcal{T}}}(A_{i})=v\}} + \sum_{\text{child $u$ of $v$ in $\mc{T}$}} \Phi_{\chi(u)}\right) H^{-1/2}A^{\top}L[\ell]^{-\top}  \\ 
&= \Phi_{\{i\mid\mathsf{low}^{{\mathcal{T}}}(A_{i})=v\}}H^{-1/2}A^{\top}L^{-\top}+\sum_{\text{child \ensuremath{u} of \ensuremath{v} in \ensuremath{\mathcal{T}}}}J_v\cdot L[\ell]^{-\top}
\end{align*}
Thus, for each $v \in \mc{T}$, we only need to compute the term
$\Phi_{\{i\mid\mathsf{low}^{{\mathcal{T}}}(A_{i})=v\}}H^{-1/2}A^{\top}L^{-\top}$.
Since the non-zero columns of $\Phi_{\{i\mid\mathsf{low}^{{\mathcal{T}}}(A_{i})=v\}}H^{-1/2}A^{\top}$
lie on $\mathcal{P}^{\mathcal{T}}(v)$, this term has $O(\tau \cdot r)$
many non-zero entries, 
and we can compute it in $O(\tau^{2} \cdot r)$
time by \cref{lem:Linv-sparsity-time}. 
Again, by iteratively computing $Z_v$ up the tree, we can compute $Z_v$ for all $v\in\mathcal{T}$ in  $O(|\mathcal{T}|\cdot\tau^2\cdot r)$ time.

Because $h$ is explicitly given, we can compute $y^{\triangledown}_v=Z_v\cdot(I-I_{\mathcal{P}^{{\tt }}(v)}) \cdot h$ in $\nnz(Z_v)$ time for each $v \in \mc{S}$; then we can compute $y^{\triangledown}_v$ for all $v\in\mathcal{T}$ in $O(|\mathcal{T}| \cdot \tau^2\cdot r)$ time. 

Combined with the fact $\tau \leq \eta$ and $|\mc{T}| = d \leq n$, the total time is bounded by $O(n \cdot \tau\cdot \eta \cdot r)$.
\end{proof}

%TODO: update the runtime?
\begin{lemma}[$\textsc{Update}$] 
 \label{lem:simple-updateW}Suppose the current state of data structure
satisfied the \cref{inv:W-simple-invariant}. Given $\new{\mathcal{W}}$
implicitly by $\new \ox$, and $\new h$, the function \textsc{update} of \textsc{SimpleSketch} updates the sketches in $\mathcal{S}$ implicitly in time $O(\|\new \ox-\ox\|_{0}\cdot\tau^{2} \cdot \eta \cdot r) + O(\|\new h-h\|_{0}\cdot \tau \cdot r)$
Moreover, the function also updates the internal states correspondingly so that \cref{inv:W-simple-invariant} is still preserved.
\end{lemma}
\begin{proof}
	Note that we can process the updates to $\mathcal{W}$ and $h$ consecutively; hence the runtime is a sum of the runtimes of the two steps. 
	
	To update $\ox$ to $\new \ox$ and thus $\mc W$ to $\new{\mc{W}}$, we again view it as a sequence of updates, where each update correspond to a single coordinate change in $\ox$, processed by the helper function \textsc{UpdateCoordinate}. The associated proof is given in \cref{lem:update-coordinate}.
	
	Similarly, we update $h$ to $\new h$ by a sequence of single-coordinate updates. 
	By \cref{lem:Z-sparsity}, the nonzero columns of $Z_v\cdot(I-I_{\mathcal{P}^{\mathcal{T}}(v)})$ lies on $\mathcal{D}^{\mathcal{T}}(v)$. Therefore, when $h_{i}$ changes, $y_v^{\triangledown}$ changes only if $i\in\mathcal{D}^{\mathcal{T}}(v)$, so it suffices to update only the sketches at nodes $v$ where $v \in \mathcal{P}^{\mathcal{T}}(i)$, and the update is given by $Z_v\cdot I_{\{i\}}\cdot(\new h-h)$, computable in $O(r)$ time. 
	Thus, updating a coordinate $h_{i}$ takes $|\mathcal{P}^\mathcal{T}(v)|O(r)=O(\tau \cdot r)$ time. 
	Summing over all changed coordinates, we can update $h$ in $O(\|\new h-h\|_{0}\cdot\tau \cdot r)$ time.
\end{proof}

\begin{lemma}[$\textsc{Query}$]
\label{lem:simple-query}Suppose \cref{inv:W-simple-invariant}
is satisfied. 
The function $\textsc{Query}(v)$ of \textsc{SimpleSketch} outputs $\Phi_{\chi(v)}\mathcal{W}^{\top}h$ in $O(\tau^{2}r)$ time. 
Moreover, \cref{inv:W-simple-invariant} is preserved 
after the function call.
\end{lemma}

\begin{proof}
\textbf{Correctness:} For the case $v\in\mathcal{S\setminus T}$,
the correctness directly follows by definition of $J_v$. Now, we
consider the case that $v\in\mathcal{T}$. The invariant maintenance
of moving $t_v$ to $\ell$ directly follows by \cref{lem:Z*-formula}.
To output $\Phi_{\chi(v)}\mathcal{W}^{\top}h = y_v$, the function computes the expression as given by \cref{eq:y_v}.

\textbf{Runtime: }For the case $v\in\mathcal{S}\setminus\mathcal{T}$,
the non-zero columns of $J_v$ lies on $\mathcal{P}^{\mathcal{T}}(\mathcal{T}(v))$ by \cref{lem:J-sparsity}. 
By \cref{lem:Linv-sparsity-time}, the term $J_v \cdot L[\ell]^{-\top}$ has $O(\tau \cdot r)$
many nonzero entries and can be computed in time $O(\tau^{2}r)$
time. Thus, we can compute $J_v \cdot L[\ell]^{-\top}h$ in time $O(\tau^{2} \cdot r)$.

For the case $v\in\mathcal{T}$, we first note that we can find $\Delta L$ in $O(\tau^{2})$ time since $|\mathcal{P}^{\mathcal{T}}(v)|\leq\tau$, 
and the column sparsity of $L$ is also bounded by $\tau$ by \cref{lem:L-column-sparsity-pattern}. By the sparsity pattern of $\Delta L$, we can compute
$(\Delta L)\cdot(Z{\triangledown}_v)^{\top}$ in  $O(\tau^{2}r)$ time.
By sparsity pattern of $L$, we can update $Z_v \leftarrow Z^*_v$ and $y^{\triangledown}_v$
in $O(\tau^{2}r)$ time by solving $O(r)$ lower triangular
system using \cref{lem:Linv-sparsity-time}.
In \cref{eq:y_v}, we can compute $Z_v \cdot I_{\mc{P}^{\mc{T}}(v)} \cdot h$ in $O(\tau\cdot r)$ time
since $|\mathcal{P}^{\mathcal{T}}(v)|\leq\tau$. Hence, the function takes $O(\tau^{2} \cdot r)$ time in total.
\end{proof}
\begin{lemma}[$\textsc{UpdateCoordinate}$]
\label{lem:update-coordinate}
Suppose the current state of the data
structure satisfies~\cref{inv:W-simple-invariant}. The
function $\textsc{UpdateCoordinate}$ of \textsc{SimpleSketch} updates the implicit representation
of $\mathcal{W}$ by updating the $i$-th coordinate of $\ox$ from $\ox_{i}$ to $\new{\ox_{i}}$ in $O(\tau^{2}\cdot\eta\cdot r)$ time.
Moreover, the function $\textsc{UpdateCoordinate}$ also updates the internal states correspondingly such that~\cref{inv:W-simple-invariant} is preserved after the function call.
\end{lemma}

\begin{proof}
\textbf{Correctness: }First, we show that for the change on $L$,
it suffices to updates all nodes on the path $\mathcal{P}^{\mathcal{T}}(\low^{\mathcal{T}}(A_{i}))$.
We note that only (iii) of \cref{inv:W-simple-invariant}
depends on the value of $L[\ell]$, so we need to update the sketch only if $(L[\ell+1]-L[t_v])\cdot I_{\mathcal{D}^{\mathcal{T}}(v)\setminus\{v\}}\neq\boldsymbol{0}$.
Since the data structure satisfies the invariants for $\ell$, we have
$(L[\ell]-L[t_v])\cdot I_{\mathcal{D}^{\mathcal{T}}(v)\setminus\{v\}}=\boldsymbol{0}$
for all $v$. Therefore, we need to update the sketch only if 
\[
(L[\ell+1]-L[\ell])\cdot I_{\mathcal{D}^{\mathcal{T}}(v)\setminus\{v\}}\neq\boldsymbol{0}.
\]
We use $\new L$ to denote $L[\ell+1]$ and use $L$ to denote $L[\ell]$,
where $\new L(\new L)^{\top}=AH^{-1}A^{\top}+cA_{i}A_{i}^{\top}$
for some $c$ and $i$. Let $\Delta L=\new L-L$. By \cref{lem:cholesky-update-sparsity-pattern}, the non-zero columns of $\Delta L$ lies on $\mathcal{P}^{\mathcal{T}}(\low^{\mathcal{T}}(A_{i}))$.
We denote $\low^{\mathcal{T}}(A_{i})$ by $u$, and rewrite 
$\Delta L$ as $\sum_{w\in\mathcal{P}^{\mathcal{T}}(u)}(\Delta L)_{w}e_{w}^{\top}$.
For each $w\in\mathcal{P}^{\mathcal{T}}(u)$, we note that $(\Delta L)_{w}e_{w}^{\top}\cdot I_{\mathcal{D}^{\mathcal{T}}(v)\setminus\{v\}}\neq\boldsymbol{0}$
only if $v\in\mathcal{P}^{\mathcal{T}}(w)\setminus w$. Hence, it
suffices to update 
\[
\bigcup_{w\in\mathcal{P}^{\mathcal{T}}(u)}\mathcal{P}^{\mathcal{T}}(w)\setminus w\subseteq\mathcal{P}^{\mathcal{T}}(u).
\]

The function then uses two helper functions \textsc{Update}$L$ and \textsc{Update}$H$, whose correctness and run-time are given in \cref{lem:update-H} and \cref{lem:update-L}.

\textbf{Runtime: }By \cref{lem:cholesky-update-time}, we can find
$\new L$ in $O(\tau^{2})$ time and $L$ changes in $\tau$ columns. By \cref{lem:update-H}, 
$H$ changes in coordinate $i$, so we can update it in $O(\tau^2 \cdot \eta \cdot r)$ time.
Since $L$ changes in $\tau$ columns, we can update $L$ in the data structure in $O(\tau^{3}r)$ time by \cref{lem:update-L}.
Hence, the function takes $O(\tau^{2} \cdot \eta \cdot r)$
time in total.
\end{proof}
\begin{lemma}[$\textsc{Update}H$]
\label{lem:update-H}Suppose \cref{inv:W-simple-invariant} is satisfied. $\textsc{Update}H$ updates $H$ to $\new H$, and implicitly adjusts the sketches in $\mc{S}$ to preserve \cref{inv:W-simple-invariant} in $O(\nnz(\Delta H)\cdot\tau^{2}\cdot \eta \cdot r)$ time. 
\end{lemma}

\begin{proof}
\textbf{Correctness:} We observe that $Z_v$ changes only if $I_{\chi(v)}\cdot\Delta H\neq\boldsymbol{0}$.
Suppose the $i$-th column of $H$ changes, and let $v$ be the node of $\mc{S}$ with $\chi(v)=\{i\}$. 
Then observe that for a change in $H_{i}$, it suffices to update
$\mathcal{P}^{\mathcal{S}}(v)$.

\textbf{Runtime: } 
For a change in $H_i$, let $\widetilde{H}\defeq(H+\Delta H\cdot I_{\{i\}})^{-1/2}-H^{-1/2}$. Then we can find $\Phi_{\chi(v)}\widetilde{H}A^{\top}$
by computing an outer product of a column of $\Phi$ with row of $A^{\top}$, 
which takes $O(\tau \cdot r)$ time by the sparsity pattern of $A$ (\cref{lem:A-sparsity-pattern}). 
Then for a node $v$, we can update $Z_v$ by compute $\Phi_{\chi(v)}\widetilde{H}A^{\top}L[t_v]^{-\top}$,
which takes $O(\tau^{2} \cdot r)$ time by \cref{lem:Linv-sparsity-time}. 
We can then update $y^{\triangledown}_v$
in $O(\tau^{2} \cdot r)$ time. 
As $\mathrm{height}(\mathcal{S}) = \eta$, and we only update along a path to the root, this
function takes $O(\tau^{2}\cdot \eta \cdot r)$ time for the update to $H_i$.
\end{proof}
\begin{lemma}[$\textsc{Update}L$]
\label{lem:update-L}
Given a set $S\subset V(\mathcal{T})$, the function $\textsc{Update}L$ updates $t_v$ to the latest time at each $v \in S$, and adjusts the implicit representation of the sketch at $v$ to preserve \cref{inv:W-simple-invariant}.
If the number of non-zero columns of $\Delta L$ is bounded by $O(\tau)$, then the function takes $O(|S|\cdot\tau^{2} \cdot r)$ time. 
\end{lemma}

\begin{proof}
\textbf{Correctness: }The correctness directly follows by \cref{lem:Z*-formula}.

\textbf{Runtime: }By the sparsity pattern of\textbf{ $L$ }(\cref{lem:L-column-sparsity-pattern}), we can compute $(L[\ell]-L[t_v])\cdot I_{\mathcal{P}^{{\tt }}(\mathcal{T}(v))}\cdot (Z_v)^{\top}$
and $\Delta L\cdot (Z_v)^{\top}$ in $O(\tau^{2} r )$ time,
and the column sparsity pattern of the result is on a path in  $\mathcal{T}$.
Then, we can update $Z_v$ and $y^{\triangledown}_v$ in  $O(\tau^{2}r)$ time
by solving $O(r)$ many lower triangular systems using 
\cref{lem:Linv-sparsity-time}. Hence, the total time is bounded by
$O(|S|\cdot\tau^{2} r)$.
\end{proof}

\subsection{\label{subsec:balanced-sampling-tree}Sketching the Multiscale Representation via Balanced Sampling Tree}
% Guanghao: I think the block version directly follows
% but still need to check it carefully.

\begin{table}[H]
\begin{centering}
\begin{tabular}{|c|c|c|}
\hline 
 & Simple Sampling Tree & Balanced Sampling Tree\tabularnewline
\hline 
Sampling tree height & $\wt O(\tau)$ & $\wt O(1)$\tabularnewline
\hline
JL dimension &  $\wt O(\tau^2)$ & $\wt O(1)$\tabularnewline
\hline 
Initialization time & $\wt O(n \tau^4 )$ & $\wt O(n\tau^{2} )$\tabularnewline
\hline 
Update $\mathcal{W}$ time & $\wt O(\|\new \ox-\ox\|_{0}\cdot\tau^{5} )$ & $\wt O(\|\new \ox-\ox\|_{0}\cdot\tau^{2} )$\tabularnewline
\hline 
Update $h$ time & $\wt O(\|\new h-h\|_{0}\cdot\tau^3 )$ & $\wt O(\| \new h-h\|_0)$\tabularnewline
\hline 
Query time  & $\wt O(\tau^{4} )$ & $\wt O(\tau^{2} )$\tabularnewline
\hline 
Query time $\times$ tree height & $\wt O(\tau^{5}  )$ & $\wt O(\tau^{2} )$\tabularnewline
\hline 
\end{tabular}
\par\end{centering}
\caption{Comparison between two sampling trees.}
\end{table}

Combining the simple sampling tree data structure with our IPM algorithm will give us a \\
$\wt O(n\tau^5\log(1/\eps))$ algorithm for solving LPs.
To make our algorithm competitive when $\tau$ is large, 
we demonstrate how to further speed up~(\cref{alg:simple-W-sketch,alg:simple-W-sketch-1}) to $\wt O(\tau^{2})$
per step in this section. More specifically, we have the following theorem:
\begin{theorem}
\label{thm:balanced-W-sketch}Given the constraint matrix $A$, its elimination tree $\mathcal{T}$ with height $\tau$, a 
JL matrix $\Phi\in\R^{r\times n}$, 
and a sampling tree $(\mathcal{S},\chi)$ constructed as in \cref{subsec:balanced-sampling-tree-construction} with height $O(\log n)$,
the data structure $\textsc{BalancedSketch}$
(\cref{alg:balanced-W-sketch,alg:balanced-W-sketch-1}) maintains $\Phi_{\chi(v)}\mathcal{W}^{\top}h$ for each $v \in V(\mathcal{S})$ through the following operations:
\begin{itemize}
\item $\textsc{Initialize}(\mathcal{S},\chi,\Phi,\ox,h)$: Initializes the data structure in $O(n\tau^{2}r \log n)$
time, so that each node $v \in \mc{S}$ maintains the sketch $\Phi_{\chi(v)} \mc{W}^\top h = \Phi_{\chi(v)} H_{\ox}^{-1/2}A^\top L_{\ox}^\top$.
\item $\textsc{Update}(\new{\ox},\new h)$: Updates all sketches in $\mathcal{S}$ implicitly to reflect $\mathcal{W}$ updating to $\new{\mathcal{W}}$ and $h$ updating to $\new h$
in $O(\|\new{\ox}-\ox\|_{0}\cdot\tau^{2}r\log n)+O(\|\new h - h\|_{0}\cdot r \log n)$ time, where $\new{\mathcal{W}}$ is given implicitly by $\new{\ox}$. 
\item $\textsc{Query}(v)$: Outputs $\Phi_{\chi(v)}\mathcal{W}^{\top}h$
in $O(\tau^{2} \cdot r)$ time.
\end{itemize}
\end{theorem}

We observe that the data structures in \cref{subsec:simple-sampling-tree}
has the same $\Omega(\tau^{3})$ bottleneck for both updating the multiscale representation and sampling: 
The data structures are always operating on some path of the sampling tree $\ts$, where we need to solve the lower triangular systems for each node on that path. 
In this section, we show how to obtain a balanced
sampling tree with $O(\log n)$ height and thus speeds up each operation to $\wt O(\tau^{2})$. 

\subsubsection{Balanced Sampling Tree Construction} \label{subsec:balanced-sampling-tree-construction}

As a first ingredient in our construction is the following lemma from Sleator and Tarjan's heavy-light decomposition.
\begin{lemma}[{Heavy-Light Decomposition \cite{DBLP:journals/jcss/SleatorT83}}]
	\label{lem:heavy-light-decomposition}
	Given a rooted tree $\mathcal{T}$,
	there exists an ordering of vertices of $V(\mathcal{T})$ such that the path between any two vertices consists of at most $O(\log n)$ many contiguous subsequences of the ordering, 
	and for any vertices $v$,
	the subtree rooted at $v$ corresponds to a single contiguous
	subsequence of the order. Moreover, such an ordering can be found in $O(n)$ time. \qed
	
	\begin{figure}
		\begin{centering}
			\includegraphics[width=10cm]{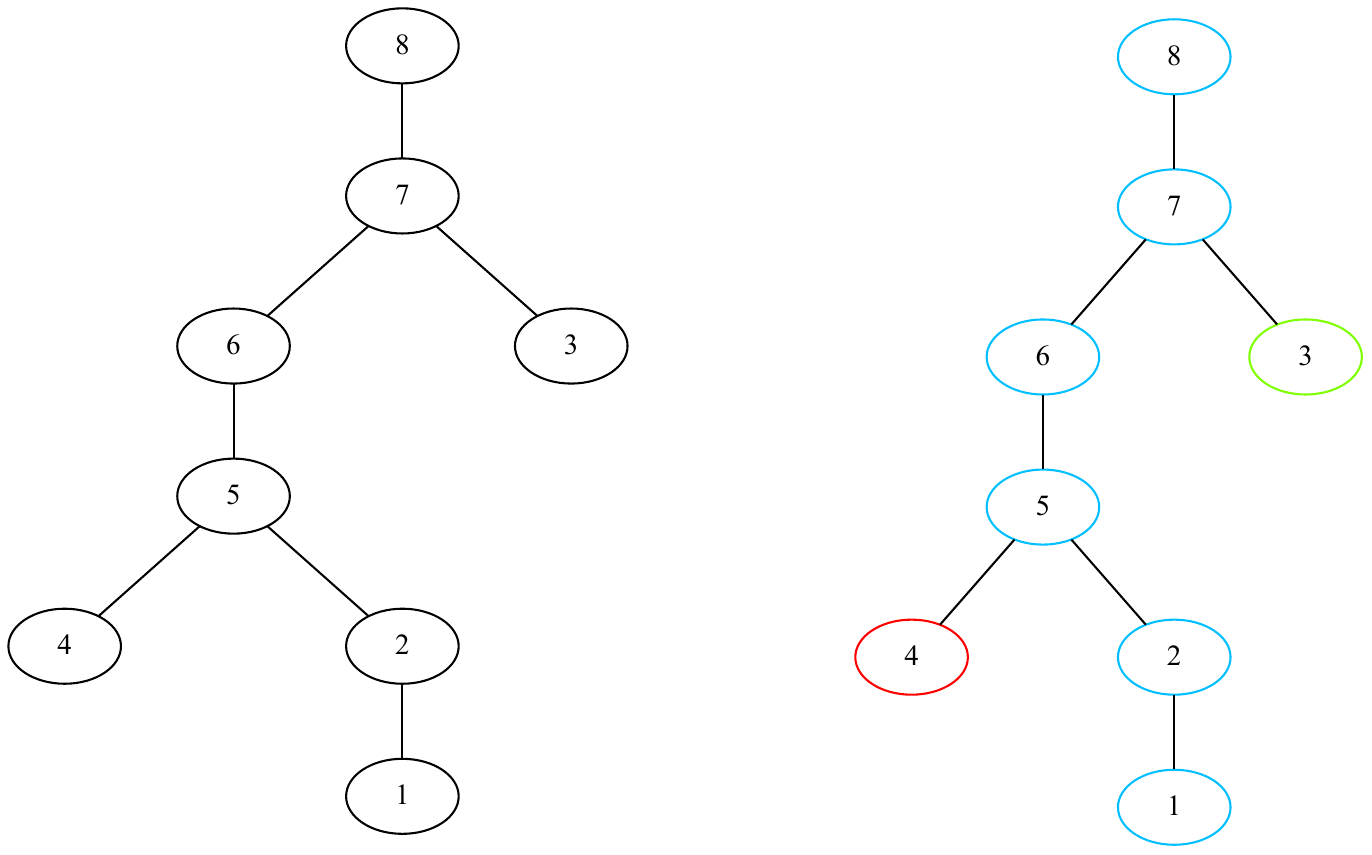}
			\par\end{centering}
		\caption{\label{fig:heavy-light}A rooted tree is given on the left, its heavy-light decomposition is shown on the right; the ordering is $[8,7,6,5,2,1,4,3]$.}
	\end{figure}
	\end{lemma}

First, we construct a sampling tree of $\R^{d}$, denoted by $(\mathcal{B},\overline{\chi})$, as follows:
We perform heavy-light decomposition on the elimination tree ${\mathcal{T}}$ with vertex set $[d]$. Let $\sigma_1,\dots,\sigma_d$ denote the vertices ordered according to \cref{lem:heavy-light-decomposition}.
Let $\mathcal{B}$ be a complete binary tree containing $d$ leaves, where the $i$-th leaf is $\sigma_i \in [d]$. 
We set $\overline{\chi}(\sigma_i) = \{\sigma_i\}$. 
For each non-leaf node $v \in\mathcal{B}$, we let $\overline{\chi}(v)=\overline{\chi}(\text{left child of \ensuremath{v}})\cup\overline{\chi}(\text{right child of \ensuremath{v}})$.
It is easy to check $(\mathcal{B},\overline{\chi})$ is a sampling
tree of $\R^{d}$ by \cref{def:sampling-tree}.

Now, we extend the sampling tree $(\mathcal{B},\overline{\chi})$
on $\R^{d}$ to $\R^{n}$ to obtain the \emph{balanced sampling tree $(\mathcal{S},\chi)$}: 
At each leaf node $v\in\mathcal{B}$, we add a complete binary tree with vertex set  %TODO: ugh low(A) is a number and chi(v) is a set
\[
\{i\in[n]\mid\text{coordinate $i$ in $j$-th block and }\block_low^{\mathcal{T}}(A_{j})=\overline{\chi}(v)\},\]
where $\block_low(A_j)$ is defined by  
\[
\block_low(A_j) =\arg \max_{i \in \{i\mid A_{ij}\neq \boldsymbol{0}\}} \depth(i).
\]%TODO: not clear $A_{ij}$ is block vector here
We denote this modified binary tree as $\mathcal{S}$. 
Then, each leaf node $u$ of $\mathcal{S}$ corresponds to some $i\in[n]$, and we set $\chi(u)=\{i\}$. We can check that 
for any leaf node $v\in\text{\ensuremath{\mathcal{B}}}$,
\[
	\chi(v)=\{i\in[n]\mid \text{coordinate $i$ in $j$-th block and }\block_low^{\mathcal{T}}(A_{j})=\overline{\chi}(v)\}.
\]
We define this $(\mathcal{S},\chi)$ to be our \emph{balanced sampling tree}. An example is shown in \cref{fig:balanced-sampling-tree-example}. 

Since the height of $\mathcal{B}$ is $\log d$, and the height of the newly
added binary trees are at most $\log n$, the height of $\mathcal{S}$ is $O(\log n)$.

\begin{theorem}\label{thm:balanced-sampling-tree-construct-time}
	Given an elimination tree $\mathcal{T}$ with height $\tau$, the balanced sampling tree can be constructed in $O(n\tau+n\log n)$ time.
\end{theorem}
\begin{proof}
	By \cref{lem:heavy-light-decomposition}, we find the heavy-light-decomposition order in linear time.
	Since $\mathcal{B}$ is a binary tree on this order, we can construct $(\mathcal{B},\overline{\chi})$ in $O(d\log d)$ time.
	For each block $A_j$ where $j\in [m]$, we can find $\block_low(A_j)$ in time $O(\tau)$ since $\nnz(A_j)=O(\tau)$ (\cref{lem:A-sparsity-pattern}). 
	Hence, we can find $\block_low(A_j)$ for all $j\in [m]$ in time $O(n\tau)$. This gives us $\chi(v)$ for all $v\in \mathcal{B}$.
	Finally, we can construct $(\mathcal{S},\chi)$ in time $O(n\log n)$.
	Hence, we can construct $(\mathcal{S},\chi)$ in $O(n\tau + n\log n)$ time.
\end{proof}

The balanced sampling tree does not preserve ancestor-descendant relationship of the vertices of $\mc{T}$. However, the following lemma about the Heavy-Light
Decomposition will help us get something close.

\begin{figure}
\begin{centering}
\includegraphics[width=16cm]{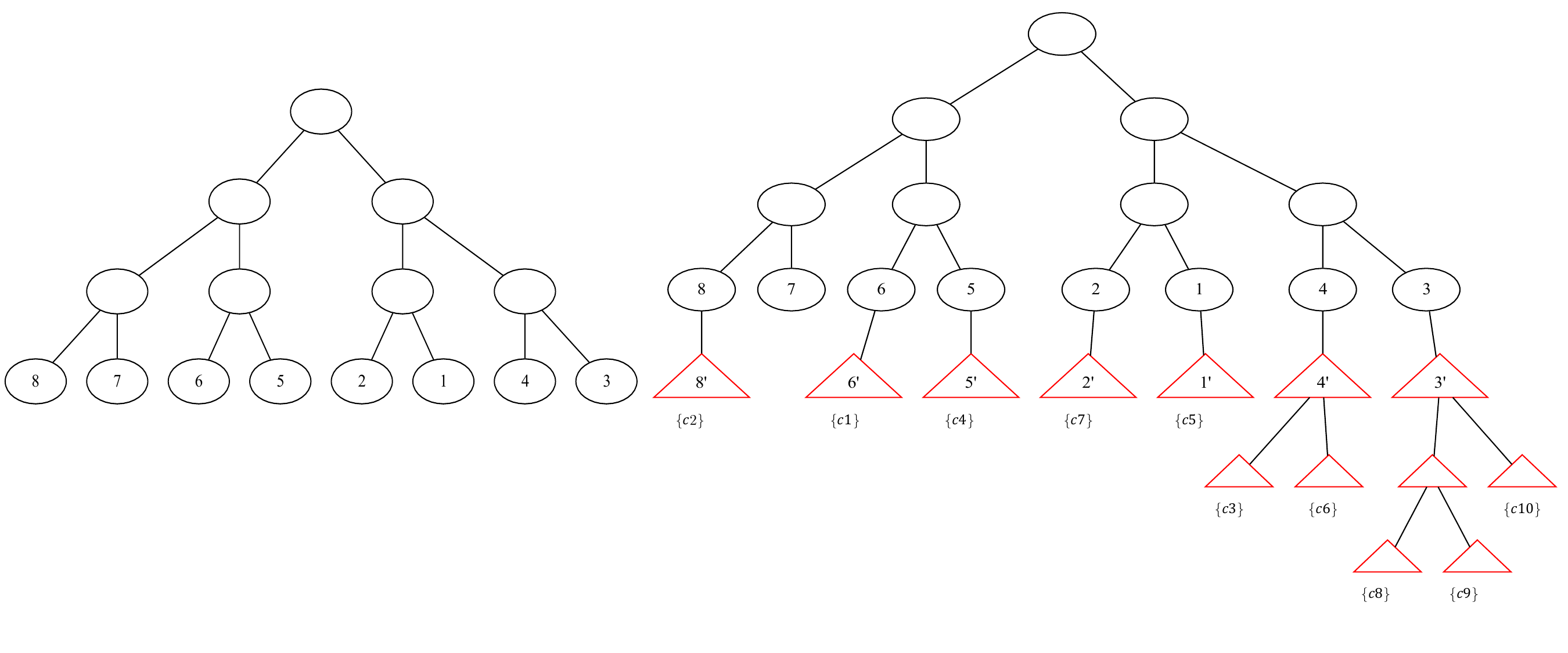}
\par\end{centering}
\caption{Examples of tree $\mathcal{B}$ (left) and $\mathcal{S}$ (right) constructed
from the elimination tree in \cref{fig:heavy-light}. 
The red triangle in the right graph denotes the newly added nodes in the tree.
The bracket under each leaf node denotes the corresponding column maintained by that node.}\label{fig:balanced-sampling-tree-example}

\end{figure}

\begin{lemma}
\label{lem:heavy-light-intersection-size}Let the sequence $a_{1},a_{2},\ldots,a_{n}$
be the order produced by Heavy-Light Decomposition on tree ${\mathcal{T}}$.
For any contiguous subsequence $a_{l},a_{l+1},\ldots,a_{r}$, we have
\[
\left|\left(\bigcup_{i\in[l,r]}\mathcal{P}^{{\mathcal{T}}}(a_{i})\right)\cap\left(\bigcup_{i\in[1,n]\setminus[l,r]}\mathcal{P}^{{\mathcal{T}}}(a_{i})\right)\right|\leq2\cdot\mathsf{height({\mathcal{T}).}}
\]
\end{lemma}

\begin{proof}
It suffices to show that for any four numbers $l_{1},l_{2},r_{1},r_{2}$
where $l_{1}\leq l_{2}\leq r_{2}\leq r_{1}$, we have $\mathcal{P}(a_{l_{1}})\cap\mathcal{P}(a_{r_{1}})\subseteq\mathcal{P}(a_{l_{2}})\cap\mathcal{P}(a_{r_{2}})$.
Indeed, when this is true, we have
\begin{align*}
\left|\left(\bigcup_{i\in[l,r]}\mathcal{P}(a_{i})\right)\cap\left(\bigcup_{i\in[n]\setminus[l,r]}\mathcal{P}(a_{i})\right)\right| & \leq\left|\left(\bigcup_{i\in[1,l-1]}\mathcal{P}(a_{i})\right)\cap\left(\bigcup_{i\in[l,r]}\mathcal{P}(a_{i})\right)\right|\\
 & \qquad+\left|\left(\bigcup_{i\in[l,r]}\mathcal{P}(a_{i})\right)\cap\left(\bigcup_{i\in[r+1,n]}\mathcal{P}(a_{i})\right)\right|,
\end{align*}
and the two terms on the right-hand side can be bounded by $|\mathcal{P}(a_{l-1})\cap\mathcal{P}(a_{l})|$
and $|\mathcal{P}(a_{r})\cap\mathcal{P}(a_{r+1})|$ respectively.

Note that $\mathcal{P}(x)\cap\mathcal{P}(y)=\mathcal{P}(\mathsf{LCA}(x,y))$,
where $\lca(x,y)$ denotes the \emph{lowest common ancestor} of $x$
and $y$ in tree ${\mathcal{T}}$. Since the ordering $a_{1,}a_{2,}\dots,a_{n}$
is produced by a depth-first traversal on $\mathcal{T},$ we observe
that the subtree rooted at $\lca(a_{l_{1}},a_{r_{1}})$ contains $a_{l_{2}},a_{r_{2}}$,
since they are both discovered during the DFS after $a_{l_{1}}$and
before $a_{r_{1}}$; consequently it contains $\lca(a_{l_{2}},a_{r_{2}})$.
It also by definition contains $a_{l_{1}}$and $a_{r_{1}}$. Therefore,
$\lca(a_{l_{1}},a_{r_{1}})$ is an ancestor of $\lca(a_{l_{2}},a_{r_{2}})$,
and it follows that $\mathcal{P}(\lca(a_{l_{1}},a_{r_{1}}))\subseteq\mathcal{P}(\lca(a_{l_{2}},a_{r_{2}})).$
\end{proof}
For the sampling tree $(\mathcal{B},\overline{\chi})$, we have the
following lemma:
\begin{lemma}[{See e.g. \cite{DBLP:books/lib/BergCKO08}}]
\label{lem:segment-tree-query}Given the complete binary sampling
tree $(\mathcal{B},\overline{\chi})$, 
let $a_{k} = \overline{\chi}(v)$, 
where $v$ is the $k$-th leaf of $\mathcal{B}$. 
For any contiguous subsequence $a_{l},a_{l+1},\ldots,a_{r}$ of the sequence $a_1,\dots,a_d$, we can find a node set $S \subseteq \mathcal{B}$ of size $O(\log d)$ such that 
\[
\bigcup_{u\in S}\overline{\chi}(u)=\{a_{i}\mid i\in[l,r]\}.
\]
 Moreover, this set $S$ can be found in $O(\log d)$ time.
\end{lemma}

\subsubsection{Data Structure for Sketching}

As our balanced sampling tree $\mathcal{S}$ does not totally preserve the ancestor-descendant relationships in $\mathcal{T}$,
we need a more complex maintenance scheme.
We first observe that for any node $v\in\mathcal{S}\setminus\mathcal{B}$,
the nonzero columns of $\Phi_{\chi(v)}H^{-1/2}A^{\top}$
lies on a path of $\mathcal{T}$. Therefore, given $J_v=\Phi_{\chi(v)}H^{-1/2}A^{\top}$,
the term $\Phi_{\chi(v)}\mathcal{W}^{\top}$ can be computed
in $\wt O(\tau^{2})$ time.
In the last section, for each node $v\in\mathcal{T}$, we delay $L$'s
updates in the columns that lie on $\mathcal{P}^{\mathcal{T}}(v)$. Here,
we define its analogy on $\mathcal{B}$:

\begin{definition}[{$\Lambda(v),\overline{\Lambda}(v)$}]
Let $\Lambda:\mathcal{B}\to2^{\mathcal{T}}$ be the function 
\[
\Lambda(v)\defeq\left(\bigcup_{i\in\overline{\chi}(v)}\mathcal{P}^{{\mathcal{T}}}(a_{i})\right)\cap\left(\bigcup_{i\in\mathcal{T}\setminus\overline{\chi}(v)}\mathcal{P}^{{\mathcal{T}}}(a_{i})\right).
\]
We also define $\overline{\Lambda}(v):\mathcal{B\to}2^{\mathcal{T}}$
to be the set of columns that are maintained up-to-date for each $v$:
\[
\overline{\Lambda}(v)\defeq\left(\bigcup_{i\in\overline{\chi}(v)}\mathcal{P}^{{\mathcal{T}}}(a_{i})\right)\setminus\Lambda(v).
\]
\end{definition}

\begin{figure}
	\begin{center}
		\includegraphics[width=15cm]{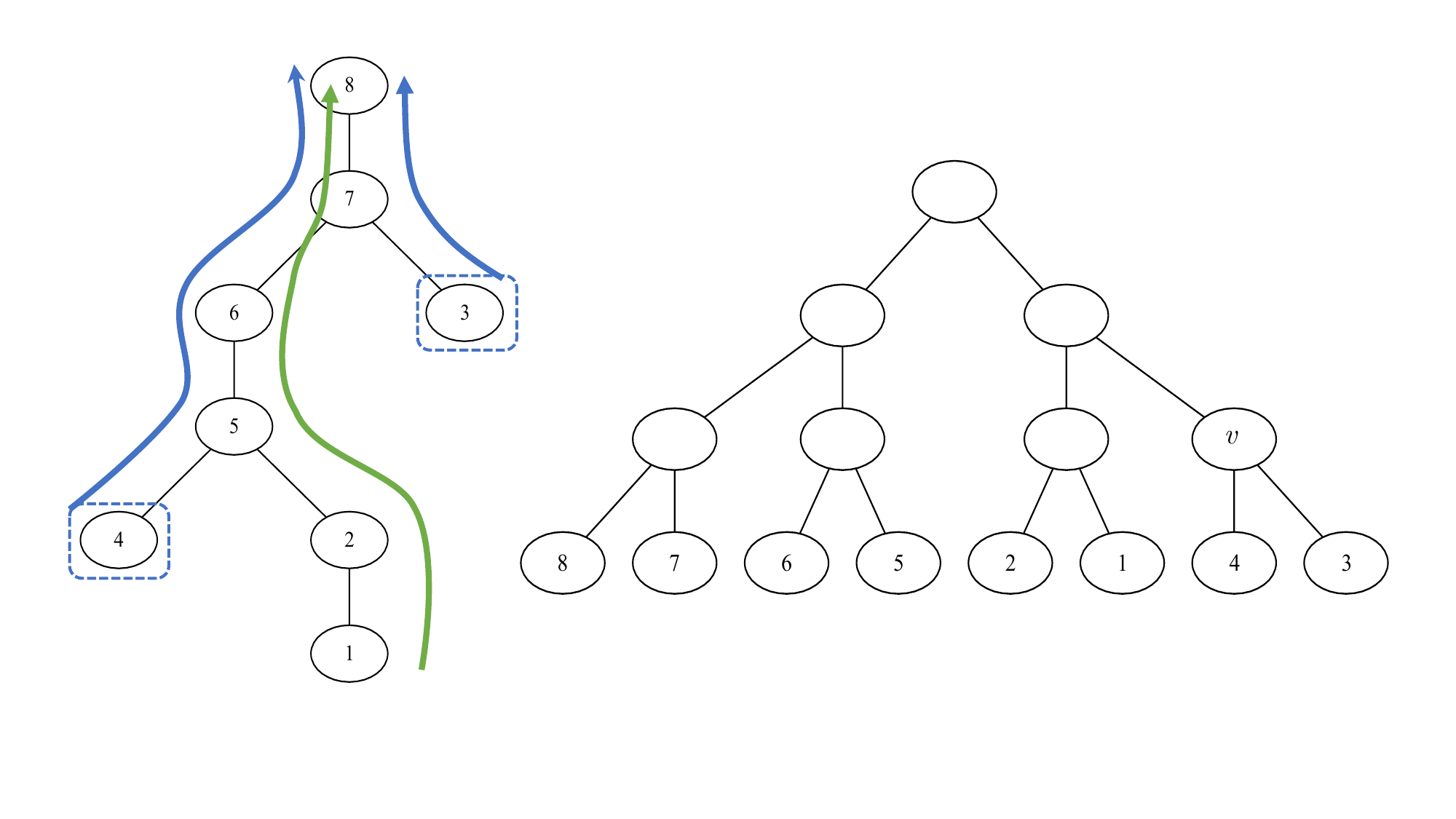}
	\end{center}
	\caption{Example of $\Lambda(v)$ and $\overline{\Lambda}(v)$: on the left is $\mathcal{B}$, where $\Lambda(v)=\{5,6,7,8\}$ is the set of nodes crossed by both the green path and blue path. $\overline{\Lambda}(v)=\{3,4\}$ is the set of nodes contained by blue boxes.}
	\end{figure}
\begin{lemma}\label{lem:Lambda-ancestors}
For any nodes $u,v\in\mathcal{B}$, if $u\in\mathcal{P}^{\mathcal{B}}(v)$,
then $\overline{\Lambda}(v)\subseteq\overline{\Lambda}(u)$.
\end{lemma}

\begin{proof}
By definition of $\overline{\Lambda}(\cdot)$ and $\Lambda(\cdot)$,
we have 
\[
\overline{\Lambda}(v)=\left(\bigcup_{i\in\overline{\chi}(v)}\mathcal{P}^{{\mathcal{T}}}(a_{i})\right)\setminus\left(\bigcup_{i\in\mathcal{T}\setminus\overline{\chi}(v)}\mathcal{P}^{{\mathcal{T}}}(a_{i})\right).
\]
Since
$(\mathcal{B},\overline{\chi})$ is a sampling tree, we have $\overline{\chi}(v)\subseteq\overline{\chi}(u)$. 
Hence, we complete the proof by noting $\left(\bigcup_{i\in\overline{\chi}(v)}\mathcal{P}^{{\mathcal{T}}}(a_{i})\right)\subseteq\left(\bigcup_{i\in\overline{\chi}(u)}\mathcal{P}^{{\mathcal{T}}}(a_{i})\right)$
and $\left(\bigcup_{i\in\mathcal{T}\setminus\overline{\chi}(u)}\mathcal{P}^{{\mathcal{T}}}(a_{i})\right)\subseteq\left(\bigcup_{i\in\mathcal{T}\setminus\overline{\chi}(v)}\mathcal{P}^{{\mathcal{T}}}(a_{i})\right)$.
\end{proof}
\begin{definition}[$\Lambda^{\clubsuit}$]
	Finally, we define the function $\Lambda^{\clubsuit}:\mathcal{T}\to\mathcal{B}$.
\[
\Lambda^{\clubsuit}(u)\defeq\low^{\mathcal{B}}(\{v\in\mathcal{B}\mid u\in\overline{\Lambda}(v)\}).
\]
Equivalently, we have
\[
\Lambda^{\clubsuit}(u)=\low^{\mathcal{B}}(\{v\in\mathcal{B}\mid\mathcal{D}^{\mathcal{T}}(u)\subseteq\overline{\chi}(v)\}).
\]
\end{definition}
Using this equivalent definition, we show that $\Lambda^{\clubsuit}(\cdot)$
is well-defined. 
\begin{lemma}
\label{lem:Lambda-club-path}For any $u\in\mathcal{T}$, the set $\{v\in\mathcal{B}\mid\mathcal{D}^{\mathcal{T}}(u)\subseteq\overline{\chi}(v)\}$
is a path on $\mathcal{B}$.
\end{lemma}

\begin{proof}
Recall that $\{\overline{\chi}(v)\mid v\in\mathcal{B}\text{ and }\depth(v)=k\}$
forms a partition of $\mathcal{T}$ for any $k\leq\text{height}(\mathcal{B})$.
Then, for any $k\leq\text{height}(\mathcal{B})$, there is at most
one node satisfying both $\mathcal{D}^{\mathcal{T}}(u)\subseteq\overline{\chi}(v)$
and $\depth(v)=k$. We complete the proof by note that if $\mathcal{D}^{\mathcal{T}}(u)\subseteq\overline{\chi}(v)$, 
then $\mathcal{D}^{\mathcal{T}}(u)\subseteq\overline{\chi}(w)$ for
any $w\in\mathcal{P}^{\mathcal{B}}(v)$ since $(\mathcal{B},\overline{\chi})$
is a sampling tree.
\end{proof}
Now, we show the equivalence of two definition:
\begin{lemma}
For any $u\in\mathcal{T}$ and $v\in\mathcal{B}$, we have $u\in\overline{\Lambda}(v)$
if and only if $\mathcal{D}^{\mathcal{T}}(u)\subseteq\overline{\chi}(v)$.
\end{lemma}

\begin{proof}
The only if direction: Suppose $u\in\overline{\Lambda}(v)$ and there
is a $w\in\mathcal{D}^{\mathcal{T}}(u)$ but $w\notin\overline{\chi}(v)$.
We note that $w\not\in\overline{\chi}(v)$ implies $w\in\mathcal{T}\setminus\overline{\chi}(v)$.
Then, we have $u\in\Lambda(v)$ since $u\in\mathcal{P}^{\mathcal{T}}(w)$.
This contradicts with our assumption that $u\in\overline{\Lambda}(v).$

The if direction: Since $\mathcal{D}^{\mathcal{T}}(u)\subseteq\overline{\chi}(v)$,
we have $u\in\overline{\chi}(v)\subseteq\left(\bigcup_{i\in\overline{\chi}(v)}\mathcal{P}^{{\mathcal{T}}}(a_{i})\right)$.
Then, it suffices to show $u\not\in\Lambda(v)$. Suppose $u\in\Lambda(v)$,
then there is a node $w\in\mathcal{D}^{\mathcal{T}}(v)$ such that
$w\notin\overline{\chi}(v)$. This contradicts with our assumption
that $\mathcal{D}^{\mathcal{T}}(u)\subseteq\overline{\chi}(v)$.
\end{proof}
Intuitively, for each node $v\in\mathcal{B}$, we need to maintain
the union of paths on the interval $[l_{v},r_{v}]=\overline{\chi}(v)$,
i.e.\,$\bigcup_{i\in[l_{v},r_{v}]}\mathcal{P}^{{\mathcal{T}}}(a_{i})$.
The set $\Lambda(v)$ to denote the set of nodes in $\mathcal{T}$
shared by other nodes the same level of binary tree $\mathcal{B}$.
\cref{lem:heavy-light-intersection-size} shows that $|\Lambda(v)|=O(\tau)$.
Hence, we never maintain them exactly in the sampling tre, but rather compute them as needed. On the other hand, for each node $v\in\mathcal{B}$,
we maintain the node $u\in\mathcal{T}$ exactly only if $a_{i}\in\mathcal{D}^{\mathcal{T}}(u)$
for $i\in[l_{v},r_{v}]$. Thus, for each node $v\in\mathcal{T}$,
it is only been explicitly maintained on a path of $\mathcal{B}$,
and $\Lambda^{\clubsuit}(v)$ denotes the lower end of that path.
In particular, we have the following lemma about $\Lambda^{\clubsuit}(\cdot)$:
\begin{lemma}
\label{lem:Lamda-clubsuit-path}Let $u,v\in\mathcal{T}$. If $u\in\mathcal{P}^{\mathcal{T}}(v)$,
then $\Lambda^{\clubsuit}(u)\in\mathcal{P}^{\mathcal{B}}(\Lambda^{\clubsuit}(v))$.
\end{lemma}

\begin{proof}
We note that $\mathcal{D^{\mathcal{T}}}(u)\subseteq\overline{\chi}(w)$
implies $\mathcal{D}^{\mathcal{T}}(v)\subseteq\overline{\chi}(w)$
for any $w\in\mathcal{B}$. Then, we have $\{w\in\mathcal{B}\mid\mathcal{D}^{\mathcal{T}}(u)\subseteq\overline{\chi}(w)\}\subseteq\{w\in\mathcal{B}\mid\mathcal{D}^{\mathcal{T}}(v)\subseteq\overline{\chi}(w)\}$.
Hence, by the equivalent definition of $\Lambda^{\clubsuit}$ and
\cref{lem:Lambda-club-path}, we have $\Lambda^{\clubsuit}(u)\in\mathcal{P}^{\mathcal{B}}(\Lambda^{\clubsuit}(v))$.
\end{proof}
Similar to the proof of~\cref{subsec:simple-sampling-tree}, ideally, we want to maintain 
\[
Z^{*}_v\defeq\Phi_{\chi(v)}H^{-1/2}A^{\top}L[{\ell}]^{-\top}\quad\text{ and }y^{*}_v\defeq Z^{*}_v \cdot h.
\]
To do so, we make use of the following properties:
\begin{invariant}\label{inv:w-balanced-invariant}
	The variables maintained in the data structure $\textsc{BalancedSketch}$, as given in \cref{alg:balanced-W-sketch}, preserve the following invariants before and after each function call:
	\begin{alignat*}{2}
		J_v & =\Phi_{\chi(v)}H^{-1/2}A^{\top} & v\in\mathcal{S} \tag{i} \\
Z_v & =J_v\cdot L[t_v]^{-\top} & v\in\mathcal{B}\tag{ii} \\
\boldsymbol{0} & =(L[\ell]-L[t_v])\cdot I_{\overline{\Lambda}(v)} & v\in\mathcal{B}\tag{iii} \\
y^{\triangledown}_v & =Z_v(I-I_{\Lambda(v)})h^{(\ell)} \qquad & v\in\mathcal{B} \tag{iv}
	\end{alignat*}
\end{invariant}

\begin{lemma}[{Sparsity Pattern of $Z_v$}]
\label{lem:balanced-Z-sparsity} Suppose $J_v$ and $Z_v$ satisfies (i) and (ii) of \cref{inv:w-balanced-invariant} for some $v\in\mathcal{B}$. 
Let $S$ be the index set of the non-zero columns
of $Z_v$, e.g.~$S=\{i\in[d]\mid (Z_v)_{i}\neq\boldsymbol{0}\}$,
then we have 
\[
S\subseteq\bigcup_{i\in\overline{\chi}(v)}\mathcal{P}^{{\mathcal{T}}}(a_{i}).
\]
\end{lemma}

\begin{proof}
By our construction of $J_v$, we note that nonzero columns of $J_v$
lies on $\bigcup_{i\in\overline{\chi}(v)}\mathcal{P}^{{\mathcal{T}}}(a_{i})$.
By \cref{lem:Linv-sparsity-time}, we have $S\subseteq\bigcup_{i\in\overline{\chi}(v)}\mathcal{P}^{{\mathcal{T}}}(a_{i})$.
\end{proof}
We have the following relationship between $Z_v$ and $Z^{*}_v$ for any $v \in \mathcal{T}$:
\begin{lemma}
\label{lem:balanced-Z*-formula}Suppose \cref{inv:w-balanced-invariant}
is satisfied for $v$, then we have 
\[
Z^{*}_v=Z_v-\left(L[\ell]^{-1}(L[\ell]-L[t_v])\cdot I_{\Lambda(v)}Z_v^{\top}\right)^{\top}.
\]
\end{lemma}

\begin{proof}
Similar to the proof of \cref{lem:Z*-formula}, we have 
\[
Z^{*}_v=Z_v-\left(L[\ell]^{-1}(L[\ell]-L[t_v])Z_v^{\top}\right)^{\top}.
\]
Let $\Delta L\defeq L[\ell]-L[t_v]$. Then, we can split $\Delta L$
into three parts:

\[
\Delta L=(I_{\Lambda(v)}+I_{\overline{\Lambda}(v)}+I_{\mathcal{T}\setminus(\Lambda(v)\cup\overline{\Lambda}(v))})\Delta L.
\]
We first note that $I_{\mathcal{T}\setminus(\Lambda(v)\cup\overline{\Lambda}(v))}\cdot Z_v^{\top}=\boldsymbol{0}.$
By \cref{lem:balanced-Z-sparsity}, the nonzero columns of $Z_v$
lies on \[
\bigcup_{i\in\overline{\chi}(v)}\mathcal{P}^{{\mathcal{T}}}(a_{i})=(\Lambda(v)\cup\overline{\Lambda}(v)).
\]
Hence, $I_{\mathcal{T}\setminus(\Lambda(v)\cup\overline{\Lambda}(v))}\cdot Z_v^{\top}=\boldsymbol{0}.$
By (iii) of \cref{inv:w-balanced-invariant}, we have $(L[\ell]-L[t_v])\cdot I_{\overline{\Lambda}(v)}=\boldsymbol{0}$.
Thus, we have 
\[
L[\ell]^{-1}(L[\ell]-L[t_v])Z_v^{\top}=L[\ell]^{-1}(L[\ell]-L[t_v])\cdot I_{\Lambda(v)}Z_v^{\top}.
\]
\end{proof}

\begin{algorithm}
	\caption{Balanced Multiscale Representation Sketching Data Structure --
		Initialize and Query\label{alg:balanced-W-sketch}}
	
	\algnewcommand{\LeftComment}[1]{\State \(\triangleright\) #1} 

	\begin{algorithmic}[1]
		
		\State \textbf{datastructure }\textsc{BalancedSketch}
		
		\State \textbf{private : members}
		
		\State \hspace{4mm} $\Phi\in\R^{r\times n}$ \Comment{JL
			matrix}
		
		\State \hspace{4mm}  sampling tree $(\mathcal{S},\chi)$ with balanced binary tree $\mathcal{B}$
		\Comment{constructed as in \cref{subsec:balanced-sampling-tree-construction}}
		
		\State \hspace{4mm} $\ell\in\mathbb{N}$ \Comment{iteration counter}
		
		\State \hspace{4mm} $h \in\R^{d}$
		
		\State \hspace{4mm} $\ox \in \R^{n}$ \Comment{$\mathcal{W}$ given implicitly by $\ox$}
		
		\State \hspace{4mm} $H\in \R^{n\times n}$\Comment{Hessian $H=\nabla^2 \phi(\ox)$}
		
		\State \hspace{4mm} {\sc List}$\{L[t]\in\R^{d\times d}\}_{t\geq 0}$ 
		\Comment{sequence of  Cholesky factor $L$ at various iterations $t$}

		\State \hspace{4mm} {\sc List} $\{J_v\in \R^{r\times d}\}_{v\in \mathcal{S}}$
		\Comment{$J_v=\Phi_{\chi(v)}H^{-1/2}A^{\top}$}
		
		\State \hspace{4mm} {\sc List} $\{Z_v\in \R^{r\times d}\}_{v\in \mathcal{B}}$
		\Comment{$Z_v=\Phi_{\chi(v)}H^{-1/2}A^{\top}L[t_v]^{-\top}$}

		\State \hspace{4mm} {\sc List} $\{y^\triangledown_v\in \R^{r}\}_{v\in \mathcal{B}}$
		\Comment{$y^{\triangledown}_v = Z_v(I-I_{\Lambda(v)})h$}

		\State \hspace{4mm} {\sc List} $\{t_v \in \mathbb{N}\}_{v\in \mathcal{B}}$
		\Comment{Last iterations during which node $v$ was updated}
		
		% \State \hspace{4mm} $Z[v_{1},v_{2},\ldots,v_{|\mathcal{B}|}]\in\R^{r\times d}$\Comment{
		% 	Maintains $Z_v=\Phi_{\chi(v)}H_{\ell}^{-1/2}A^{\top}L[t_v]^{-\top}$
		% }
		
		% \State \hspace{4mm} $y_{\triangledown}[v_{1},v_{2},\ldots,v_{|\mathcal{B}|}]\in\R^{r}$
		% \Comment{An array index by node in $V$}
		
		% \State \hspace{4mm} $t[v_{1},v_{2},\ldots,v_{|\mathcal{B}|}]\in\mathbb{N}$
		% \Comment{Time stamp for last time update $Z$ and $y_{\triangledown}$
		% }
		\State \textbf{end members}
		
		\Procedure{$\textsc{Initialize}$}{$\mathcal{S},\chi, \Phi\in\R^{r\times n},\ox\in\R^{n},h\in\R^{d}$}\Comment{\cref{lem:balanced-initialization}}
		
		\State $(\mathcal{S},\chi)\leftarrow (\mathcal{S},\chi)$
	
		\State $\Phi\leftarrow\Phi$

		\State $\ell\leftarrow0,h\leftarrow h$
		
		\State Find $\Lambda(v)$ for all $v\in\mathcal{B}$
		
		\State Compute $H\leftarrow\nabla^{2}\phi(\ox)$
		
		\State Find lower Cholesky factor $L[\ell]$ of $AH^{-1}A^{\top}$

		\ForAll{$v\in\mathcal{S}$}
		
			\State $J_v\leftarrow \Phi_{\chi(v)}H^{-1/2}A^\top$ \Comment{compute $J_v$ for all $v\in \mathcal{S}$}

		\EndFor

		\ForAll{$v\in\mathcal{B}$}
		
		\State $Z_v\leftarrow J_v L[\ell]^{-\top}$
		
		\State $y^{\triangledown}_v \leftarrow Z_v (I-I_{\Lambda(v)})h$
		
		\State $t_v \leftarrow \ell$
		\EndFor
				
		\EndProcedure
		
		\Procedure{$\textsc{Query}$}{$v\in\mathcal{S}$}\Comment{\cref{lem:balanced-query}}
		
		\If{$v\in\mathcal{S}\setminus\mathcal{B}$}
		
		\State \Return $J_v\cdot L[\ell]^{-\top}h$ \Comment{directly compute the value of sketch}
		
		\EndIf

		\LeftComment{For $v\in \mathcal{T}$, we make use of existing partial computations}
		
		\State $\Delta L \leftarrow (L[\ell]-L[t_v])\cdot I_{\Lambda(v)}$

		\State $Z_v\leftarrow Z_v-(L[\ell]^{-1} \cdot \Delta L\cdot Z_v^{\top})^{\top}$
		
		\State $y^{\triangledown}_v\leftarrow Z_v\cdot(I-I_{\Lambda(v)})h$
		
		\State $t_v\leftarrow\ell$ \Comment{update the time stamp for node $v$}
		
		\State $y^{\vartriangle}_v\leftarrow Z_v\cdot I_{\Lambda(v)}\cdot h$
		
		\State \Return $y^{\vartriangle}_v+y^{\triangledown}_v$
				
		\EndProcedure
		
	\end{algorithmic}
\end{algorithm}

\begin{algorithm}
\caption{Balanced Multiscale Representation Sketching Data Structure --
Updates\label{alg:balanced-W-sketch-1}}
\algnewcommand{\LeftComment}[1]{\State \(\triangleright\) #1} 

\begin{algorithmic}[1]

\State\textbf{datastructure} \textsc{BalancedSketch}

\Procedure{$\textsc{Update}$}{$\new{\ox}\in\R^{n},\new h \in \R^{d}$}
\Comment{\cref{lem:balanced-update}}

	\For{$i\in[m]$ where $\new x_{i}\neq x_{i}$}

		\State $\textsc{UpdateBlock}(\new{x_{i}})$

	\EndFor

	\ForAll{$\new{h_{i}}\neq h_{i}$}

		\State $v\leftarrow\Lambda^{\clubsuit}(i)$

		\ForAll{$u\in\mathcal{P}^{\mathcal{B}}(v)$}

			\State $y^{\triangledown}_u\leftarrow Z_u\cdot I_{\{i\}}\cdot(\new h-h)$

		\EndFor

	\EndFor

	\State $h \leftarrow\new h$

\EndProcedure

\Procedure{$\textsc{UpdateBlock}$}{$\new{\ox_{i}}\in \R^{n_i}$} \Comment{\cref{lem:balanced-update-coordinate}}

	\State $\ell\leftarrow\ell+1$

	\State $\ox_{i}\leftarrow\new{\ox_{i}}$

	\State $\new H=\nabla^{2}\phi(\ox)$

	\State Find lower Cholesky factor $L[\ell]$ of $A(\new H)^{-1}A^{\top}$

	\State $S\leftarrow\mathcal{P}^{\mathcal{B}}(\Lambda^{\clubsuit}(\low(A_{i})))$

	\State $\textsc{Update}L(S)$

	\State $\textsc{Update}H(\new H)$

\EndProcedure

\Procedure{$\textsc{Update}L$}{$S\subset \mathcal{B}$} \Comment{\cref{lem:balanced-update-L}}

	\ForAll{$v\in S$ }
	
		\LeftComment{We update $Z_v$ in two steps: first from $L[t_v]$ to $L[\ell-1]$, then from $L[\ell-1]$ to $L[\ell]$}

		\State $Z_v\leftarrow Z_v-\left(L[\ell-1]^{-1}(L[\ell-1]-L[t_v])_{\Lambda(v)}\cdot Z_v^{\top}\right)^{\top}$

		\State $Z_v\leftarrow Z_v-(L[\ell]^{-1}\cdot(L[\ell]-L[\ell-1])\cdot Z_v^{\top})^{\top}$

		\State $y^{\triangledown}_v\leftarrow Z_v(I-I_{\Lambda(v)})h$

		\State $t_v\leftarrow\ell$

	\EndFor

\EndProcedure

\Procedure{$\textsc{Update}H$}{$\new H$}\Comment{\cref{lem:balanced-update-H}}

	\State $\Delta H = \new H - H$

	\ForAll{$i\in[m]$ such that $(\Delta H)_{i}\neq\boldsymbol{0}$}

		\State Find set $S$ such that $S=\{v \in \mathcal{S} \mid \chi(v)=\{j\}\text{ and coordinate $j$ in $i$-th block}\}$

		\ForAll{$u\in\bigcup_{v\in S}\mathcal{P}^{\mathcal{S}}(v)$}

			\State $J_v\leftarrow\Phi_{\chi(v)}(H+\Delta H\cdot I_{\{i\}})^{-1/2}A^{\top}$

		\If{$v\in\mathcal{B}$}

			\State $Z_v\leftarrow J_v \cdot L[t_v]^{-\top}$

			\State $y^{\triangledown}_v\leftarrow Z_v\cdot(I-I_{\Lambda(v)})h$

		\EndIf

		\EndFor

	\EndFor
	\State $H \leftarrow \new H$

\EndProcedure

\end{algorithmic}
\end{algorithm}

\begin{lemma}[{\sc Initialize}]
\label{lem:balanced-initialization} Given initial $\ox$ and $h$,
the JL matrix $\Phi\in\R^{r\times n}$, and the elimination
tree $\mathcal{T}$, 
the data structure \textsc{BalancedSketch} initializes in time $O(n\tau^{2} r\log n)$.
Moreover, the internal state of the data structure satisfies the \cref{inv:w-balanced-invariant}
after initialization.
\end{lemma}

\begin{proof}
The correctness directly follows by \cref{inv:w-balanced-invariant}.

By \cref{cor:chol-time}, we can find the initial Cholesky decomposition $L[0]$
in time $O(n\tau^{2})$ time. For computing $J_v$, we note that
$J_v=\sum_{\text{child \ensuremath{c} of \ensuremath{v}}}J_c$.
When $\chi(v)=\{i\}$ for some $i$, we can compute $J_v$ in time
$O(\tau r)$ by column sparsity of $A$. Since the height of
tree is bounded by $O(\log n)$, we can compute all $J_v$ in time
$O(|\mathcal{S}|\cdot\tau r\log n)$. For $Z_v$, we note that
\[
Z_v=\sum_{\text{child \ensuremath{c} of \ensuremath{v}}}J_c \cdot L[\ell]^{-\top}
\]
Hence, it suffices to compute $J_v \cdot L[\ell]^{-\top}$ for all leaf
node $v\in\mathcal{B}$, which takes $O(n\tau^{2} r)$ time.
By \cref{lem:Linv-sparsity-time}, the solution of $J_v \cdot L[\ell]^{-\top}$
has $O(\tau r)$ nonzero entries. Hence, we can compute $Z_v$
for all $v\in\mathcal{B}$ in time $O(n\tau^{2} r\log n)$ time. Compute
$y^{\triangledown}_v$ takes time $O(n\tau^{2}r\log n)$
time. 
\end{proof}

\begin{lemma}[{\sc Update}]\label{lem:balanced-update}
	Suppose the current state of data structure satisfies \cref{inv:w-balanced-invariant}. 
	Given $\new{\mathcal{W}}$ implicitly by $\new \ox$, and $\new h$, the function {\sc Update}
	of {\sc BalancedSketch} updates the sketches in $\mathcal{S}$ implicitly by in time
	$O(\|\new \ox - \ox\|_0 \cdot \tau^2r\log n)+O(\|\new h - h\|_0 r\log n)$. 
	Moreover, the function also updates the internal states correspondingly so that 
	\cref{inv:w-balanced-invariant} is still preserved.
\end{lemma}
\begin{proof}
	To update $\ox$ to $\new \ox$, we view it as a sequence of updates, 
	where each update corresponding to a single block change in $\ox$, and use the helper function {\sc UpdateBlock}.
	The proof of correctness and runtime are given in \cref{lem:balanced-update-coordinate}.

	Similarly, we update $h$ to $\new h$ block-wise. Suppose $h$ changes in coordinate $j$. We note that $(Z_v)_{j}\neq\boldsymbol{0}$
	only if $j\in\bigcup_{i\in[l_{v},r_{v}]}\mathcal{P}^{\mathcal{T}}(i)$.
	Then $Z_v(I-I_{\Lambda(v)})e_{j}\neq\boldsymbol{0}$ only if $j\in\bigcup_{i\in\overline{\chi}(v)}\mathcal{P}^{\mathcal{T}}(i)$
	and $j\not\in\Lambda(v)$. Hence, all such $v$ lies on the path $\mathcal{P}^{\mathcal{B}}(\Lambda^{\clubsuit}(j))$. 

	For each coordinate $j$ of $h$ that changes, it suffices to compute $Z_v\cdot I_{\{j\}}\cdot(\new h-h)$. 
	We note this can be done in $O(r)$ time. 
	Thus, for each $h_{j}$, it takes $O(r \log n)$ time since $|\mathcal{P}^{\mathcal{B}}(\Lambda^{\clubsuit}(j))|=O(\log n)$.
	Hence, we can update $h$ in $O(\|\new h-h\|_{0}\cdot r \log n)$ time in total.
\end{proof}

\begin{lemma}[$\textsc{Query}$]
\label{lem:balanced-query}
Suppose the state of the data structure satisfies \cref{inv:w-balanced-invariant} immediately before a call to \textsc{Query}. Then calling  
$\textsc{Query}(v)$ of \textsc{BalancedSketch} returns  $\Phi_{\chi(v)}\mathcal{W}^{\top}h$
in $O(\tau^{2}r)$ time. 
Moreover, \cref{inv:W-simple-invariant} is preserved at the end of the function call.
\end{lemma}

\begin{proof}
When $v\in\mathcal{S}\setminus\mathcal{B}$, we directly compute $\Phi_{\chi(v)}H^{-1/2}A^{\top}L[\ell]^{-\top}h=J_v\cdot L[\ell]^{-\top}h^{(\ell)}$.
Let $u$ be the lowest ancestor node of $v$ in $\mathcal{B}$. Note
the row sparsity pattern of $J_v$ lies on $\mathcal{P}^{\mathcal{T}}(u)$
by the construction of $\mathcal{S}$. Hence, we can solve $L[\ell]^{-1}J_v^{\top}$
in $O(\tau^{2}r)$ time by \cref{lem:Linv-sparsity-time}. Since $h^{(\ell)}$ is given explicitly, we compute
$J_v\cdot L[\ell]^{-\top}h^{(\ell)}$ in $O(\tau^{2}r)$ time.

In the other case where $v\in\mathcal{B}$, the correctness follows by \cref{inv:w-balanced-invariant}. 
For the runtime, by \cref{lem:heavy-light-intersection-size},
we have $|\Lambda(v)|=O(\tau)$. By the sparsity pattern of $L$, 
we can compute $(L[\ell]-L[t_v])\cdot I_{\Lambda(v)}\cdot Z_{\triangledown}[v]^{\top}$ in $O(\tau^{2} r)$ time, and the column sparsity pattern of the result is on two paths of $\mathcal{T}$. Hence, we can update $Z_v$
and $y^{\triangledown}_v$ in $O(\tau^{2}r)$ time. Since
$|\Lambda(v)|=O(\tau)$, computing $y_{\vartriangle}$ takes  $O(r\cdot\tau)$ time. In total, we can compute  $\Phi_{\chi(v)}\mathcal{W}^{\top}h$ in $O(\tau^{2}r)$ time.
\end{proof}

\begin{lemma}[$\textsc{UpdateBlock}$]
\label{lem:balanced-update-coordinate}
Suppose the current state
of the data structure satisfies \cref{inv:w-balanced-invariant}.
The function $\textsc{UpdateBlock}$ of \textsc{BalancedSketch} updates the implicit representation
of $\mathcal{W}$ by updating $i$-th
block coordinate of $\ox$, $\ox_{i}$ to $\new{\ox_{i}}$, in time $O(\tau^{2}r \log n)$.
Moreover, \cref{inv:w-balanced-invariant} is preserved after the function call.
\end{lemma}

\begin{proof}
\textbf{Correctness: }To update $L$, it suffices
to update the nodes in the set $S=\mathcal{P}^{\mathcal{B}}(\Lambda^{\clubsuit}(u))$: Indeed, note that only (iii) of \cref{inv:w-balanced-invariant}
depends on $L[\ell]$. Thus, we need to update the sketch only if
\[
(L[\ell+1]-L[\ell])\cdot I_{\overline{\Lambda}(v)}\neq\boldsymbol{0}.
\]
We use $\Delta L$ to denote $L[\ell+1]-L[\ell]$. By \cref{lem:cholesky-update-sparsity-pattern}, the non-zero columns of $\Delta L$ lies on $\mathcal{P}^{\mathcal{T}}(\low^{\mathcal{T}}(A_{i}))$.
Let $u=\low^{\mathcal{T}}(A_{i})$, then we have 
\[
\Delta L=\sum_{w\in\mathcal{P}^{\mathcal{T}}(u)}(\Delta L)_{w}e_{w}^{\top}.
\]
We note that $(\Delta L)_{w}e_{w}^{\top}\cdot I_{\overline{\Lambda}(v)}\neq\boldsymbol{0}$
only if $w\in\overline{\Lambda}(v)$. By definition of $\Lambda^{\clubsuit}$
and \cref{lem:Lambda-club-path}, we have $(\Delta L)_{w}e_{w}^{\top}\cdot I_{\overline{\Lambda}(v)}\neq\boldsymbol{0}$
only if $v\in\mathcal{P}^{\mathcal{B}}(\Lambda^{\clubsuit}(w))$.
Thus, we need to update the set 
\[
\bigcup_{w\in\mathcal{P}^{\mathcal{T}}(u)}\mathcal{P}^{\mathcal{B}}(\Lambda^{\clubsuit}(w)).
\]
By \cref{lem:Lamda-clubsuit-path}, we have $\Lambda^{\clubsuit}(w)\in\mathcal{P}^{\mathcal{B}}(\Lambda^{\clubsuit}(u))$.
Hence, $\bigcup_{w\in\mathcal{P}^{\mathcal{T}}(u)}\mathcal{P}^{\mathcal{B}}(\Lambda^{\clubsuit}(w))=\mathcal{P}^{\mathcal{B}}(\Lambda^{\clubsuit}(u))$.

\textbf{Runtime: }By \cref{lem:cholesky-update-time}, we can perform the rank-1 updates
on $L$ in $O(\tau^{2})$ time. Since $S=\mathcal{P}^{\mathcal{B}}(\Lambda^{\clubsuit}(u))$,
we have $|S|=O(\log n).$ We note when $x_{i}$ updates, we only need
to update one diagonal block of $H$, hence $\nnz(\Delta H)=O(1)$.
By \cref{lem:balanced-update-H,lem:balanced-update-L},
we can update $H$ and $L$ for the data structure in time $O(\tau^{2}r \log n)$.
Hence, the function takes $O(\tau^{2}r \log n)$ time in total.
\end{proof}
\begin{lemma}[$\textsc{Update}L$]
\label{lem:balanced-update-L}Suppose \cref{inv:w-balanced-invariant}
is satisfied. Given the set $S\subset\mathcal{\mathcal{B}}$, then the function $\textsc{Update}L$ of \textsc{BalancedSketch}
updates the sketches implicitly and $t_v$ to the current time at each $v\in S$. 
If the number of non-zero column of $\Delta L$ is bounded by $O(\tau)$, then the function $\textsc{Update}L$ takes  $O(|S|\cdot\tau^{2}\cdot r)$ time. 
\end{lemma}

\begin{proof}
\textbf{Correctness: }The correctness directly follows by \cref{lem:balanced-Z*-formula}.

\textbf{Runtime: }
By \cref{lem:heavy-light-intersection-size}, we have $|\Lambda(v)|=O(\tau)$.
By \cref{lem:L-column-sparsity-pattern,lem:cholesky-update-sparsity-pattern}, 
we can compute $(L[\ell]-L[t_v])\cdot I_{\Lambda(v)}\cdot Z_{\triangledown}[v]^{\top}$
and $\Delta L\cdot Z_{\triangledown}[v]^{\top}$ in time $O(\tau^{2}r)$
and the column sparsity pattern of the result is on a path of $\mathcal{T}$.
Hence, we can update $Z_v$ and $y^{\triangledown}_v$ in time $O(\tau^{2}r)$.
Hence, the total time is bounded by $O(|S|\cdot\tau^{2}r)$. 
\end{proof}
\begin{lemma}[$\textsc{Update}H$]
\label{lem:balanced-update-H} Suppose \cref{inv:w-balanced-invariant}
is satisfied for $\ell$. Given $\Delta H$ such that $H[\ell+1]=H[\ell]+\Delta H$,
then the function $\textsc{Update}H$ updates $H$ and the internal state of the data structure in $O(\nnz(\Delta H)\cdot\tau^{2} r \log n)$ time. %TODO: more specific here probably
\end{lemma}

\begin{proof}
\textbf{Correctness: }We observe that $Z_v$ changes only if $I_{\chi(v)}\cdot\Delta H\neq\boldsymbol{0}$.
Suppose the $i$-th block of $H$ changes, and let $v\in\mathcal{S}$
where $\chi(v)=\{i\}$. For changes on $H_{i}$, it suffices to update
$\mathcal{P}^{\mathcal{S}}(v)$. 

\textbf{Runtime: }Since $H$ is a block-diagonal matrix, let $\widetilde{H}\defeq(H[\ell]+\Delta H\cdot I_{\{i\}})^{-1/2}-H[\ell]^{-1/2}$,
then the nonzero pattern of $\widetilde{H}$ is an $n_{i}\times n_{i}$
submatrix on the block diagonal. Recall $n_{i}=O(1)$, hence, we can find $\Phi_{\chi(v)}\widetilde{H}A^{\top}$
by computing $O(1)$ many outer products of columns of $\Phi$ and
row of $A^{\top}$, which takes $O(\tau r)$ time by sparsity
pattern of $A$ (\cref{lem:A-sparsity-pattern}). Then, we can
update $Z_v$ by compute $\Phi_{\chi(v)}\widetilde{H}A^{\top}L[t_v]^{-\top}$,
which takes $O(\tau^{2}r)$ time by \cref{lem:Linv-sparsity-time}. Thus, we can update $Z_v$ and $y^{\triangledown}_v$
in time $O(\tau^{2}r)$ time for each $v$. As the 
height of $\mathcal{S}$ is bounded by $O(\log n)$, this function
takes time $O(\nnz(\Delta H)\cdot\tau^{2}r \log n)$.
\end{proof}

\subsection{Proof of \cref{thm:central-path-algo}}

\begin{algorithm}
	\caption{Robust Central Path Maintenance -- Initialize, MultiplyAndMove, Output
		\label{alg:CentralPath}}
	
	\algnewcommand{\LeftComment}[1]{\State \(\triangleright\) #1} 
	
	\begin{algorithmic}[1]
		
		\State \textbf{data structure }$\textsc{CentralPathMaintenance}$
		\textbf{extends }$\textsc{MultiscaleRepresentation}$
		
		\State \textbf{private : member}
		\State \hspace{4mm} \textsc{BalancedSketch} $\mathsf{Sketch}\mathcal{W}^\top\eps_x,\;
		\mathsf{Sketch}\mathcal{W}^\top\eps_s,\;
		\mathsf{Sketch}\mathcal{W}^\top h$
		
		\Comment{maintains $\mathcal{W}^{\top}\eps_{x},\mathcal{W}^{\top}\eps_{s}$, and $\mathcal{W}^\top h$,
		\cref{thm:balanced-W-sketch}}
		
		\State \hspace{4mm} 
		$\textsc{VectorSketch}$ $\textsf{Sketch}H^{1/2}\widehat{x},\; \textsf{Sketch}H^{-1/2}\widehat{s},\;\textsf{Sketch}H^{-1/2}c_x$

		\Comment{maintains $H_{\ox}^{1/2}\widehat{x}$, $H_{\ox}^{-1/2}\widehat{s}$, and $H_{\ox}^{-1/2}c_x$,
		\cref{thm:vector-sketch}}
		
		\State \hspace{4mm}  $\textsc{\ensuremath{\ell_{\infty}-}Approximates}$
		$\textsf{Approx}H_{\ox}^{1/2}x,\;\textsf{Approx}H_{\ox}^{-1/2}s$
		\State \Comment{maintains $\ell_\infty$-approximations of $H_{\ox}^{1/2}x$ and $H_{\ox}^{-1/2}s$, \cref{thm:approximate-ell-infty}}
		
		\State \hspace{4mm}  Sampling tree $(\mathcal{S},\chi)$
		
		\State \hspace{4mm} $\ell \in \mathbb{N}$ \Comment{central path step counter}

		\State \hspace{4mm} $N\in \mathbb{N}$ \Comment{upper bound on total number of steps}
		
		\State \hspace{4mm} $k \leftarrow \sqrt{n}$ \Comment{number of steps supported before a restart}

		\State \hspace{4mm} $r\leftarrow \Theta(\log^3(N))$ 

		\State \hspace{4mm} $\Phi\in\R^{r\times n}$ \Comment{JL matrix}
		
		%TODO: w\leftarrow \boldsymbol{1}_m, other central path parameters 
		\State \textbf{end members}
		
		\Procedure{$\textsc{Initialize}$}{$x\in\R^{n},s\in\R^{n},t\in\R_{+},\overline{\eps}\in(0,1)$}
		
		\State $\textsf{super}.\textsc{Initialize}(x,s,x,s,t)$
		%\Comment{calls $\textsc{Initialize}$ in $\textsc{MultiscaleRepresentation}$}
		
		\State Initialize $\Phi\in\R^{r\times n}$ by letting each
		entry be i.i.d.~samples from $\mathcal{N}(0,\frac{1}{\sqrt{r}})$
		
		\State Construct sampling tree $(\mathcal{S},\chi)$ as in \cref{subsec:balanced-sampling-tree-construction}.
		
		\State $\textsc{InitializeSketch}()$

		\State $\ell \leftarrow 0$
		
		\State $\eps_{\mathrm{apx}}^{(x)}\leftarrow \frac{\overline{\eps}}{\max_{i}n_{i}},\zeta^{(x)}\leftarrow 2\alpha,\delta_{\mathrm{apx}}\leftarrow \frac{N}{20k}$
		\Comment{setting the appropriate approximation tolerances}

		\State $\eps_{\mathrm{apx}}^{(s)}\leftarrow \frac{\overline{\eps}\cdot \overline{t}}{2\max_{i}n_i },\zeta^{(s)}\leftarrow 2\alpha \overline{t}$
		\Comment{$\alpha, n_i$ as in~\cref{alg:IPM_framework}}

		\State $\textsf{Approx}H_{\ox}^{1/2}x.\textsc{Initialize}(\mathcal{S},\chi,\eps_{\mathrm{apx}}^{(x)},\delta_{\mathrm{apx}},\zeta^{(x)},k)$
		\State $\textsf{Approx}H_{\ox}^{-1/2}s.\textsc{Initialize}(\mathcal{S},\chi,\eps_{\mathrm{apx}}^{(s)},\delta_{\mathrm{apx}},\zeta^{(s)}, k)$
		
		% \State Generate $\overline{w}$ by $\overline{w}_{i}\defeq w_{j}$
		% where $i$-th coordinate is in the $j$-th block.
		%\Comment{We define $\overline{w}^{\pm1/2}=(\overline{w}_{1}^{\pm1/2},\overline{w}_{2}^{\pm1/2},\ldots,\overline{w}_{n}^{\pm1/2})$}
		% Since $w$ is uniform here we don't need to set weight
		\EndProcedure
		
		\Procedure{$\textsc{MultiplyAndMove}$}{$t\in \R_+$}

		\State $\ell \leftarrow \ell +1$

		\If{$|\overline{t}-t|>\overline{t}\cdot \eps_t$ or $\ell > k$} \Comment{restarts entire data structure} %TODO: eps_t not defined yet

		\State $\textsc{Initialize}(x,s,t,\overline{\eps})$ 

		\EndIf
		
		\State \textsf{super}.$\textsc{Move}()$ %\Comment{calls $\textsc{Move}$ in $\textsc{MultiscaleRepresentation}$}
		
		\State \Comment {Oracle $\mathcal{O}_x, \mc{O}_s$ for the $\ell_\infty$-\textsc{Approximates} data structures, \cref{lem:oracles}}
		\State $\new\ox\leftarrow H_{\ox}^{-1/2}\cdot  \textsf{Approx}H_{\ox}^{1/2}x.\textsc{Query}(\mathcal{O}_x \{H_{\ox}^{1/2}\widehat{x}+\beta_x \cdot c_{x}-\mathcal{W}^{\top}(\beta_x h+\eps_{x})\})$
		\label{line:main-update-ox}
		\State $\new\os\leftarrow H_{\ox}^{1/2}\cdot\mathsf{Approx}H_{\ox}^{-1/2}s.\textsc{Query}(\mathcal{O}_s\{H_{\ox}^{-1/2}\widehat{s}+\mathcal{W}^{\top}(\beta_s h+\eps_{s})\})$
		\label{line:main-update-os}
		\State $\textsf{super}.\textsc{Update}(\new\ox,\new\os)$ 
		%\Comment{calls $\textsc{Update}$ in $\textsc{MultiscaleRepresentation}$}
		
		\State $\textsc{UpdateSktech}()$
		
		\EndProcedure
		
		\Procedure{$\textsc{Output}$}{}
		
		\State \Return $\widehat{x}+H_{\ox}^{-1/2}\beta_x c_{x}-H_{\ox}^{-1/2}\mathcal{W}^{\top}(\beta_{x} h+\eps_{x}),\;\;
		\widehat{s}+H_{\ox}^{1/2}\mathcal{W}^{\top}(\beta_{s} h+\eps_{s})$ 
		
		\EndProcedure
		
	\end{algorithmic}
\end{algorithm}

\begin{algorithm}
	\caption{Robust Central Path Maintenance -- Helper Functions\label{alg:CentralPath-1}}
	
	\begin{algorithmic}[1]
		
		\State \textbf{datastructure }$\textsc{CentralPathMaintenance}$
		\textbf{extends }$\textsc{MultiscaleRepresentation}$
		
		\Procedure{$\textsc{InitializeSketch}$}{}
		
		\State$\mathsf{Sketch}\mathcal{W}^\top\eps_x.
		\textsc{Initialize}(\ts, \chi, \Phi,x,\eps_{x})$
		
		\State$\mathsf{Sketch}\mathcal{W}^\top\eps_s.
		\textsc{Initialize}(\ts, \chi, \Phi,x,\eps_{s})$
		
		\State $\mathsf{Sketch}\mathcal{W}^\top h.
		\textsc{Initialize}(\ts, \chi, \Phi,x,h)$
		
		\State
		$\textsf{Sketch}H^{-1/2}c_x.
		\textsc{Initialize}(\mathcal{S},\chi,\Phi,H^{-1/2}c_{x})$ 
		
		\State 
		$\textsf{Sketch}H^{1/2}\widehat{x}.\textsc{Initialize}(\mathcal{S},\chi,\Phi,H_{\ox}^{1/2}\widehat{x})$
		
		\State 
		$\textsf{Sketch}H^{-1/2}\widehat{s}.\textsc{Initialize}(\mathcal{S},\chi,\Phi,H_{\ox}^{-1/2}\widehat{s})$
		
		\EndProcedure
		
		\Procedure{$\textsc{UpdateSketch}$}{}
		
		\State $\mathsf{Sketch}\mathcal{W}^\top\eps_x.\textsc{Update}(\ox, \eps_{x})$
		
		\State $\mathsf{Sketch}\mathcal{W}^\top\eps_s.\textsc{Update}(\ox, \eps_{s})$
		
		\State $\mathsf{Sketch}\mathcal{W}^\top h.\textsc{Update}(\ox, h)$
		
		\State $\textsf{Sketch}H^{-1/2}c_x.\textsc{Update}(H_{\ox}^{-1/2}c_{x})$
		
		\State $\textsf{Sketch}H^{1/2}\widehat{x}.\textsc{Update}(H_{\ox}^{1/2}\widehat{x})$
		
		\State $\textsf{Sketch}H^{-1/2}\widehat{s}.\textsc{Update}(H_{\ox}^{-1/2}\widehat{s})$
		
		\EndProcedure
		
		\State 
		
		\State \textbf{Oracle}  $\mathcal{O}_x\{H_{\ox}^{1/2}\widehat{x}+\beta_x \cdot c_{x}-\mathcal{W}^{\top}(\beta_x h+\eps_{x})\}$
		\label{line:main-alg-oracle-begin}
		
		\Procedure{\textsc{TypeI}}{$v \in \mc{S}$}
		\State \Return 
		$\textsf{Sketch}H^{1/2}\widehat{x}.\textsc{Query}(v)+
		\beta_x\cdot\textsf{Sketch}H^{-1/2}c_x.\textsc{Query}(v)
		-$\\ 
		\hspace{6em}
		$\beta_x\cdot\mathsf{Sketch}\mathcal{W}^\top h.\textsc{Query}(v) - \mathsf{Sketch}\mathcal{W}^\top\eps_x.\textsc{Query}(v)$
		\EndProcedure
		\Procedure{\textsc{TypeII}}{$i \in [n]$}
		\State \hspace{4mm} \Return  $e_{i}^{\top}(H_{\ox}^{1/2}\widehat{x}+\beta_x \cdot c_{x}-\mathcal{W}^{\top}(\beta_x h+\eps_{x}))$
		\EndProcedure

		\State 
		\State \textbf{Oracle}  $\mathcal{O}_s\{H_{\ox}^{-1/2}\widehat{s}+\mathcal{W}^{\top}(\beta_s h+\eps_{s})\}$
		
		\Procedure{\textsc{TypeI}}{$v \in \mc{S}$}
		\State \Return 
		$
		\textsf{Sketch}H^{-1/2}\widehat{s}.\textsc{Query}(v)+\beta_s \cdot \mathsf{Sketch}\mathcal{W}^\top h .\textsc{Query}(v)+\mathsf{Sketch}\mathcal{W}^\top\eps_s.\textsc{Query}(v)$
		\EndProcedure
		
		\Procedure{\textsc{TypeII}}{$i \in [n]$}
			\State \hspace{4mm} \Return  $e_{i}^{\top}(H_{\ox}^{-1/2}\widehat{s}+\mathcal{W}^{\top}(\beta_s h+\eps_{s}))$
		\EndProcedure \label{line:main-alg-oracle-end}

	\end{algorithmic}
\end{algorithm}

	\textbf{Proof of Correctness:\footnote{
		\cite{gu2022faster} has noted an error pertaining to the maintenance of the approximate solutions (see their remark 4.19). 
		Specifically, the guarantee of \cref{eq:approx-guarantee-bug} is insufficient to satisfy the conditions of 
		\cref{thm:approximate-ell-infty} for using the $\ell_\infty$-\textsc{Approximates} data structure. 
		A fix is provided by \cite{gu2022faster}, where the sampling procedure along the sampling tree is replaced by a top-down BFS.
		An alternative fix is to modify the $\ell_\infty$-\textsc{Approximates} data structure, so that rather than approximating $y = H_{\ox}^{1/2} x$, where $y$ is treated as a black-box vector, the data structure should maintain the term $H_{\ox}^{1/2}$ separately from $x$. The correct implementation of this approach can be found in \cite[Section 6]{dong2022nested}.
		In this manuscript, we have left the error as is.}}
	
	At every iteration of \textsc{MultiplyAndMove}, we call \textsf{super.}\textsc{Move} followed by \textsf{super.}\textsc{Update} with the updated $\new \ox, \new \os$ values.
	Therefore, we correctly maintain the implicit representation $(x,s)$ directly as a result of \cref{thm:W-representation}. 
	
	Now, we show that 
	\[
	\|\ox_{i}-x_{i}\|_{\ox_{i}}\leq\overline{\eps}\quad\text{and}\quad\|\os_i-s_i\|_{\ox_i}^* \leq t\overline{\eps}\quad\text{for all \ensuremath{i\in[m]}}.
	\]
	By construction of $\ell_\infty$\textsc{-Approximates} data structure,  $\mathsf{Approx}H_{\ox}^{1/2}x$ maintains an $\ell_\infty$-approximation of
	\[
		H_{\ox}^{1/2}\widehat{x}+\beta_x \cdot c_{x}-\mathcal{W}^{\top}(\beta_x h+\eps_{x})=H_{\ox}^{1/2}x.
	\]
	For any non-negative integer $\ell\leq k$, we have the guarantee
	\begin{equation} \label{eq:approx-guarantee-bug}
		\|H_{\ox}^{1/2}x^{(\ell+1)}-H_{\ox}^{1/2}x^{(\ell)}\|_{2} =\| H_{\ox}^{1/2}(x^{(\ell+1)}-x^{(\ell)})\|_{2}
		=\|\delta_{x}^{(\ell)}\|_{\ox}
		 \leq\frac{9}{8}\alpha\leq \zeta^{(x)},
	\end{equation}
	where we used \cref{lem:alpha_i} for the first inequality.
	Since $\eta = O(\log n)$, $k=\sqrt{n}$, and $\delta_{\mathrm{apx}}=\frac{N}{2k}$, we can choose $r=\Theta(\log^3 N)$.
	
	By \cref{thm:approximate-ell-infty},
	if $\widetilde{x}$ denotes the output of $\textsf{Approx}\ox.\textsc{Query}$, then it satisfies 
	%\yintat{I am a bit confused, technically this theorem assume linf of input is off by 1 and the error is bounded by eps, so this is a scaled version of that data structure? Then we should make the theorem support this formally just to be safe?}
	\begin{align*}
	\| H_{\ox}^{1/2}x-\widetilde{x}\|_{\infty} &\leq \eps_{\mathrm{apx}}^{(x)} = \frac{\overline{\eps}}{\max_{i}n_{i}}.\\
	\intertext{In \cref{line:main-update-ox}, we set $\overline{x}=H^{-1/2}\widetilde{x}$, so} 
	\|H^{1/2}(\ox-x)\|_{\infty}&\leq\frac{\overline{\eps}}{\max_{i}n_{i}}\cdot \\
	\intertext{Therefore, we have the desired error bound}
	\|\ox_{i}-x_{i}\|_{\ox_{i}}=\|H_{\ox_{i}}^{1/2}(\ox_{i}-x_{i})\|_{2} &\leq\sqrt{n_{i}}\cdot\|H_{\ox_{i}}^{1/2}(\ox_{i}-x_{i})\|_{\infty}\leq\overline{\eps}.
	\end{align*}
	
	Similarly, $\textsf{Approx}H_{\ox}^{-1/2}s$ maintains an $\ell_\infty$-approximation of 
	\[
		H_{\ox}^{-1/2}\widehat{s}+\mathcal{W}^{\top}(\beta_s h+\eps_{s})=H_{\ox}^{-1/2}s.
	\]
	By \cref{lem:alpha_i}, for any non-negative integer $\ell\leq k$
	we have 
	\[
	\|H_{\ox}^{-1/2}\delta_{s}^{(\ell)}\|\leq\frac{9}{8}\alpha\cdot t \leq \frac98 \alpha \overline{t} \leq \zeta^{(s)},
	\]
	where we used $t\leq\overline{t}$ at every step of the algorithm.
	By \cref{thm:approximate-ell-infty}, if $\widetilde{s}$ denotes the output of $\textsf{Approx}\os.\textsc{Query}$, then in \cref{line:main-update-os}, $\os=H_{\ox}^{1/2}\widetilde{s}$, and so
	\[
	\|H_{\ox}^{-1/2}(s-\overline{s})\|_{\infty}\leq \eps_{\mathrm{apx}}^{(s)} = \frac{\overline{t}\cdot \overline{\eps}}{2\max_{i}n_{i}},
	\]
	Therefore, we have the desired error bound
	\[
	\|\os_{i}-s_{i}\|_{\ox_{i}}^{*}=\|H_{\ox_{i}}^{-1/2}(s_{i}-\os_{i})\|_{2}\leq\sqrt{n_{i}}\cdot\frac{\overline{t}\cdot\overline{\eps}}{2\max_{i}n_{i}}\leq t\overline{\eps},
	\]
	where the last step follows by $0<\eps_t<\frac{1}{2}$ and hence $t\in (\overline{t}/2,\overline{t}]$.
	
	By our choice of $\delta_{\mathrm{apx}}$, $\textsf{Approx}H_{\ox}^{1/2}x$ and $\textsf{Approx}H_{\ox}^{-1/2}s$ succeed with probability at least $1-\frac{N}{10k}$. Taking the union bound over $\frac{N}{k}$ many restarts, the data structure succeeds with probability at least $0.9$ after $N$ total steps of central path.

	Lastly, we ensure our oracle implementations are correct. For simplicity, we check $\mc{O}_x$:
	\begin{lemma}\label{lem:oracles}
		Oracles $\mc{O}_x$ and $\mc{O}_s$ are implemented correctly on \crefrange{line:main-alg-oracle-begin}{line:main-alg-oracle-end} of \cref{alg:CentralPath-1} for the latest query to the $\ell_\infty$-\textsc{Approximates} data structure $\textsf{Approx}H_{\ox}^{1/2}x$ and $\textsf{Approx}H_{\ox}^{-1/2}s$.
	\end{lemma}
	\begin{proof}
		The input vector is $H_{\ox}^{1/2}\widehat{x}+\beta_x \cdot c_{x}-\mathcal{W}^{\top}(\beta_x h+\eps_{x})$. 
		A type-I access at $v \in \mc{S}$ should return $\Phi{\chi(v)}(H_{\ox}^{1/2}\widehat{x}+\beta_x \cdot c_{x}-\mathcal{W}^{\top}(\beta_x h+\eps_{x}))$. By linearity of $\Phi$, and by construction of the sketching data structures, this is precisely what the oracle implements. 
		A type-II access should return coordinate $i$ of the input vector, which the oracle does correctly. 

		The proof for $\mc{O}_s$ is identical, we omit it here.
	\end{proof}

	\textbf{Proof of Runtime:} 
	We split this proof into \cref{lem:central-path-initialize,lem:central-path-MultiplyAndMove,lem:centralPath-output}.

\begin{lemma}[Initialization time]\label{lem:central-path-initialize}
	The initialization time of $\textsc{CentralPathMaintenance}$
	is $O(n\tau^{2}\log^4 N)$.
\end{lemma}

\begin{proof}
	By \cref{thm:W-representation}, initializing $\textsc{MultiscaleRepresentation}$
	takes $O(n\tau^{2})$ time. 
	We can construct the balanced sampling tree in time $O(n\tau+n\log n)$ by \cref{thm:balanced-sampling-tree-construct-time}.
	By \cref{thm:vector-sketch,thm:balanced-W-sketch} and our choice of $r$,
	the initialization of each sketch takes $O(n\tau^{2}\log^4 N)$ time.
	Hence, the total initialization time is bounded by $O(n\tau^{2}\log^4 N)$.
\end{proof}
\begin{lemma}[\textsc{MultiplyAndMove} time]\label{lem:central-path-MultiplyAndMove}
	Suppose that the function is called at most $N$ times and $t$ is monotonic decreasing,
	the total running time of $\textsc{MultiplyAndMove}$ is 
	\[
		O\left(\left(\frac{Nn^{1/2}}{\overline{\eps}^4}+n\frac{\log(t_\textnormal{max}/t_\textnormal{min})}{\eps_t}\right)\tau^2\poly\log(n/\overline{\eps})\right).	
	\]
	% \[
	% 	O\left(\frac{k^{2}\tau^{2}}{\overline{\eps}^{4}}\cdot\poly\log(n,k,1/\overline{\eps})+n\tau^2\cdot \frac{\log(t_\textnormal{max}/t_\textnormal{min})}{\log(1+\eps_t)}\right).
	% \]
\end{lemma}

\begin{proof}
	By \cref{thm:W-representation}, $\textsc{Move}$ takes time
	$O(1)$ time for each call. 
	By \cref{thm:balanced-W-sketch},
	the sampling tree has height $\eta = O(\log n)$. 
	Then, each $\textsc{Update}h$ takes time $O(\log^4 N)$ per coordinate change by \cref{thm:vector-sketch,thm:balanced-W-sketch}. 
	By \cref{lem:CentralPath-Query-Time,lem:Central-Path-Compute-time}, 
	each type-I query takes $O(\tau^{2}\log^3 N)$ time and type-II query takes $O(\tau^{2})$ time. 
	Thus, the running time of $\textsf{Approx}H_{\ox}^{1/2}x$
	and $\textsf{Approx}H_{\ox}^{-1/2}s$ is bounded by $O(n\tau^{2}\cdot\poly\log(N))$ for every $k := \sqrt{n}$ steps
	% \yintat{where is this $\eps^4$ comes from? In fact $\eps_{apx}$ looks like a constant? so, there is no epsilon here?} 
	by \cref{thm:approximate-ell-infty} and our choice of $\alpha$ and $\overline{\eps}$ in \cref{alg:IPM_framework}. This also implies 
	\[
	\sum_{\ell=\ell_0}^{\ell_0+k}\|\ox^{(\ell+1)}-\os^{(\ell)}\|_{0}+\|\os^{(\ell+1)}-\os^{(\ell)}\|_{0}=O(n\cdot\poly\log(N)).
	\] 
	Hence, by $\textsc{Update}$ time in \cref{thm:W-representation},
	the running time for this function during the algorithm is 
	$$
	\frac{N}{k}\cdot O(n\cdot\tau^{2}\cdot\poly\log(N)) = O({N n^{1/2}}\tau^{2}\cdot\poly\log(N)).
	$$
	and the total number of entries change during algorithm for each variable
	in \cref{eq:implicit-representation} is 
	$$
	O\left({Nn^{1/2}}\tau\cdot\poly\log(N)\right).
	$$
	
	We note that $\textsf{Approx}H_{\ox}^{1/2}x$ (resp. $\textsf{Approx}H_{\ox}^{-1/2}s$) requires oracle queries to previous versions of variables maintained \textsc{CentralPathMaintenance}, including all the sketching data structures and the variables maintained in \textsc{MultiscaleRepresentation}. 
	We resolve this by using persistent data structures throughout, costing an $O(\log N)$ multiplicative factor in all run-times, see e.g.\ \cite{DBLP:journals/jcss/DriscollSST89}. 
	
	Hence, by \cref{thm:vector-sketch,thm:balanced-W-sketch}, the total
	running time of $\textsc{UpdateSktech}$ during algorithm is bounded
	by $O(Nn^{1/2}\tau^{2}\cdot\poly\log(N))$. 

	Note that the algorithm will restart whenever $|\overline{t}-t|>\overline{t}\cdot \eps_t$ or $\ell> k = \sqrt{n}$. Hence, we can bound the total number of restart by 
	${\log_{1-\eps_t}(t_\textnormal{min}/t_\textnormal{max})}+\frac{N}{k}=O(\frac{N}{k}+\log(t_\textnormal{max}/t_\textnormal{min})/\eps_t)$. Each restart takes $O(n\tau^2 \log^4 N)$ time by \cref{lem:central-path-initialize}.
	
	Thus, the runtime of $\textsc{MultiplyAndMove}$ is bounded by 
	\begin{align*}
		&\; O\left({Nn^{1/2}}\tau ^2\poly\log(N)+n\tau^2\log^4 N\left(\frac{\log(t_\textnormal{max}/t_\textnormal{min})}{\eps_t}+\frac{N}{k}\right)  \right) \\
		=& O\left(\left({Nn^{1/2}}+n\frac{\log(t_\textnormal{max}/t_\textnormal{min})}{\eps_t}\right)\tau^2\poly\log(N)\right) \\ 
		=& O\left(\left({Nn^{1/2}}+n{\log(t_\textnormal{max}/t_\textnormal{min})}\right)\tau^2\poly\log(N)\right).
	\end{align*}
	where the last step follows by the choice of $\eps_t$ in \cref{alg:IPM_framework}. 
\end{proof}
\begin{lemma}[Output time]\label{lem:centralPath-output}
	$\textsc{Output}$ runs in $O(n\tau^{2})$ time.
\end{lemma}
\begin{proof}
	We note that we compute $\beta_x h+\eps_{x}$ exactly in time $O(n)$. Recall that $\mathcal{W}=L_{\ox}^{-1}AH_{\ox}^{-1/2}$,
	we can compute $\mathcal{W}^{\top}(\beta_x h+\eps_{x})$ in time $O(n\tau^{2})$
	by \cref{lem:Linv-path}. Hence, we can compute
	$x$ in time $O(n\tau^2)$. The analysis for $s$ is identical,
	we omit it here.
\end{proof}
\begin{lemma}[Query time]
	\label{lem:CentralPath-Query-Time} Type-I queries to the oracles $\mc{O}_x$ and $\mc{O}_s$ run in $O(\tau^{2}\cdot\log^3 N)$ time.
\end{lemma}

\begin{proof}
	By the runtime of $\textsc{Query}$ in \cref{thm:vector-sketch,thm:balanced-W-sketch} and $r=\Theta(\log^3 N)$,
	the total query time is bounded by $O(\tau^{2}\cdot \log^3N)$.
\end{proof}
\begin{lemma}[Compute time]
	\label{lem:Central-Path-Compute-time}
	Type-II queries to the oracles $\mc{O}_x$ and $\mc{O}_s$  run in $O(\tau^{2})$ time.
\end{lemma}

\begin{proof}
	We show the claim for $\mc{O}_x$: 
	Since $H_{\ox}$ is a block-diagonal matrix and $n_i=O(1)$,
	we can compute $e_{i}^{\top}(H_{\ox}^{1/2}\widehat{x}+\beta_x \cdot c_{x})$
	in $O(1)$ time. Now, it suffices to show we can compute $e_{i}^\top\mathcal{W}^{\top}(\beta_x h+\eps_{x})$
	in $O(\tau^{2})$ time. By the definition of $\mathcal{W}$, we have
	\[
	e_{i}^{\top}\mathcal{W}^{\top}(\beta_x h+\eps_{x})=(\beta_x h+\eps_{x})^{\top}L_{\ox}^{-1}AH_{\ox}^{-1/2}e_{i}.
	\]
	By \cref{lem:A-sparsity-pattern,lem:Linv-sparsity-time}, we can compute
	$y=L_{\ox}^{-1}AH_{\ox}^{-1/2}e_{i}$ in $O(\tau^{2})$ time and $y$ has $O(\tau)$ many non-zero entries. 
	Then, we can compute the product
	$(\beta_x h+\eps_{x})^{\top}y$ in $O(\tau)$ time. 
	Hence, the total runtime for a type-II query to $\mc{O}_x$ is $O(\tau^2)$. 
	The proof for $\mc{O}_s$ is identical; we omit it here.
\end{proof}

\section{Acknowledgment}
We thank Aaron Sidford for discussing the optimization on thick path and Anup B. Rao for discussing the convex regression problem. The authors are supported by NSF awards CCF-1749609, CCF-1740551, DMS-1839116, DMS-2023166, Microsoft Research Faculty Fellowship, Sloan Research Fellowship, and Packard Fellowships.

\newpage 
\bibliographystyle{alpha}
\bibliography{references}
\newpage  

\appendix
\section{Robust Interior Point Algorithm for General Convex Sets}

\label{sec:Robust-IPM}

In this section, we give a robust interior point method for the optimization
problem 
\begin{equation}
\tag{CP}\min_{Ax=b,x_{i}\in K_{i}\text{ for }i\in[m]}c^{\top}x\label{eq:problem}
\end{equation}
where $A$ is a $d\times n$ matrix, $x_{i}\in K_{i}\subset\R^{n_{i}}$,
and $x$ is the concatenation of $x_{i}$ lying inside the domain
$K\defeq\prod_{i=1}^{m}K_{i}\subset\R^{n}$ with $n=\sum_{i=1}^{m}n_{i}$.
The main result of this section is the following: 

\begin{restatable}{theorem}{thmIPMframework}\label{thm:IPM_framework}
Consider the convex program \cref{eq:problem}. Given $\nu_{i}$-self-concordant
barriers $\phi_{i}:K_{i}\rightarrow\R$ with its minimum $x_{i}$.
Define the following parameters of the convex problem:
\begin{enumerate}
\item Inner radius $r$: There exists a $z$ such that $Az=b$ and $B(z,r)\subset K$.
\item Outer radius $R$: We have $K\subset B(x,R)$ for some $x\in\Rn$.
\item Lipschitz constant $L$: $\|c\|_{2}\leq L$.
\end{enumerate}
Let $w\in\R_{\geq1}^{m}$ be any weight vector, and $\kappa=\sum_{i=1}^{m}w_{i}\nu_{i}$.
For any $0<\eps\leq1/2$, \cref{alg:IPM_framework} outputs an approximate
solution $x$ in $O(\sqrt{\kappa}\log(m)\log(\frac{n\kappa R}{\eps r}))$
steps, such that $Ax=b$, $x\in K$ and 
\[
c^{\top}x\leq\min_{Ax=b,\;x\in K}c^{\T}x+\eps LR.
\]

\end{restatable}

\begin{rem}\label{rem:IPM-universal-barrier}If the barrier functions
$\phi_{i}$ is not given, we can use $w_{i}=1$ and universal barrier
functions $\phi_{i}$ for $K_{i}$ \cite{nesterov1994interior, lee2018universal}.
In this case, the algorithm takes $O(\sqrt{n}\log n\log(\frac{n\kappa R}{\eps r}))$
steps, and the cost of computing a good enough approximation of $\nabla\phi_{i}$
and $\nabla^{2}\phi_{i}$ both takes $n_{i}^{O(1)}\log(\frac{nR}{r})$
time for each $i$, assuming the following mild conditions: 
\begin{enumerate}
\item We can check if $x_{i}$ is in $K_{i}$ in time $n_{i}^{O(1)}$. 
\item We are given $x_{i}$ such that $B(x_{i},r)\subset K$. 
\end{enumerate}
\end{rem}

Our algorithm and the proof is a simplified but strengthen version
of \cite{DBLP:conf/colt/LeeSZ19}. We introduce approximate $t$ in
the algorithm to simplify our main data structure. We introduce a
new reduction for finding initial point, which allows us to output
$x$ exactly satisfying $Ax=b$. We used the potential $\cosh(\|\cdots\|)$
instead of $\exp(\|\cdots\|)$ as in \cite{DBLP:conf/colt/LeeSZ19}
and this simplifies the proof and the algorithm for the data structure.

Although we will simply use $w_{i}=1$ for all $i$ in this paper,
we support the use of other weights in case it is useful in the future.
Another improvement over \cite{DBLP:conf/colt/LeeSZ19} is that our
bound is tight even for the case some $\nu_{i}$ is much larger than
other $\nu_{i}$. We note that it is an interesting open question
to extend it to dynamic weighted barriers such as the Lee-Sidford
barrier \cite{DBLP:journals/corr/abs-1910-08033} (beyond the case
$n_{i}=1$).

\subsection{Overview}

Our algorithm is based on interior point methods which follow some
path $x(t)$ inside the the interior of the domain $K$. The path
starts at some interior point of the domain $x(1)$ and ends at the
solution $x(0)$ we want to find. One commonly used path is defined
by 
\begin{equation}
x(t)=\arg\min_{Ax=b}c^{\top}x+t\phi(x)\quad\text{with }\phi(x)\defeq\sum_{i=1}^{m}w_{i}\phi_{i}(x_{i})\label{eq:ipm_formula}
\end{equation}
where $\phi_{i}$ are self-concordant barrier functions on $K_{i}$.
The weights $w\in\R_{>0}^{m}$ are fixed throughout the algorithm.
\begin{definition}[\cite{Nesterov1998}]\label{defn:self-concordant}
A function $\phi$ is a $\nu$-self-concordant barrier for a non-empty
open convex set $K$ if $\dom\phi=K$, $\phi(x)\rightarrow+\infty$
as $x\rightarrow\partial K$, and for any $x\in K$ and for any $u\in\R^{n}$
\[
D^{3}\phi(x)[u,u,u]\leq2\|u\|_{\nabla^{2}\phi(x)}\text{ and }\|\nabla\phi(x)\|_{(\nabla^{2}\phi(x))^{-1}}\leq\sqrt{\nu}.
\]
A function $\phi$ is a self-concordant barrier if the first condition
holds.

\end{definition}

For many convex sets, we have an explicit barrier with $\nu=O(n)$.
For the case of linear programs, the convex set $K_{i}=[\ell_{i},u_{i}]$
and one can use the log barrier $-\log(u_{i}-x)-\log(x-\ell_{i})$.
It has self-concordance $1$. Throughout this section, we only use
the fact that $\nu\geq1$ to simplify formulas. \begin{lemma}[{\cite[Corollary 4.3.1]{Nesterov1998}}]The
self-concordance $\nu$ is larger than $1$ for any barrier function. 

\end{lemma}

Since $\phi_{i}$ blows up on $\partial K_{i}$, $x(t)$ lies in the
interior of the domain for $t>0$ (if the interior is non-empty).
By the definition of $x(t)$, $x(0)$ is a minimizer of the problem
\cref{eq:problem}. In \cref{sec:IPMAlgo} to \cref{subsec:GD_on_Phi},
we explain how to follow the path from $x(t)$ to $x(0)$ assuming
$x(t)$ is given for some $t$. In \cref{sec:initial-point-reduction},
we show how to find the initial point $x(t)$ (for some $t$) quickly
by reformulating the problem into an equivalent form. 

\subsection{Interior Point Algorithm}

\label{sec:IPMAlgo}

\begin{algorithm}
\caption{A Robust Interior Point Method for \cref{eq:problem}\label{alg:IPM_framework}}

\algnewcommand{\LeftComment}[1]{\State \(\triangleright\) #1} 

\begin{algorithmic}[1]

\Procedure{\textsc{InteriorPointMethod}}{}

\State \textbf{Input}: linear program $A\in\mathbb{R}^{d\times n},b\in\R^{d},c\in\R^{n}$
with inner radius $r$ and outer radius $R$

\State \textbf{Input}: $\nu_{i}$ self-concordant barrier functions
$\phi_{i}:\R^{n_{i}}\rightarrow\R$ for $i\in[m]$ and its weight
$w\in\R_{\geq1}^{m}$

\State Let $\phi(x)\defeq\sum_{i=1}^{m}w_{i}\phi_{i}(x_{i})$, $L=\|c\|_{2},\kappa=\sum_{i=1}^{m}w_{i}\nu_{i}$

\LeftComment{Modify the convex program and obtain an initial $(x,s)$
according to \cref{thm:IPM_interior}}

\State Let $t=2^{16}(n+\kappa)^{5}\cdot\frac{LR}{\delta}\cdot\frac{R}{r}$
with $\delta=1/128$

\State Compute $x_{c}=\arg\min_{x\in K}c^{\top}x+t\phi(x)$ and $x_{\circ}=\arg\min_{Ax=b}\|x-x_{c}\|_{2}$

\State Let $x=(x_{c},3R+x_{\circ}-x_{c},3R)$ and $s=(-t\nabla\phi(x_{c}),\frac{t}{3R+x_{\circ}-x_{c}},\frac{t}{3R})$

\State Let the new matrix $\new A=[A,A,-A]$, the new barrier and
new weight
\[
\new{\phi}_{i}=\begin{cases}
\phi_{i} & \text{if }i\in[m]\\
-\log x & \text{elses}
\end{cases}\quad\text{and}\quad\new w_{i}=\begin{cases}
w_{i} & \text{if }i\in[m]\\
1 & \text{elses}
\end{cases}
\]

\LeftComment{Find an initial $(x,s)$ for the original linear program.}

\State $((x^{(1)},x^{(2)},x^{(3)}),(s^{(1)},s^{(2)},s^{(3)}))\leftarrow\textsc{Centering}(\new A,\new{\phi},\new w,x,s,t,LR)$

\State $(x,s)\leftarrow(x^{(1)}+x^{(2)}-x^{(3)},s^{(1)})$

\LeftComment{Optimize the original linear program.}

\State $(x,s)\leftarrow\textsc{Centering}(A,\phi,w,x,s,LR,\frac{\eps}{4\sum_{i}w_{i}\nu_{i}})$

\State Return $x$

\EndProcedure

\Procedure{\textsc{Centering}}{$A,\phi,w,x,s,t_{\mathrm{start}},t_{\mathrm{end}}$}

\LeftComment{Definitions}

\State $\lambda=64\log(256m\sum_{i=1}^{m}w_{i})$, $\overline{\eps}=\frac{1}{1440\lambda}$,
$\alpha=\frac{\overline{\eps}}{2}$

\State $\eps_{t}=\frac{\overline{\eps}}{4}(\min_{i}\frac{w_{i}}{w_{i}+\nu_{i}})$,
$h=\frac{\alpha}{64\sqrt{\sum_{i=1}^{m}w_{i}\nu_{i}}}$ where $\nu_{i}$
is the self-concordance of $\phi_{i}$

\State $\mu_{i}^{t}(x,s)\defeq s_{i}/t+w_{i}\nabla\phi_{i}(x_{i})$,
$\gamma_{i}^{t}(x,s)\defeq\|\mu_{i}^{t}(x,s)\|_{x_{i}}^{*}$

\State $c_{i}^{t}(x,s)\defeq\frac{\sinh(\frac{\lambda}{w_{i}}\gamma_{i}^{t}(x,s))}{\gamma_{i}^{t}(x,s)\cdot\sqrt{\sum_{j=1}^{m}w_{j}^{-1}\cosh^{2}(\frac{\lambda}{w_{i}}\gamma_{j}^{t}(x,s))}}$

\State $\Psi_{\lambda}(r)\defeq\sum_{i=1}^{m}\cosh(\lambda r_{i}/w_{i})$,
$\Phi^{t}(x,s)\defeq\Psi_{\lambda}(\gamma^{t}(x,s))$

\State $P_{x}\defeq H_{x}^{-1/2}A^{\top}(AH_{x}^{-1}A^{\top})^{-1}AH_{x}^{-1/2}$
and $H_{x}\defeq\nabla^{2}\phi(x)$

\LeftComment{Main Loop}

\State $\ot=t=t_{\mathrm{start}},\ox=x,\os=s,k=0$.

\While{$t\geq t_{\mathrm{end}}$}\label{line:IPM-iterations}

\State Maintain $\overline{x},\overline{s},\overline{t}$ such that
$\|\overline{x}_{i}-x_{i}\|_{\overline{x}_{i}}\leq\overline{\eps}$,
$\|\overline{s}_{i}-s_{i}\|_{\overline{x}_{i}}^{*}\leq\ot\overline{\eps}w_{i}$
and $|\ot-t|\leq\eps_{t}\ot$\label{line:maintain_xs}

\State $\delta_{\mu,i}\leftarrow-\alpha\cdot c_{i}^{\ot}\oxs\cdot\mu_{i}^{\ot}\oxs$
for all $i\in[m]$

\State Pick $\delta_{x}$ and $\delta_{s}$ such that $A\delta_{x}=0$,
$\delta_{s}\in\mathrm{Range}(A^{\top})$ and\label{line:IPMchoosexs}
\begin{align*}
\|H_{\ox}^{1/2}\delta_{x}-(I-P_{\ox})H_{\ox}^{-1/2}\delta_{\mu}\|_{2} & \leq\overline{\eps}\alpha\\
\|\ot^{-1}H_{\ox}^{-1/2}\delta_{s}-P_{\ox}H_{\ox}^{-1/2}\delta_{\mu}\|_{2} & \leq\overline{\eps}\alpha
\end{align*}

\State $k\leftarrow k+1$, $t\leftarrow\max((1-h)t,t_{\mathrm{end}})$,
$x\leftarrow x+\delta_{x}$, $s\leftarrow s+\delta_{s}$\label{line:move_xs}

\EndWhile

\State Return $(x,s)$

\EndProcedure

\end{algorithmic} 
\end{algorithm}

In this section, we discuss how to follow the path $x(t)$ efficiently.
To lower the cost of each step, we maintain our $(x,s)$ implicitly.
Throughout the algorithm, we only access an approximation of $(x,s)$,
which called $(\ox,\os)$. Our algorithm takes $O(\sqrt{\sum_{i}w_{i}\nu_{i}}\log(1/\eps))$
steps and each step involves solving some linear system according
to $(\ox,\os)$.

To analyze the central path, we use the norm induced by the Hessian
of $\Phi$ throughout this paper. 

\begin{definition}[Induced Norms]\label{def:norm}For each block
$K_{i}$, we define $\|v\|_{x_{i}}\defeq\|v\|_{\nabla^{2}\phi_{i}(x_{i})}$,
$\|v\|_{x_{i}}^{*}\defeq\|v\|_{(\nabla^{2}\phi_{i}(x_{i}))^{-1}}$
for $v\in\R^{n_{i}}$. For the whole domain $K=\prod_{i=1}^{m}K_{i}$,
we define $\|v\|_{x}\defeq\|v\|_{\nabla^{2}\phi(x)}=\sqrt{\sum_{i}w_{i}\|v_{i}\|_{x_{i}}^{2}}$
and $\|v\|_{x}^{*}\defeq\|v\|_{(\nabla^{2}\phi(x))^{-1}}=\sqrt{\sum_{i}w_{i}^{-1}(\|v_{i}\|_{x_{i}}^{*})^{2}}$
for $v\in\Rn$.

\end{definition}

This norm depends on the Hessian and so it changes as the parameter
$x$ changes. The following lemma about self-concordance implies when
the parameter $x$ is not changed rapidly, then the approximate solution
for previous iteration will not be too far from the solution of next
iteration. 

\begin{lemma}[{\cite[Theorem 4.1.6]{Nesterov1998}}] \label{lem:phi_properties}Given
a self-concordant barrier $\phi$. For any $x\in\dim\phi$ and any
$y$ such that $\|y-x\|_{x}<1$, we have $y\in\dom\phi$ and that
\[
(1-\|y-x\|_{x})^{2}\hes{\phi}x\preceq\hes{\phi}y\preceq\frac{1}{(1-\|y-x\|_{x})^{2}}\hes{\phi}x.
\]
\end{lemma}

Instead of following the path $x(t)$ exactly, we follow the path
\begin{align}
s/t+w\nabla\phi(x) & =\mu,\label{eq:KKT}\\
Ax & =b,\nonumber \\
A^{\top}y+s & =c\nonumber 
\end{align}
where $\mu$ is close to $0$ in $(\nabla^{2}\phi(x))^{-1}$ norm.
We enforce $\mu$ close to $0$ using the following potential.

\begin{definition}[Potential Function] \label{def:potential}For
each $i\in[m]$, we define the $i$-th coordinate error 
\begin{align}
\mu_{i}^{t}(x,s)\defeq\frac{s_{i}}{t}+w_{i}\nabla\phi_{i}(x_{i})\label{eq:mu_def}
\end{align}
and its norm $\gamma_{i}^{t}(x,s)\defeq\|\mu_{i}^{t}(x,s)\|_{x_{i}}^{*}$.
We define the soft-max function by 
\[
\Psi_{\lambda}(r)\defeq\sum_{i=1}^{m}\cosh(\lambda\frac{r_{i}}{w_{i}})
\]
for some $\lambda>0$ and finally the potential function is the soft-max
of the norm of the error of each coordinate 
\[
\Phi^{t}(x,s)=\Psi_{\lambda}(\gamma^{t}(x,s)).
\]
\end{definition}

When $(x,s)$ or $t$ is clear in the context, we may ignore them
in the notation. The algorithm alternates between decreasing $t$
multiplicatively and a Newton-like step on \cref{eq:KKT} and the proof
simply shows the potential $\Phi$ is bounded throughout. In \cref{subsec:GD_on_Psi}
and \cref{subsec:GD_on_Phi}, we explain how we design our Newton step.
In \cref{subsec:Phi_xs}, we bound how $\Phi$ changes under our Newton
step. Finally, we give the proof of \cref{thm:IPM_framework} in \cref{subsec:IPM_main}.

\subsection{Gradient Descent on $\Psi_{\lambda}$}

\label{subsec:GD_on_Psi}

Since our goal is to bound $\Phi(x,s)=\Psi_{\lambda}(\gamma)$, we
first discuss how to decreases $\Psi_{\lambda}(r)$ by directly controlling
$r$. Suppose we can make step $r\leftarrow r+\delta_{r}$ with step
size $\sum_{i}w_{i}^{-1}\delta_{r,i}^{2}\leq\alpha^{2}$. Then, a
natural choice is the steepest descent direction\footnote{We use the $*$ to highlight this is the ideal step and to distinguish
with the step we will take.}: 
\[
\delta_{r}^{*}=\arg\min_{\sum_{i}w_{i}^{-1}\delta_{r,i}^{2}\leq\alpha^{2}}\left\langle \nabla\Psi_{\lambda}(r),\delta_{r}\right\rangle .
\]
Using that $\Psi_{\lambda}(r)=\sum_{i=1}^{m}\cosh(\lambda\frac{r_{i}}{w_{i}})$,
we have $\nabla_{r}\Psi_{\lambda}(r)=\frac{\lambda}{w_{i}}\sinh(\frac{\lambda}{w_{i}}r_{i})$
and hence 
\[
\delta_{r}^{*}=\frac{-\alpha\cdot\sinh(\frac{\lambda}{w_{i}}r_{i})}{\sqrt{\sum_{j}w_{j}^{-1}\sinh^{2}(\frac{\lambda}{w_{j}}r_{j})}}.
\]
The following Lemma shows that the direction $\delta_{r}^{*}$ indeed
decreases $\Psi_{\lambda}$. Furthermore, this step is robust under
$\ell_{\infty}$ perturbation of $r$ and $\ell_{2}$ perturbation
of $\delta_{r}^{*}$. To avoid the extra difficulties arising from
$0$ divided by $0$, we replace the $\sinh$ by $\cosh$ in the denominator. 

\begin{lemma} \label{lem:Phi_potential_decrease}Fix any $r\in\R^{m}$
and $w\in\R_{\geq1}^{m}$. Given any $\overline{r}\in\R^{m}$ with
$|r_{i}-\overline{r}_{i}|\leq\frac{w_{i}}{8\lambda}$ for all $i$
and 
\begin{equation}
\delta_{r}=\frac{-\alpha\cdot\sinh(\frac{\lambda}{w_{i}}\overline{r}_{i})}{\sqrt{\sum_{j}w_{j}^{-1}\cosh^{2}(\frac{\lambda}{w_{j}}\overline{r}_{j})}}+\eps_{r}\label{eq:delta_r_cond}
\end{equation}
with $\sqrt{\sum_{i}w_{i}^{-1}\eps_{r,i}^{2}}\leq\frac{\alpha}{8}$.
For any step size $0\leq\alpha\leq\frac{1}{8\lambda}$, we have that
\[
\Psi_{\lambda}(r+\delta_{r})\leq\Psi_{\lambda}(r)-\frac{\alpha\lambda}{2}\sqrt{\sum_{i}w_{i}^{-1}\cosh^{2}(\lambda\frac{r_{i}}{w_{i}})}+\alpha\lambda\sqrt{\sum_{i}w_{i}^{-1}}.
\]
\end{lemma}

\begin{proof} By Taylor expansion, we have 
\begin{equation}
\Psi_{\lambda}(r+\delta_{r})=\Psi_{\lambda}(r)+\left\langle \nabla\Psi_{\lambda}(r),\delta_{r}\right\rangle +\frac{1}{2}\delta_{r}^{\top}\nabla^{2}\Psi_{\lambda}(\widetilde{r})\delta_{r}\label{eq:phi_upper}
\end{equation}
where $\widetilde{r}=r+t\delta_{r}$ for some $t\in[0,1]$.

For the first order term $\left\langle \nabla\Psi_{\lambda}(r),\delta_{r}-\eps_{r}\right\rangle $
in \cref{eq:phi_upper}, we have that 
\[
\left\langle \nabla\Psi_{\lambda}(r),\delta_{r}-\eps_{r}\right\rangle =-\alpha\lambda\frac{\sum_{i}w_{i}^{-1}\sinh(\frac{\lambda}{w_{i}}\overline{r}_{i})\sinh(\frac{\lambda}{w_{i}}r_{i})}{\sqrt{\sum_{j}w_{j}^{-1}\cosh^{2}(\frac{\lambda}{w_{j}}\overline{r}_{j})}}.
\]
Using \cref{lem:sinh-x-y} and the assumption $|r_{i}-\overline{r}_{i}|<\frac{w_{i}}{8\lambda}$,
we have 
\[
\sinh(\frac{\lambda}{w_{i}}\overline{r}_{i})\sinh(\frac{\lambda}{w_{i}}r_{i})\geq\frac{6}{7}\sinh^{2}(\frac{\lambda}{w_{i}}\overline{r}_{i})-\frac{1}{7}\left|\sinh(\frac{\lambda}{w_{i}}\overline{r}_{i})\right|.
\]
Hence, we have 
\begin{align}
 & \left\langle \nabla\Psi_{\lambda}(r),\delta_{r}-\eps_{r}\right\rangle \nonumber \\
\leq & -\frac{6}{7}\alpha\lambda\frac{\sum_{i}w_{i}^{-1}\sinh^{2}(\frac{\lambda}{w_{i}}\overline{r}_{i})}{\sqrt{\sum_{j}w_{j}^{-1}\cosh^{2}(\frac{\lambda}{w_{j}}\overline{r}_{j})}}+\frac{1}{7}\alpha\lambda\frac{\sum_{i}w_{i}^{-1}\left|\sinh(\frac{\lambda}{w_{i}}\overline{r}_{i})\right|}{\sqrt{\sum_{j}w_{j}^{-1}\cosh^{2}(\frac{\lambda}{w_{j}}\overline{r}_{j})}}\nonumber \\
\leq & -\frac{6}{7}\alpha\lambda\frac{\sum_{i}w_{i}^{-1}\cosh^{2}(\frac{\lambda}{w_{i}}\overline{r}_{i})}{\sqrt{\sum_{j}w_{j}^{-1}\cosh^{2}(\frac{\lambda}{w_{j}}\overline{r}_{j})}}+\frac{6}{7}\alpha\lambda\frac{\sum_{i}w_{i}^{-1}}{\sqrt{\sum_{j}w_{j}^{-1}\cosh^{2}(\frac{\lambda}{w_{j}}\overline{r}_{j})}}+\frac{1}{7}\alpha\lambda\frac{\sum_{i}w_{i}^{-1}\left|\sinh(\frac{\lambda}{w_{i}}\overline{r}_{i})\right|}{\sqrt{\sum_{j}w_{j}^{-1}\sinh^{2}(\frac{\lambda}{w_{j}}\overline{r}_{j})}}\nonumber \\
\leq & -\frac{6}{7}\alpha\lambda\sqrt{\sum_{i}w_{i}^{-1}\cosh^{2}(\frac{\lambda}{w_{i}}\overline{r}_{i})}+\alpha\lambda\sqrt{\sum_{i}w_{i}^{-1}}\label{eq:grad_Phi_lower}
\end{align}
Using \cref{lem:sinh-x-y} and the assumption $|r_{i}-\overline{r}_{i}|<\frac{w_{i}}{8\lambda}$
again, we have 
\begin{align*}
\sqrt{\sum_{i}w_{i}^{-1}\cosh^{2}(\frac{\lambda}{w_{i}}\overline{r}_{i})} & \geq\frac{6}{7}\sqrt{\sum_{i}w_{i}^{-1}\cosh^{2}(\frac{\lambda}{w_{i}}r_{i})}.
\end{align*}
Putting this into \cref{eq:grad_Phi_lower}, we have 
\begin{align}
\left\langle \nabla\Psi_{\lambda}(r),\delta_{r}-\eps_{r}\right\rangle  & \leq-\frac{36}{49}\alpha\lambda\sqrt{\sum_{i}w_{i}^{-1}\cosh^{2}(\frac{\lambda}{w_{i}}r_{i})}+\alpha\lambda\sqrt{\sum_{i}w_{i}^{-1}}.\label{eq:grad_phi_1}
\end{align}

For the first order term $\left\langle \nabla\Psi_{\lambda}(r),\eps_{r}\right\rangle $
in \cref{eq:phi_upper}, we have that 
\begin{align}
\left\langle \nabla\Psi_{\lambda}(r),\eps_{r}\right\rangle  & =\sum_{i}\frac{\lambda}{w_{i}}\sinh(\frac{\lambda}{w_{i}}r_{i})\eps_{r,i}\nonumber \\
 & \leq\lambda\cdot\sqrt{\sum_{i}w_{i}^{-1}\sinh^{2}(\frac{\lambda}{w_{i}}r_{i})}\sqrt{\sum_{i}w_{i}^{-1}\eps_{r,i}^{2}}\nonumber \\
 & \leq\frac{1}{8}\alpha\lambda\sqrt{\sum_{i}w_{i}^{-1}\cosh^{2}(\frac{\lambda}{w_{i}}r_{i})}.\label{eq:grad_phi_2}
\end{align}

For the second order term $\delta_{r}^{\top}\nabla^{2}\Psi_{\lambda}(\widetilde{r})\delta_{r}$
in \cref{eq:phi_upper}, we note that 
\begin{align*}
\delta_{r}^{\top}\nabla^{2}\Psi_{\lambda}(\widetilde{r})\delta_{r} & =\lambda^{2}\sum_{i}w_{i}^{-2}\delta_{r,i}^{2}\cosh(\lambda\frac{\widetilde{r}_{i}}{w_{i}}).
\end{align*}
Note that 
\begin{align}
\sqrt{\sum_{i}w_{i}^{-1}\delta_{r,i}^{2}} & \leq\sqrt{\sum_{i}w_{i}^{-1}\left(\frac{\alpha\cdot\sinh(\frac{\lambda}{w_{i}}\overline{r}_{i})}{\sqrt{\sum_{j}w_{j}^{-1}\cosh^{2}(\frac{\lambda}{w_{j}}\overline{r}_{j})}}\right)^{2}}+\sqrt{\sum_{i}w_{i}^{-1}\eps_{r,i}^{2}}\nonumber \\
 & \leq\alpha+\frac{\alpha}{8}=\frac{9\alpha}{8}.\label{eq:hess_phi_1}
\end{align}
In particular, this shows that $|\delta_{r,i}|\leq\frac{9\alpha}{8}\sqrt{w_{i}}\leq\frac{9\alpha}{8}w_{i}$.
Using this and \cref{eq:hess_phi_1}, we have 
\begin{align}
\delta_{r}^{\top}\nabla^{2}\Psi_{\lambda}(\widetilde{r})\delta_{r} & =\lambda^{2}\sum_{i}w_{i}^{-2}\delta_{r,i}^{2}\cosh(\lambda\frac{\widetilde{r}_{i}}{w_{i}})\nonumber \\
 & \leq\frac{9\alpha}{8}\lambda^{2}\sum_{i}w_{i}^{-1}|\delta_{r,i}|\cosh(\lambda\frac{\widetilde{r}_{i}}{w_{i}})\nonumber \\
 & \leq\frac{9\alpha}{8}\lambda^{2}\sqrt{\sum_{i}w_{i}^{-1}\delta_{r,i}^{2}}\sqrt{\sum_{i}w_{i}^{-1}\cosh^{2}(\lambda\frac{\widetilde{r}_{i}}{w_{i}})}\nonumber \\
 & \leq(\frac{9\alpha}{8})^{2}\lambda^{2}\left(\sqrt{\sum_{i}w_{i}^{-1}\cosh^{2}(\lambda\frac{\widetilde{r}_{i}}{w_{i}})}\right)\nonumber \\
 & \leq(\frac{9\alpha}{8})^{2}\lambda^{2}\left(\frac{8}{7}\sqrt{\sum_{i}w_{i}^{-1}\cosh^{2}(\lambda\frac{r_{i}}{w_{i}})}\right)\label{eq:hess_phi_2}
\end{align}
where we used \cref{eq:hess_phi_1} at the third last inequality and
\cref{lem:sinh-x-y} at the second last inequality.

Putting \cref{eq:grad_phi_1}, \cref{eq:grad_phi_2}, and \cref{eq:hess_phi_2}
into \cref{eq:phi_upper} gives 
\begin{align*}
\Psi_{\lambda}(r+\delta_{r})= & \Psi_{\lambda}(r)+\left\langle \nabla\Psi_{\lambda}(r),\delta_{r}\right\rangle +\frac{1}{2}\delta_{r}^{\top}\nabla_{\lambda}^{2}(\widetilde{r})\delta_{r}\\
\leq & \Psi_{\lambda}(r)-\frac{36}{49}\alpha\lambda\sqrt{\sum_{i}w_{i}^{-1}\cosh^{2}(\lambda\frac{r_{i}}{w_{i}})}+\alpha\lambda\sqrt{\sum_{i}w_{i}^{-1}}\\
 & +(\frac{1}{8}\alpha\lambda+\frac{8}{7}(\frac{9\alpha}{8})^{2}\lambda^{2})\sqrt{\sum_{i}w_{i}^{-1}\cosh^{2}(\lambda\frac{r_{i}}{w_{i}})}
\end{align*}
Using $\alpha\leq\frac{1}{8\lambda}$, we can simplify it to 
\begin{align*}
\Psi_{\lambda}(r+\delta_{r}) & \leq\Psi_{\lambda}(r)-\left(\frac{36}{49}-\frac{1}{8}-\frac{1}{2}(\frac{9}{8})^{2}\frac{1}{7}\right)\alpha\lambda\sqrt{\sum_{i}w_{i}^{-1}\cosh^{2}(\lambda\frac{r_{i}}{w_{i}})}+\alpha\lambda\sqrt{\sum_{i}w_{i}^{-1}}\\
 & \leq\Psi_{\lambda}(r)-\frac{\alpha\lambda}{2}\sqrt{\sum_{i}w_{i}^{-1}\cosh^{2}(\lambda\frac{r_{i}}{w_{i}})}+\alpha\lambda\sqrt{\sum_{i}w_{i}^{-1}}.
\end{align*}
\end{proof}

\subsection{Gradient Descent on $\Phi$}

\label{subsec:GD_on_Phi}

In the last section, we discussed how to decrease $\Psi_{\lambda}$
by changing the input $r$ directly. But our real potential $\Phi^{t}(x,s)=\Psi_{\lambda}(\gamma^{t}(x,s))$
is defined indirectly using $(x,s)$. In this section, we discuss
how to design the Newton-like step for $(x,s)$. Note that the non-linear
equation \cref{eq:KKT} has an unique solution for any vector $\mu$.
In particular, the solution $x$ is the solution of the optimization
problem $\min_{Ax=b}c^{\top}x+t\sum_{i=1}^{m}w_{i}\phi_{i}(x_{i})-t\mu^{\top}x$.
Hence, we can move $\mu$ arbitrarily while maintaining \cref{eq:KKT}
by moving $x$ and $s$.

Since our goal is to decrease $\Phi(x,s)=\Psi_{\lambda}(\gamma)$,
similar to \cref{subsec:GD_on_Psi}, a natural choice is the steepest
descent direction: 
\begin{equation}
\delta_{\mu}^{*}=\arg\min_{\|\delta_{\mu}\|_{x}^{*}=\alpha}\langle\nabla_{\mu}\Psi_{\lambda}(\|\mu_{i}\|_{x_{i}}^{*}),\mu+\delta_{\mu}\rangle\label{eq:delta_mu}
\end{equation}
with step size $\alpha$. We can view this as a gradient descent step
on $\Phi$ for $\mu$ with step size $\alpha$. Recall that $\Psi_{\lambda}(r)=\sum_{i=1}^{m}\cosh(\lambda\frac{r_{i}}{w_{i}})$.
Hence, $\nabla_{\|\mu_{i}\|_{x_{i}}^{*}}\Psi_{\lambda}(\|\mu_{i}\|_{x_{i}}^{*})=\frac{\lambda}{w_{i}}\sinh(\frac{\lambda}{w_{i}}\|\mu_{i}\|_{x_{i}}^{*})$
and 
\[
\nabla_{\mu_{i}}\Psi_{\lambda}(\|\mu_{i}\|_{x_{i}}^{*})=\frac{\lambda\sinh(\frac{\lambda}{w_{i}}\|\mu_{i}\|_{x_{i}}^{*})}{w_{i}\|\mu_{i}\|_{x_{i}}^{*}}\cdot\nabla\phi_{i}(x_{i})^{-1}\mu_{i}=\frac{\lambda\sinh(\frac{\lambda}{w_{i}}\gamma_{i}^{t}(x,s))}{w_{i}\gamma_{i}^{t}(x,s)}\cdot\nabla\phi_{i}(x_{i})^{-1}\mu_{i}
\]
Solve \cref{eq:delta_mu}\footnote{The derivation of the formula is not used in the main proof as this
is just a motivation for the choice of the step. Therefore, we skip
the proof of this. An alternative choice is the gradient step on $\min_{Ax=b,A^{\top}y+s=c}\Phi^{t}(x,s)$.
This step will be very similar to the step we use in this paper. But
it contains few more terms and may make the proof longer.}, we get 
\[
\delta_{\mu,i}^{*}(x,s)=-\frac{\alpha\sinh(\frac{\lambda}{w_{i}}\gamma_{i}^{t}(x,s))}{\gamma_{i}^{t}(x,s)\cdot\sqrt{\sum_{j=1}^{m}w_{j}^{-1}\sinh^{2}(\frac{\lambda}{w_{j}}\gamma_{j}^{t}(x,s))}}\cdot\mu_{i}^{t}(x,s).
\]

To move $\mu$ to $\mu+\delta_{\mu}$ approximately, we take Newton
step $(\delta_{x}^{*},\delta_{s}^{*})$\footnote{We use the $*$ to highlight this is the ideal step and to distinguish
with the step we will take..}: 
\begin{align*}
\frac{1}{t}\delta_{s}^{*}+\hes{\phi}x\delta_{x}^{*} & =\delta_{\mu}^{*}(x,s),\\
A\delta_{x}^{*} & =0,\\
A^{\top}\delta_{y}^{*}+\delta_{s}^{*} & =0.
\end{align*}
Using $H_{x}$ to denote $\hes{\phi}x$, we solve the system above,
and get 
\begin{align*}
\delta_{x}^{*} & =H_{x}^{-1}\delta_{\mu}^{*}-H_{x}^{-1}A^{\top}(AH_{x}^{-1}A^{\top})^{-1}AH_{x}^{-1}\delta_{\mu}^{*}(x,s),\\
\delta_{s}^{*} & =tA^{\top}(AH_{x}^{-1}A^{\top})^{-1}AH_{x}^{-1}\delta_{\mu}^{*}(x,s).
\end{align*}
Let the orthogonal projection matrix $P_{x}\defeq H_{x}^{-1/2}A^{\top}(AH_{x}^{-1}A^{\top})^{-1}AH_{x}^{-1/2}$,
then we can rewrite it as 
\begin{align*}
\delta_{x}^{*} & =H_{x}^{-1/2}(I-P_{x})H_{x}^{-1/2}\delta_{\mu}^{*}(x,s),\\
\delta_{s}^{*} & =tH_{x}^{1/2}P_{x}H_{x}^{-1/2}\delta_{\mu}^{*}(x,s).
\end{align*}
Our robust algorithm only uses $H_{\ox}$, $P_{\ox}$ and $\delta_{\mu}^{*}(\ox,\os)$
where $(\overline{x},\overline{s})$ is some approximation of $(x,s)$.
Formally, our step on $x$ and $s$ is defined in \cref{line:IPMchoosexs}.
Note that we allow for an extra error for $(\delta_{x},\delta_{s})$
on top of the error due to $\ox$ and $\os$. Also, we replace $\sinh$
by $\cosh$ in the denominator as in \cref{lem:Phi_potential_decrease}.

\subsection{Bounding $\Phi$ under Changes of $x$ and $s$}

\label{subsec:Phi_xs}

To use \cref{lem:Phi_potential_decrease} to bound the potential, we
need to verify $|\gamma_{i}^{t}(\new x,\new s)-\gamma_{i}^{t}(x,s)|\leq\frac{w_{i}}{8\lambda}$
and \cref{eq:delta_r_cond}.

\subsubsection{Verifying conditions of \cref{lem:Phi_potential_decrease}}

Recall that the ideal step we want to take is 
\[
\delta_{\mu,i}^{*}=-\alpha\cdot c_{i}^{t}(x,s)\cdot\mu_{i}^{t}(x,s).
\]
where $\alpha$ is the step size. A rough calculation shows 
\begin{align*}
\gamma_{i}^{t}(\new x,\new s) & =\|\mu_{i}+\delta_{\mu,i}^{*}\|_{x_{i}}^{*}\\
 & \sim\|\mu_{i}\|_{x_{i}}^{*}-\frac{\alpha}{\|\mu_{i}\|_{x_{i}}^{*}}\cdot c_{i}^{t}(x,s)\cdot\mu_{i}^{\top}\nabla^{2}\phi_{i}(x)^{-1}\mu_{i}\\
 & =\gamma_{i}^{t}(x,s)-\alpha\cdot c_{i}^{t}(x,s)\cdot\gamma_{i}^{t}(x,s)
\end{align*}
This shows that \cref{eq:delta_r_cond} should roughly holds. Formally,
in \cref{lem:change_of_gamma}, we prove this holds for the step we
take in \cref{alg:IPM_framework}. First, we bound the step size for
each block $\delta_{x,i}$. 

\begin{lemma}[Step size of $\delta_{x}$] \label{lem:alpha_i}Let
$\alpha_{i}\defeq\|\delta_{x,i}\|_{\ox_{i}}$, then 
\[
\sqrt{\sum_{i=1}^{m}w_{i}\alpha_{i}^{2}}\leq\frac{9}{8}\alpha.
\]
In particular, we have $\alpha_{i}\leq\frac{9}{8}\alpha$. Similarly,
we have $\sqrt{\sum_{i=1}^{m}w_{i}^{-1}(\|\delta_{s,i}\|_{\ox_{i}}^{*})^{2}}\leq\frac{9}{8}\alpha\cdot t.$

\end{lemma}

\begin{proof}We have 
\[
\sqrt{\sum_{i=1}^{m}w_{i}\alpha_{i}^{2}}=\norm{\delta_{x}}_{\ox}\leq\|(I-P_{\overline{x}})H_{\overline{x}}^{-1/2}\delta_{\mu}\|_{2}+\overline{\eps}\alpha\leq\|H_{\overline{x}}^{-1/2}\delta_{\mu}\|_{2}+\overline{\eps}\alpha\leq\alpha+\overline{\eps}\alpha\leq\frac{9}{8}\alpha,
\]
where the first inequality follows by the choice that $\delta_{x}\approx(I-P_{\overline{x}})H_{\overline{x}}^{-1/2}\delta_{\mu}$,
the second inequality follows by $I-P_{\ox}$ is an orthogonal projection
matrix and, second last equality follows by the step size for $\delta_{\mu}$
and the last equality follows by $\overline{\eps}\leq\frac{1}{8}$.

For $\delta_{s}$, we note that 
\[
\sqrt{\sum_{i=1}^{m}w_{i}^{-1}(\|\delta_{s,i}\|_{\ox_{i}}^{*})^{2}}=\norm{\delta_{s}}_{\ox}^{*}\leq\ot\|P_{\overline{x}}H_{\overline{x}}^{-1/2}\delta_{\mu}\|_{2}+\overline{\eps}\alpha\ot\leq\ot\|H_{x}^{-1/2}\delta_{\mu}\|_{2}+\overline{\eps}\alpha\ot\leq\frac{9}{8}\alpha t
\]
where we used $\ot\leq\frac{33}{32}t$ and $\overline{\eps}\leq\frac{1}{32}$.\end{proof}

To bound the change of $\gamma$, we first show that $\new{\mu}$
is close to $\mu+\delta_{\mu}$. 

\begin{lemma}[Change in $\mu$]\label{lem:mu_change}Let $\mu_{i}^{t}(x^{\text{new}},s^{\text{new}})=\mu_{i}^{t}(x,s)+\delta_{\mu,i}+\eps_{i}^{(\mu)}$
with $\beta_{i}\defeq\|\eps_{i}^{(\mu)}\|_{x_{i}}^{*}$, we have $\sqrt{\sum_{i=1}^{m}w_{i}^{-1}\beta_{i}^{2}}\leq15\overline{\eps}\alpha$.

\end{lemma}

\begin{proof}

Let $\eps_{1}=H_{\ox}^{1/2}\delta_{x}-(I-P_{\ox})H_{\ox}^{-1/2}\delta_{\mu}$
and $\eps_{2}=\ot^{-1}H_{\ox}^{-1/2}\delta_{s}-P_{\ox}H_{\ox}^{-1/2}\delta_{\mu}$.
By definition of $\mu$, we have 
\begin{align}
\mu_{i}^{t}(\new x,\new s)= & \frac{\new{s_{i}}}{t}+w_{i}\nabla\phi_{i}(\new x)\nonumber \\
= & \mu_{i}^{t}(x,s)+\frac{1}{t}\delta_{s}+w_{i}(\nabla\phi_{i}(\new x)-\nabla\phi_{i}(x_{i}))\nonumber \\
= & \mu_{i}^{t}(x,s)+\delta_{\mu,i}+\underbrace{w_{i}(\nabla\phi_{i}(\new x)-\nabla\phi_{i}(x_{i})-\nabla^{2}\phi_{i}(\ox_{i})\delta_{x})}_{\eps_{i}^{(\mu,1)}}\nonumber \\
 & +\underbrace{\left(H_{\ox}^{1/2}(\eps_{1}+\eps_{2})\right)_{i}}_{\eps_{i}^{(\mu,2)}}+\underbrace{(\frac{1}{t}-\frac{1}{\ot})\delta_{s}}_{\eps_{i}^{(\mu,3)}}\label{eq:mu_change_term}
\end{align}
where the last step follows by $\delta_{\mu,i}=\frac{1}{\ot}\delta_{s,i}+w_{i}\nabla^{2}\phi_{i}(\ox_{i})\delta_{x,i}-(w_{i}\nabla^{2}\phi_{i}(\ox_{i}))^{1/2}(\eps_{1}+\eps_{2})$.

To bound $\eps_{i}^{(\mu,1)}$, let $x^{(u)}=u\new x+(1-u)x$, then
we have 
\begin{align*}
\eps_{i}^{(\mu,1)}/w_{i} & =\nabla\phi_{i}(x_{i}^{\text{new}})-\nabla\phi_{i}(x_{i})-\nabla^{2}\phi_{i}(\ox_{i})\delta_{x,i}\\
 & =\int_{0}^{1}\left(\nabla^{2}\phi_{i}(x_{i}^{(u)})-\nabla^{2}\phi_{i}(\ox_{i})\right)\delta_{x,i}\d u.
\end{align*}
By \cref{lem:phi_properties}, we have 
\begin{equation}
(1-\|x_{i}^{(u)}-\ox_{i}\|_{\ox_{i}})^{2}\nabla^{2}\phi_{i}(\ox_{i})\preceq\nabla^{2}\phi_{i}(x^{(u)})\preceq\frac{1}{(1-\|x_{i}^{(u)}-\ox_{i}\|_{\ox_{i}})^{2}}\nabla^{2}\phi_{i}(\ox_{i}).\label{eq:phi_pro_instance}
\end{equation}
Note that 
\[
\|x_{i}^{(u)}-\ox_{i}\|_{\ox_{i}}\leq\|x_{i}^{(u)}-x_{i}\|_{\ox_{i}}+\|x_{i}-\ox_{i}\|_{\ox_{i}}\leq u\|\delta_{x,i}\|_{\ox_{i}}+\overline{\eps}\leq\alpha_{i}+\overline{\eps}\leq\frac{9}{8}\alpha+\overline{\eps}\leq\frac{25}{16}\overline{\eps},
\]
where we used $\|x_{i}-\ox_{i}\|_{\ox_{i}}\leq\overline{\eps}$, $\alpha_{i}\leq\frac{9}{8}\alpha$
(\cref{lem:alpha_i}) and $2\alpha\leq\overline{\eps}$ (by the algorithm
description). Combine two inequalities above and using that $\overline{\eps}\leq\frac{1}{8}$,
we get 
\begin{equation}
-5\overline{\eps}\nabla^{2}\phi_{i}(\ox_{i})\preceq\nabla^{2}\phi_{i}(x^{(u)})-\nabla^{2}\phi_{i}(\ox_{i})\preceq5\overline{\eps}\nabla^{2}\phi_{i}(\ox_{i}).\label{eq:hessphi_upper}
\end{equation}
Using this, \cref{eq:phi_pro_instance} and the algorithm description,
we have 
\begin{align*}
 & \left(\nabla^{2}\phi_{i}(x^{(u)})-\nabla^{2}\phi_{i}(\ox_{i})\right)\left(\nabla^{2}\phi_{i}(x_{i})\right)^{-1}\left(\nabla^{2}\phi_{i}(x^{(u)})-\nabla^{2}\phi_{i}(\ox_{i})\right)\\
\preceq & \frac{1}{(1-\frac{25}{16}\frac{1}{8})^{2}}\left(\nabla^{2}\phi_{i}(x^{(u)})-\nabla^{2}\phi_{i}(\ox_{i})\right)\left(\nabla^{2}\phi_{i}(\ox_{i})\right)^{-1}\left(\nabla^{2}\phi_{i}(x^{(u)})-\nabla^{2}\phi_{i}(\ox_{i})\right)\\
\preceq & \frac{(5\overline{\eps})^{2}}{(1-\frac{25}{16}\frac{1}{8})^{2}}\nabla^{2}\phi_{i}(\ox_{i})\preceq40\overline{\eps}^{2}\nabla^{2}\phi_{i}(\ox_{i}).
\end{align*}
This implies 
\begin{align*}
 & \left\Vert \left(\nabla^{2}\phi_{i}(x^{(u)})-\nabla^{2}\phi_{i}(x_{i})\right)\delta_{x,i}\right\Vert _{x_{i}}^{*}\\
= & \sqrt{\delta_{x,i}^{\top}\left(\nabla^{2}\phi_{i}(x^{(u)})-\nabla^{2}\phi_{i}(\ox_{i})\right)^{\top}\left(\nabla^{2}\phi_{i}(x_{i})\right)^{-1}\left(\nabla^{2}\phi_{i}(x^{(u)})-\nabla^{2}\phi_{i}(\ox_{i})\right)\delta_{x,i}}\\
\leqslant & \sqrt{40\overline{\eps}^{2}\delta_{x,i}^{\top}\nabla^{2}\phi_{i}(\ox_{i})\delta_{x,i}^{\top}}\\
\leqslant & \sqrt{40}\cdot\overline{\eps}\|\delta_{x,i}\|_{\ox_{i}}\\
= & \frac{9\sqrt{40}}{8}\cdot\overline{\eps}\alpha_{i}.
\end{align*}
Hence, 
\begin{align}
\|\eps_{i}^{(\mu,1)}\|_{x_{i}}^{*} & \leqslant w_{i}\int_{0}^{1}\left\Vert \left(\nabla^{2}\phi_{i}(x^{(u)})-\nabla^{2}\phi_{i}(\ox_{i})\right)\delta_{x,i}\right\Vert _{x_{i}}^{*}\d u\leqslant7.2\overline{\eps}w_{i}\alpha_{i}.\label{eq:eps_mu_1}
\end{align}

To bound the term $\eps_{i}^{(\mu,2)}$ in \cref{eq:mu_change_term},
we use the definition of induced norm (\cref{def:norm}) and \cref{eq:phi_pro_instance}
to get 
\begin{align}
\sqrt{\sum_{i}w_{i}^{-1}(\|\eps_{i}^{(\mu,2)}\|_{x_{i}}^{*})^{2}} & =\|\eps^{(\mu,2)}\|_{x}^{*}=\|H_{\ox}^{1/2}(\eps_{1}+\eps_{2})\|_{x}^{*}\nonumber \\
 & \leq\frac{1}{1-\frac{25}{16}\frac{1}{8}}\|H_{\ox}^{1/2}(\eps_{1}+\eps_{2})\|_{\ox}^{*}\nonumber \\
 & \leq2\|\eps_{1}+\eps_{2}\|_{2}\leq4\overline{\eps}\alpha.\label{eq:eps_mu_2}
\end{align}
where we used $\|\eps_{1}\|_{2}\leq\overline{\eps}\alpha$ and $\|\eps_{2}\|_{2}\leq\overline{\eps}\alpha$
at the end according to the algorithm description.

To bound the term $\eps_{i}^{(\mu,3)}$ in \cref{eq:mu_change_term},
we note that
\begin{align}
\sqrt{\sum_{i}w_{i}^{-1}(\|(\frac{1}{t}-\frac{1}{\ot})\delta_{s,i}\|_{x_{i}}^{*})^{2}} & =\frac{1}{t}|\frac{\ot-t}{\ot}|\sqrt{\sum_{i}w_{i}^{-1}(\|\delta_{s,i}\|_{x_{i}}^{*})^{2}}\nonumber \\
 & \leq\frac{3}{2t}|\frac{\ot-t}{\ot}|\sqrt{\sum_{i}w_{i}^{-1}(\|\delta_{s,i}\|_{\ox_{i}}^{*})^{2}}\nonumber \\
 & \leq2\alpha\eps_{t}t\label{eq:eps_mu_3}
\end{align}
where we used \cref{eq:hessphi_upper} on the second inequality, \cref{lem:alpha_i}
$|t-\ot|\leq\eps_{t}\ot$ at the end.

Using $\eps^{(\mu)}=\eps^{(\mu,1)}+\eps^{(\mu,2)}+\eps^{(\mu,3)}$,
\cref{eq:eps_mu_1} and \cref{eq:eps_mu_2}, we have 
\begin{align*}
\sqrt{\sum_{i}w_{i}^{-1}(\|\eps_{i}^{(\mu)}\|_{x_{i}}^{*})^{2}} & \leq7.2\overline{\eps}\sqrt{\sum_{i}w_{i}\alpha_{i}^{2}}+4\overline{\eps}\alpha+2\alpha\eps_{t}t\leq15\overline{\eps}\alpha.
\end{align*}
\end{proof}

Now, we can check the condition $|r_{i}-\overline{r}_{i}|\leq\frac{w_{i}}{8\lambda}$
in \cref{lem:Phi_potential_decrease}. The following Lemma shows that
it is true when $\gamma_{i}^{t}(x,s)\leq w_{i}$ for all $i$, which
holds when $\Phi$ is small enough. 

\begin{lemma}\label{lem:mu_distance}Assume $\gamma_{i}^{t}(x,s)\leq w_{i}$
for all $i$. For all $i\in[m]$, we have 
\[
\|\mu_{i}^{t}(x,s)-\mu_{i}^{\ot}(\overline{x},\overline{s})\|_{x_{i}}^{*}\leq3\overline{\eps}w_{i}.
\]
Furthermore, we have that $|\gamma_{i}^{t}(x,s)-\gamma_{i}^{\ot}(\ox,\os)|\leq5\overline{\eps}w_{i}.$

\end{lemma}

\begin{proof}For the first result, note that 
\[
\|\mu_{i}^{t}(x,s)-\mu_{i}^{t}(\overline{x},\overline{s})\|_{\ox_{i}}^{*}\leq\frac{1}{t}\|s_{i}-\os_{i}\|_{\ox_{i}}^{*}+w_{i}\|\nabla\phi_{i}(x_{i})-\nabla\phi_{i}(\ox_{i})\|_{\ox_{i}}^{*}.
\]
Let $x^{(u)}=ux_{i}+(1-u)\ox_{i}$. By \cref{eq:hessphi_upper}, we
have $\nabla^{2}\phi_{i}(x_{i}^{(u)})\preceq(1+5\overline{\eps})\nabla^{2}\phi_{i}(\ox_{i})\preceq\frac{5}{8}\nabla^{2}\phi_{i}(\ox_{i})$
and hence 
\[
\nabla^{2}\phi_{i}(x_{i}^{(u)})(\nabla^{2}\phi_{i}(\ox_{i}))^{-1}\nabla^{2}\phi_{i}(x_{i}^{(u)})\preceq\frac{25}{64}\nabla^{2}\phi_{i}(\ox_{i}).
\]
Therefore, we have 
\begin{equation}
\|\nabla\phi_{i}(x_{i})-\nabla\phi_{i}(\ox_{i})\|_{\ox_{i}}^{*}=\left\Vert \int_{0}^{1}\nabla^{2}\phi_{i}(x_{i}^{(u)})(x_{i}-\ox_{i})\d u\right\Vert _{\ox_{i}}^{*}\leq\frac{5}{8}\left\Vert x_{i}-\ox_{i}\right\Vert _{\ox_{i}}.\label{eq:grad_phi_change}
\end{equation}
Using $\|s_{i}-\os_{i}\|_{\ox_{i}}^{*}\leq\ot\overline{\eps}w_{i}$,
we have $\|\mu_{i}^{t}(x,s)-\mu_{i}^{t}(\overline{x},\overline{s})\|_{\ox_{i}}^{*}\leq2\overline{\eps}w_{i}$.

Finally, we note that $\gamma_{i}^{t}(x,s)\leq w_{i}$ and $\|\nabla\phi_{i}(x_{i})\|_{\ox_{i}}^{*}\leq2\|\nabla\phi_{i}(x_{i})\|_{x_{i}}^{*}\leq2\nu_{i}$.
This implies that 
\[
\|\frac{s_{i}}{t}-\frac{s_{i}}{\ot}\|_{\ox_{i}}\leq(1-\frac{t}{\ot})\|\frac{s_{i}}{t}\|_{\ox_{i}}\leq2(\frac{t-\ot}{\ot})(w_{i}+\nu_{i})\leq\frac{1}{2}\overline{\eps}w_{i}.
\]
and hence the result.

For the second result, note that 
\begin{align*}
|\gamma_{i}^{t}(x,s)-\gamma_{i}^{\ot}(\ox,\os)| & \leq\|\mu_{i}^{t}(x,s)-\mu_{i}^{\ot}(\overline{x},\overline{s})\|_{\ox_{i}}^{*}+|\|\mu_{i}^{t}(x,s)\|_{x_{i}}^{*}-\|\mu_{i}^{t}(x,s)\|_{\ox_{i}}^{*}|\\
 & \leq3\overline{\eps}w_{i}+2\|x_{i}-\ox_{i}\|_{\ox_{i}}\|\mu_{i}^{t}(x,s)\|_{x_{i}}^{*}\\
 & =3\overline{\eps}w_{i}+2\overline{\eps}\gamma_{i}^{t}(x,s)\leq5\overline{\eps}w_{i}
\end{align*}
where we used the algorithm description and \cref{lem:phi_properties}
\end{proof}

\subsubsection{First Order Approximation of $\gamma$}

In this subsection, we will show that $\gamma_{i}$ is close to $\gamma_{i}^{t}(x,s)-\alpha\cdot c_{i}^{t}(\ox,\os)\cdot\gamma_{i}^{t}(\ox,\os)$.
First, we need the following helper lemma to bound $\gamma_{i}$,
$\sum_{i=1}^{m}w_{i}^{-1}\sinh^{2}(\frac{\lambda}{w_{i}}\gamma_{i}^{t}(\ox,\os))$
and $c(\ox,\os)$. In this helper lemma, we assume that $\Phi$ is
not too large, which is the invariant maintained throughout the algorithm.
\begin{lemma}\label{lem:c_bound}Suppose that $\Phi^{t}(x,s)\leq\cosh(\lambda)$,
then we have 
\begin{itemize}
\item $\gamma_{i}^{t}(x,s)\leq w_{i}$. and $\gamma_{i}^{\ot}(\ox,\os)\leq2w_{i}$. 
\item $0\leq c_{i}^{\ot}(\ox,\os)\leq\lambda$. 
\end{itemize}
\end{lemma}

\begin{proof}For the first inequality, note that $\Phi^{t}(x,s)\leq\cosh(\lambda)$
implies that $\gamma_{i}^{t}(x,s)\leq w_{i}$ for all $i$. Hence,
\cref{lem:mu_distance} shows that 
\[
|\gamma_{i}^{t}(x,s)-\gamma_{i}^{\ot}(\ox,\os)|\leq5w_{i}\overline{\eps}\leq\frac{5w_{i}}{8\lambda}.
\]
Hence, we have $\gamma_{i}^{\ot}(\ox,\os)\leq2w_{i}$.

For the second inequality, we note that 
\[
c_{i}^{\ot}(\ox,\os)=\frac{\sinh(\frac{\lambda}{w_{i}}\gamma_{i}^{\ot}(\ox,\os))}{\gamma_{i}^{\ot}(\ox,\os)\cdot\sqrt{\sum_{j=1}^{m}w_{j}^{-1}\cosh^{2}(\frac{\lambda}{w_{j}}\gamma_{j}^{\ot}(\ox,\os))}}.
\]
Since $\gamma_{i}\geq0$ (by definition), we have $c_{i}^{\ot}\geq0$.
If $\gamma_{i}^{\ot}(\ox,\os)\geq\frac{w_{i}}{\lambda}$, we have
that 
\[
c_{i}^{\ot}(\ox,\os)\leq\frac{w_{i}\sinh(\frac{\lambda}{w_{i}}\gamma_{i}^{\ot}(\ox,\os))}{\gamma_{i}^{\ot}(\ox,\os)\cdot\cosh(\frac{\lambda}{w_{i}}\gamma_{i}^{\ot}(\ox,\os))}\leq\lambda
\]
where we used $w_{i}\geq1$. If $\gamma_{i}^{\ot}(\ox,\os)\leq\frac{w_{i}}{\lambda}$,
we use that $|\sinh(x)|\leq2|x|$ for all $|x|\leq1$ and get 
\[
c_{i}^{\ot}(\ox,\os)\leq\frac{2\frac{\lambda}{w_{i}}\gamma_{i}^{\ot}(\ox,\os)}{\gamma_{i}^{\ot}(\ox,\os)\cdot\sqrt{\sum_{j=1}^{m}w_{j}^{-1}\cosh^{2}(\frac{\lambda}{w_{j}}\gamma_{j}^{\ot}(\ox,\os))}}\leq\frac{2\lambda}{w_{i}\sqrt{4\sum_{j=1}^{m}w_{j}^{-1}}}\leq\lambda.
\]
\end{proof}

Finally, we can bound the distance between $\new{\gamma}$ and $\gamma-\alpha c\gamma$.
Here we crucially use the fact that $\sinh(x)/x$ is bounded at $x=0$
and it makes our argument slightly simpler than~\cite{DBLP:conf/colt/LeeSZ19}. 

\begin{lemma}[Change in $\gamma$]\label{lem:change_of_gamma} Assume
$\Phi^{t}(x,s)\leq\cosh(\lambda)$. For all $i\in[m]$, let 
\[
\eps_{r,i}\defeq\gamma_{i}^{t}(\new x,\new s)-\gamma_{i}^{t}(x,s)+\alpha\cdot c_{i}^{\ot}(\ox,\os)\cdot\gamma_{i}^{\ot}(\ox,\os).
\]
Then, we have 
\[
\sqrt{\sum_{i=1}^{m}w_{i}^{-1}\eps_{r,i}^{2}}\leq90\overline{\eps}\lambda\alpha+4\max_{i}\left(\frac{\gamma_{i}^{t}(x,s)}{w_{i}}\right)\alpha
\]
\end{lemma}

\begin{proof}For notation simplicity, we write $\overline{c}_{i}=c_{i}^{\ot}(\ox,\os)$.
Also, we use $\gamma_{i}^{t}(x,z,s)$ to denote $\|\mu_{i}(x,s)\|_{z_{i}}^{*}$.
Using $\delta_{\mu,i}=-\alpha\cdot\overline{c}_{i}\cdot\mu_{i}^{\ot}(\overline{x},\overline{s})$,
we have 
\begin{align}
\gamma_{i}^{t}(\new x,x,\new s)= & \|\mu_{i}^{t}(x,s)+\delta_{\mu,i}+\eps_{i}^{(\mu)}\|_{x_{i}}^{*}\nonumber \\
= & \|\mu_{i}^{t}(x,s)-\alpha\overline{c}_{i}\mu_{i}^{\ot}(\overline{x},\overline{s})\|_{x_{i}}^{*}\pm\|\eps_{i}^{(\mu)}\|_{x_{i}}^{*}\nonumber \\
= & \|\mu_{i}^{t}(x,s)-\alpha\overline{c}_{i}\mu_{i}^{t}(x,s)\|_{x_{i}}^{*}\pm\alpha\overline{c}_{i}\cdot\|\mu_{i}^{t}(x,s)-\mu_{i}^{\ot}(\overline{x},\overline{s})\|_{x_{i}}^{*}\pm\|\eps_{i}^{(\mu)}\|_{x_{i}}^{*}\nonumber \\
= & (1-\alpha\overline{c}_{i})\gamma_{i}^{t}(x,s)\pm\alpha\overline{c}_{i}\cdot\|\mu_{i}^{t}(x,s)-\mu_{i}^{\ot}(\overline{x},\overline{s})\|_{x_{i}}^{*}\pm\|\eps_{i}^{(\mu)}\|_{x_{i}}^{*}\label{eq:gamma_diff}
\end{align}
where we used that $0\leq\alpha\overline{c}_{i}\leq\alpha\lambda\leq1$
at the end \cref{lem:c_bound}).

In particular, we have that 
\begin{align}
\gamma_{i}^{t}(\new x,x,\new s) & \leq\gamma_{i}^{t}(x,s)+\alpha\overline{c}_{i}\|\mu_{i}^{t}(x,s)-\mu_{i}^{\ot}(\overline{x},\overline{s})\|_{x_{i}}^{*}+\|\eps_{i}^{(\mu)}\|_{x_{i}}^{*}\nonumber \\
 & \leq\gamma_{i}^{t}(x,s)+4\alpha\overline{c}_{i}\overline{\eps}w_{i}+\beta_{i}\label{eq:gamma_upper}
\end{align}
where we used \cref{lem:mu_distance} and \cref{lem:mu_change} at the
end. Hence, we have 
\begin{align}
\left|\gamma_{i}^{t}(\new x,\new x,\new s)-\gamma_{i}^{t}(\new x,x,\new s)\right|= & \left|\|\mu_{i}^{t}(\new x,\new s)\|_{\new{x_{i}}}-\|\mu_{i}^{t}(\new x,\new s)\|_{x_{i}}\right|\nonumber \\
\leq & 2\|\new{x_{i}}-x_{i}\|_{x_{i}}\|\mu_{i}^{t}(\new x,\new s)\|_{x_{i}}\nonumber \\
\leq & 3\|\delta_{x,i}\|_{\ox_{i}}\gamma_{i}^{t}(\new x,x,\new s)=3\alpha_{i}\gamma_{i}^{t}(\new x,x,\new s)\nonumber \\
\leq & 3\alpha_{i}\gamma_{i}^{t}(x,s)+12\alpha\overline{c}_{i}\overline{\eps}w_{i}+3\beta_{i}\label{eq:gamma_diff1}
\end{align}
where we used \cref{lem:phi_properties} on the first inequality, $\new{x_{i}}-x_{i}=\delta_{x,i}$
on the second inequality, the definition of $\alpha_{i}$ on the second
equality, \cref{eq:gamma_upper} and $\alpha_{i}\leq1$ on the last
inequality.

Using \cref{eq:gamma_diff}, we have 
\begin{align}
 & \left|\gamma_{i}^{t}(\new x,x,\new s)-\gamma_{i}^{t}(x,s)+\alpha\overline{c}_{i}\gamma_{i}^{\ot}(\ox,\os)\right|\nonumber \\
\leq & |(1-\alpha\overline{c}_{i})\gamma_{i}^{t}(x,s)-\gamma_{i}^{t}(x,s)+\alpha\overline{c}_{i}\gamma_{i}^{\ot}(\ox,\os)|\nonumber \\
 & +\alpha\overline{c}_{i}\|\mu_{i}^{t}(x,s)-\mu_{i}^{\ot}(\overline{x},\overline{s})\|_{x_{i}}^{*}+\|\eps_{i}^{(\mu)}\|_{x_{i}}^{*}\nonumber \\
\leq & \alpha\overline{c}_{i}|\gamma_{i}^{t}(x,s)-\gamma_{i}^{\ot}(\ox,\os)|+\alpha\overline{c}_{i}\|\mu_{i}^{t}(x,s)-\mu_{i}^{\ot}(\overline{x},\overline{s})\|_{x_{i}}^{*}+\|\eps_{i}^{(\mu)}\|_{x_{i}}^{*}\nonumber \\
\leq & \alpha\overline{c}_{i}(5\overline{\eps}w_{i})+\alpha\overline{c}_{i}(4\overline{\eps}w_{i})+\beta_{i}\nonumber \\
\leq & 9\alpha\overline{c}_{i}\overline{\eps}w_{i}+\beta_{i}\label{eq:gamma_diff2}
\end{align}
where we used \cref{lem:mu_distance}, \cref{lem:mu_change} and $\gamma_{i}^{t}(x,s)\leq w_{i}$
at the second last inequality.

Combining \cref{eq:gamma_diff1} and \cref{eq:gamma_diff2}, we have
\begin{align}
|\eps_{r,i}| & \leq\left|\gamma_{i}^{t}(\new x,x,\new s)-\gamma_{i}^{t}(x,s)+\alpha\overline{c}_{i}\gamma_{i}^{\ot}(\ox,\os)\right|+\left|\gamma_{i}^{t}(\new x,\new x,\new s)-\gamma_{i}^{t}(\new x,x,\new s)\right|\nonumber \\
 & \leq9\alpha\overline{c}_{i}\overline{\eps}w_{i}+\beta_{i}+3\alpha_{i}\gamma_{i}^{t}(x,s)+12\alpha\overline{c}_{i}\overline{\eps}w_{i}+3\beta_{i}\nonumber \\
 & \leq21\alpha\overline{c}_{i}\overline{\eps}w_{i}+3\alpha_{i}\gamma_{i}^{t}(x,s)+4\beta_{i}\label{eq:v_bound}
\end{align}
where we used \cref{lem:mu_distance} at the end.

Now, we bound the $\|\eps_{r}\|_{w^{-1}}$. We first note that 
\begin{align*}
\sum_{i=1}^{m}w_{i}\overline{c}_{i}^{2} & =\frac{\sum_{i=1}^{m}w_{i}\frac{\sinh^{2}(\frac{\lambda}{w_{i}}\gamma_{i}^{\ot}(\ox,\os))}{\gamma_{i}^{\ot}(\ox,\os)^{2}}}{\sum_{j=1}^{m}w_{j}^{-1}\cosh^{2}(\frac{\lambda}{w_{j}}\gamma_{j}^{\ot}(\ox,\os))}\\
 & =\lambda^{2}\frac{\sum_{i=1}^{m}w_{i}^{-1}\frac{w_{i}^{2}}{\lambda^{2}\gamma_{i}^{\ot}(\ox,\os)^{2}}\sinh^{2}(\frac{\lambda}{w_{i}}\gamma_{i}^{\ot}(\ox,\os))}{\sum_{j=1}^{m}w_{j}^{-1}\cosh^{2}(\frac{\lambda}{w_{j}}\gamma_{j}^{\ot}(\ox,\os))}\\
 & \leq\lambda^{2}\frac{\sum_{i=1}^{m}w_{i}^{-1}\cosh^{2}(\frac{\lambda}{w_{i}}\gamma_{i}^{\ot}(\ox,\os))}{\sum_{j=1}^{m}w_{j}^{-1}\cosh^{2}(\frac{\lambda}{w_{j}}\gamma_{j}^{\ot}(\ox,\os))}=\lambda^{2}
\end{align*}
where we used that $\frac{\sinh^{2}(x)}{x^{2}}\leq\cosh^{2}(x)$ for
all $x$ at the second last inequality. Using this and $\sum_{i}w_{i}\alpha_{i}^{2}\leq\frac{9}{8}\alpha^{2}$
(\cref{lem:alpha_i}) into \cref{eq:v_bound}, we have 
\begin{align*}
\sqrt{\sum_{i=1}^{m}w_{i}^{-1}\eps_{r,i}^{2}}\leq & 21\sqrt{2}\alpha\lambda\overline{\eps}+3\max\left(\frac{\gamma_{i}^{t}(x,s)}{w_{i}}\right)\cdot\sqrt{\sum_{i=1}^{m}w_{i}\alpha_{i}^{2}}+4\sqrt{\sum_{i=1}^{m}w_{i}^{-1}\beta_{i}^{2}}\\
\leq & 21\sqrt{2}\alpha\lambda\overline{\eps}+4\max\left(\frac{\gamma_{i}^{t}(x,s)}{w_{i}}\right)\alpha+60\alpha\overline{\eps}
\end{align*}
where we used $\sqrt{\sum_{i=1}^{m}w_{i}\alpha_{i}^{2}}\leq\frac{9}{8}\alpha$
(\cref{lem:alpha_i}) and $\sqrt{\sum_{i=1}^{m}w_{i}^{-1}\beta_{i}^{2}}\leq15\overline{\eps}\alpha$
(\cref{lem:mu_change}) at the end.

\end{proof}

\subsubsection{Bounding the Movement of $\Phi$}

After verifying conditions in \cref{lem:Phi_potential_decrease}, we
are ready to bound the change of $\Phi$ in one step of $(x,s)$. 

\begin{lemma}[Change of $\Phi$ after $(x,s)$ step] \label{lem:change_of_xs}Assume
$\Phi^{t}(x,s)\leq\cosh(\lambda/64)$. We have 
\[
\Phi^{t}(\new x,\new s)\leq\Phi^{t}(x,s)-\frac{\alpha\lambda}{2}\sqrt{\sum_{i=1}^{m}w_{i}^{-1}\cosh^{2}(\frac{\lambda}{w_{i}}\gamma_{i}^{t}(x,s))}+\alpha\lambda\sqrt{\sum_{i}w_{i}^{-1}}.
\]
\end{lemma}

\begin{proof}Let $r_{i}=\gamma_{i}^{t}(x,s)$, $\overline{r}_{i}=\gamma_{i}^{\ot}(\ox,\os)$
and $\delta_{r,i}=\gamma_{i}^{t}(\new x,\new s)-\gamma_{i}^{t}(x,s)$.
Now, we verify the conditions in \cref{lem:Phi_potential_decrease}
for $r_{i}$, $\overline{r}_{i}$ and $\delta_{r}$. \cref{lem:mu_distance}
shows that 
\[
|r_{i}-\overline{r}_{i}|\leq5w_{i}\overline{\eps}\leq\frac{w_{i}}{8\lambda}
\]
where we used the assumption $\overline{\eps}\leq\frac{1}{40\lambda}$.
\[
\sqrt{\sum_{i=1}^{m}w_{i}^{-1}\eps_{r,i}^{2}}\leq90\alpha\lambda\overline{\eps}+4\max\left(\frac{\gamma_{i}^{t}(x,s)}{w_{i}}\right)
\]

Next, \cref{lem:change_of_gamma} shows that 
\[
\delta_{r,i}=-\alpha\cdot c_{i}^{\ot}(\ox,\os)\cdot\gamma_{i}^{\ot}(\ox,\os)+\eps_{r,i}
\]
with 
\[
\sqrt{\sum_{i=1}^{m}w_{i}^{-1}\eps_{r,i}^{2}}\leq90\alpha\lambda\overline{\eps}+4\max\left(\frac{\gamma_{i}^{t}(x,s)}{w_{i}}\right)\alpha\leq\frac{90\alpha\lambda}{1440\lambda}+\frac{4}{64}\alpha\leq\frac{\alpha}{8}
\]
where we used $|\gamma_{i}^{t}(x,s)|\leq\frac{w_{i}}{64}$ (due to
$\Phi^{t}(x,s)\leq\cosh(\lambda/64)$), $\overline{\eps}\leq\frac{1}{1440\lambda}$.
Using the formula of $c_{i}^{\ot}(\ox,\os)$, we have 
\[
\delta_{r,i}=-\frac{\alpha\sinh(\frac{\lambda}{w_{i}}\gamma_{i}^{\ot}(\ox,\os))}{\sqrt{\sum_{j=1}^{m}w_{j}^{-1}\cosh^{2}(\frac{\lambda}{w_{j}}\gamma_{j}^{\ot}(\ox,\os))}}+\eps_{r,i}=\frac{-\alpha\sinh(\frac{\lambda}{w_{i}}\overline{r}_{i})}{\sqrt{\sum_{j=1}^{m}w_{j}^{-1}\cosh^{2}(\frac{\lambda}{w_{j}}\overline{r}_{j})}}+\eps_{r,i}
\]
and this exactly satisfies the conditions in \cref{lem:Phi_potential_decrease}.

Now, \cref{lem:Phi_potential_decrease} shows that 
\begin{align*}
\Phi^{t}(\new x,\new s) & \leq\Phi^{t}(x,s)-\frac{\alpha\lambda}{2}\sqrt{\sum_{i=1}^{m}w_{i}^{-1}\cosh^{2}(\frac{\lambda}{w_{i}}\gamma_{i}^{t}(x,s))}+\alpha\lambda\sqrt{\sum_{i}w_{i}^{-1}}.
\end{align*}
\end{proof}Now, we bound the change of $\Phi$ after changing $t$. 

\begin{lemma}[Change of $\Phi$ after $t$ step] \label{lem:change_of_t}
Assume that $\Phi^{t}(x,s)\leq\cosh(\lambda)$. Let $\new t\leftarrow(1-h)t$
for $h\leq\frac{1}{8\lambda\sqrt{\max_{i}\nu_{i}}}$. We have 
\[
\Phi^{\new t}(x,s)\leq\Phi^{t}(x,s)+16h\lambda\sum_{i=1}^{m}\sqrt{\nu_{i}}\cosh(\lambda\gamma_{i}^{\new t}(x,s)/w_{i}).
\]
\end{lemma}

\begin{proof}By definition of $\gamma$, we have 
\begin{align}
\gamma_{i}^{\new t}(x,s) & =\|\frac{s_{i}}{\new t}+w_{i}\nabla\phi_{i}(x_{i})\|_{x_{i}}^{*}=\|\frac{s}{t(1-h)}+w_{i}\nabla\phi_{i}(x_{i})\|_{x_{i}}^{*}\nonumber \\
 & \leq\frac{1}{1-h}\gamma_{i}^{t}(x,s)+\left(\frac{1}{1-h}-1\right)w_{i}\|\nabla\phi_{i}(x_{i})\|_{x_{i}}^{*}\nonumber \\
 & \leq(1+2h)\gamma_{i}^{t}(x,s)+2h\sqrt{\nu_{i}}w_{i}\label{eq:gamma_new_change}
\end{align}
where the last inequality follows by the definition of self-concordance,
$h\leq\frac{1}{8\lambda\sqrt{\max_{i}\nu_{i}}}$ and $\nu_{i}\geq1$
.

For $\Phi^{\new t},$we have 
\[
\Phi^{\new t}(x,s)=\sum_{i=1}^{m}\cosh(\lambda\gamma_{i}^{\new t}/w_{i})\leq\sum_{i=1}^{m}\cosh(\lambda\gamma_{i}^{t}/w_{i}+2h\lambda(\gamma_{i}^{t}/w_{i}+\sqrt{\nu_{i}})).
\]
Since $\Phi^{t}(x,s)\leq\cosh(\lambda)$, we have $\gamma_{i}^{t}/w_{i}\leq1$.
Since we have self-concordance $\nu_{i}\geq1$ , we have 
\begin{align*}
\Phi^{\new t}(x,s) & \leq\sum_{i=1}^{m}\cosh(\lambda\gamma_{i}^{t}/w_{i}+4h\lambda\sqrt{\nu_{i}})\\
 & \leq\Phi^{t}(x,s)+8h\lambda\sum_{i=1}^{m}\sqrt{\nu_{i}}\cosh(\lambda\gamma_{i}^{t}/w_{i})
\end{align*}
where the last inequality follows by \cref{lem:cosh-x-y}.

Similar to the argument in \cref{eq:gamma_new_change}, we have 
\[
\gamma_{i}^{\new t}(x,s)\geq(1+h)\gamma_{i}^{t}(x,s)-2h\sqrt{\nu_{i}}w_{i}.
\]
Hence, we have $\gamma_{i}^{t}(x,s)\leq\gamma_{i}^{\new t}(x,s)+2h\sqrt{\nu_{i}}w_{i}$.
By \cref{lem:cosh-x-y} again, we have 
\begin{equation}
\cosh(\lambda\gamma_{i}^{t}/w_{i})\leq2\cosh(\lambda\gamma_{i}^{\new t}/w_{i}).\label{eq:cosh_bound_change}
\end{equation}
This gives the result. 

\end{proof}

Combining the bound of $\Phi$ under $(x,s)$ change (\cref{lem:change_of_xs})
and the bound of $\Phi$ under $t$ change (\cref{lem:change_of_t}),
we get the bound on $\Phi$ after 1 step. 

\begin{theorem}\label{thm:potential-decrease}Assume $\Phi^{t}(x,s)\leq\cosh(\lambda/64)$.
Then for any $0\leq h\leq\frac{\alpha}{64\sqrt{\sum_{i=1}^{m}w_{i}\nu_{i}}}$,
we have 
\[
\Phi^{\new t}(\new x,\new s)\leq(1-\frac{\alpha\lambda}{8\sqrt{\sum_{i}w_{i}}})\Phi^{t}(x,s)+\alpha\lambda\sqrt{\sum_{i}w_{i}^{-1}}.
\]
In particular, for any $\cosh(\lambda/128)\leq\Phi^{t}(x,s)\leq\cosh(\lambda/64)$,
we have that 
\[
\Phi^{\new t}(\new x,\new s)\leq\Phi^{t}(x,s).
\]

\end{theorem}

\begin{proof}By \cref{lem:change_of_t} and \cref{lem:change_of_xs},
we have 
\begin{align*}
 & \Phi^{\new t}(\new x,\new s)\\
\leq & \Phi^{\new t}(x,s)-\frac{\alpha\lambda}{2}\sqrt{\sum_{i=1}^{m}w_{i}^{-1}\cosh^{2}(\frac{\lambda}{w_{i}}\gamma_{i}^{\new t}(x,s))}+\alpha\lambda\sqrt{\sum_{i}w_{i}^{-1}}\\
\leq & \Phi^{t}(x,s)+16h\lambda\sum_{i=1}^{m}\sqrt{\nu_{i}}\cosh(\frac{\lambda}{w_{i}}\gamma_{i}^{\new t}(x,s))-\frac{\alpha\lambda}{2}\sqrt{\sum_{i=1}^{m}w_{i}^{-1}\cosh^{2}(\frac{\lambda}{w_{i}}\gamma_{i}^{\new t}(x,s))}+\alpha\lambda\sqrt{\sum_{i}w_{i}^{-1}}.
\end{align*}
By Cauchy Schwarz inequality, we have
\begin{align*}
\frac{\alpha}{4}\sqrt{\sum_{i=1}^{m}w_{i}^{-1}\cosh^{2}(\frac{\lambda}{w_{i}}\gamma_{i}^{\new t}(x,s))}\geq & \frac{\alpha}{4}\frac{\sum_{i=1}^{m}\sqrt{\nu_{i}}\cosh(\frac{\lambda}{w_{i}}\gamma_{i}^{\new t}(x,s))}{\sqrt{\sum_{i=1}^{m}w_{i}\nu_{i}}}\\
\geq & 16h\sum_{i=1}^{m}\sqrt{\nu_{i}}\cosh(\frac{\lambda}{w_{i}}\gamma_{i}^{\new t}(x,s)).
\end{align*}
Hence, we have that
\begin{align*}
\Phi^{\new t}(\new x,\new s)\leq & \Phi^{t}(x,s)-\frac{\alpha\lambda}{4}\sqrt{\sum_{i=1}^{m}w_{i}^{-1}\cosh^{2}(\frac{\lambda}{w_{i}}\gamma_{i}^{\new t}(x,s))}+\alpha\lambda\sqrt{\sum_{i}w_{i}^{-1}}\\
\leq & \Phi^{t}(x,s)-\frac{\alpha\lambda}{8}\sqrt{\sum_{i=1}^{m}w_{i}^{-1}\cosh^{2}(\frac{\lambda}{w_{i}}\gamma_{i}^{t}(x,s))}+\alpha\lambda\sqrt{\sum_{i}w_{i}^{-1}}\\
\leq & \Phi^{t}(x,s)-\frac{\alpha\lambda}{8}\frac{\Phi^{t}(x,s)}{\sqrt{\sum_{i}w_{i}}}+\alpha\lambda\sqrt{\sum_{i}w_{i}^{-1}}
\end{align*}
where we used \cref{eq:cosh_bound_change} at the second inequality.

If $\Phi^{t}(x,s)\geq\cosh(\lambda/128)$, we have
\begin{align*}
\frac{\Phi^{t}(x,s)}{8\sqrt{\sum_{i}w_{i}}} & \geq\frac{\cosh(\lambda/128)}{8\sqrt{\sum_{i}w_{i}}}=\frac{\exp(\lambda/128)}{16\sqrt{\sum_{i}w_{i}}}\\
 & =\frac{\exp(64\log(256m\sum_{i}w_{i})/128)}{16\sqrt{\sum_{i}w_{i}}}\\
 & =\frac{16\sqrt{m\sum w_{i}}}{16\sqrt{\sum_{i}w_{i}}}=\sqrt{m}\geq\sqrt{\sum_{i}w_{i}^{-1}}.
\end{align*}
Hence, we have $\Phi^{\new t}(\new x,\new s)\leq\Phi^{t}(x,s)$.

\end{proof}

\subsection{Initial Point Reduction}

\label{sec:initial-point-reduction}

Since \cref{alg:IPM_framework} requires a point on the central path,
we modify the convex program to make it happen. To satisfy the constraint
$x\in K$, we start the algorithm by solving $\min_{x\in K}c^{\top}x+t\phi(x)$
for some parameter $t$. Since $x$ may not satisfy the constraint
$Ax=b$, we write $\new x=x^{(1)}+x^{(2)}-x^{(3)}$ where $x^{(1)}$
acts as the original variable and $x^{(2)},x^{(3)}\in\R_{\geq0}^{n}$
are the extra variables. We put a large cost vector on $x^{(2)}$
and $x^{(3)}$ to ensure the solution is roughly the same. The proof
shows that if we optimize this new program well enough, we will have
$\new x=x^{(1)}+x^{(2)}-x^{(3)}\in K$ and hence $\new x$ gives a
starting point of the original program. The precise formulation of
the modified linear program is as follows:

\begin{definition}\label{def:IPM_interior_modified}Given a convex
program $\min_{Ax=b,x\in K}c^{\top}x$ with inner radius $r$, outer
radius $R$ and Lipschitz constant $L$. For any $t$, we define the
modified convex program by
\[
\min_{(x^{(1)},x^{(2)},x^{(3)})\in\mathcal{P}_{t}}\left\langle c^{(1)},x^{(1)}\right\rangle +\left\langle c^{(2)},x^{(2)}\right\rangle +\left\langle c^{(3)},x^{(3)}\right\rangle 
\]
where $\mathcal{P}_{t}=\{x^{(1)}\in K,(x^{(2)},x^{(3)})\in\R_{\geq0}^{2n}:A(x^{(1)}+x^{(2)}-x^{(3)})=b\}$,
$c^{(1)}=c$, $c^{(2)}=\frac{t}{3R+x_{\circ}-x_{c}}$, $c^{(3)}=\frac{t}{3R}$,
$x_{c}=\arg\min_{x\in K}c^{\top}x+t\phi(x)$ and $x_{\circ}=\arg\min_{Ax=b}\|x-x_{c}\|_{2}$.
We define the corresponding dual set by
\[
\mathcal{D}_{t}=\{s^{(1)}\in K^{*},(s^{(2)},s^{(3)})\in\R_{\geq0}^{2n}:A^{\top}y+s^{(1)}=c^{(1)},A^{\top}y+s^{(2)}=c^{(2)},-A^{\top}y+s^{(3)}=c^{(3)}\text{ for }y\in\R^{d}\}.
\]
We define the corresponding central path problem 
\begin{equation}
\min_{(x^{(1)},x^{(2)},x^{(3)})\in\mathcal{P}_{t}}f_{t}(x^{(1)},x^{(2)},x^{(3)})\label{eq:IPM_interior_ft}
\end{equation}
where $f_{t}(x^{(1)},x^{(2)},x^{(3)})=\left\langle c^{(1)},x^{(1)}\right\rangle +\left\langle c^{(2)},x^{(2)}\right\rangle +\left\langle c^{(3)},x^{(3)}\right\rangle +t\phi(x^{(1)})-t\sum_{i=1}^{n}\log x_{i}^{(2)}-t\sum_{i=1}^{n}\log x_{i}^{(3)}$.

\end{definition}

The main result about the modified program is the following.

\begin{theorem}\label{thm:IPM_interior}Given a convex program $\min_{Ax=b,x\in K}c^{\top}x$
with inner radius $r$, outer radius $R$ and Lipschitz constant $L$.
For any $0\leq\delta\leq\frac{1}{2}$, the modified linear program
(\cref{def:IPM_interior_modified}) with $t\geq2^{16}(n+\kappa)^{5}\cdot\frac{LR}{\delta}\cdot\frac{R}{r}$
has the following properties:
\begin{itemize}
\item The point $(x_{c},3R+x_{\circ}-x_{c},3R)$ is the minimizer of \cref{eq:IPM_interior_ft}.
The corresponding $s$ variables are $(-t\nabla\phi(x_{c}),\frac{t}{3R+x_{\circ}-x_{c}},\frac{t}{3R})$.
\item Given any primal $(x^{(1)},x^{(2)},x^{(3)})\in\mathcal{P}_{t}$ and
dual $(s^{(1)},s^{(2)},s^{(3)})\in\mathcal{D}_{t}$ that approximately
minimizes $f_{t'}$ at $t'=LR$ as promised by \cref{alg:IPM_framework}:
\begin{align}
\|s_{i}^{(1)}/t'+w_{i}\nabla\phi_{i}(x_{i}^{(1)})\|_{x_{i}^{(1)}}^{*} & \leq\frac{w_{i}}{16}\text{ for all }i\in[m],\label{eq:modified_LP_apx_xs}\\
x_{i}^{(j)}s_{i}^{(j)} & \in[1\pm\frac{1}{16}]t'\text{ for all }i\in[n],j\in\{2,3\}.\nonumber 
\end{align}
Let $\new x=x^{(1)}+x^{(2)}-x^{(3)}$ and $\new s=s^{(1)}$, then
we have that $A\new x=b$, $\new x\in K$, $A^{\top}y+\new s=c$ for
some $y$ and 
\[
\|\new s_{i}/t'+w_{i}\nabla\phi_{i}(\new x_{i})\|_{\new x}^{*}\leq\|s_{i}^{(1)}/t'+w_{i}\nabla\phi_{i}(x_{i}^{(1)})\|_{x_{i}^{(1)}}^{*}+\delta.
\]
\end{itemize}
\end{theorem}

\begin{proof}The proof is separated into \cref{lem:ipm_interior_modified}
and \cref{lemma:IPM_center_original}\end{proof}

First, we prove the first conclusion in \cref{thm:IPM_interior}.

\begin{lemma}\label{lem:ipm_interior_modified}The point $x\defeq(x_{c},3R+x_{\circ}-x_{c},3R)$
is the minimizer of $f_{t}$ over $\mathcal{P}_{t}$ (\cref{eq:IPM_interior_ft}).
The corresponding $s$ variables are $(-t\nabla\phi(x_{c}),\frac{t}{3R+x_{\circ}-x_{c}},\frac{t}{3R})$.

\end{lemma}

\begin{proof}

We will show that $x\in\mathcal{P}_{t}$ and that it minimizes $f_{t}$
over $\R^{3n}$, not just $\mathcal{P}_{t}$.

For the set constraints, we note that $x^{(1)}=x_{c}\in K$ by the
definition of $x_{c}$ and $x^{(3)}=3R\geq0$ by the definition. For
$x^{(2)}$, we note that $z\in K$ with $Az=b$ and hence $\|x_{\circ}-x_{c}\|_{2}\leq\|z-x_{c}\|_{2}\leq2R$
(since $K$ has radius $R$). Hence, $x_{i}^{(2)}\geq R$ for all
$i\in[n]$. Hence, $(x^{(1)},x^{(2)},x^{(3)})\in\mathcal{P}_{t}$.

For the optimality, we note that
\begin{align*}
\nabla_{x^{(1)}}f_{t}(x) & =c^{(1)}+t\nabla\phi(x^{(1)})\\
 & =c+t\nabla\phi(x_{c})=0
\end{align*}
where we used that $x_{c}=\arg\min_{x\in K}c^{\top}x+t\phi(x)$. We
note that
\[
\nabla_{x^{(2)}}f_{t}(x)=c^{(2)}-\frac{t}{x^{(2)}}=0
\]
and similarly $\nabla_{x^{(3)}}f_{t}(x)=0$. Hence $x$ is the minimizer
of $f_{t}$.

\end{proof}

Next, we show that the minimizer of $f_{t'}(x)$ for $t'=LR$ is far
from the boundary of $K$ for $x^{(1)}$ and has small $x^{(2)}$
and $x^{(3)}$. The proof for both involves the same idea: use the
optimality condition of $f_{t'}$. Throughout the rest of the section,
we are given $(x^{(1)},x^{(2)},x^{(3)})\in\mathcal{P}_{t}$ and $(s^{(1)},s^{(2)},s^{(3)})\in\mathcal{D}_{t}$
satisfying \cref{eq:modified_LP_apx_xs}. The following lemma shows
that $(x^{(1)},x^{(2)},x^{(3)})$ is the minimizer of some function
$g$ and we use it to prove the properties of $x$.

\begin{lemma}\label{lem:ipm_interior_modified_KKT}$(x^{(1)},x^{(2)},x^{(3)})$
is the minimizer of the function
\[
g(x^{(1)},x^{(2)},x^{(3)})\defeq\left\langle \widetilde{c},x^{(1)}\right\rangle +\left\langle c^{(2)},x^{(2)}\right\rangle +\left\langle c^{(3)},x^{(3)}\right\rangle +t'\phi(x^{(1)})-\sum_{i=1}^{n}\mu_{i}^{(2)}\log x_{i}^{(2)}-\sum_{i=1}^{n}\mu_{i}^{(3)}\log x_{i}^{(3)}
\]
over the domain $\mathcal{P}_{t}$ for some $\widetilde{c}=c^{(1)}-t'(s^{(1)}/t'+\nabla\phi(x^{(1)}))$,
$\frac{15}{16}t'\leq\mu_{i}^{(2)}\leq\frac{17}{16}t'$, $\frac{15}{16}t'\leq\mu_{i}^{(3)}\leq\frac{17}{16}t'$

\end{lemma}

\begin{proof}

Let $\mu^{(2)}=x^{(2)}s^{(2)}$ and $\mu^{(3)}=x^{(3)}s^{(3)}$. By
the definition of $\mathcal{P}_{t}\times\mathcal{D}_{t}$, we have
that
\begin{align*}
\nabla_{x^{(2)}}g(x) & =c^{(2)}-\frac{\mu^{(2)}}{x^{(2)}}=c^{(2)}-s^{(2)}=A^{\top}y,\\
\nabla_{x^{(3)}}g(x) & =c^{(3)}-\frac{\mu^{(3)}}{x^{(3)}}=c^{(3)}-s^{(3)}=-A^{\top}y
\end{align*}
for some $y$. For the gradient with respect to $x^{(1)}$, we note
that
\[
\nabla_{x^{(1)}}g(x)=\widetilde{c}+t'\nabla\phi(x^{(1)})=c^{(1)}-s^{(1)}=A^{\top}y.
\]
This shows that $x$ satisfies the optimality condition for $g$,
namely $\nabla g(x)=[A,-A,A]^{\top}y$.

\end{proof}

The gradient of $g$ is a bit complicated. We avoid it by considering
the directional derivative at $x$ on the direction ``$z-x$'' for
some $z\in K$ promised by the definition of inner radius. Since our
domain is in $\mathcal{P}_{t}\subset\R^{3n}$, we need to lift $z$
to higher dimension. Now, we define the point
\begin{align*}
z^{(1)}= & z,\\
z^{(2)}= & z^{(3)}=\frac{t'}{t}R.
\end{align*}
By construction, we have that $z\in\mathcal{P}_{t}$. Now, we define
the path $p(\beta)=(1-\beta)\cdot(x^{(1)},x^{(2)},x^{(3)})+\beta\cdot(z^{(1)},z^{(2)},z^{(3)})$.
Since $p(0)$ minimizes $g$, we have that $\frac{d}{d\beta}g(p(\beta))|_{\beta=0}\geq0$.
In particular, we have
\begin{align}
0\leq & \frac{d}{d\beta}g(p(\beta))|_{\beta=0}\nonumber \\
= & (\widetilde{c}+t'\nabla\phi(x^{(1)}))^{\top}(z^{(1)}-x^{(1)})\label{eq:LP_exact_3_term}\\
 & +\sum_{i=1}^{n}(c_{i}^{(2)}-\frac{\mu_{i}^{(2)}}{x_{i}^{(2)}})(z_{i}^{(2)}-x_{i}^{(2)})+\sum_{i=1}^{n}(c_{i}^{(3)}-\frac{\mu_{i}^{(3)}}{x_{i}^{(3)}})(z_{i}^{(3)}-x_{i}^{(3)}).\nonumber 
\end{align}
Now, we bound the terms one by one. To bound the first term in \eqref{eq:LP_exact_3_term},
we need following lemmas relating the self-concordance barrier and
the distance to the boundary.

\begin{lemma}[{\cite[Theorem 4.1.6, Theorem 4.2.6]{Nesterov1998}}]\label{lem:phi_properties_2}Given
a $\nu$-self-concordant barrier $\phi$. For any $x,y\in\dim\phi$
such that $\nabla\phi(x)^{\top}(y-x)\geq0$, we have $\|y-x\|_{x}\leq\nu+2\sqrt{\nu}$.
In particular, for $x^{*}=\argmin_{x}\phi(x)$, we have
\[
\{x:\|x-x^{*}\|_{x^{*}}\leq1\}\subset\dom\phi\subset\{x:\|x-x^{*}\|_{x^{*}}\leq\nu+2\sqrt{\nu}\}.
\]

\end{lemma}

\begin{lemma}\label{lem:barrier_distance}Given a $\nu$-self-concordant
barrier $\phi$ for the interval $[\alpha,\beta]$. For any $x,z\in(\alpha,\beta)$,
we have that
\[
\sqrt{\phi''(x)}\leq\frac{3\nu}{\min(x-\alpha,\beta-x)}
\]
and
\[
\phi'(x)(z-x)+\frac{1}{16}\sqrt{\phi''(x)}|z-x|\leq4\nu^{2}-\frac{1}{16}\max(\frac{z-\alpha}{x-\alpha},\frac{\beta-z}{\beta-x}).
\]

\end{lemma}

\begin{proof}

For the first result, we bound $\phi''$ in two case. If $\phi'(x)\geq0$,
then $\phi'(x)(x-\alpha)\geq0$ and \cref{lem:phi_properties_2} shows
that $|\alpha-x|\sqrt{\phi''(x)}\leq\nu+2\sqrt{\nu}\leq3\nu$. Hence,
we have $\sqrt{\phi''(x)}\leq\frac{3\nu}{x-\alpha}$. If $\phi'(x)\leq0$,
similar argument shows that $\sqrt{\phi''(x)}\leq\frac{3\nu}{\beta-x}$.

For the second result, we split into four cases. First, we note that
both sides on the equation is invariant under affine transformation.
Hence, we can assume $\alpha=0$ and $\beta=1$.

Case 1) $\phi'(x)(z-x)\geq0$.

\cref{lem:phi_properties_2} shows that 
\[
\sqrt{\phi''(x)}|z-x|\leq\nu+2\sqrt{\nu}\leq3\nu.
\]
Together with the fact that $|\phi'(x)|\leq\sqrt{\nu\phi''(x)}$,
we have
\[
\phi'(x)(z-x)+\frac{1}{16}\sqrt{\phi''(x)}|z-x|\leq2\nu^{2}.
\]

Case 2) $x\in[\frac{1}{12\nu},1-\frac{1}{12\nu}]$.

Since $\phi'(x)(z-x)\leq0$ and $z,x\in[0,1]$, we have
\[
\phi'(x)(z-x)+\frac{1}{16}\sqrt{\phi''(x)}|z-x|\leq\frac{1}{16}\sqrt{\phi''(x)}\leq\frac{1}{16}\cdot36\nu^{2}=3\nu^{2}
\]
where we used the first result at the end.

Case 3) $x\leq\frac{1}{12\nu}$

Let $x^{*}=\arg\min_{x\in[0,1]}\phi(x)$. \cref{lem:phi_properties_2}
shows that there is an interval $I=[-\gamma,\gamma]$ such that 
\[
x^{*}+I\subset[0,1]\subset x^{*}+(\nu+2\sqrt{\nu})I\subset x^{*}+3\nu I.
\]
In particular, this implies that $x^{*}\in[\frac{1}{6\nu},1-\frac{1}{6\nu}]$.
Since $x\leq\frac{1}{12\nu}$, we have that $x\leq x^{*}-x$. 

Now we use this to show $\phi'(x)\leq-\frac{1}{8}\sqrt{\phi''(x)}$.
Note that
\[
\phi'(x)=\phi'(x^{*})-\int_{x}^{x^{*}}\phi''(t)dt=-\int_{x}^{x^{*}}\phi''(t)dt.
\]
\cref{lem:phi_properties} shows that $[x-\frac{1}{\sqrt{\phi''(x)}},x+\frac{1}{\sqrt{\phi''(x)}}]$
lies in $\dom\phi$. In particular, this implies that
\begin{equation}
\frac{1}{\sqrt{\phi''(x)}}\leq x\leq x^{*}-x\label{eq:phi_properties_3_lower}
\end{equation}
and hence $x^{*}\geq x+\frac{1}{\sqrt{\phi''(x)}}$. Hence, we have
\begin{align*}
\phi'(x) & \leq-\int_{x}^{x+(\phi''(x))^{-1/2}/2}\phi''(t)dt\\
 & \leq-\frac{1}{4}\phi''(x)\cdot\frac{(\phi''(x))^{-1/2}}{2}\\
 & =-\frac{1}{8}\sqrt{\phi''(x)}.
\end{align*}
where we used $\phi''(t)\geq\frac{1}{4}\phi''(x)$ for all $|t-x|\leq\frac{1}{2\sqrt{\phi''(x)}}$
(\cref{lem:phi_properties}).

Since $\phi'(x)(z-x)\leq0$ and $\phi'(x)\leq0$, we have $z\geq x$
and 
\begin{align*}
\phi'(x)(z-x)+\frac{1}{16}\sqrt{\phi''(x)}(z-x) & \leq-\frac{1}{16}\sqrt{\phi''(x)}(z-x)\\
 & \leq-\frac{z-x}{16x}=\frac{1}{16}-\frac{z}{16x}
\end{align*}
where we used $x\geq\frac{1}{\sqrt{\phi''(x)}}$ at the end (\cref{eq:phi_properties_3_lower}).

Case 4) $x\geq1-\frac{1}{12\nu}$

By the same argument as case 3, we have $\phi'(x)(z-x)+\frac{1}{16}\sqrt{\phi''(x)}(z-x)\leq\frac{1}{16}-\frac{1-z}{16(1-x)}.$

Combining all the cases, we have the result.

\end{proof}

Now, we can bound the first term in \eqref{eq:LP_exact_3_term}.

\begin{lemma}\label{lem:LP_exact_3_term_1}We have that $(\widetilde{c}+t'\nabla\phi(x^{(1)}))^{\top}(z^{(1)}-x^{(1)})\leq(6\kappa^{2}-\frac{r}{16\eta})LR$
where $\eta$ is the minimum distance between $x^{(1)}$ to the boundary
of some $K_{i}$, i.e. $\eta=\min_{i}\min_{q\in\partial K_{i}}\|q-x_{i}^{(1)}\|_{2}.$

\end{lemma}

\begin{proof}

Recall that $\widetilde{c}=c^{(1)}-t'\alpha$ with $\alpha=s^{(1)}/t'+\nabla\phi(x^{(1)})$.
By the assumption on $(x,s)$, we have that
\[
\|\alpha_{i}\|_{x_{i}^{(1)}}^{*}\leq\frac{w_{i}}{16}\text{ for all }i\in[m].
\]
Hence, we have
\begin{align}
 & (\widetilde{c}+t'\nabla\phi(x^{(1)}))^{\top}(z^{(1)}-x^{(1)})\nonumber \\
= & c^{(1)\top}(z^{(1)}-x^{(1)})-t'\sum_{i=1}^{m}\alpha_{i}^{\top}(z_{i}^{(1)}-x_{i}^{(1)})+t'\sum_{i=1}^{m}w_{i}\nabla\phi_{i}(x^{(1)})^{\top}(z_{i}^{(1)}-x_{i}^{(1)}).\nonumber \\
\leq & 2LR+\frac{t'}{16}\sum_{i=1}^{m}w_{i}\|z_{i}^{(1)}-x_{i}^{(1)}\|_{x_{i}^{(1)}}+t'\sum_{i=1}^{m}w_{i}\nabla\phi_{i}(x^{(1)})^{\top}(z_{i}^{(1)}-x_{i}^{(1)})\label{eq:LP_exact_3_term_1_1}
\end{align}
where we used $\|c^{(1)}\|_{2}\leq L$ and $\|z^{(1)}-x^{(1)}\|_{2}\leq2R$
(the radius of $K$ is bounded by $R$).

To bound the last two terms, we define $\widetilde{\phi}$ be the
$\phi_{i}$ restricted on the line between $z_{i}^{(1)}$ and $x_{i}^{(1)}$.
Note that $\widetilde{\phi}$ is a $\nu_{i}$-self-concordant barrier
function on some interval $[\alpha,\beta]$. Let $z$ and $x$ be
the scalar such that $\widetilde{\phi}(z)$ and $\widetilde{\phi}(x)$
corresponding to $\phi_{i}(z_{i}^{(1)})$ and $\phi_{i}(x_{i}^{(1)})$.
Then, we have that
\begin{align*}
u_{i}\defeq\nabla\phi_{i}(x^{(1)})^{\top}(z_{i}^{(1)}-x_{i}^{(1)})+\frac{1}{16}\|z_{i}^{(1)}-x_{i}^{(1)}\|_{x_{i}^{(1)}} & =\widetilde{\phi}'(x)(z-x)+\frac{1}{16}\sqrt{\widetilde{\phi}''(x)}|z-x|\\
 & \leq4\nu_{i}^{2}-\frac{1}{16}\max(\frac{z-\alpha}{x-\alpha},\frac{\beta-z}{\beta-x}).
\end{align*}
Let $\eta_{i}=\min_{q\in\partial K_{i}}\|q-x_{i}^{(1)}\|_{2}$ and
$q_{i}$ be a minimizing $q$. Suppose $\alpha\leq x\leq z$ (the
other case is similar). Since $K_{i}$ is convex, there is a hyperplane
separating $q_{i}$ and $K_{i}$. Let $h$ be the $\ell_{2}$ distance
to the hyperplane. Note that $h$ is linear on $K$ and that $h(\alpha)\geq0$,
$h(z)\geq r$ (because $B(z,r)\subset K_{i}$). Hence,
\[
\eta_{i}=h(x)=\frac{x-\alpha}{z-\alpha}h(z)+\frac{z-x}{z-\alpha}h(\alpha)\geq\frac{x-\alpha}{z-\alpha}r.
\]
Hence, we have
\[
\frac{z-\alpha}{x-\alpha}\geq\frac{r}{\eta_{i}}.
\]
This shows $u_{i}\leq4\nu_{i}^{2}-\frac{r}{16\eta_{i}}$. In particular,
we know that $u_{i}\leq4\nu_{i}^{2}-\frac{r}{16\eta}$ for one of
the $i$. For other terms, we can simply by it by $4\nu_{i}^{2}$.
Putting these into \cref{eq:LP_exact_3_term_1_1}, we have
\begin{align*}
(\widetilde{c}+t'\nabla\phi(x^{(1)}))^{\top}(z^{(1)}-x^{(1)}) & \leq2LR+4t'\sum_{i=1}^{m}w_{i}\nu_{i}^{2}-\frac{rt'}{16\eta}\\
 & \leq2LR+4t'\kappa^{2}-\frac{rt'}{16\eta}.
\end{align*}
Using $t'=LR$ and $\kappa\geq1$, we have the result.

\end{proof}

For the second term and the third term in \eqref{eq:LP_exact_3_term},
we have the following

\begin{lemma}\label{lem:LP_exact_3_term_2}We have that 
\[
\sum_{i=1}^{n}(c_{i}^{(j)}-\frac{\mu_{i}^{(j)}}{x_{i}^{(j)}})(z_{i}^{(j)}-x_{i}^{(j)})\leq3LRn-\frac{t}{5R}\sum_{i=1}^{n}x_{i}^{(j)}
\]
 for both $j=2$ and $3$.\end{lemma}

\begin{proof}

We only prove the case $j=2$. The proof for $j=3$ is similar. As
proved in \cref{lem:ipm_interior_modified}, $\|x_{\circ}-x_{c}\|_{2}\leq2R$.
Hence $c_{i}^{(2)}=\frac{t}{3R+x_{\circ,i}-x_{c,i}}\in[\frac{t}{5R},\frac{t}{R}]$.
Hence, we have
\begin{align*}
(c_{i}^{(2)}-\frac{\mu_{i}^{(2)}}{x_{i}^{(2)}})(z_{i}^{(2)}-x_{i}^{(2)})= & c_{i}^{(2)}z_{i}^{(2)}-\frac{\mu_{i}^{(2)}}{x_{i}^{(2)}}z_{i}^{(2)}-c_{i}^{(2)}x_{i}^{(2)}+\mu_{i}^{(2)}\\
\leq & \frac{t}{R}\cdot\frac{t'}{t}R-\frac{t}{5R}\cdot x_{i}^{(2)}+2t'\\
\leq & 3t'-\frac{t}{5R}\cdot x_{i}^{(2)}.
\end{align*}
Summing over all $i$ and using $t'=LR$ gives the result.

\end{proof}

Combining \eqref{eq:LP_exact_3_term}, \cref{lem:LP_exact_3_term_1}
and \cref{lem:LP_exact_3_term_2}, we have
\begin{align*}
0 & \leq(6\kappa^{2}-\frac{r}{16\eta})LR+6LRn-\frac{t}{5R}\sum_{i=1}^{n}(x_{i}^{(2)}+x_{i}^{(3)}).
\end{align*}
Hence, this shows that $(x^{(1)},x^{(2)},x^{(3)})$ satisfies \cref{eq:modified_LP_apx_xs}
implies that it is far from $\partial K$ (i.e. $\eta$ is large)
and $x^{(2)},x^{(3)}$ are small:
\[
\frac{r}{16\eta}+\frac{t}{5LR^{2}}\sum_{i=1}^{n}(x_{i}^{(2)}+x_{i}^{(3)})\leq6n+6\kappa^{2}.
\]
In particular, this shows the following:

\begin{lemma}\label{lem:LP_exact_distance_OPT}We have that $\eta\geq\frac{r}{96(n+\kappa^{2})}$
and $\sum_{i=1}^{n}(x_{i}^{(2)}+x_{i}^{(3)})\leq30(n+\kappa^{2})\cdot\frac{LR}{t}\cdot R$.\end{lemma}

Now, we are ready to prove the second conclusion of Theorem \cref{thm:IPM_interior}.

\begin{lemma}\label{lemma:IPM_center_original}Let $\new x=x^{(1)}+x^{(2)}-x^{(3)}$
and $\new s=s^{(1)}$, then we have that $A\new x=b$, $\new x\in K$,
$A^{\top}y+\new s=c$ for some $y$ and 
\[
\|\new s_{i}/t'+w_{i}\nabla\phi_{i}(\new x_{i})\|_{\new x_{i}}^{*}\leq\|s_{i}^{(1)}/t'+w_{i}\nabla\phi_{i}(x_{i}^{(1)})\|_{x_{i}^{(1)}}^{*}+\delta.
\]

\end{lemma}

\begin{proof}

Note that $A\new x=b$ by definition. \cref{lem:LP_exact_distance_OPT}
shows that $x^{(1)}$ is $\eta\geq\frac{r}{96(n+\kappa^{2})}$ far
from $\partial K$. Since $\new x=x^{(1)}+x^{(2)}-x^{(3)}$, we have
\[
\|\new x-x^{(1)}\|_{2}\leq\|x^{(2)}\|_{1}+\|x^{(3)}\|_{1}\leq30(n+\kappa^{2})\cdot\frac{LR}{t}\cdot R
\]
where we used $x^{(2)},x^{(3)}\geq0$ and \cref{lem:LP_exact_distance_OPT}.
Hence, we have that
\[
\frac{\|\new x-x^{(1)}\|_{2}}{\eta}\leq2^{12}(n+\kappa)^{4}\cdot\frac{R}{r}\cdot\frac{LR}{t}<1
\]
by the choice of $t$. In particular, this shows that $\new x\in K$. 

Next, $A^{\top}y+\new s=A^{\top}y+s^{(1)}=c$ by construction.

Finally, to bound $s/t+w\nabla\phi$, we note that \cref{lem:barrier_distance}
shows that $\nabla^{2}\phi_{i}(x^{(1)})\preceq\frac{9\nu^{2}}{\eta^{2}}$.
This gives
\begin{align*}
\|\new x_{i}-x_{i}^{(1)}\|_{x_{i}^{(1)}} & \leq\frac{3\nu}{\eta}\cdot\|\new x_{i}-x_{i}^{(1)}\|_{2}\\
 & \leq3\nu\cdot(\frac{r}{96(n+\kappa^{2})})^{-1}\cdot30(n+\kappa^{2})\cdot\frac{LR}{t}\cdot R\\
 & \leq2^{14}(n+\kappa)^{5}\cdot\frac{LR}{t}\cdot\frac{R}{r}\leq\frac{\delta}{4}
\end{align*}
where we used our choice of $\delta$. Using this and \cref{lem:phi_properties}
gives $\|v\|_{\new x_{i}}\leq(1+\frac{\delta}{2})\|v\|_{x_{i}^{(1)}}$
and $\|\nabla\phi_{i}(\new x_{i})-\nabla\phi_{i}(x_{i}^{(1)})\|_{x_{i}^{(1)}}^{*}\leq\frac{\delta}{2}$.
Hence
\begin{align*}
 & \|\new s_{i}/t'+w_{i}\nabla\phi_{i}(\new x_{i})\|_{\new x_{i}}^{*}\\
\leq & (1+\frac{\delta}{2})\|\new s_{i}/t'+w_{i}\nabla\phi_{i}(\new x_{i})\|_{x_{i}^{(1)}}^{*}\\
= & (1+\frac{\delta}{2})\|s_{i}^{(1)}/t'+w_{i}\nabla\phi_{i}(x_{i}^{(1)})+(\nabla\phi_{i}(\new x_{i})-\nabla\phi_{i}(x_{i}^{(1)}))\|_{x_{i}^{(1)}}^{*}\\
\leq & (1+\frac{\delta}{2})\|s_{i}^{(1)}/t'+w_{i}\nabla\phi_{i}(x_{i}^{(1)})\|_{x_{i}^{(1)}}^{*}+(1+\frac{\delta}{2})\|\nabla\phi_{i}(\new x_{i})-\nabla\phi_{i}(x_{i}^{(1)})\|_{x_{i}^{(1)}}^{*}\\
\leq & \|s_{i}^{(1)}/t'+w_{i}\nabla\phi_{i}(x_{i}^{(1)})\|_{x_{i}^{(1)}}^{*}+\delta.
\end{align*}

\end{proof}

\subsection{Main Result\label{subsec:IPM_main}}

To prove \cref{thm:IPM_framework}, we first need the following lemma
showing that the iterate $x$ is a good solution when $t$ is small
enough. 

\begin{lemma}[{\cite[Lemma D.3]{DBLP:conf/colt/LeeSZ19}}]\label{lem:erm-error}Let
$\phi_{i}(x)$ be a $\nu_{i}$-self-concordant barrier for $K_{i}$.
Suppose we have $\|\frac{s_{i}}{t}+\nabla\phi_{i}(x_{i})\|_{x_{i}}^{*}\leq1$
for $i\in[m]$, $A^{\top}y+s=c$ and $Ax=b$. Then, we have 
\[
c^{\top}x\leq\min_{Ax=b,x\in\prod_{i=1}^{m}K_{i}}c^{\top}x+4t\sum\nu_{i}.
\]
\end{lemma}

\thmIPMframework*

\begin{proof}

\cref{thm:IPM_interior} gives an explicit point on the central path.
Hence, we have $\Phi^{t}(x,s)=m\leq\cosh(\lambda/128)$ initially.
\cref{thm:potential-decrease} shows that $\Phi^{t}(x,s)\leq\cosh(\lambda/128)$
throughout the first call of $\textsc{Centering}$. After we obtain
the approximate central path point $((x^{(1)},x^{(2)},x^{(3)}),(s^{(1)},s^{(2)},s^{(3)}))$
at $t=LR$ for the modified convex program, \cref{thm:IPM_interior}
shows that $(x^{(1)}+x^{(2)}-x^{(3)},s^{(1)})$ is an approximate
central path point at $t=LR$ for the original convex program. Furthermore,
$\gamma_{i}^{t}$ is increased by $\delta=\frac{1}{128}$ for all
$i$. Hence, $\Phi^{LR}$ is increased by at most $\exp(\frac{\lambda}{128})$
factor. Hence, we have $\Phi^{LR}(x,s)\leq\cosh(\lambda/64)$. Now,
\cref{thm:potential-decrease} shows that $\Phi^{t}(x,s)\leq\cosh(\lambda/64)$
throughout the second call of $\textsc{Centering}$.

Now, we verify the output. Note that $A\delta_{x}=0$ and $\delta_{s}\in\text{Im}A^{\top}$.
Hence, throughout the algorithm, we have $Ax=b$ and $c-s\in\text{Im}A^{\top}$.
Finally, for the optimality, we note that $w_{i}\phi_{i}$ are $w_{i}\nu_{i}$
self-concordant. \cref{lem:erm-error} shows that 
\[
c^{\top}x'\leq\min_{Ax=b,x\in\prod_{i=1}^{m}K_{i}}c^{\top}x+4t_{\mathrm{end}}\sum_{i=1}^{m}w_{i}\nu_{i}.
\]
Since the algorithm terminates at $t_{\mathrm{end}}=\eps/(4\sum_{i=1}^{m}w_{i}\nu_{i})$,
we have the error bounded.

\end{proof}

\subsection{Using the Universal Barrier}

\label{subsec:IPM-universal-barrier}

In this subsection, we discuss the case if the barriers $\phi_{i}:K_{i}\rightarrow\mathbb{R}$
is not given. In this case, we can use the universal barrier, which
has self concordance $n_{i}$.

\begin{theorem}[\cite{nesterov1994interior, lee2018universal}]For
any convex set $K$, the universal barrier function $\phi(x)=\log\mathrm{Vol}(K^{\circ}(x))$
is a $n$ self-concordant barrier where $K^{\circ}(x)=\{y\in\R^{n}:y^{\top}(z-x)\leq1,\forall z\in K\}$.

\end{theorem}

The gradient and Hessian of the universal barrier function $\phi$
can be computed using the center of gravity and the covariance of
$K^{\circ}(x)$.

\begin{lemma}[{\cite[Lemma1]{lee2018universal}}] For any convex set
$K\subset\mathbb{R}^{n}$ and any $x\in\mathrm{int}(K)$, we have
\begin{align*}
\nabla\phi(x) & =-(n+1)\mu(K^{\circ}(x)),\\
\nabla^{2}\phi(x) & =(n+1)(n+2)\mathrm{Cov}(K^{\circ}(x))+(n+1)\mu(K^{\circ}(x))\mu(K^{\circ}(x))^{\top}.
\end{align*}
where $\mu(K^{\circ}(x))$ is the center of gravity of $K^{\circ}(x)$
and $\mathrm{Cov}(K^{\circ}(x))$ is the covariance matrix of $K^{\circ}(x)$.

\end{lemma}

Computing center of gravity and covariance takes polynomial time.
See for example \cite{lee2018kannan} for a survey. 

\begin{theorem}[\cite{dyer1991random,srivastava2013covariance}]Given
a membership oracle for a convex set $K\subset\R^{n}$ with cost $\mathcal{T}$.
Assuming $B(0,r)\subset K\subset B(0,R)$, we can compute $x$ and
$A$ such that 
\[
\|x-\mu(K)\|_{\mathrm{Cov}(K)^{-1}}\leq\eps\qquad\text{and}\qquad(1-\eps)A\preceq\mathrm{Cov}(K)\preceq(1+\eps)A
\]
in time $O(n^{O(1)}\mathcal{T}\log(R/r)/\eps^{2})$.

\end{theorem}

Next, note that the membership oracle of $K^{\circ}(x)$ involves
optimizing one linear function over the convex set $k$ and it can
be done using membership oracle of $K$ and the ellipsoid method.
Therefore, for any $x$, we can compute an approximate gradient $g$
and the Hessian $H$ of the universal barrier function such that 
\[
\|g-\nabla\phi(x)\|_{\nabla^{2}\phi(x)^{-1}}\leq\eps\qquad\text{and}\qquad(1-\eps)H\preceq\nabla^{2}\phi(x)\leq(1+\eps)H
\]
in time $O(n^{O(1)}\mathcal{T}\log(R/r)/\eps^{2})$ where $\mathcal{T}$
is the cost of the membership oracle of $K$. 

Finally, we note that as long as $\eps\leq\frac{1}{\log^{c}m}$ for
some large enough $c$, our robust interior point method works with
those approximate gradient and the Hessian with the same guarantee.
Since the proof is essentially same, we skip the analysis here. We
note there are known explicit barrier functions with good self-concordance
for most commonly used convex sets and in this case, we do not need
heavy machinery like the above to compute them. 

\subsection{Hyperbolic Function Lemmas}

\begin{lemma} \label{lem:sinh-x-y}For any $x,y\in\R$ with $|y|\leq\frac{1}{8}$,
we have 
\[
|\sinh(x+y)-\sinh(x)|\leq\frac{1}{7}|\sinh(x)|+\frac{1}{7}.
\]
Similarly, we have $|\cosh(x+y)-\cosh(x)|\leq\frac{1}{7}\cosh(x).$

\end{lemma}

\begin{proof} Note that $\sinh(x+y)=\sinh(x)\cosh(y)+\cosh(x)\sinh(y)$.
Using that $||\cosh(x)|-|\sinh(x)||\leq1$, we have 
\begin{align*}
|\sinh(x+y)-\sinh(x)| & \leq|\sinh(x)||\cosh(y)-1|+\cosh(x)\sinh(y)\\
 & \leq|\sinh(x)|(|\cosh(y)-1|+|\sinh(y)|)+|\sinh(y)|
\end{align*}
The first result follows from this and $|\cosh(y)-1|+|\sinh(y)|\leq\frac{1}{7}$
for $|y|\leq\frac{1}{8}$.

For the second result, note that $\cosh(x+y)=\cosh(x)\cosh(y)+\sinh(x)\sinh(y)$.
Hence, 
\begin{align*}
\left|\cosh(x+y)-\cosh(x)\right| & \leq(\cosh(y)-1)\cosh(x)+\sinh(x)\sinh(y)\\
 & \leq(\cosh(y)-1+|\sinh(y)|)\cosh(x)\\
 & \leq\frac{1}{7}\cosh(x).
\end{align*}
\end{proof}

\begin{lemma} \label{lem:cosh-x-y}For any $x\geq0$ and $0\leq y\leq1$,
we have 
\[
\cosh(x+y)\leq(1+2y)\cosh(x)
\]
\end{lemma}

\begin{proof} Note that $\cosh(x+y)=\cosh(x)\cosh(y)+\sinh(x)\sinh(y)$
and $\exp(x)=\sinh(x)+\cosh(x)$, then we have 
\begin{align*}
\cosh(x+y) & =\cosh(x)\left[\exp(y)-\sinh(y)\right]+\sinh(x)\sinh(y)\\
 & \leq\cosh(x)\left[\exp(y)-\sinh(y)\right]+\cosh(x)\sinh(y)\\
 & =\cosh(x)\exp(y)\\
 & \leq\cosh(x)+2y\cosh(x),
\end{align*}
where we use $\exp(y)\leq1+2y$ for $0\leq y\leq1$. \end{proof}

\section{Treewidth vs. Problem Size in Netlib Instances}

\begin{figure}[H] \label{fig:practice}
	\centering
\begin{tikzpicture}[scale=0.85] 
\begin{axis}[
    title={Treewidth vs Problem size},
    xlabel={Number of variables $n$},
    ylabel={Upper bound on treewidth $\tau$},
    xmin=0, xmax=1000000, xmode=log,
    ymin=0, ymax=10000, ymode=log,
    legend pos=north west,
    grid style=dashed,
]

\addplot[only marks, scatter,
    color=blue,
    mark=*,
    mark options={scale=0.5}
    ]
    coordinates {
(1876,157)
(12061,169)
(138,21)
(51,9)
(615,137)
(758,162)
(758,162)
(472,64)
(295,48)
(114,39)
(1586,123)
(4486,340)
(334,60)
(303,71)
(482,76)
(7248,125)
(77137,967)
(6411,115)
(73948,739)
(3371,174)
(3562,30)
(5831,283)
(6184,313)
(757,129)
(2604,336)
(12230,1254)
(472,80)
(816,135)
(1028,136)
(1064,62)
(1049,23)
(1677,627)
(10524,25)
(13525,3000)
(1706,103)
(1160,30)
(5598,192)
(5598,192)
(645,80)
(946,107)
(301,60)
(316,151)
(68,23)
(3602,36)
(21349,78)
(42659,122)
(154699,433)
(366,46)
(1966,126)
(9408,1035)
(1620,85)
(25067,79)
(54797,77)
(104374,78)
(243246,78)
(7716,138)
(29351,518)
(49932,862)
(108175,1972)
(1506,162)
(4860,524)
(1123,121)
(6680,837)
(2267,281)
(2928,137)
(2446,252)
(8856,2289)
(22275,4552)
(1632,620)
(204,32)
(163,20)
(317,29)
(78,15)
(78,15)
(671,44)
(185,28)
(600,61)
(1200,66)
(1800,71)
(466,36)
(1275,81)
(760,35)
(1350,59)
(2750,65)
(660,44)
(2500,113)
(3340,98)
(253,37)
(162,35)
(1777,49)
(2166,28)
(1506,33)
(4363,42)
(2467,50)
(5533,28)
(2869,35)
(2735,31)
(614,101)
(1274,69)
(1383,55)
(1274,60)
(165,27)
(3045,92)
(23541,142)
(8806,182)
(628,74)
(346,56)
(2595,205)
(8418,155)
    };
    \legend{Netlib LP}
    
\addplot [dashed,
    domain=10:1000000, 
    samples=100, 
    color=red,
]
{x^(0.5)};
\addlegendentry{$\tau = n^{1/2}$}

\addplot [dashed,
    domain=10:1000000, 
    samples=100, 
    color=blue,
]
{x^(3/4)};
\addlegendentry{$\tau = n^{3/4}$}
\end{axis}
\end{tikzpicture}
\qquad
\begin{tikzpicture}[scale=0.85]
\begin{axis}[
    title={Time vs Problem size},
    xlabel={Number of non-zeros $s = \mathrm{nnz}(A)$},
    ylabel={Time complexity $\mathcal{T} = n \tau^2$},
    xmin=50, xmax=2000000, xmode=log,
    ymin=1000, ymax=1000000000000, ymode=log,
    legend pos=north west,
    grid style=dashed,
]
\addplot[only marks, scatter,
    color=blue,
    mark=*,
    mark options={scale=0.5}
    ]
    coordinates {
(23264,344474221)
(424,60858)
(102,4131)
(2862,11542935)
(4740,19892952)
(4756,19892952)
(2494,1933312)
(3408,679680)
(522,173394)
(5532,23994594)
(14996,518581600)
(1448,1202400)
(2202,1527423)
(1896,2784032)
(18168,113250000)
(260785,72129960193)
(15977,84785475)
(246614,40384555708)
(21234,102060396)
(10708,3205800)
(33081,466998959)
(37704,605840296)
(4201,12597237)
(25432,293981184)
(35632,19231870680)
(2768,3020800)
(2537,14871600)
(6401,19013888)
(2760,4090016)
(13427,554921)
(9868,659277333)
(129042,6577500)
(50284,121725000000)
(6937,18098954)
(2445,1044000)
(31070,206364672)
(31070,206364672)
(5620,4128000)
(8252,10830754)
(2612,1083600)
(2443,7205116)
(313,35972)
(8404,4668192)
(49058,129887316)
(97246,634936556)
(358171,29004360811)
(1136,774456)
(10137,31212216)
(144848,10078084800)
(3168,11704500)
(144812,156443147)
(317097,324891413)
(604488,635011416)
(1408073,1479908664)
(16571,146943504)
(63220,7875577724)
(107605,37101673008)
(232647,420669209200)
(6148,39523464)
(44375,1334439360)
(5264,16441843)
(74949,4679800920)
(14977,179004587)
(9265,54955632)
(13331,155330784)
(38304,46401197976)
(94950,461553681600)
(7296,627340800)
(687,208896)
(340,65200)
(665,266597)
(160,17550)
(148,17550)
(1725,1299056)
(465,145040)
(2732,2232600)
(5469,5227200)
(8206,9073800)
(1534,603936)
(3288,8365275)
(2388,931000)
(4316,4699350)
(8584,11618750)
(1872,1277760)
(7334,31922500)
(9734,32077360)
(1179,346357)
(777,198450)
(3558,4266577)
(6380,1698144)
(4400,1640034)
(12882,7696332)
(7194,6167500)
(16276,4337872)
(8284,3514525)
(8001,2628335)
(4003,6263414)
(3230,6065514)
(3338,4183575)
(3878,4586400)
(501,120285)
(9357,25772880)
(72721,474680724)
(27836,291689944)
(4561,3438928)
(1051,1085056)
(70216,109054875)
(37487,202242450)
    };
    \legend{Netlib LP}
    
\addplot [dashed,
    domain=10:10000000, 
    samples=100, 
    color=red,
]
{x^(3/2)};
\addlegendentry{$\mathcal{T} = s^{3/2}$}

\addplot [dashed,
    domain=10:10000000, 
    samples=100, 
    color=blue,
]
{x^(2)};
\addlegendentry{$\mathcal{T} = s^{2}$}
\end{axis}
\end{tikzpicture}
	\caption{The left plot shows some upper bound of treewidth vs $d$ for all 109 feasible
		linear program instances in Netlib repository. We compute a upper
		bound of treewidth using \cite{karypis1998fast}. This shows that treewidth is between $n^{1/2}$ and $n^{3/4}$ for many linear programs in this data set. The right plot shows that the runtime $n \tau^2$ is sub-quadratic in the input size $\mathrm{nnz}(A)$ for many linear programs in this data set.}
\end{figure}
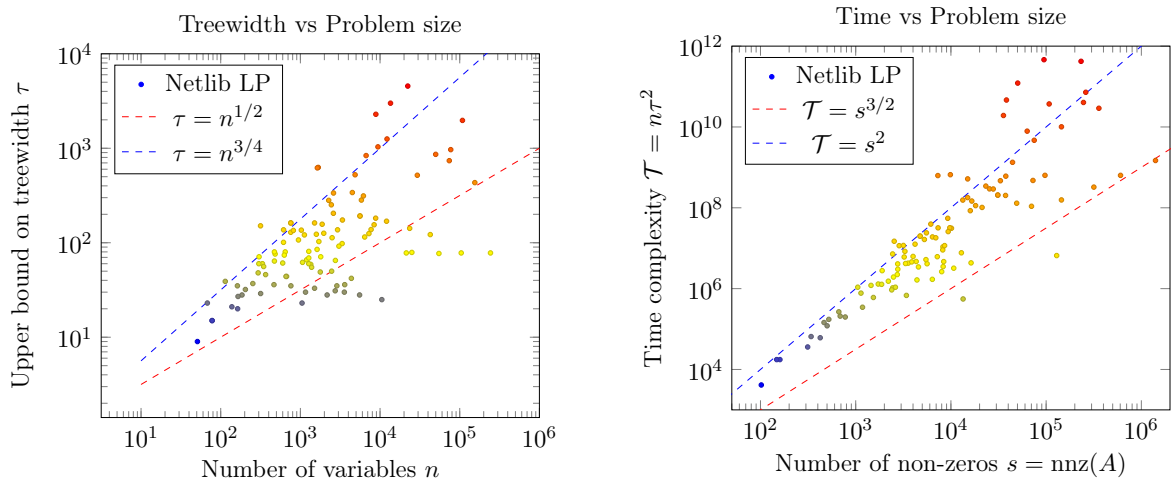

\end{document}